\documentclass[openany]{now} 
\pdfoutput=1 

\usepackage{amsfonts}
\usepackage{amsmath}
\usepackage{latexsym}
\usepackage{mathtools}
\usepackage{bbm}
\usepackage{relsize}
\usepackage{graphicx}
\usepackage{epstopdf}
\usepackage{tikz}

\usetikzlibrary{calc,shapes}

\makeatletter
 \def\@textbottom{\vskip \z@ \@plus 12pt}
 \let\@texttop\relax
\makeatother

\definecolor{White}{rgb}{1,1,1}
\definecolor{Black}{rgb}{0,0,0}
\definecolor{LightGray}{rgb}{.81,.81,.81}
\colorlet{ChannelColor}{LightGray}
\colorlet{ChannelTextColor}{Black}
\colorlet{ReadoutColor}{White}

\newcommand{\reg}[1]{\textsf{#1}}
\newcommand{\class}[1]{\textup{#1}}

\newcommand{\setft}[1]{\mathrm{#1}}
\newcommand{\Density}{\setft{D}}
\newcommand{\Pos}{\setft{Pos}}
\newcommand{\Proj}{\setft{Proj}}
\newcommand{\Channel}{\setft{C}}
\newcommand{\Unitary}{\setft{U}}
\newcommand{\Herm}{\setft{Herm}}
\newcommand{\Lin}{\setft{L}}
\newcommand{\Trans}{\setft{T}}

\newcommand{\I}{\mathbbm{1}}

\newcommand{\complex}{\mathbb{C}}
\newcommand{\real}{\mathbb{R}}
\renewcommand{\natural}{\mathbb{N}}

\newcommand{\rational}{\mathbb{Q}}

\newcommand{\tinyspace}{\mspace{1mu}}
\newcommand{\microspace}{\mspace{0.5mu}}

\newcommand{\op}[1]{\operatorname{#1}}
\newcommand{\tr}{\operatorname{Tr}}

\renewcommand{\int}{\operatorname{int}}

\newcommand{\fid}{\operatorname{F}}

\renewcommand{\t}{{\scriptscriptstyle\mathsf{T}}}

\newcommand{\abs}[1]{\lvert #1 \rvert}
\newcommand{\bigabs}[1]{\bigl\lvert #1 \bigr\rvert}

\newcommand{\Biggabs}[1]{\Biggl\lvert #1 \Biggr\rvert}

\newcommand{\ip}[2]{\langle #1 , #2\rangle}
\newcommand{\bigip}[2]{\bigl\langle #1, #2 \bigr\rangle}

\newcommand{\biggip}[2]{\biggl\langle #1, #2 \biggr\rangle}

\newcommand{\floor}[1]{\lfloor #1 \rfloor}

\newcommand{\norm}[1]{\lVert\tinyspace #1 \tinyspace\rVert}
\newcommand{\bignorm}[1]{\bigl\lVert\tinyspace #1 \tinyspace\bigr\rVert}
\newcommand{\Bignorm}[1]{\Bigl\lVert\tinyspace #1 \tinyspace\Bigr\rVert}
\newcommand{\biggnorm}[1]{\biggl\lVert\tinyspace #1 \tinyspace\biggr\rVert}

\newcommand{\ket}[1]{
  \lvert\microspace #1 \microspace \rangle}

\newcommand{\bigket}[1]{
  \bigl\lvert\microspace #1 \microspace \bigr\rangle}

\newcommand{\bra}[1]{
  \langle\microspace #1 \microspace \rvert}

\newcommand{\bigbra}[1]{
  \bigl\langle\microspace #1 \microspace \bigr\rvert}

\newcommand\X{\mathcal{X}}
\newcommand\Y{\mathcal{Y}}
\newcommand\Z{\mathcal{Z}}
\newcommand\W{\mathcal{W}}
\newcommand\V{\mathcal{V}}

\newcommand\C{\mathcal{C}}

\newcommand\yes{\textup{yes}}
\newcommand\no{\textup{no}}

\newenvironment{mylist}[1]{\begin{list}{}{
	\setlength{\leftmargin}{#1}
	\setlength{\rightmargin}{0mm}
	\setlength{\labelsep}{2mm}
	\setlength{\labelwidth}{8mm}
	\setlength{\itemsep}{0mm}}}
	{\end{list}}

\newtheorem{cor}[theorem]{Corollary}
\newtheorem{prop}[theorem]{Proposition}

\newcommand{\field}{\mathbb{F}}
\newcommand{\val}{\omega}

\newcommand{\Tr}{\mbox{\rm Tr}}

\newcommand{\rW}{\ensuremath{\textsf{W}}}
\newcommand{\rY}{\ensuremath{\textsf{Y}}}
\newcommand{\rX}{\ensuremath{\textsf{X}}}
\newcommand{\rZ}{\ensuremath{\textsf{Z}}}

\newcommand{\QMA}{\mathrm{QMA}}

\newcommand{\IP}{\mathrm{IP}}
\newcommand{\QIP}{\mathrm{QIP}}
\newcommand{\PSPACE}{\mathrm{PSPACE}}
\newcommand{\MIP}{\mathrm{MIP}}
\newcommand{\QMIP}{\mathrm{QMIP}}
\newcommand{\NEXP}{\mathrm{NEXP}}
\newcommand{\EXP}{\mathrm{EXP}}

\newcommand{\eps}{\varepsilon}

\DeclareMathOperator{\poly}{poly}

\newcommand{\cons}{\textsc{cons}}
\newcommand{\conf}{\textsc{conf}}

\newcommand{\CHSH}{\mathrm{CHSH}}
\newcommand\opn{\mathrel{\ooalign{$\subseteq$\cr
  \hidewidth\raise.325ex\hbox{$?\mkern.5mu$}\cr}}}

\title{Quantum Proofs}

\author{
  Thomas Vidick \\
  California Institute of Technology \\
  vidick@cms.caltech.edu
  \and
  John Watrous \\
  University of Waterloo \\
  john.watrous@uwaterloo.ca
}

\begin{document}

 \copyrightowner{T. Vidick and J. Watrous}
 \volume{11}
 \issue{1--2}
 \pubyear{2015}
 \isbn{978-1-68083-126-9}
 \doi{0.1561/0400000068}
 \firstpage{1}
 \lastpage{215}

\maketitle
\tableofcontents
\mainmatter

\begin{abstract}
  Quantum information and computation provide a fascinating twist on the notion
  of proofs in computational complexity theory. 
  For instance, one may consider a quantum computational analogue of the
  complexity class \class{NP}, known as \class{QMA}, in which a quantum state
  plays the role of a proof (also called a certificate or witness), and is
  checked by a polynomial-time quantum computation.
  For some problems, the fact that a quantum proof state could be a
  superposition over exponentially many classical states appears to offer
  computational advantages over classical proof strings.
  In the interactive proof system setting, one may consider a verifier and one
  or more provers that exchange and process quantum information rather than
  classical information during an interaction for a given input string, giving
  rise to quantum complexity classes such as \class{QIP}, \class{QSZK}, and
  $\class{QMIP}^{\ast}$ that represent natural quantum analogues of \class{IP},
  \class{SZK}, and \class{MIP}.
  While quantum interactive proof systems inherit some properties from their
  classical counterparts, they also possess distinct and uniquely quantum
  features that lead to an interesting landscape of complexity classes based on
  variants of this model.

  In this survey we provide an overview of many of the known results concerning
  quantum proofs, computational models based on this concept, and properties of
  the complexity classes they define.
  In particular, we discuss non-interactive proofs and the complexity class
  \class{QMA}, single-prover quantum interactive proof systems and the
  complexity class \class{QIP}, statistical zero-knowledge quantum interactive
  proof systems and the complexity class \class{QSZK}, and multiprover
  interactive proof systems and the complexity classes \class{QMIP},
  $\class{QMIP}^{\ast}$, and $\class{MIP}^{\ast}$.
\end{abstract}

\chapter{Introduction}
\label{chapter:Introduction}

\noindent
The topic of this survey, \emph{quantum interactive proof systems}, draws upon
three different notions---quantum information, interaction, and proofs---whose
combination forms a fascinating recipe best presented in the reverse order.


We begin with the notion of \emph{proofs} in complexity theory.
This notion has been central to complexity theory from its early
beginnings, relating closely to the fundamental distinction between
\emph{efficient construction} and \emph{efficient verification}.
In greater detail, it has long been recognized that for some computational
problems whose solutions may be difficult to obtain, it may nevertheless be
possible to efficiently \emph{verify} the correctness of a solution, given some
additional information (representing a \emph{proof}) that aids in this
verification.
The complexity class \class{NP} represents a formalization of this notion---it
includes those decision problems for which positive instances can be efficiently
verified given a suitable proof string (and for which negative instances are
never incorrectly verified as positive ones).

The distinction between efficient construction and efficient verification
appears, for instance, in work of Edmonds~\cite{Edmonds65a} from 1965 (although
not in his more famous 1965 paper \cite{Edmonds65}), where he describes the
\emph{principle of the absolute supervisor}:
a supervisor can ask his or her assistant to carry out a potentially lengthy
search procedure for some computational problem (potentially ``killing'' the
assistant with work!), and at the end of the day the assistant is expected to
provide sufficient information so that his or her solution can be ``verified
with ease'' by the supervisor.

The more modern terminology used to describe this situation is that of
a \emph{prover} and \emph{verifier}: the prover represents the assistant,
while the verifier represents the supervisor in Edmonds' story.
With respect to this terminology, our sympathies are generally reversed:
the verifier, faced with limitations on its computational abilities, simply
wants to know whether or not a given input is a positive instance of a fixed
decision problem, while the computationally unrestricted prover is
untrustworthy and will try to convince the verifier that the input is a
positive instance, irrespective of the truth.

The importance of what is now known as the \class{P} vs \class{NP} question,
which essentially asks if there are indeed problems for which the efficient
construction of a solution is impossible while an efficient verification is
possible, was in fact implicitly noted some time prior to Edmonds' work---in a
letter written to John von Neumann in the mid-1950s, Kurt G\"odel observed the
striking consequences that would result from an efficient solution to a certain
problem in first-order logic that is now known to be \class{NP}-complete.
The development of the theory of NP-completeness, by Cook \cite{Cook71}, 
Levin \cite{Levin73}, and Karp \cite{Karp72} in the early 1970s, placed the
notion of proofs in computational complexity on a firm mathematical foundation.


Next, we add a second ingredient: \emph{interaction}. 
The notion of an \emph{interactive proof system} was introduced independently by
Goldwasser, Micali, and Rackoff \cite{GoldwasserMR85,GoldwasserMR89} and Babai
\cite{Babai85,BabaiM88} in the 1980s. 
Babai was following a similar line of thought that led to the introduction of
\class{P} and \class{NP}: the identification of structural features that allow
a fine classification of the difficulty of solving classes of computational
problems (in this case, problems related to groups). 
Goldwasser, Micali, and Rackoff arrived at the notion from a different angle.
They introduced a notion of ``knowledge complexity'' of an interactive proof
(informally, the amount of information about a problem instance conveyed by the
interaction beyond the problem's solution) and gave an example of a simple
problem (testing quadratic residuosity) for which there existed a
\emph{zero-knowledge} interactive proof. 

The simplest type of interactive proof system represents an interaction
between a prover and verifier, which are similar characters to the ones
introduced in the non-interactive setting above, except that now we imagine that
they may engage in a discussion rather than the prover simply providing the
verifier with information.
In particular, the verifier may ask the prover questions and demand acceptable
responses in order to be satisfied.
As before, one views that the prover's aim is to convince the verifier that
a given input string is a positive instance to a fixed decision problem (or,
equivalently, that an input string possesses a fixed property of interest).
The verifier's goal is to check the validity of the prover's argument,
accepting only in the event that it is indeed convinced that the input string
is a positive problem instance, and rejecting if not.

It turns out that the (classical) interactive proof system model only represents
a departure from the non-interactive setting described above when the verifier
makes use of \emph{randomness}---in which case we must generally be satisfied
with the verifier gathering overwhelming statistical evidence, but not having
absolute certainty, in order to conclude that the prover's argument is valid.
(When no randomness is used, the prover may as well attempt to convince the
verifier to accept non-interactively by simply presenting a complete
transcript of the conversation they would have had by interacting, which the
verifier can efficiently check for validity by itself.)
As in the non-interactive case, we also make the standard assumption that
the prover's computational abilities are greater than the verifier's (or, at the
very least, that the prover has access to information that the verifier lacks).
The class \class{IP} is representative of the case in which the verifier
is required to run in polynomial time and the prover is computationally
unrestricted.
The characterization \class{IP}=\class{PSPACE}~\cite{LundFKN92,Shamir92} cements
the tight relationship between interactive proofs and computation, justifying
its position as a fundamental concept in computational complexity theory.

Many variants of interactive proof systems have been considered that impose
additional conditions on the interaction, place more stringent limits on the
prover's abilities, or consider interactions between more diverse sets of
parties, such as a verifier interacting with multiple cooperating or competing
provers.
Prominent examples include the class \class{SZK} of problems that have
zero-knowledge interactive proofs and the class \class{MIP} of problems whose
solution can be determined by a polynomial-time verifier interacting with
multiple cooperating provers, restricted only in their inability to communicate
with one other.


Finally, we finish off with a curious catalyst: \emph{quantum information}.
The Church--Turing thesis plays a foundational role in computer science by
postulating that computability is model independent: whether based on the
concept of a Turing machine, first-order logic, or any ``purely mechanical
process,'' the classes of functions whose values can be ``effectively
calculated'' are identical. 
The development of quantum computing in the 1990s posed the first serious
threat to this thesis.
Impetus for the consideration of computational procedures based on the laws of
quantum mechanics was provided by Shor's discovery of an efficient
\emph{quantum} algorithm for factoring~\cite{Shor94, Shor97}, a problem for
which no efficient classical probabilistic algorithm is known.
The study of the relation between \class{P} (or \class{BPP}) and \class{BQP},
the class of problems that can be decided in polynomial time by a quantum
Turing machine, is among the most interesting and mysterious problems in modern
complexity theory.
The difficulty of this question prompts the introduction of 
``quantum analogues'' of the most important classical complexity classes in an
attempt to identify problems for which the consideration of quantum processes
induces a strict separation. 

One prominent example is the complexity class \class{QMA} of decision problems
whose positive instances have \emph{quantum proofs} that can be verified by an
efficient quantum procedure.
Aside from the fundamental problem of understanding the physical substrate of
computation, the consideration of quantum mechanical states as proofs provides
a fascinating window into some of the most subtle features of quantum physics.
An essential way in which quantum states differ from their classical
counterparts is in one's ability to recover information that is present in the
mathematical description of the state. 
In quantum mechanics this ability is limited by the uncertainty
principle---for example, both the momentum and position of an electron can
be determined with high precision in principle, but there is a fundamental limit
to the accuracy with which those two properties can be \emph{simultaneously}
determined. 
Thus, the study of \class{QMA} sheds light on the many areas of physics in
which the properties of quantum states play an important role, from the theory
of superconductors to that of black holes.


Stir vigorously, and you have a recipe for quantum interactive proofs. 
Beyond the class \class{QMA} already discussed, quantum interactive proofs
reflect the richness of the classical model on which they are based, providing
a powerful lens on the properties of quantum mechanics and quantum
information.
For example, single-prover quantum interactive proofs, corresponding to the
class \class{QIP}, have the distinguishing property that they can be
parallelized to three message interactions, and this property (unlikely to hold
for classical interactive proofs) makes crucial use of the superposition
principle of quantum mechanics.
The no-cloning theorem plays an important role in the study of the class
\class{QSZK} of problems having quantum zero-knowledge interactive proofs
by hindering the construction of ``simulators'' essential to the study of
classical zero-knowledge.
By allowing multiple cooperating provers to share quantum entanglement, the
class $\class{QMIP}^{\ast}$ provides a complexity-theoretic viewpoint on the
nonlocal properties of entanglement. 


Having set a rather ambitious stage for this survey, we proceed with a more
concrete description of what is to come.

Chapter~\ref{chapter:preliminaries} introduces some preliminary material.
While it is assumed that the reader will be familiar with the basics of
complexity theory and quantum computing, we have made an effort to state and
explain the facts that play an important role in the results to be discussed,
directing the reader to standard textbooks for background material.

In Chapter~\ref{chapter:QMA} we begin with the consideration of the class
\class{QMA} of languages that have efficiently verifiable quantum proofs. This
class satisfies many of the desirable features of \class{NP}, such as strong
error amplification procedures and a rich set of complete problems. 
It also has many variants restricting, or extending, the types of proofs
allowed and the power of the verifier; a small but representative set of such
variants is discussed in the chapter.

Chapter~\ref{chapter:single-prover} considers single-prover quantum interactive
proof systems. 
An important tool in the study of the associated class \class{QIP} is a
semidefinite programming formulation of the verifier's maximum acceptance
probability. 
We introduce this formulation and use it to establish a parallel repetition
property of \class{QIP} as well as to give an essentially self-contained proof
of the characterization \class{QIP}=\class{PSPACE}.

In Chapter~\ref{chapter:QSZK} we consider the class \class{QSZK} of quantum
zero-knowledge interactive proofs.
One aspect in which these proof systems differ from their classical
counterparts is the difficulty of extending the key techniques (such as
rewinding) that are systematically used in the classical setting, and we
describe known quantum analogues for such techniques.

The final chapter, Chapter~\ref{chapter:multiple-provers}, is devoted to
quantum multi-prover interactive proofs. 
It will be seen that the consideration of entanglement between multiple provers
leads to a failure of the most basic intuition on which the classical theory is
built (most important of which are the technique of oracularization and the
characterization \class{MIP}=\class{NEXP}). 
We describe ways to work around this failure by fighting fire with fire,
devising techniques that make positive use of the provers' ability to share
entanglement.


This survey is mainly intended for non-specialists having a basic background in
complexity theory and quantum information.
A typical reader may be a student or researcher in either area desiring to
learn about the fundamentals of the (actively developing) theory of quantum
interactive proofs.
In most cases we have not included full proofs of the main results we present,
but whenever possible we have either included detailed sketches of the key
ideas behind the proofs, or have attempted to describe their most salient
elements in simplified settings.
Each chapter ends with notes that provide references for the results discussed
in the chapter as well as a brief survey of related results and pointers to the
literature.

\chapter{Preliminary Notions} \label{chapter:preliminaries}

The purpose of this chapter is to summarize various notions, primarily
concerning basic complexity theory and quantum computation, that we will rely
upon in subsequent chapters and consider as requisite background material.
The chapter is mainly intended to clarify our notation and terminology; readers
unfamiliar with the notions summarized will likely find other sources,
including textbooks and surveys focusing on this material, to better serve as a
first introduction.
Suggested references will be mentioned when appropriate.

\section{Complexity theoretic notions}

We assume the reader is familiar with standard classical complexity classes,
such as 
\class{NC},
\class{P}, 
\class{BPP}, 
\class{NP}, 
\class{AM}, 
\class{PSPACE}, and
\class{NEXP},
as well as the quantum complexity class \class{BQP}.
The textbook of Arora and Barak \cite{AroraB09} may be consulted for
definitions and basic properties of these classes.
All computational problems considered in this survey will be assumed to be
encoded over the binary alphabet $\{0,1\}$, which is hereafter
denoted~$\Sigma$.

It is convenient for us to consider computational decision problems as
\emph{promise problems}, where input strings may be assumed to be drawn from
some subset of all possible input strings. 
More formally, a promise problem is a pair $A = (A_{yes}, A_{no})$, where
$A_{yes}, A_{no}\subseteq \Sigma^{\ast}$ are disjoint sets of strings.
The strings contained in the sets $A_{yes}$ and $A_{no}$ are called the
\emph{yes-instances} and \emph{no-instances} of the problem, and a correct
answer to any such instance of the problem $A$ requires that it be properly
classified as a yes-instance or no-instance.
All strings lying outside of $A_{yes}\cup A_{no}$ may be considered as
``don't care'' inputs, and no requirements whatsoever are placed on
computations for such strings.
All of the complexity classes mentioned above may be considered as classes of
promise problems, as opposed to classes of languages (which are essentially
promise problems for which $A_{yes}\cup A_{no} = \Sigma^{\ast}$).

Karp reductions (also called polynomial-time many-one reductions), as well
as the notion of \emph{completeness} of a problem for a complexity class, are
defined for promise problems in the same way as for languages.
More precisely, a promise problem $A=(A_{yes},A_{no})$ is 
\emph{Karp reducible} to a promise problem $B=(B_{yes},B_{no})$ if there exists
a polynomial-time computable function $f$ that maps every string $x\in A_{yes}$
to $f(x)\in B_{yes}$, and every string $x\in A_{no}$ to $f(x)\in B_{no}$. 
In this case, the notation $A \leq_m B$ (where the ``m'' is short for
``many-one'') may be used to indicate this relationship.
A promise problem $A$ is said to be \emph{complete} for a certain class
$\class{C}$ of promise problems if $A\in\class{C}$ and every promise problem in
$\class{C}$ is Karp reducible to $A$.

There are several occasions in which we speak of functions defined on the
nonnegative integers $\natural = \{0,1,2,\ldots\}$, taking either nonnegative
integer values or real-number values.
In particular, the following terminology will be used: 
\begin{mylist}{\parindent}
\item[1.] 
  A function of the form $f:\natural\to\natural$ is said to be
  \emph{polynomially bounded} if there exists a polynomial-time deterministic
  Turing machine that outputs $1^{f(n)}$ on input $1^n$, for every
  $n\in\natural$.
  We will use the notation $\poly$ to denote an arbitrary function that is
  polynomially bounded.
\item[2.]
  A function of the form $g:\natural\to\rational$ is said to be 
  \emph{polynomial-time computable} if there exists a
  polynomial-time deterministic Turing machine that outputs $g(n)$,
  expressed as a ratio of integers written in binary notation,
  on input $1^n$, for each $n\in\natural$.
\end{mylist}

\section{Quantum states, channels, and measurements}
\label{section:states-channels-measurements}

When discussing quantum interactive proof systems, and quantum computations
more generally, it is useful to make use of some basic concepts of the theory
of quantum information.
Readers unfamiliar with quantum information and computation are referred to the
books of Nielsen and Chuang \cite{NielsenC00} and Kaye, Laflamme, and Mosca
\cite{KayeLM07}.

\subsection{Linear algebra notation}

We use calligraphic letters $\X,\Y,\W$ to denote Hilbert spaces.
All Hilbert spaces considered in this survey are finite-dimensional and will
usually correspond to systems comprised of zero or more qubits (in which case
their dimension is a power of 2).
It is assumed that an orthonormal \emph{standard basis} has been fixed for each
such space $\X$.

In the typical case in which $\X$ is the Hilbert space corresponding to $n$
qubits, the standard basis is expressed using Dirac notation as
\begin{equation}
  \bigl\{\ket{x}\,:\,x\in\Sigma^n\bigr\},
\end{equation}
where $\ket{x}$ denotes the column vector with a 1 in the entry indexed by $x$
and zero for all other entries.
An arbitrary vector $\ket{\psi}$ in such a space may be expressed as
\begin{equation}
  \ket{\psi} = \sum_{x\in\Sigma^n} \alpha_x \ket{x}
\end{equation}
for some choice of complex coefficients $\{\alpha_x\,:\,x\in\Sigma^n\}$.
The \emph{conjugate-transpose} of this vector is given by
\begin{equation}
  \bra{\psi} = \sum_{x\in\Sigma^n} \overline{\alpha_x} \bra{x},
\end{equation}
where $\bra{x}$ denotes the row vector (as opposed to a column vector) with a 1
in the entry indexed by $x$ and zero for all other entries.
The inner product of two vectors, $\ket{\phi}$ and $\ket{\psi}$, is
written $\langle \phi | \psi \rangle$, and the Euclidean norm of $\ket{\psi}$
is defined as $\|\ket{\psi}\| = \langle \psi | \psi \rangle^{1/2}$.

Given two Hilbert spaces $\X$ and $\Y$, the space of all linear mappings (or
\emph{operators}) from $\X$ to $\Y$ is denoted $\Lin(\X,\Y)$, and in the case
that $\X = \Y$ the shorthand $\Lin(\X)$ is used in place of $\Lin(\X,\X)$.
The identity element of $\Lin(\X)$ is written $\I_\X$, and the trace of an
operator $X\in\Lin(\X)$ is defined as
\begin{equation}
  \tr(X) = \sum_x \bra{x} X \ket{x},
\end{equation}
where the sum is taken over the standard basis elements of $\X$.
(Any other orthonormal basis would yield the same value.)
With respect to the standard bases of $\X$ and $\Y$, operators in the set
$\Lin(\X,\Y)$ may be identified with matrices in the usual way, with the
$(x,y)$ entry of the matrix corresponding to an operator $A$ being given by
$\langle x | A | y\rangle$.

The \emph{adjoint}, or \emph{conjugate transpose}, of an operator
$A\in\Lin(\X,\Y)$ is the operator $A^{\ast}\in\Lin(\Y,\Y)$ defined by the
condition
\begin{equation}
  \bra{y} A^{\ast} \ket{x} = \overline{\bra{x} A \ket{y}}
\end{equation}
for every choice of standard basis states $\ket{x}\in \X$ and $\ket{y}\in\Y$.
An inner product is defined on the space $\Lin(\X,\Y)$ as
\begin{equation}
  \ip{A}{B} = \tr(A^{\ast} B)
\end{equation}
for all $A,B\in\Lin(\X,\Y)$.

The following set of operators will often be mentioned throughout this
survey:
\begin{mylist}{\parindent}
\item[1.]
  An operator $U\in\Lin(\X)$ is \emph{unitary} if $U^{\ast} U = \I_{\X}$, which
  is an equivalent condition to $U U^{\ast} = \I_{\X}$.
  The set of all such operators is denoted $\Unitary(\X)$.
  More generally, an operator $A \in \Lin(\X,\Y)$ is an \emph{isometry}
  if $A^{\ast} A = \I_{\X}$, and the set of all such operators
  is denoted $\Unitary(\X,\Y)$.
\item[2.]
  An operator $H\in\Lin(\X)$ is \emph{Hermitian} if $H = H^{\ast}$.
  The set of all such operators is denoted $\Herm(\X)$.
\item[3.]
  An operator $P\in\Lin(\X)$ is \emph{positive semidefinite} if
  it is Hermitian and has only nonnegative eigenvalues.
  The set of all such operators is denoted $\Pos(\X)$.
\item[4.]
  An operator $\Pi\in\Lin(\X)$ is a \emph{projection operator} if
  it is Hermitian and satisfies $\Pi^2 = \Pi$.
  The set of all such operators is denoted $\Proj(\X)$.
\item[5.]
  An operator $\rho\in\Lin(\X)$ is a \emph{density operator} if
  it is positive semidefinite and satisfies $\tr(\rho) = 1$.
  The set of all such operators is denoted $\Density(\X)$.
\end{mylist}

\subsection{States and registers}

A \emph{quantum state} is represented by a density operator
$\rho\in\Density(\X)$, for some Hilbert space $\X$ that has been associated
with the physical system whose state is being described.
A state is \emph{pure} if it is represented by a density operator of the
form $\rho = \ket{\psi}\bra{\psi}$, for $\ket{\psi}$ a unit vector in $\X$.
Equivalently, a state is pure if its associated density operator 
$\rho \in \Density(\X)$ is an extreme point of the (convex) set $\Density(\X)$.
A state is said to be \emph{mixed} if it is not pure.
When a unit vector $\ket{\psi}$ is referred to as being a state of a system, it
is to be understood that one is speaking of the pure state
$\ket{\psi}\bra{\psi}$.

It is convenient to refer to physical systems that store quantum
information as \emph{registers}.
%
Names such as $\reg{X}$, $\reg{Y}$, and $\reg{Z}$, and other capital letters
written in a \emph{sans serif} font, are commonly used for this purpose.
With a given register $\reg{X}$, one associates a Hilbert space $\X$, so the
set of possible quantum states of $\reg{X}$ coincides with $\Density(\X)$.
As a general convention, we use the same letter in different fonts to refer to
a register and its associated Hilbert space.

Pairs or $k$-tuples of registers, such as $(\reg{X},\reg{Y})$ or
$(\reg{X}_1,\ldots,\reg{X}_k)$, are often considered, and may themselves be
treated as single registers.
The Hilbert space associated with such a compound register is obtained by
taking the tensor product of the Hilbert spaces associated with the individual
registers.
For example, if $\reg{Z} = (\reg{X},\reg{Y})$, then the Hilbert space $\Z$
corresponding to $\reg{Z}$ is given by $\Z = \X\otimes\Y$, for $\X$ and $\Y$
being the Hilbert spaces associated with $\reg{X}$ and $\reg{Y}$, respectively.
The standard basis of $\Z$ in this case is obtained by tensoring the elements
of the standard bases of $\X$ and $\Y$:
\begin{equation}
  \bigl\{\ket{x}\ket{y}\,:\,x\in\Sigma^n,\:y\in\Sigma^m\bigr\},
\end{equation}
assuming here that $\reg{X}$ is an $n$-qubit register and $\reg{Y}$ is an
$m$-qubit register.
(In general, a juxtaposition of vectors, such as $\ket{\phi}\ket{\psi}$ for
$\ket{\phi}\in\X$ and $\ket{\psi}\in\Y$, denotes a tensor product:
$\ket{\phi}\ket{\psi} = \ket{\phi}\otimes\ket{\psi}$.)

Every nonzero vector $\ket{\gamma}\in\X\otimes\Y$ of a bipartite tensor product
space has a decomposition
\begin{equation}
  \ket{\gamma} = \sum_{i = 1}^r \alpha_i \ket{\phi_i} \ket{\psi_i}
\end{equation}
for $\alpha_1,\ldots,\alpha_r$ being positive real numbers and
$\ket{\phi_1},\ldots,\ket{\phi_r}\in\X$ and
$\ket{\psi_1},\ldots,\ket{\psi_r}\in\Y$ being vectors for which both of the
collections
$\{\ket{\phi_1},\ldots,\ket{\phi_r}\}$ and
$\{\ket{\psi_1},\ldots,\ket{\psi_r}\}$ are orthonormal.
This decomposition, which is called the \emph{Schmidt decomposition}, is
closely related to the singular value decomposition of matrices.
The values $\alpha_1,\ldots,\alpha_r$ are the \emph{Schmidt coefficients} of
$\ket{\gamma}$, and are uniquely determined by $\ket{\gamma}$.
A pure state corresponding to a unit vector $\ket{\gamma} \in \X\otimes\Y$ is
called a \emph{product} state if it takes the form 
$\ket{\gamma} = \ket{\psi}\ket{\phi}$ for vectors $\ket{\psi}\in\X$ and
$\ket{\phi}\in\Y$, which is equivalent to $\ket{\gamma}$ having a single
(nonzero) Schmidt coefficient $\alpha_1=1$.

A mixed state $\rho \in \Density(\X\otimes\Y)$ is \emph{separable} if it has a
decomposition
\begin{equation}
  \rho = \sum_{i=1}^N p_i \,\sigma_i \otimes \tau_i
\end{equation}
for $\sigma_1,\ldots,\sigma_N\in\Density(\X)$ and
$\tau_1,\ldots,\tau_N\in\Density(\Y)$, and $(p_1,\ldots,p_N)$ being a vector
of probabilities.
A state that is not separable is called \emph{entangled}.
A pure state is separable if and only if it is a product state, but the
mixed-state case is more complicated: it is an \class{NP}-hard problem to
decide if a given density operator is separable \cite{Gurvits03}.

If $\rho\in\Density(\X\otimes\Y)$ is a state of a pair of registers
$(\reg{X},\reg{Y})$, then the \emph{reduced state} of $\reg{X}$ is obtained by
taking the \emph{partial trace} over~$\Y$:
\begin{equation}
  \tr_{\Y}(\rho) = \sum_y \bigl(\I_{\X}\otimes \bra{y}\bigr)
  \rho \bigl(\I_{\X}\otimes \ket{y}\bigr),
\end{equation}
where the sum ranges over the elements of the standard basis of $\Y$.

Given a state $\rho\in\Density(\X)$ of a register $\reg{X}$, a 
\emph{purification} of $\rho$ is any pure state of a pair of
registers $(\reg{X},\reg{Y})$ whose reduced state on $\reg{X}$ is $\rho$.
That is, such a state is represented by a unit vector
$\ket{\gamma}\in\X\otimes\Y$ such that
\begin{equation}
  \rho = \tr_{\Y}\bigl( \ket{\gamma} \bra{\gamma} \bigr).
\end{equation}
Given a spectral decomposition
\begin{equation}
  \rho = \sum_{i=1}^r \lambda_i \ket{\psi_i}\bra{\psi_i}
\end{equation}
with each $\lambda_i > 0$, one may obtain a purification
\begin{equation}
  \ket{\gamma} = \sum_i \sqrt{\lambda_i} \ket{\psi_i}\ket{\phi_i}
\end{equation}
provided $\Y$ has dimension at least the rank of $\rho$, allowing for the
existence of an orthonormal collection $\{\ket{\phi_i}\}$.
The following theorem concerning purifications has fundamental importance.

\begin{theorem}[Unitary equivalence of purifications]
  \label{thm:unitary-equivalence}
  Let $\rho\in \Density(\X)$ and suppose $\ket{\psi},\ket{\phi}\in\X\otimes\Y$
  satisfy
  \begin{equation}
    \tr_{\Y}\bigl( \ket{\psi} \bra{\psi} \bigr) = 
    \rho = \tr_{\Y}\bigl( \ket{\phi} \bra{\phi} \bigr).
  \end{equation}
  There exists a unitary operator $U\in \Unitary(\Y)$ such that
  $\ket{\phi} = (\I_\X \otimes U) \ket{\psi}$. 
\end{theorem}

\subsection{Channels and measurements}

\emph{Quantum channels} describe discrete-time changes in the states of
registers that, in an idealized sense, may be considered physically
implementable.
Given registers $\reg{X}$ and $\reg{Y}$ having associated Hilbert spaces
$\X$ and $\Y$, the set of all quantum channels transforming states of
$\reg{X}$ into states of $\reg{Y}$, denoted $\Channel(\X,\Y)$, can be
characterized as the set of all linear maps of the form
\begin{equation}
  \Phi: \Lin(\X) \rightarrow \Lin(\Y)
\end{equation}
that are \emph{completely positive} and \emph{trace-preserving}.
An equivalent way to describe the two conditions of being completely positive
and trace-preserving is to require that, for every finite-dimensional Hilbert
space $\Z$, it holds that
\begin{equation}
  \bigl(\Phi \otimes \I_{\Lin(\Z)}\bigr)(\rho) \in \Density(\Y\otimes\Z)
\end{equation}
for every density operator $\rho\in\Density(\X\otimes\Z)$.
(Here, the mapping $\I_{\Lin(\Z)}$ denotes the identity mapping on $\Lin(\Z)$.)

A channel $\Phi\in \Channel(\X,\Y)$ should be understood as representing a
physical transformation of register $\reg{X}$ into register $\reg{Y}$.
That is, if the state of $\reg{X}$ is given by $\rho\in\Density(\X)$ and the
channel $\Phi\in\Channel(\X,\Y)$ is performed, the register $\reg{X}$ is
transformed into the register $\reg{Y}$, whose state is then $\Phi(\rho)$.
The two registers $\reg{X}$ and $\reg{Y}$ never simultaneously co-exist in this
situation, so it is not meaningful to consider their joint state.
Of course, nothing prevents one from taking $\reg{Y}$ to be equal to $\reg{X}$,
and in this situation it is natural to view that the channel has simply changed
the state of $\reg{X}$, as opposed to transforming $\reg{X}$ into a new
register.

A convenient representation of quantum channels is the 
\emph{Stinespring representation}. 
Given an arbitrary channel $\Phi\in\Channel(\X,\Y)$ there always exists a 
Hilbert space $\Z$, which can be chosen to have dimension at most the product
of the dimensions of $\X$ and $\Y$, along with a linear operator
$A\in\Lin(\X,\Y\otimes \Z)$, such that
\begin{mylist}{\parindent}
\item[1.]
  $\Phi(X)=\Tr_{\Z}(A X A^{\ast})$ for every $X\in\Lin(\X)$, and
\item[2.]
  $A^{\ast} A = \I_{\X}$ (i.e., $A$ is an isometry).
\end{mylist}
The fact that every channel can be represented in this way is a consequence of
a theorem known as \emph{Stinespring's dilation theorem}.

It is instructive to consider the case in which a channel $\Phi$ transforms an
$n$-qubit register $\reg{X}$ into an $m$-qubit register $\reg{Y}$.
In this case, the Hilbert spaces corresponding to these registers are such that
$\X$ has dimension $2^n$ and $\Y$ has dimension $2^m$; and when considering
a Stinespring representation of $\Phi$, one may take $\reg{Z}$ to be an
$(n+m)$-qubit register, so that $\Z$ has dimension $2^{n+m}$.
One finds that there must exist a Stinespring representation
\begin{equation}
  \Phi(X) = \tr_{\Z} (A X A^{\ast}),
\end{equation}
for $A$ being an isometry of the form $A\in\Lin(\X,\Y\otimes\Z)$, as
is illustrated in Figure~\ref{figure:Stinespring}.
One may further observe that it must be possible to express any isometry
$A$ of the form $A\in\Lin(\X,\Y\otimes\Z)$ as
\begin{equation}
  A = U \bigl(\I_{\X} \otimes \ket{0^{2m}}\bigr)
\end{equation}
for $U$ being a unitary operator acting on $n + 2m$ qubits.
It is therefore possible to implement an arbitrary channel transforming
$n$ qubits to $m$ qubits by means of a unitary operation on $n + 2m$ qubits,
as is also illustrated in Figure~\ref{figure:Stinespring}.

\begin{figure}
  \begin{center}
    \begin{minipage}[c]{46mm}
      \begin{tikzpicture}[scale=0.45,
          channel/.style={draw, minimum height=14mm, minimum width=10mm,
            fill = ChannelColor, text=ChannelTextColor},
          isometry/.style={draw, minimum height=20mm, minimum width=10mm,
            fill = ChannelColor, text=ChannelTextColor},
          invisible/.style={minimum height=14mm, minimum width=8mm},
          >=latex]
        
        \node (W) at (0,0) [isometry] {$A$};
        \node (Left) at (-3,0) [invisible] {$\rho$};
        \node (Right) at (3,0) [invisible] {%
          \makebox[0mm][c]{\raisebox{12mm}{$\quad\Phi(\rho)$}}};
        
        \foreach \y in {-6,-4,-2,0,2,4,6} {
          \draw ([yshift=\y mm]Left.east) -- ([yshift=\y mm]W.west) {};
        }
        
        \foreach \y in {-16,-14,-12,-10,-8,-6,-4,-2,0,2,4,10,12,14,16} {
          \draw ([yshift=\y mm]W.east) -- ([yshift=\y mm]Right.west) {};
        }
        
        \node at (4.15,-0.62) {$\left.\rule{0mm}{6mm}\right\}$\hspace{-2mm}
          \footnotesize
          \begin{tabular}{c}
            traced\\[-1mm]
            out
        \end{tabular}};
      \end{tikzpicture}
    \end{minipage}\hspace{4mm}
    \begin{minipage}[c]{52mm}
      \begin{tikzpicture}[scale=0.45,
          channel/.style={draw, minimum height=14mm, minimum width=10mm,
            fill = ChannelColor, text=ChannelTextColor},
          unitary/.style={draw, minimum height=20mm, minimum width=10mm,
            fill = ChannelColor, text=ChannelTextColor},
          invisible/.style={minimum height=14mm, minimum width=8mm},
          >=latex]
        
        \node (U) at (0,0) [unitary] {$U$};
        \node (Left) at (-3,0) [invisible] {%
          \makebox[0mm][c]{\raisebox{9mm}{$\rho$}}};
        \node (Right) at (3,0) [invisible] {%
          \makebox[0mm][c]{\raisebox{12mm}{$\quad\Phi(\rho)$}}};
                
        \foreach \y in {16,14,12,10,8,6,4,-2,-4,-6,-8,-10,-12,-14,-16} {
          \draw ([yshift=\y mm]Left.east) -- ([yshift=\y mm]U.west) {};
        }
        
        \foreach \y in {-16,-14,-12,-10,-8,-6,-4,-2,0,2,4,10,12,14,16} {
          \draw ([yshift=\y mm]U.east) -- ([yshift=\y mm]Right.west) {};
        }
        
        \node at (4.15,-0.62) {$\left.\rule{0mm}{6mm}\right\}$\hspace{-2mm}
          \footnotesize
          \begin{tabular}{c}
            traced\\[-1mm]
            out
        \end{tabular}};
        
        \node at (-3.6,-0.9) 
              {$\scriptstyle{\ket{0\,\cdots\,0}\left\{\rule{0mm}{5mm}\right.}$};
              
      \end{tikzpicture}
    \end{minipage}
  \end{center}
  \caption{Given a quantum channel $\Phi$ transforming an $n$ qubit register
    $\reg{X}$ into an $m$ qubit register $\reg{Y}$, one may always find an
    isometry $A$, transforming pure states of $\reg{X}$ to pure states
    of $(\reg{Y},\reg{Z})$, for $\reg{Z}$ being an $n+m$ qubit register, so
    that $\Phi(\rho) = \tr_{\Z}(A \rho A^{\ast})$ for every state $\rho$ of
    $\reg{X}$.
    The isometry $A$ can be realized as a unitary operation on $n + 2m$ qubits,
    with all but the first $n$ qubits being initialized to the $\ket{0}$
    state.}
  \label{figure:Stinespring}
\end{figure}
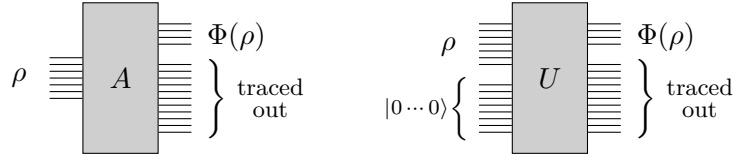

An analogous fact to the unitary equivalence of purifications holds for
Stinespring representations, as the following theorem states.

\begin{theorem}
  \label{theorem:Stinespring-equivalence}
  Let $\Phi\in\Channel(\X,\Y)$ and let
  $A, B \in \Lin(\X,\Y\otimes\Z)$ be isometries such that
  \begin{equation}
    \Tr_\Z(AXA^{\ast}) = \Phi(X) = \Tr_\Z(B X B^{\ast})
  \end{equation}
  for all $X\in\Lin(\X)$. 
  There exists a unitary operator $U\in\Unitary(\Z)$ such that
  $B = (\I_\Y\otimes U)A$.
\end{theorem} 

A channel $\Phi$ is said to be a \emph{unitary channel} if there
exists a unitary operator $U$ such that
\begin{equation}
  \Phi(X) = U X U^{\ast}
\end{equation}
for all $X\in\Lin(\X)$, and $\Phi$ is said to be an \emph{isometric channel}
if there exists a linear isometry $A$ such that
\begin{equation}
  \Phi(X) = A X A^{\ast}
\end{equation}
for all $X\in\Lin(\X)$.

When considering channels and related notions, it is sometimes convenient to
consider linear maps of the form
\begin{equation}
  \label{eq:linear-map-on-operator-spaces}
  \Phi: \Lin(\X) \rightarrow \Lin(\Y)
\end{equation}
that might be neither completely positive nor trace-preserving.
The notation $\Trans(\X,\Y)$ is used to denote the set of all linear maps of
the form \eqref{eq:linear-map-on-operator-spaces}.
For example, if $\Phi_0,\Phi_1\in\Channel(\X,\Y)$ are channels, it may be
useful to consider the linear map $\Phi = \Phi_0 - \Phi_1 \in \Trans(\X,\Y)$
as a way of representing the difference between the two channels.
The \emph{adjoint} of a map $\Phi\in\Trans(\X,\Y)$ is the uniquely defined
map $\Phi^{\ast} \in \Trans(\Y,\X)$ satisfying
\begin{equation}
  \ip{Y}{\Phi(X)} = \ip{\Phi^{\ast}(Y)}{X}
\end{equation}
for all $X\in\Lin(\X)$ and $Y\in\Lin(\Y)$.
\label{sec:channel-purification}

A \emph{measurement} is a process through which classical information is
obtained from a register in a quantum state.
For the purposes of this survey it will be sufficient to consider measurements
that are defined as the following special cases of quantum channels:
\begin{mylist}{\parindent}
\item[1.]
  A \emph{standard-basis measurement} of a register $\reg{X}$ is
  described by the so-called \emph{completely phase-damping} channel
  $\Lambda\in\Channel(\X)$, which is defined as
  \begin{equation}
    \Lambda(\rho) = \sum_{x} \bra{x}\rho\ket{x} \ket{x}\bra{x},
  \end{equation}
  for the sum ranging over the standard basis elements of $\X$.
  When $\reg{X}$ is measured in this way, its resulting state is represented by
  a diagonal density operator, which is naturally associated with a classical,
  probabilistic state.
  That is, measuring $\rho$ yields each outcome $x$ with probability
  $\bra{x} \rho \ket{x}$.
\item[2.]
  A general measurement of a register $\reg{X}$ can always be described as the
  composition of a channel $\Phi$ transforming $\reg{X}$ into $\reg{Y}$,
  followed by a standard-basis measurement of $\reg{Y}$.
  A composition of a channel and a standard-basis measurement of this sort can
  always be written as
  \begin{equation}
    (\Lambda \Phi)(X) = \sum_y \bigip{P_y}{X} \ket{y}\bra{y},
  \end{equation}
  for $\{P_y\}$ being a collection of positive semidefinite 
  \emph{measurement operators} satisfying
  \begin{equation}
    \sum_y P_y = \I_{\X}.
  \end{equation}
\end{mylist}

\subsection{Distance measures on quantum states and channels}

The space $\Lin(\X,\Y)$ is equipped with the \emph{operator norm}
(or \emph{spectral norm}), derived from the Euclidean norm on $\X$ and $\Y$ as
\begin{equation}
  \norm{X} = \max\bigl\{
  \norm{X\ket{\psi}} \,:\, \ket{\psi}\in\X,\:\norm{\ket{\psi}} \leq 1\bigr\}  .
\end{equation}
We will also make use of the \emph{trace norm}, defined as
\begin{equation}
  \norm{X}_1 = \Tr \Bigl( (XX^{\ast})^{1/2} \Bigr).
\end{equation}
Equivalently, $\norm{X}_1$ is equal to the sum of the singular values of $X$.
The operator norm and trace norm are dual to one another, meaning that the
following relationships hold:
\begin{equation}
  \begin{aligned}
    \norm{X} & = 
    \max\bigl\{\abs{\ip{Y}{X}}\,:\,Y\in\Lin(\X,\Y),\:\norm{Y}_1\leq 1\bigr\},\\
    \norm{X}_1 & = 
    \max\bigl\{\abs{\ip{Y}{X}}\,:\,Y\in\Lin(\X,\Y),\:\norm{Y}\leq 1\bigr\}.
  \end{aligned}
\end{equation}
It is sometimes convenient to make use of the fact that, for any operator
$X\in\Lin(\X)$, it holds that
\begin{equation}
  \norm{X}_1 = \max_{U\in\Unitary(\X)} \abs{\ip{U}{X}},
\end{equation}
and furthermore if $X$ is Hermitian this maximization may be restricted to
operators $U$ that are both unitary and Hermitian.

The most standard notion of distance between two quantum states $\rho$ and
$\sigma$ is the \emph{trace distance}
\begin{equation}
  \frac{1}{2} \|\rho-\sigma\|_{1}.
\end{equation}
(The factor $\frac{1}{2}$ ensures that the distance between two states
lies in the interval $[0,1]$.)
A theorem called the \emph{Holevo--Helstrom theorem} implies that
\begin{equation}
  \begin{multlined}
    \max \Bigl\{
    \frac{1}{2} \ip{P_0}{\rho_0} + 
    \frac{1}{2} \ip{P_1}{\rho_1}\,:\,P_0,P_1\geq 0,\:P_0 + P_1 = \I_{\X}
    \Bigr\} \\
    = \frac{1}{2} + \frac{1}{4} \norm{\rho_0 - \rho_1}_1,
  \end{multlined}
\end{equation}
which has the interpretation that the trace distance expresses the maximum
bias with which any measurement can correctly distinguish between two states
$\rho_0$ and $\rho_1$ given with equal probability.

It is sometimes convenient to refer to one state as being an
\emph{$\varepsilon$-approximation} to another when the trace-distance between
the states is bounded from above by $\varepsilon$, as the following definition
makes precise.
\begin{definition}
  \label{definition:state-approximation}
  Let $\rho$ and $\sigma$ be states on the same space.
  It is said that $\sigma$ is an \emph{$\varepsilon$-approximation} to $\rho$
  if
  \begin{equation}
    \frac{1}{2}\norm{\rho - \sigma}_1 \leq \varepsilon.
  \end{equation}
\end{definition}

Another important measure of distance between quantum states is the
\emph{fidelity}, defined for density operators $\rho,\sigma\in\Density(\X)$ as
\begin{equation}
  \fid(\rho,\sigma) = \Tr\bigg(\sqrt{\sqrt{\rho}\sigma\sqrt{\rho}}\bigg)
  = \bignorm{\sqrt{\rho}\sqrt{\sigma}}_1.
\end{equation}
(The second expression makes it apparent that 
$\fid(\rho,\sigma)=\fid(\sigma,\rho)$, as $\sqrt{\rho}\sqrt{\sigma}$ and
$\sqrt{\sigma}\sqrt{\rho} = (\sqrt{\rho}\sqrt{\sigma})^{\ast}$ must share the
same singular values.)
When $\sigma = \ket{\psi}\bra{\psi}$ is a pure state, the expression simplifies
to
\begin{equation}
  \fid(\rho,\ket{\psi}\bra{\psi}) = \sqrt{\bra{\psi}\rho\ket{\psi}}. 
\end{equation}
The following theorem gives an alternative characterization of the fidelity.
\begin{theorem}[Uhlmann's theorem]\label{theorem:Uhlmann}
  Let $\rho,\sigma\in\Density(\X)$ be density operators.
  It holds that
  \begin{equation}
    \fid(\rho,\sigma) = \max_{\ket{\psi},\ket{\phi}}
    \abs{\langle\psi|\phi\rangle},
  \end{equation}
  where the maximization is over all purifications $\ket{\psi}$ and
  $\ket{\phi}$ of $\rho$ and $\sigma$, respectively.
\end{theorem}

The fidelity is related to the trace distance by the
\emph{Fuchs--van de Graaf inequalities}:
\begin{equation}\label{eq:fuchs-graaf}
  1-\fid(\rho,\sigma)
  \leq \frac{1}{2}\|\rho-\sigma\|_1 \leq \sqrt{1-\fid(\rho,\sigma)^2},
\end{equation}
for all density operators $\rho$ and $\sigma$.

There is a notion of distance between quantum channels that is analogous to the
trace distance between quantum states.
This notion of distance is defined by the \emph{diamond norm} (also known as
the \emph{completely bounded trace norm}) of linear maps of the form
$\Phi\in\Trans(\X,\Y)$.
This norm is defined as
\begin{equation}
  \norm{\Phi}_{\diamond}
  = \max \bigl\{
  \bignorm{(\Phi\otimes\I_{\Lin(\W)})(X)}_1 : X\in \Lin(\X\otimes\W),\:
  \norm{X}_1 \leq 1\bigr\},
\end{equation}
where $\W$ is any Hilbert space having dimension at least as large as $\X$.
(Changing the dimension of $\W$ does not change the value of the norm, so long
as $\dim(\W)\geq\dim(\X)$.)
An analogous theorem to the Holevo--Helstrom theorem establishes that the
\emph{diamond norm distance}
\begin{equation}
  \frac{1}{2} \norm{\Phi_0 - \Phi_1}_{\diamond}
\end{equation}
between two channels $\Phi_0,\Phi_1\in\Channel(\X,\Y)$ describes the maximum
bias with which a physical process (consisting of an arbitrary state
preparation, followed by a channel evaluation, followed by a measurement)
can distinguish between $\Phi_0$ and $\Phi_1$ given with equal probability.

As for states, it is sometimes convenient to refer to one channel as being an
\emph{$\varepsilon$-approximation} to another when the diamond norm distance
between the two channels is bounded from above by $\varepsilon$.

\begin{definition}
  \label{definition:channel-approximation}
  Let $\Phi$ and $\Psi$ be channels sharing the same input spaces and the same
  output spaces.
  It is said that $\Phi$ is an \emph{$\varepsilon$-approximation} to $\Psi$ if
  \begin{equation}
    \frac{1}{2}\bignorm{\Psi - \Phi}_{\Diamond} \leq \varepsilon.
  \end{equation}
\end{definition}

The following alternate characterization of the diamond norm will prove useful:
if $\Phi\in\Trans(\X,\Y)$ is a map specified by
\begin{equation}
  \Phi(X) = \Tr_\Z(A_0 X A_1^{\ast})
\end{equation}
for all $X\in\Lin(\X)$, for operators $A_0, A_1\in \Lin(\X,\Y\otimes\Z)$, then 
\begin{equation}
  \norm{\Phi}_{\diamond} = 
  \max_{\rho_0,\rho_1\in\Density(\X)}
  \fid\big(\Psi_0(\rho_0),\Psi_1(\rho_1)\big),
\end{equation}
where
\begin{equation}
  \Psi_0(X) = \Tr_\Y(A_0 X A_0^{\ast})
  \quad\text{and}\quad
  \Psi_1(X) = \Tr_\Y(A_1 X A_1^{\ast}).
\end{equation}

\section{Quantum circuits} 
\label{sec:circuits}

The primary model of computation used throughout this survey is the
\emph{quantum circuit} model.
A quantum circuit is an acyclic network of \emph{quantum gates}
connected by \emph{wires}.
The quantum gates represent quantum channels while the wires represent
qubits on which the channels act.
In general, we allow the quantum channels implemented by the gates of a
quantum circuit to be potentially non-unitary, as first suggested by
Aharonov, Kitaev, and Nisan \cite{AharonovKN98}.
This general variant of the quantum circuit model has a fairly straightforward
connection to the more commonly used model of unitary quantum circuits, by
virtue of the Stinespring representation of channels, as will be discussed
shortly.

An example of a quantum circuit having three input qubits and two output qubits
is pictured in Figure~\ref{fig:quantum-circuit}.
In general, a quantum circuit may have $n$ input qubits and $m$ output qubits
for any choice of integers $n,m\geq 0$.
Such a circuit induces a quantum channel from $n$ qubits to $m$ qubits,
determined by composing the actions of the individual gates in the appropriate
way.
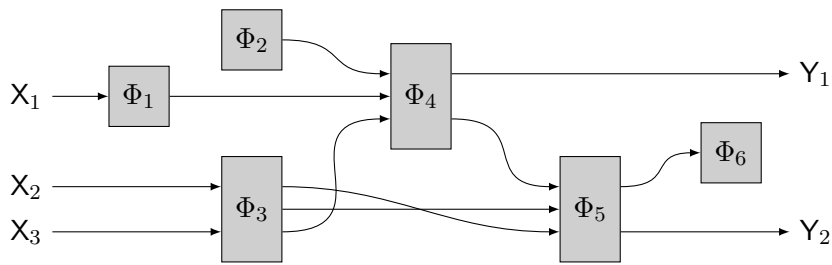
\begin{figure}
  \begin{center}
    \begin{tikzpicture}[scale=0.75,
        smallgate/.style={draw, minimum size=8mm,
          fill = ChannelColor, text=ChannelTextColor},
        biggate/.style={draw, minimum width=8mm, minimum height=14mm,
          fill = ChannelColor, text=ChannelTextColor},
        point/.style={inner sep=0pt, fill=none, draw=none},
        >=latex]
      
      \node (Phi2) at (-2,4) [smallgate] {$\Phi_2$};
      \node (Phi1) at (-4,3) [smallgate] {$\Phi_1$};
      \node (Phi3) at (-2,1) [biggate] {$\Phi_3$};
      \node (Phi4) at (1,3) [biggate] {$\Phi_4$};
      \node (Phi5) at (4,1) [biggate] {$\Phi_5$};
      \node (Phi6) at (6.5,2) [smallgate] {$\Phi_6$};

      \node (X1) at (-6,3) {$\reg{X}_1$};
      \node (X2) at (-6,1.4) {$\reg{X}_2$};
      \node (X3) at (-6,0.6) {$\reg{X}_3$};
      \node (Y1) at (8,3.4) {$\reg{Y}_1$};
      \node (Y2) at (8,0.6) {$\reg{Y}_2$};
      
      \draw[->] 
      (Phi2.0) 
      .. controls +(right:12mm) and +(left:12mm).. 
      ([yshift=4mm]Phi4.west);
      
      \draw[->] 
      (Phi1.east) 
      .. controls +(right:20mm) and +(left:20mm) .. 
      (Phi4.west);

      \draw[->] 
      ([yshift=-4mm]Phi3.east) 
      .. controls +(right:20mm) and +(left:20mm) .. 
      ([yshift=-4mm]Phi4.west);

      \draw[->] 
      (Phi3.east) 
      .. controls +(right:20mm) and +(left:20mm) .. 
      (Phi5.west);

      \draw[->] 
      ([yshift=-4mm]Phi4.east) 
      .. controls +(right:15mm) and +(left:15mm) .. 
      ([yshift=4mm]Phi5.west);

      \draw[->] 
      ([yshift=4mm]Phi3.east) 
      .. controls +(right:20mm) and +(left:20mm) .. 
      ([yshift=-4mm]Phi5.west);

      \draw[->] 
      (X1) 
      .. controls +(right:20mm) and +(left:20mm) ..
      (Phi1.west);
      
      \draw[->] (X2.east) -- ([yshift=4mm]Phi3.west);

      \draw[->] (X3.east) -- ([yshift=-4mm]Phi3.west);

      \draw[->] 
      ([yshift=4mm]Phi4.east) 
      .. controls +(right:20mm) and +(left:20mm) ..
      (Y1);

      \draw[->] 
      ([yshift=4mm]Phi5.east) 
      .. controls +(right:10mm) and +(left:15mm) ..
      (Phi6);

      \draw[->] ([yshift=-4mm]Phi5.east) -- (Y2);
    \end{tikzpicture}
  \end{center}
  \caption{An example of a quantum circuit.
    The input qubits are labelled $\reg{X}_1,\reg{X}_2,\reg{X}_3$, the output
    qubits are labelled $\reg{Y}_1$ and $\reg{Y}_2$, and the gates are labelled
    by (hypothetical) quantum channels $\Phi_1,\ldots,\Phi_6$.}
  \label{fig:quantum-circuit}
\end{figure}
The \emph{size} of a quantum circuit is the total number of gates plus
the total number of input and output qubits. 

Restrictions must be placed on the gates from which quantum circuits
may be composed if the quantum circuit model is to be used for
complexity theory---for without such restrictions it cannot be argued
that each quantum gate corresponds to an operation with unit cost.
For the remainder of this survey, quantum circuits may be assumed to
be composed of gates from the following list (representing a standard choice
for a gate set):
\begin{mylist}{\parindent}
\item[1.] \emph{Toffoli gates}.
  A Toffoli gate is a three-qubit unitary gate $\Phi_T$ identified
  with the unitary transformation
  \begin{equation}
    T:\ket{a}\ket{b}\ket{c} \mapsto \ket{a}\ket{b}\ket{c \oplus ab}.
  \end{equation}
\item[2.] \emph{Hadamard gates}.
  A Hadamard gate is a single-qubit unitary gate $\Phi_H$ identified
  with the unitary transformation
  \begin{equation}
    H:\ket{a} \mapsto \frac{1}{\sqrt{2}} \ket{0} + 
    \frac{(-1)^a}{\sqrt{2}} \ket{1}.
  \end{equation}
\item[3.] \emph{Phase-shift gates}.
  A Phase-shift gate is a single-qubit unitary gate $\Phi_P$ identified with
  the unitary transformation
  \begin{equation}
    P:\ket{a} \mapsto i^a \ket{a}.
  \end{equation}
\item[4.] \emph{Ancillary gates}.
  Ancillary gates are non-unitary gates that take no input and
  produce a single qubit in the state $\ket{0}$ as output.
\item[5.] \emph{Erasure gates}.
  Erasure gates are non-unitary gates that take a single qubit as
  input and produce no output.
  Their effect is represented by the partial trace on the qubit they
  take as input.  
\end{mylist}

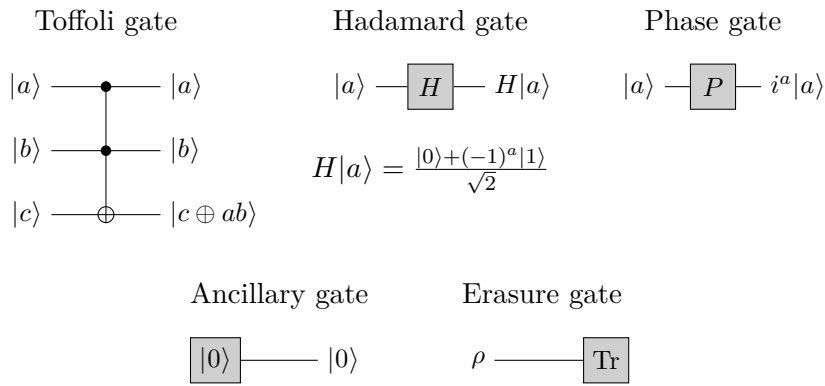
\begin{figure}
  \begin{center}
    \hfill
    \begin{minipage}[b][3cm][t]{0.28\textwidth}
      \begin{center}
        \begin{tikzpicture}[scale=0.85]
          \node at (0,2) {Toffoli gate};
          \small          
          \node (In1) at (-1,1) {\makebox(0,0)[r]{$\ket{a}$}};
          \node (In2) at (-1,0) {\makebox(0,0)[r]{$\ket{b}$}};
          \node (In3) at (-1,-1) {\makebox(0,0)[r]{$\ket{c}$}};
          \node (Out1) at (1,1) {\makebox(0,0)[l]{$\ket{a}$}};
          \node (Out2) at (1,0) {\makebox(0,0)[l]{$\ket{b}$}};
          \node (Out3) at (1,-1) {%
            \makebox(0,0)[l]{$\ket{c\oplus ab}$}};
          \node (Phantom1) at (-1.3,-1) {};
          \node (Phantom1) at (2.2,-1) {};
          \draw (In1) -- (Out1);
          \draw (In2) -- (Out2);
          \draw (In3) -- (Out3);
          \node[circle, fill, minimum size = 4pt, inner sep=0mm]
          (Control1) at (0,1) {};
          \node[circle, fill, minimum size = 4pt, inner sep=0mm]
          (Control2) at (0,0) {};
          \node[circle, draw, minimum size = 6pt, inner sep=0mm]
          (Not) at (0,-1) {};
          \draw (Not.north) -- (Not.south);
          \draw (Control1.north) -- (Not.north);
        \end{tikzpicture}
        \vspace{2mm}
      \end{center}
    \end{minipage}
    \hfill
    \begin{minipage}[b][3cm][t]{0.28\textwidth}
      \begin{center}
        \begin{tikzpicture}[scale=0.85]
          \node at (0,1) {Hadamard gate};
          \small
          \node (In1) at (-1,0) {\makebox(0,0)[r]{$\ket{a}$}};
          \node (Out1) at (1,0) {\makebox(0,0)[l]{$H\ket{a}$}};
          \node (Phantom1) at (-1.6,0) {};
          \node (Phantom1) at (1.7,0) {};
          \node (H) at (0,0) [draw, fill = ChannelColor, 
            text=ChannelTextColor, minimum size=6mm] {$H$};
          \draw (In1) -- (H) -- (Out1);
          \node (Equation) at (0,-1.3) {\makebox(0,0){\normalsize
              $H\ket{a} = \frac{\ket{0} + (-1)^a\ket{1}}{\sqrt{2}}$}};
        \end{tikzpicture}
      \end{center}
    \end{minipage}
    \hfill
    \begin{minipage}[b][3cm][t]{0.28\textwidth}
      \begin{center}
        \begin{tikzpicture}[scale=0.85]
          \node at (0,1) {Phase gate};
          \small
          \node (In1) at (-0.9,0) {\makebox(0,0)[r]{$\ket{a}$}};
          \node (Out1) at (0.9,0) {\makebox(0,0)[l]{$i^a\ket{a}$}};
          \node (Phantom1) at (1.7,-0.5) {};
          \node (P) at (0,0) [draw, fill = ChannelColor, 
            text=ChannelTextColor, minimum size=6mm] {$P$};
          \draw (In1) -- (P) -- (Out1);
        \end{tikzpicture}
      \end{center}
    \end{minipage}\hfill\\[6mm]
    \begin{minipage}[b][1cm][t]{0.28\textwidth}
      \begin{center}
        \begin{tikzpicture}[scale=0.85]
          \node at (0,1) {Ancillary gate};
          \small          
          \node (Out) at (1,0) {$\ket{0}$};
          \node (Ancilla) at (-1,0) [draw, fill = ChannelColor, 
            text=ChannelTextColor, minimum size=6mm] {$\ket{0}$};
          \node (Phantom1) at (-1.5,0) {};
          \node (Phantom1) at (1.5,0) {};
          \draw (Ancilla) -- (Out);
        \end{tikzpicture}
        \vspace{2mm}
      \end{center}
    \end{minipage}
    \begin{minipage}[b][1cm][t]{0.28\textwidth}
      \begin{center}
        \begin{tikzpicture}[scale=0.85]
          \node at (0,1) {Erasure gate};
          \small
          \node (In) at (-1,0) {$\rho$};
          \node (Tr) at (1,0) [draw, fill = ChannelColor, 
            text=ChannelTextColor, minimum size=6mm] {$\op{Tr}$};
          \node (Phantom1) at (-1.5,0) {};
          \node (Phantom1) at (1.5,0) {};
          \draw (In) -- (Tr);
        \end{tikzpicture}
      \end{center}
    \end{minipage}
  \end{center}
  \caption{Commonly used notation for denoting gates from the universal gate
    set.}
  \label{figure:universal-gates}
\end{figure}

\noindent
We note that it is not essential that one chooses this particular set of
gates, and we will not often refer specifically to these gates in this
survey---but it is convenient nevertheless to assume that reversible
computations and Hadamard gates can be performed without error.
Figure~\ref{figure:universal-gates} illustrates the notation that is commonly
used to describe these gates within quantum circuits.

The above gate set is \emph{universal} in a strong sense: every
quantum channel mapping qubits to qubits can be approximated to
within any desired degree of accuracy by some quantum circuit composed
of gates from this set.
The following theorem expresses this fact in more precise terms.

\begin{theorem}[Universality Theorem] \label{theorem:universality}
  Let $\Phi$ be any quantum channel from $n$ qubits to $m$ qubits. 
  For every $\varepsilon>0$ there exists a quantum circuit $Q$
  with $n$ input qubits and $m$ output qubits that implements a
  $\varepsilon$-approximation to $\Phi$.
  Moreover, for a fixed choice of $n$ and $m$, the circuit $Q$ may be taken to
  satisfy $\op{size}(Q) = \op{poly}(\log(1/\varepsilon))$.
\end{theorem}

A \emph{unitary quantum circuit} is a quantum circuit in which all of the gates
correspond to unitary quantum channels, so that the channel associated to the
entire circuit is therefore unitary as well.
Naturally this requires that every gate, and hence the circuit itself, has an
equal number of input and output qubits.
It is common in the study of quantum computing that one works entirely with
unitary quantum circuits. 
The equivalence between the unitary and general models of quantum circuits is
made straightforward by the universal gate set described above.
Suppose $Q$ is a quantum circuit taking input qubits
$(\reg{X}_1,\ldots,\reg{X}_n)$ and producing output qubits
$(\reg{Y}_1,\ldots,\reg{Y}_m)$, and assume there are $j$ ancillary
gates and $k$ erasure gates among the gates of~$Q$.
A new quantum circuit $R$ may then be formed by removing the ancillary
and erasure gates, and to account for the removal of these gates the
circuit $R$ takes $j$ additional input qubits
$(\reg{Z}_1,\ldots,\reg{Z}_j)$ and produces $k$ additional
output qubits $(\reg{W}_1,\ldots,\reg{W}_k)$.
Figure~\ref{fig:purification} illustrates this simple process.
\begin{figure}
  \begin{center}
    \begin{minipage}[b][3.6cm][t]{0.47\textwidth}
      \begin{center}
      \begin{tikzpicture}[scale=0.7,
          smallgate/.style={draw, minimum size=6mm,
            fill = ChannelColor, text=ChannelTextColor, inner sep=0mm},
          control/.style={circle, fill, minimum size = 4pt, inner sep=0mm},
          target/.style={circle, draw, minimum size = 6pt, inner sep=0mm},
          >=latex]
        \small

        \node (X1) at (-3.5,2) {$\reg{X}_1$};
        \node (X2) at (-3.5,1) {$\reg{X}_2$};
        \node (X3) at (-3.5,0) {$\reg{X}_3$};
        
        \node (Y1) at (3.5,2) {$\reg{Y}_1$};
        \node (Y2) at (3.5,1) {$\reg{Y}_2$};
        
        \node (A1) at (-2.5,-1) [smallgate] {$\ket{0}$};
        \node (A2) at (-1.5,-2) [smallgate] {$\ket{0}$};

        \node (Tr1) at (2.5,0) [smallgate] {\textup{Tr}};
        \node (Tr2) at (0.5,-1) [smallgate] {\textup{Tr}};
        \node (Tr3) at (2.5,-2) [smallgate] {\textup{Tr}};

        \draw (X1) -- (Y1);
        \draw (X2) -- (Y2);
        \draw (X3) -- (Tr1);
        \draw (A1) -- (Tr2);
        \draw (A2) -- (Tr3);

        \node at (-0.5,2) [smallgate] {$H$};
        \node at (-0.5,1) [smallgate] {$P$};
        \node at (0.5,-2) [smallgate] {$H$};
        \node at (1.5,2) [smallgate] {$H$};
        
        \node (C11) at (-1.5,2) [control] {};
        \node (C12) at (-1.5,1) [control] {};
        \node (T13) at (-1.5,-1) [target] {};
        \draw (T13.south) -- (C11);

        \node (C21) at (-0.5,0) [control] {};
        \node (C22) at (-0.5,-1) [control] {};
        \node (T23) at (-0.5,-2) [target] {};
        \draw (T23.south) -- (C21);

        \node (C31) at (0.5,1) [control] {};
        \node (C32) at (0.5,0) [control] {};
        \node (T33) at (0.5,2) [target] {};
        \draw (T33.north) -- (C32);

        \node (C41) at (1.5,0) [control] {};
        \node (C42) at (1.5,-2) [control] {};
        \node (T43) at (1.5,1) [target] {};
        \draw (T43.north) -- (C42);

      \end{tikzpicture}
      \end{center}
    \end{minipage}
    \hspace{2mm}
    \begin{minipage}[b][3.6cm][t]{0.47\textwidth}
      \begin{center}
      \begin{tikzpicture}[scale=0.7,
          smallgate/.style={draw, minimum size=6mm,
            fill = ChannelColor, text=ChannelTextColor},
          control/.style={circle, fill, minimum size = 4pt, inner sep=0mm},
          target/.style={circle, draw, minimum size = 6pt, inner sep=0mm},
          >=latex]
        \small
        \node (X1) at (-3,2) {$\reg{X}_1$};
        \node (X2) at (-3,1) {$\reg{X}_2$};
        \node (X3) at (-3,0) {$\reg{X}_3$};
        \node (Z1) at (-3,-1) {$\reg{Z}_1$};
        \node (Z2) at (-3,-2) {$\reg{Z}_2$};
        
        \node (Y1) at (3,2) {$\reg{Y}_1$};
        \node (Y2) at (3,1) {$\reg{Y}_2$};
        \node (W1) at (3,0) {$\reg{W}_1$};
        \node (W2) at (3,-1) {$\reg{W}_2$};
        \node (W3) at (3,-2) {$\reg{W}_3$};

        \draw (X1) -- (Y1);
        \draw (X2) -- (Y2);
        \draw (X3) -- (W1);
        \draw (Z1) -- (W2);
        \draw (Z2) -- (W3);

        \node at (-0.5,2) [smallgate] {$H$};
        \node at (-0.5,1) [smallgate] {$P$};
        \node at (0.5,-2) [smallgate] {$H$};
        \node at (1.5,2) [smallgate] {$H$};
        
        \node (C11) at (-1.5,2) [control] {};
        \node (C12) at (-1.5,1) [control] {};
        \node (T13) at (-1.5,-1) [target] {};
        \draw (T13.south) -- (C11);

        \node (C21) at (-0.5,0) [control] {};
        \node (C22) at (-0.5,-1) [control] {};
        \node (T23) at (-0.5,-2) [target] {};
        \draw (T23.south) -- (C21);

        \node (C31) at (0.5,1) [control] {};
        \node (C32) at (0.5,0) [control] {};
        \node (T33) at (0.5,2) [target] {};
        \draw (T33.north) -- (C32);

        \node (C41) at (1.5,0) [control] {};
        \node (C42) at (1.5,-2) [control] {};
        \node (T43) at (1.5,1) [target] {};
        \draw (T43.north) -- (C42);
      \end{tikzpicture}
      \end{center}
    \end{minipage}
  \end{center}
  \caption{A general quantum circuit (left) and its unitary
    purification (right).}
  \label{fig:purification}
\end{figure}
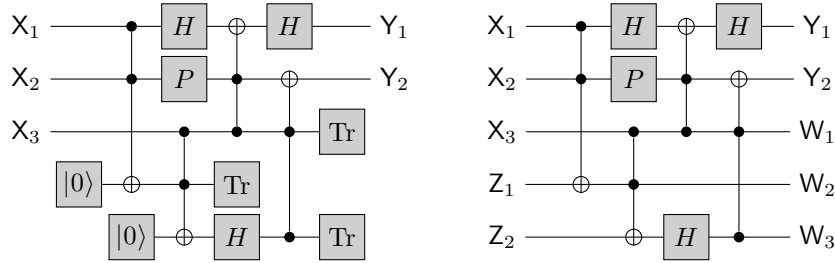
The circuit $R$ is said to be a \emph{unitary purification} of $Q$.
\label{sec:circuit-purification}
It holds that $R$ is equivalent to $Q$, provided the qubits 
$(\reg{Z}_1,\ldots,\reg{Z}_j)$ are each initially set to the $\ket{0}$
state and the qubits $(\reg{W}_1,\ldots,\reg{W}_k)$ are ignored after
the circuit is run.
Despite the simplicity of this process, it is often useful to refer to
unitary purifications of general quantum circuits obtained in this way.

Along similar lines, one may consider \emph{isometric quantum circuits}, which
are quantum circuits composed of unitary gates and ancillary gates, but no
erasure gates.
Such circuits implement isometric quantum channels.
By performing the process described above, but only for the erasure gates
of a general quantum circuit and not the ancillary gates, one obtains an
isometric extension of a general quantum circuit.

Any quantum circuit formed from the universal gate set described above can be
\emph{encoded} as a binary string, with respect to any number of different
possible encoding schemes.
As is the case when uniform families of classical Boolean circuits are studied,
many specific details of such encoding schemes are not important; for the
sake of brevity we will leave it to the reader to imagine that a sensible and
efficient encoding scheme for quantum circuits has been fixed.
Naturally it is assumed that a circuit's size and its encoding length are
polynomially related.

\chapter{Non-Interactive Quantum Proofs} \label{chapter:QMA}

The traditional notion of a proof in mathematics does not require an
interaction to take place between an individual proving a theorem and one
who verifies the proof (beyond the obvious requirement that the individual
verifying the proof has come into possession of it).
The same is true in a complexity-theoretic setting, in which proofs are
typically abstracted as strings of symbols to be checked by computationally
efficient verification procedures.
The standard definition of \class{NP} in terms of polynomial-length proofs
(or \emph{certificates}), checked by polynomial-time deterministic
computations, is representative of this traditional notion.

One natural way to generalize this notion to the quantum setting is to allow a
proof to be a \emph{quantum state} rather than a classical string of symbols.
Such a state is to be verified by a computationally efficient procedure, as
in the classical setting, but in this case the verification procedure will be
an efficient \emph{quantum} computation.
The most natural and well-studied complexity class to be defined through
this notion is \class{QMA}, which stands for ``quantum Merlin--Arthur.''
Rather than being a direct quantum variant of \class{NP}, the class \class{QMA}
is more accurately described as being a quantum computational analogue of the
complexity class \class{MA}, which is essentially \class{NP} with a
bounded-error polynomial-time probabilistic verifier rather than a
deterministic one.

The study of the class \class{QMA} provides a complexity-theoretic perspective
on the properties of quantum states and their relative power compared to
classical states, when seen as untrusted advice to be given to a quantum or
classical verifier respectively. 
The following points are among those to be discussed in this chapter:

\begin{mylist}{\parindent}
\item[1.] In Section~\ref{sec:gnm} the \emph{group non-membership} problem is
  shown to be included in $\class{QMA}$.
  This problem illustrates the potential advantages of quantum over classical
  proofs.
  
\item[2.] Two procedures for error reduction---\emph{parallel}
  error reduction and \emph{witness-preserving} error reduction--- for
  \class{QMA} are presented in Section~\ref{section:QMA-error-reduction}. 
  One consequence of witness-preserving error reduction is the inclusion
  $\class{QMA}\subseteq\class{PP}$, providing an upper-bound on \class{QMA}
  in terms of a well-studied classical complexity class.

\item[3.] A few complete problems for $\class{QMA}$ are introduced in
  Section~\ref{sec:QMA-complete}.
  Some of these problems, such as the \emph{local Hamiltonian problem},
  illustrate the relevance of \class{QMA} for natural problems that arise in
  the study of quantum systems in physics or chemistry.
  
\item[4.] A selection of variants of \class{QMA} is presented, including
  ones with additional promises on the proof, such as it being \emph{trusted},
  \emph{classical}, or \emph{separable} with respect to some fixed partition.
  The study of these variants probes some of the essential features of the
  class \class{QMA}.
  Some of these variants are shown to be equivalent to \class{QMA}, while for
  others there is evidence that they may differ from \class{QMA}.
\end{mylist}

\section{Definitions of quantum verification procedures and \class{QMA}}
\label{section:QMA-definition}

This section of the chapter is primarily concerned with the definition of the
class \class{QMA}, along with some of its most basic mathematical aspects, such
as its relationship to measurements and quantum circuits.
A computational problem known as the \emph{group non-membership problem} is
shown to be contained in \class{QMA}, providing a simple example of how quantum
proofs may be useful in a computational setting.

\subsection{Definition of \class{QMA}}

With the intuitive picture of a quantum proof that has been suggested above in
mind, we wish to formalize the notion of an efficient quantum verification
procedure that takes as input a quantum state, playing the role of a proof, and
outputs a single binary value, indicating acceptance or rejection of the
proof.
A natural way to model such a procedure is as a quantum circuit that takes $k$
qubits as input, where $k$ denotes the length of the proof, and produces
a single output qubit.
Rather than stipulating that this output qubit must represent a classical
state, indicating whether acceptance or rejection has occurred, we will simply
assume that the qubit is to be measured with respect to the
computational basis after being output by the circuit.
The outcome of this measurement will indicate whether the proof has been
accepted or rejected (with 1 indicating acceptance, 0 indicating rejection).

Suppose, somewhat more generally, that $\Phi$ is an arbitrary channel having
$k$ input qubits and $1$ output qubit, and consider the following scenario.
An individual (the \emph{prover}) aims to provide $\Phi$ with an input state
$\rho$ on $k$ qubits maximizing the probability that a standard basis
measurement performed on the output qubit of $\Phi$ produces the outcome~1.
In effect, the channel $\Phi$ followed by a standard basis measurement of its
output qubit describes a general binary-valued measurement on the $k$ qubits
input to $\Phi$. 
Indeed, it is straightforward to specify the measurement operators
$P_0, P_1 \in \Pos\big((\complex^2)^{\otimes k}\big)$ associated with such a
measurement, which are
\begin{equation}\label{eq:p0-p1-def}
  P_0 = \Phi^{\ast}(\ket{0}\bra{0})
  \quad\text{and}\quad
  P_1 = \Phi^{\ast}(\ket{1}\bra{1}),
\end{equation}
where $\Phi^{\ast}$ denotes the adjoint mapping to $\Phi$.

With this scenario in mind, we define the \emph{value}\footnote{
  The term \emph{value} is not a standard term in this particular setting---but
  we use it nevertheless, as it is analogous to the standard usage of this
  term in the context of other models to be considered in subsequent chapters
  of this survey.}
of $\Phi$ as
\begin{equation}
  \omega(\Phi) = \max_{\rho}\,\bra{1} \Phi(\rho) \ket{1},
\end{equation}
where the maximum is over all density operators $\rho$ on $k$ qubits.
In the case that $\Phi$ is the channel described by a circuit functioning as 
a verification procedure, the value of $\Phi$ is the maximum probability with
which a quantum proof may lead this procedure to accept. 
The value coincides with the largest eigenvalue of the measurement operator
$P_1$ defined above, as a short calculation reveals:
\begin{equation}
  \omega(\Phi) = \max_\rho \bra{1} \Phi(\rho) \ket{1}
  = \max_{\rho} \bigip{\rho}{\Phi^{\ast}(\ket{1}\bra{1})}
  = \lambda_1(P_1).
\end{equation}
(In general, we write $\lambda_1(H),\,\lambda_2(H),\,\ldots,\,\lambda_n(H)$ to
denote the eigenvalues of a Hermitian operator $H$, sorted from largest to
smallest, so $\lambda_1(P_1)$ denotes the largest eigenvalue of $P_1$.) 
A prover wishing to maximize the probability of obtaining the outcome 1
may as well take $\rho$ to be a pure state $\rho = \ket{\psi}\bra{\psi}$
for $\ket{\psi}$ being any unit eigenvector corresponding to the largest
eigenvalue of~$P_1$.

With the definition of the value of a channel in hand, we may define the
class \class{QMA} as follows:

\begin{definition}
  A promise problem $A = (A_{\yes},A_{\no})$ is contained in
  the complexity class $\class{QMA}_{a,b}$ if there exists a
  polynomial-time computable function $V$ possessing the following properties:
  \begin{mylist}{\parindent}
  \item[1.]
    For every string $x\in A_{\yes}\cup A_{\no}$, one has that $V(x)$ is an
    encoding of a quantum circuit implementing a channel $\Phi_x$ having
    $1$ output qubit.
  \item[2.]
    \emph{Completeness}. For every string $x\in A_{\yes}$, the value of the
    channel $\Phi_x$ satisfies $\omega(\Phi_x) \geq a$.
  \item[3.]
    \emph{Soundness}. For every string $x\in A_{\no}$, the value of the channel
    $\Phi_x$ satisfies $\omega(\Phi_x) \leq b$.
  \end{mylist}
\end{definition}

\noindent
In this definition, $a,b\in[0,1]$ may be constant values or functions of the
length of the input string $x$.
When they are omitted, it is to be assumed that they are $a = 2/3$ and
$b = 1/3$:
\begin{equation}
  \class{QMA} = \class{QMA}_{\frac{2}{3},\frac{1}{3}}\,.
\end{equation}
As usual, the bounds $2/3$ and $1/3$ on the maximum probability of the
verifier outputting 1 are taken as being representative of statistically 
distinguishable experiments.
Methods for reducing errors in quantum verification procedures are discussed in
the section following this one, and these methods will allow the completeness
and soundness bounds $a$ and $b$ to be taken as any functions exponentially
close to $1$ and $0$ respectively. 
It is not known whether the completeness parameter can always be taken to equal
$1$ without changing the complexity class that results;
a one-sided variant of $\class{QMA}$, denoted
\begin{equation}
  \class{QMA}_1 \,=\, \class{QMA}_{1,\frac{1}{3}}\,,
\end{equation}
clearly satisfies $\class{QMA}_1\subseteq\class{QMA}$,
but the two classes are not known to be equal.

We will sometimes identify the circuit encoding $V(x)$ with the channel that it
implements in a self-explanatory way, writing $\omega(V(x))$ to mean
$\omega(\Phi_x)$ for $\Phi_x$ being the channel implemented by $V(x)$.
Along similar lines, we may write $\omega(V)$ to refer to the function whose
value is $\omega(V(x)) = \omega(\Phi_x)$ for each input string $x$.

It will be instructive and useful later to consider the actions of a
unitary circuit $Q$ that purifies a circuit implementation of a channel $\Phi$
as above.
As illustrated in Figure~\ref{fig:unitary-verification}, such a circuit will
take as input two registers:
a $k$-qubit register $\reg{X}$, which initially contains the state $\rho$
representing the quantum proof, along with an $m$-qubit register $\reg{Y}$
initially containing the pure state $\ket{0^m}$, which represents the so-called
\emph{ancillary qubits} used by the circuit $Q$.
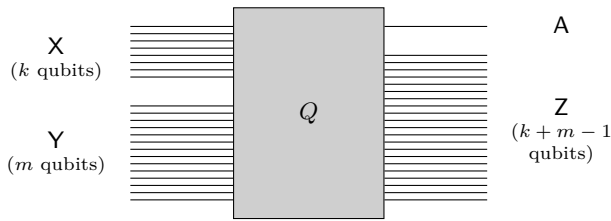
\begin{figure}
  \begin{center}
    \footnotesize
    \begin{tikzpicture}[scale=0.48, 
        circuit/.style={draw, minimum height=28mm, minimum width=20mm,
          fill = ChannelColor, text=ChannelTextColor},
        empty/.style={minimum width=10mm},
        >=latex]
      
      \node (Q) at (0,0) [circuit] {$Q$};
      \node (in) at (-6,0) [empty] {};
      \node (out) at (6,0) [empty] {};

      \foreach \y in {-24,-22,...,2} {
        \draw ([yshift=\y mm]in.east) -- ([yshift=\y mm]Q.west) {};
      }

      \foreach \y in {24,22,...,10} {
        \draw ([yshift=\y mm]in.east) -- ([yshift=\y mm]Q.west) {};
      }

      \foreach \y in {-24,-22,...,16} {
        \draw ([yshift=\y mm]Q.east) -- ([yshift=\y mm]out.west) {};
      }

      \draw ([yshift=24mm]Q.east) -- ([yshift=24mm]out.west) {};
      
      \node at (-7,1.5) {\begin{tabular}{c}
          $\reg{X}$\\[-1mm]
          \scriptsize{($k$ qubits)}
        \end{tabular}};
      
      \node at (-7,-1.1) {\begin{tabular}{c}
          $\reg{Y}$\\[-1mm]
          \scriptsize{($m$ qubits)}
        \end{tabular}};
      
      \node at (7,-0.5) {\begin{tabular}{c}
          $\reg{Z}$\\[-1mm]
          \scriptsize{($k+m-1$}\\[-1mm]
          \scriptsize{qubits)}
        \end{tabular}};
      
      \node at (7,2.45) {$\reg{A}$};

    \end{tikzpicture}
  \end{center}
  \caption{A unitary circuit implementing a verification procedure.}
  \label{fig:unitary-verification}
\end{figure}
(As explained in Section~\ref{sec:circuit-purification} the value of $m$ will
be at most linear in the number of gates required by the circuit implementation
of~$\Phi$.)
The output qubit of $\Phi$ will be named $\reg{A}$, and the remaining $k+m-1$
qubits output by $Q$ comprise a register~$\reg{Z}$.

\subsection{Example: group non-membership}
\label{sec:gnm}

We will now describe an example of a computational problem, known as the
\emph{group non-membership problem}, that illustrates one potential way in
which quantum proofs and verification procedures may gain advantages over
classical proofs and verification procedures.\footnote{%
  Although we do not go into the details, from the description in this section
  it will be evident that the group non-membership problem can be used to
  provide an \emph{oracle separation} between the complexity classes \class{MA}
  and \class{QMA}.}
The group non-membership problem is perhaps most naturally described within
the \emph{black-box group} setting \cite{BabaiS84}.
Here, one considers that there is an underlying finite group $G_n$ that has
been specified for each positive integer $n$, and elements of $G_n$ are encoded
as binary strings of length $n$ (so that it must necessarily hold that
$\abs{G_n}\leq 2^n$).
Not every string needs to encode a group element, but we will make the
assumption (which is not always in place in the black-box group setting) that
each group element has a unique binary string encoding.
A \emph{group oracle} is made available to perform the two group operations
at unit cost, and in the quantum setting one assumes that the group oracle
functions reversibly.
For example, the group oracle may operate in the following way:
\begin{equation}
  \ket{g}\ket{h}\ket{b} \mapsto
  \begin{cases}
    \ket{h g}\ket{h}\ket{b} & \text{if $b = 0$}\\[1mm]
    \ket{h^{-1} g}\ket{h}\ket{b} & \text{if $b = 1$},
  \end{cases}
\end{equation}
assuming that $b\in\{0,1\}$ and $g, h \in G_n$ are identified with their
$n$-bit encodings.
(We may assume that the group oracle acts as the identity operator when
given a string that does not encode a group element.)
When we consider the group non-membership problem below, it is to be assumed
that group elements are given as $n$-bit strings and the group operations
are determined by a fixed group oracle.

Positive results in the setting of black-box groups generally imply analogous
positive results in concrete settings in which the group oracle can be
implemented algorithmically.
For the particular case at hand, the fact that the group non-membership problem
is in \class{QMA} for black-box groups implies that it is also in \class{QMA}
for concrete realizations of groups for which the unique encoding assumption is
met and for which the group operations can be implemented efficiently.
Matrix groups over finite fields represent a fairly general class of
examples in this category.

The group non-membership problem is as follows:\vspace{1mm}
\begin{center}
  \begin{minipage}{0.95\textwidth}
    \begin{center}
      \underline{Group non-membership (GNM)}\\[2mm]
      \begin{tabular}{lp{0.8\textwidth}}
        \emph{Input:} & 
        Group elements $g_1,\ldots,g_m\in G_n$ and $a\in G_n$ (for some choice
        of $n$).\\[1mm]
        \emph{Yes:} &
        $a\not\in \langle g_1,\ldots, g_m\rangle$.\\[1mm]
        \emph{No:} &
        $a\in \langle g_1,\ldots, g_m\rangle$.
      \end{tabular}
    \end{center}
  \end{minipage}
\end{center}

\noindent
Here, the notation $\langle g_1,\ldots, g_m\rangle$ means the subgroup of $G_n$
generated by the elements $g_1,\ldots,g_m$.

Some might argue that the \emph{subgroup non-membership problem} would be
a more fitting name than the \emph{group non-membership problem}, as the
problem concerns membership in the subgroup $\langle g_1,\ldots, g_m\rangle$
rather than membership in the group $G$.
It is, however, reasonable to view that $\langle g_1,\ldots, g_m\rangle$ is
the group of interest in this problem, while $G$ is a supergroup that
happens to contain $g_1,\ldots, g_m$ and $a$.

Before discussing quantum proofs and verification procedures for this problem,
it is fitting to mention what is known in the classical setting.
It is known how to design an efficient classical verification procedure
for the complementary problem to GNM, in which the yes-instances are those with
$a\in \langle g_1,\ldots, g_m\rangle$ and the no-instances are those with 
$a\not\in \langle g_1,\ldots, g_m\rangle$.
Intuitively speaking, a short classical proof that 
$a\in \langle g_1,\ldots, g_m\rangle$ may consist of a list of instructions for
obtaining $a$ from $g_1,\ldots,g_m$ through the group operations.
(One cannot simply give a sequence of elements selected from 
the set $\{g_1,\ldots,g_m,g_1^{-1},\ldots,g_m^{-1}\}$ whose product is $a$,
because such a list might need to be as long as the size of $G_n$ itself---but a
so-called \emph{straight-line program} can be used instead.
The \emph{reachability lemma} of Babai and Szemer\'edi \cite{BabaiS84}
guarantees that a short straight-line program to generate $a$ from
$g_1,\ldots,g_m$ must always exist when $a\in\langle g_1,\ldots,g_m\rangle$.)
For some groups, including permutation groups, there exist efficient classical
verification procedures for the GNM problem, but in the black-box group setting
one can prove that no efficient classical verification procedure exists.

Using quantum proofs, however, the solution becomes elementary.
A quantum proof certifying that a given element $a$ is not contained in the
subgroup $K = \langle g_1,\ldots, g_m\rangle$, for a quantum verification
procedure to be described shortly, is the state
\begin{equation}\label{eq:subgroup-witness}
  \ket{K} = \frac{1}{\sqrt{\abs{K}}} \sum_{g\in K} \ket{g},
\end{equation}
i.e., a uniform superposition over the elements in $K$.
This state is independent of $a$, and will function correctly as a proof
that $a\not\in K$ for all such choices of $a$.

Now, if one truly had a copy of the state $\ket{K}$, it would not be difficult
to test membership in $K$ with bounded, one-sided error. 
If it is the case that $h\in K$ for some group element $h$, then the state
\begin{equation}
  \ket{h K} =  \frac{1}{\sqrt{\abs{K}}} \sum_{g\in K} \ket{hg}
\end{equation}
satisfies $\ket{h K} = \ket{K}$.
On the other hand, if $h\not\in K$, then $\ket{h K} \perp \ket{K}$. 
The following test, which we call the \emph{controlled-unitary test},\footnote{
  We are not aware of a standard name for this test, and have selected a name
  for the sake of convenience.
  Irrespective of the name, it is a very commonly used test in quantum
  algorithms and complexity \cite{CleveEMM98}, and can be viewed as a
  low-precision form of the \emph{eigenvalue estimation} procedure associated
  with Shor's algorithms for factoring and computing discrete logarithms.
} can be used to distinguish between
the two cases:
\begin{center}
  \underline{Controlled-unitary test}\\[2mm]
  \begin{tabular}{lp{0.78\textwidth}}
      \emph{Given:} & An $n$-qubit state $\ket{\psi}$ and a quantum 
      circuit specifying an $n$-qubit unitary $U$.\\[1mm]
      \emph{Outcome:} & A classical bit that is $0$ with probability 
      $$p=\frac{1+\Re({\bra{\psi}U\ket{\psi}})}{2}$$ and $1$ with probability
      $1-p$.\\[1mm]
      \emph{Procedure:} & See Figure~\ref{fig:control-u}. 
  \end{tabular}
\end{center}
The circuit described in Figure~\ref{fig:control-u} implements the
controlled-unitary test.
The measurement illustrated in the figure is a standard basis measurement.

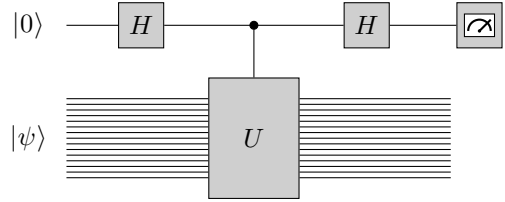
\begin{figure}
  \begin{center}
    \begin{tikzpicture}[scale=0.75]
      \small
      \node[minimum size = 6mm, draw, fill=ChannelColor, text=ChannelTextColor]
      (H) at (-2,1) {$H$};
      
      \node[minimum size = 6mm, draw, fill=ChannelColor, text=ChannelTextColor]
      (QFT) at (2,1) {$H$};
      
      \node[rounded corners = 1.5pt, draw, fill = Black, 
        inner sep = 0pt, minimum width = 3pt, minimum height = 3pt]
      (control) at (0,1) {};
      
      \node[minimum height = 1.25cm, minimum width = 1cm]
      (InTop) at (-4,1) {$\ket{0}$};
      
      \node[minimum size = 6mm, draw, fill=ChannelColor]
      (Measure) at (4,1) {};
      
      \node[minimum height = 1.25cm, minimum width = 1cm]
      (InBottom) at (-4,-1) {$\ket{\psi}$};
      
      \node[minimum height = 1.25cm, minimum width = 0.75cm]
      (OutBottom) at (4,-1) {};
      
      \draw (H.east) -- (QFT.west) {};
      \draw (InTop.east) -- (H.west) {};
      \draw (QFT.east) -- (Measure.west) {};
      
      \node[minimum height = 16mm, minimum width = 12mm, draw,
        fill=ChannelColor, text=ChannelTextColor]
      (U) at (0,-1) {$U$};
      
      \foreach \y in {-7,-6,-5,-4,-3,-2,-1,0,1,2,3,4,5,6,7} {
        \draw ([yshift=\y mm]InBottom.east) -- ([yshift=\y
          mm]U.west) {};
        \draw ([yshift=\y mm]U.east) -- ([yshift=\y mm]OutBottom.west) {};
      }
      
      \node[draw, minimum width=4mm, minimum height=3mm, fill=ReadoutColor]
      at (4,1) {};
      
      \draw[thick] (4.2,0.9) arc (0:180:2mm);
      \draw[thick] (4,0.9) -- (4.2,1.1);
      \draw[fill] (4,0.9) circle (0.2mm);
      
      \draw (U.north) -- (control.south);
      
    \end{tikzpicture}
  \end{center}
  \caption{Circuit implementing the controlled-unitary test on state
    $\ket{\psi}$ and unitary $U$.}
  \label{fig:control-u}
\end{figure}

In the circuit, the unitary $U$ is controlled by the top qubit, so that the
action of the entire controlled-unitary gate is as follows:
\begin{equation}
   \ket{0}\ket{\psi} \mapsto \ket{0}\ket{\psi}
   \quad\text{and}\quad
   \ket{1}\ket{\psi} \mapsto \ket{1}U\ket{\psi}.
\end{equation}
After applying the second Hadamard gate, the state of the $n+1$ qubits is
\begin{equation}\label{eq:ctl-u}
 \frac{\ket{0}+\ket{1}}{2}\otimes\ket{\psi} + 
 \frac{\ket{0}-\ket{1}}{2}\otimes U\ket{\psi},
\end{equation}
and measuring the top qubit with respect to the computational basis produces the
outcome $0$ with probability
\begin{equation}
  \frac{1}{4}\big\| \ket{\psi} + U\ket{\psi} \big\|^2 
  \,=\, \frac{1 + \Re({\bra{\psi}U\ket{\psi}})}{2}.
\end{equation}

In the setting of the group non-membership problem, the controlled-unitary test
will be applied to the unitary operation $M_a$ that multiplies (on the left) by
$a$:
\begin{equation}
  M_a \ket{g} = \ket{ag}.
\end{equation}
This is a reversible (and therefore unitary) operation that can be implemented
efficiently using the group oracle (or under the assumption that the group
operations can be computed efficiently).
By executing the controlled-unitary test on the state $\ket{\psi}=\ket{K}$ and
the unitary $U=M_a$, the verifier will obtain a bit that is $0$ with certainty
if $a\in K$, and is uniformly distributed if $a\notin K$.
If the test is run, leading to the outcome 1, then it must therefore be the
case that $a\not\in K$, and if $a\not\in K$ then this event will happen with
probability 1/2.

There is, of course, a problem with treating the procedure just described as a
verification procedure for the GNM problem, which is that one cannot trust that
a given quantum proof really is the state $\ket{K}$.
One could, for instance, substitute the state $\ket{1}$ (where $1$ denotes the
identity element in the group) for $\ket{K}$, which would lead to the
incorrect conclusion that all non-identity elements of $K$ are not contained in
$K$.

One solution to this problem is to first choose a collection of elements
$h_1,\ldots,h_N$ from $K$, and to sequentially run the verification procedure
on $h_1,\ldots,h_N$ taken in place of $a$.
In addition to the measurement outcome, the membership test outputs the
qubits that initially contained the proof state, and after each test these
qubits are supplied as the proof state to the next test.
Naturally, if an element $h\in K$ is selected, one would expect that
running the membership test on $h$ would reveal that $h$ is indeed contained in
$K$---so if the membership test were to reveal that $h$ is not contained in
$K$, then the proof state must have been invalid and can be rejected.
To see that this reasoning is valid not only for the first test, but for each
of the tests in sequence, from~\eqref{eq:ctl-u} one may observe that
conditioned on the test outputting 0 (which indicates a positive test for
membership) for a particular choice of $h\in G_n$ and a proof state
$\ket{\psi}$, the new proof state output by the test is
\begin{equation}
  \ket{\psi} + M_h \ket{\psi} \quad \text{(normalized)}.
\end{equation}
Because $M_h\ket{K}=\ket{hK}=\ket{K}$ holds when $h\in K$, the proof state
$\ket{K}$ is unchanged by any such test.

Now suppose that the membership test is run on a given proof state $\ket{\psi}$
for a sequence of group elements $h_1,\ldots,h_N\in K$.
Under the assumption that every one of the membership tests is consistent with
the fact that $h_1,\ldots,h_N\in K$, the resulting proof state becomes
\begin{equation}
  \sum_{k_1,\ldots,k_N\in\{0,1\}}
  M_{h_1}^{k_1} M_{h_2}^{k_2} \cdots M_{h_N}^{k_N}
  \ket{\psi} \quad \text{(normalized)}.
\end{equation}
Regardless of whether or not the original state $\ket{\psi}$ was close to
$\ket{K}$, the state above must be nearly invariant under left-multiplication
by elements of $K$, provided that $h_1,\ldots,h_N$ were chosen well.
In particular, if they represent a so-called
\emph{$\varepsilon$-uniform Erd\H{o}s--R\'enyi generating sequence} for $K$,
which means that the distribution of group elements
\begin{equation}
  {h_1}^{k_1} {h_2}^{k_2} \cdots {h_N}^{k_N}
\end{equation}
is $\varepsilon$-close to uniform over $K$ for
$(k_1,\ldots,k_N)\in\{0,1\}^N$ chosen uniformly at random, then the above state
will function in approximately the same way as $\ket{K}$ with respect to the
membership test described above.
There is a classical randomized procedure due to Babai \cite{Babai91} that
produces such a sequence with high probability.

As a remark, it is important to realize that the ability to uniformly sample
from $K$ is not known to allow one to efficiently prepare the state $\ket{K}$.
By performing the sampling in superposition it is possible to prepare a state
\begin{equation}
  \ket{K'} = \frac{1}{2^{R/2}} \sum_{r\in\{0,1\}^R} \ket{r}\ket{h(r)},
\end{equation}
where $r$ denotes the randomness used by the sampling procedure and $h(r)$
denotes the sampled group element.
In order to obtain a good approximation to $\ket{K}$ it would be necessary to
``erase,'' or uncompute, the string $r$ based on $h(r)$, and this may not be
possible (for instance if $r\to h(r)$ is not one-to-one).

The final procedure is described in Figure~\ref{fig:GNM-verification}.
One of the error-reduction procedures to be described in the next section can
be applied to this procedure to yield error bounded by 1/3, or even
exponentially small error if desired.

\begin{figure}
  \noindent\hrulefill

  \begin{mylist}{8mm}
  \item[1.]
    Input $n$-bit encodings of group elements $g_1,\ldots,g_m,a\in G_n$ and a
    proof state contained in an $k$-qubit register $\reg{X}$.
    Set $\varepsilon$ to be a small positive constant (such as
    $\varepsilon = 1/16$).
  \item[2.]
    Randomly select elements $h_1,\ldots,h_N$ so that, with probability at
    least $1 - \varepsilon$, the sequence $h_1,\ldots,h_N$ is an
    $\epsilon$-uniform Erd\H{o}s--R\'enyi generating sequence for
    $K = \langle g_1,\ldots,g_m\rangle$.
  \item[3.]
    For each $j = 1,\ldots,N$, perform the controlled-unitary test described in
    Figure~\ref{fig:control-u} on the state contained in $\reg{X}$ and the
    unitary $M_{h_j}$.
    If any of these tests results in the outcome 1, indicating non-membership,
    then reject.
  \item[4.]
    Run the membership test for $a$ on $\reg{X}$ and accept if the outcome
    is~1 (indicating non-membership), reject if the outcome is 0.
  \end{mylist}
  \vspace{-2mm}

  \noindent\hrulefill
  \caption{Verification procedure for the group non-membership problem}
  \label{fig:GNM-verification}
\end{figure}

\section{Error reduction}
\label{section:QMA-error-reduction}

In a classical setting, error reduction for polynomial-time bounded-error
verification procedures can be handled in a straightforward way:
the verification procedure is independently run multiple times on a given
proof string, and is accepted if and only if the number of acceptances obtained
by the individual runs meets or exceeds some suitably chosen threshold value.
With respect to the analysis of such a method, no significant new challenges
arise as compared with the standard analysis of error reduction for
bounded-error algorithms.

In the quantum setting, this strategy does not work---running a verification
procedure on a quantum proof will generally change it, so the original proof
may no longer be available after the first verification. 
For instance, if the measurement that is performed is a binary projective
measurement $\{\Pi_0,\Pi_1\}$, then the post-measurement state is either
$\Pi_0\ket{\psi}$ or $\Pi_1\ket{\psi}$ (properly normalized), so that repeating
the measurement will result in the same outcome with certainty.

Two solutions to this problem are known.
One solution is to request multiple, independent copies of the original proof,
one for each run of the verification procedure.
This requires an analysis to verify that no advantage may be found in
correlating the registers that are supposed to contain these
independent proof copies.
Another solution, which has the advantage that it leads to a reduction in
error without an increase in proof size, involves repeatedly running a unitary
quantum circuit implementation of the verification procedure forward and
backward in a manner reminiscent of Grover's quantum search algorithm
\cite{Grover96}.
The two methods are described in the subsections that follow.

\subsection{Parallel error reduction}
\label{sec:parallel-error-reduction}

Assume that a verification procedure is given that takes as input a $k$-qubit
register $\reg{X}$ and outputs a single qubit, which is measured with respect to
the standard basis after being output.
We will refer to the verification procedure as $V$, with the understanding that
$V$ refers to the actions of a verifier on some fixed input string that will
not be named explicitly.
The first strategy for error reduction for \class{QMA} operates as follows, for
$T$ and $t$ being positive integers satisfying $t\leq T$ to be selected later.

\begin{mylist}{\parindent}
\item[1.] Receive registers $\reg{X}_1,\ldots,\reg{X}_T$, each comprising $k$
  qubits.
\item[2.] Run $V$ independently on each of the registers
  $\reg{X}_1,\ldots,\reg{X}_T$, and let $a_1,\ldots,a_T\in\{0,1\}$ be the
  resulting binary-valued measurement outcomes.
\item[3.] Accept (i.e., output 1) if and only if $a_1+\cdots+a_T\geq t$.
\end{mylist}

\noindent
This strategy is illustrated in Figure~\ref{fig:parallel-error-reduction},
for $T = 5$, and where the circuit labelled $F$ denotes the classical threshold
value computation for some choice of $t$ (which is not specified in the
figure).

\begin{figure}
  \begin{center}
    \footnotesize
    \begin{tikzpicture}[scale=0.45, 
        circuit/.style={draw, minimum height=8mm, minimum width=8mm,
          fill = ChannelColor, text=ChannelTextColor},
        empty/.style={minimum width=6mm},
        >=latex]
      
      \node (Q1) at (0,6) [circuit] {$V$};
      \node (Q2) at (0,3) [circuit] {$V$};
      \node (Q3) at (0,0) [circuit] {$V$};
      \node (Q4) at (0,-3) [circuit] {$V$};
      \node (Q5) at (0,-6) [circuit] {$V$};

      \node (in1) at (-4,6) [empty] {$\reg{X}_1$};
      \node (in2) at (-4,3) [empty] {$\reg{X}_2$};
      \node (in3) at (-4,0) [empty] {$\reg{X}_3$};
      \node (in4) at (-4,-3) [empty] {$\reg{X}_4$};
      \node (in5) at (-4,-6) [empty] {$\reg{X}_5$};

      \node (F) at (5,0) [circuit] {$F$};
      \node (out) at (8,0) [empty] {};

      \foreach \t in {1,2,3,4,5}{
        \foreach \y in {-6,-4,-2,0,2,4,6} {
          \draw ([yshift=\y mm]in\t.east) -- ([yshift=\y mm]Q\t.west) {};
        }
      }
      
      \draw (Q1.east) .. controls +(right:20mm) and +(left:20mm) .. 
      ([yshift=4mm]F.west);
      
      \draw (Q2.east) .. controls +(right:20mm) and +(left:20mm) .. 
      ([yshift=2mm]F.west);
      
      \draw (Q3.east) .. controls +(right:20mm) and +(left:20mm) .. 
      ([yshift=0mm]F.west);
      
      \draw (Q4.east) .. controls +(right:20mm) and +(left:20mm) .. 
      ([yshift=-2mm]F.west);
      
      \draw (Q5.east) .. controls +(right:20mm) and +(left:20mm) .. 
      ([yshift=-4mm]F.west);
      
      \draw (F.east) -- (out.west);

    \end{tikzpicture}
  \end{center}
  \caption{Parallel error reduction}
  \label{fig:parallel-error-reduction}
\end{figure}
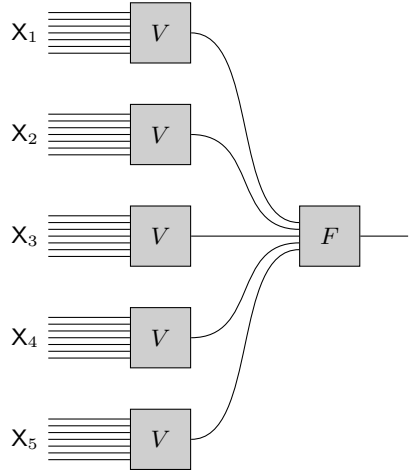

With the original verification procedure $V$, one can associate a binary-valued
measurement, as was described in the previous section.
Such a measurement may be represented by two $2^k\times 2^k$ dimensional
positive semidefinite operators $P_0$ and $P_1$ satisfying $P_0 + P_1 = \I$. 
Along similar lines, one can associate a binary-valued measurement with the
procedure described above---in this case described by two 
$2^{kT} \times 2^{kT}$ positive semidefinite operators $Q_0$ and $Q_1$
as follows:
\begin{equation}
  \begin{split}
    Q_0 & = \sum_{\substack{a_1,\ldots,a_T\in\{0,1\}\\
        a_1+ \cdots + a_T <t}}
    P_{a_1}\otimes\cdots\otimes P_{a_T},\\
    Q_1 & = \sum_{\substack{a_1,\ldots,a_T\in\{0,1\}\\
        a_1+ \cdots + a_T \geq t}}
    P_{a_1}\otimes\cdots\otimes P_{a_T}.
  \end{split}
\end{equation}
The spectra of these operators can be analyzed without difficulty due to the
fact that $P_0$ and $P_1$ necessarily commute:
$P_1 = \I - P_0$, and therefore
\begin{equation}
  P_0 P_1 = P_0 (\I - P_0) = P_0 - P_0^2 = (\I - P_0) P_0 = P_1 P_0.
\end{equation}
In particular, supposing that $\{\ket{\psi_1},\ldots,\ket{\psi_K}\}$ is an
orthonormal basis of eigenvectors of $P_1$ having corresponding eigenvalues
\begin{equation}
  \lambda_1(P_1) \geq \cdots \geq \lambda_K(P_1),
\end{equation}
where $K = 2^k$, one has that an orthonormal basis of eigenvectors of $Q_1$ is
obtained by tensoring these eigenvectors of $P_1$ (which are also eigenvectors
of $P_0$) in all possible combinations:
\begin{equation}
  \big\{
  \ket{\psi_{j_1}}\cdots\ket{\psi_{j_T}}\,:\,
  1\leq j_1,\ldots,j_T\leq K
\bigr\}.
\end{equation}
The value of the procedure that has been obtained is equal to the largest
eigenvalue of $Q_1$, which is
\begin{equation}
  \lambda_1 (Q_1) =
  \sum_{\substack{a_1,\ldots,a_T\in\{0,1\}\\ a_1+ \cdots + a_T \geq t}}
  \bra{\psi_1}P_{a_1}\ket{\psi_1}
  \cdots 
  \bra{\psi_1}P_{a_T}\ket{\psi_1}.
\end{equation}
One now sees that an optimal choice of a proof for the procedure is one in
which each of the registers $\reg{X}_1,\ldots,\reg{X}_T$ is independently
prepared in the optimal proof state $\ket{\psi_1}$ for the original verification
procedure $V$.

At this point, a suitable selection of $T$ and $t$ provides for an exponential
reduction in error, based on standard bounds on the tails of binomial
distributions.
For instance, if a given verifier $V$ has completeness and soundness
probability bounds $a$ and $b$, respectively, before error reduction, then
choosing
\begin{equation}
  T = \frac{r}{(a - b)^2}
  \quad\text{and}\quad
  t = \left\lceil \frac{a+b}{2} T \right\rceil
\end{equation}
results in a new verifier $V'$ having completeness and soundness probabilities
exponentially close to 1 and 0, respectively, with respect to a chosen error
parameter $r$.

\subsection{Witness-preserving error reduction}
\label{sec:witness-preserving-error-reduction}

The second method for error reduction of \class{QMA} is sequential, but has the
advantage that no increase in proof size is required as the error is reduced.
Before describing the method in precise terms, it will be helpful to first
discuss the intuition behind it.

Suppose that a verification procedure $V$, taking as input a $k$-qubit register
$\reg{X}$ and outputting a single qubit, is given.
As described in the previous section, one can consider a unitary circuit 
implementation $Q$ of $V$, which is a unitary procedure taking as input two
registers: the $k$-qubit proof register $\reg{X}$ along with an $m$-qubit
ancillary register $\reg{Y}$.
The output qubits of $Q$ are split between a single qubit register $\reg{A}$,
which corresponds to the output qubit of $V$, along with an $(m+k-1)$-qubit
register $\reg{Z}$ (which could be viewed as a ``garbage'' register, although
it will not be treated as garbage by the error reduction procedure).

Now, suppose that a pure state $\ket{\psi}$ on $k$ qubits has been selected as
a quantum proof, and $Q$ is run on the input state $\ket{\psi}\ket{0^m}$.
The resulting state may be expressed as 
\begin{equation}
  Q \ket{\psi}\ket{0^m} 
  = \sqrt{p_0(\psi)}\,\ket{0}\ket{\phi_0(\psi)}
  + \sqrt{p_1(\psi)}\,\ket{1}\ket{\phi_1(\psi)},
\end{equation}
where the numbers $p_0(\psi)$ and $p_1(\psi)$ represent the probabilities for a
measurement of the register $\reg{A}$ with respect to the standard basis to
give the outcomes 0 and 1, respectively.
Measuring the output qubit gives a single sample, 0 or 1, from a Bernoulli
distribution $(p_0(\psi),p_1(\psi))$ that one would ideally like to sample
multiple times.
It is natural to ask if the original proof state $\ket{\psi}$ can be recovered,
so as to allow for more samples, and perhaps the first thing one would be
inclined to do to try to recover $\ket{\psi}$ is to run $Q$ in reverse.
This yields one of the two states
\begin{equation}
  \label{eq:reversed-states-QMA-error-reduction}
  Q^{\ast}\ket{0}\ket{\phi_0(\psi)}
  \quad\text{or}\quad
  Q^{\ast}\ket{1}\ket{\phi_1(\psi)},
\end{equation}
depending on whether the outcome of the first measurement was 0 or~1.

It is not clear that the states \eqref{eq:reversed-states-QMA-error-reduction}
allow for a reconstruction of $\ket{\psi}$, or if they are useful at all for
that matter---but under a simple assumption on the original quantum proof state
$\ket{\psi}$, a recovery of $\ket{\psi}$ will generally be possible.
The assumption is that $\ket{\psi}$ is a common \emph{eigenvector} of the
two measurement operators $P_0$ and $P_1=\I-P_0$ corresponding to the
binary-valued measurement implemented by $V$.
Expanding on~\eqref{eq:p0-p1-def}, these measurement operators may be expressed
explicitly in terms of $Q$ as follows:
\begin{equation}\label{eq:verifier-measurement-operators}
  \begin{split}
    P_0 & = \bigl(\I \otimes \bra{0^m}\bigr) Q^{\ast}
    \bigl(\ket{0}\bra{0} \otimes \I\bigr) Q
    \bigl(\I \otimes \ket{0^m}\bigr),\\
    P_1 & = \bigl(\I \otimes \bra{0^m}\bigr) Q^{\ast}
    \bigl(\ket{1}\bra{1} \otimes \I\bigr) Q
    \bigl(\I \otimes \ket{0^m}\bigr).
  \end{split}
\end{equation}
It may be noted that the assumption that $\ket{\psi}$ is an eigenvector of
these operators is not a significant restriction;
a choice of $\ket{\psi}$ that maximizes $p_1(\psi)$ will necessarily be an
eigenvector of these operators, as discussed previously.

Now, to see why the condition that $\ket{\psi}$ is an eigenvector of these
operators is relevant, one may consider the states
\eqref{eq:reversed-states-QMA-error-reduction}; as they can be analyzed
similarly, we will consider $Q^{\ast}\ket{0}\ket{\phi_0(\psi)}$.
Imagine that the $m$-qubit ancillary register $\reg{Y}$ for this state is
measured with respect to the standard basis, and let us focus on the case in
which the measurement outcome is the all-zero string $0^m$.
Before normalization, the state remaining in the register $\reg{X}$ is
\begin{equation}
  \begin{aligned}
  \bigl(\I \otimes \bra{0^m}\bigr)Q^{\ast}\ket{0}\ket{\phi_0(\psi)}
  \hspace{-3cm}\\
  & = \frac{1}{\sqrt{p_0(\psi)}}
  \bigl(\I \otimes \bra{0^m}\bigr) Q^{\ast}
  \bigl(\ket{0}\bra{0} \otimes \I\bigr) Q
  \ket{\psi} \ket{0^m}\\
  & = \frac{1}{\sqrt{p_0(\psi)}} P_0 \ket{\psi}\\
  & = \sqrt{p_0(\psi)} \ket{\psi}.
  \end{aligned}
\end{equation}
Based on this calculation, one concludes that the all-zero measurement outcome
occurs with probability $p_0(\psi)$, and conditioned on this outcome the
original proof state $\ket{\psi}$ is available in the register $\reg{X}$ to be
tested again.
For the state $Q^{\ast}\ket{1}\ket{\phi_1(\psi)}$, the probability of 
obtaining the all-zero measurement outcome $0^m$ is $p_1(\psi)$, and again
the original proof state $\ket{\psi}$ is available in the register $\reg{X}$ to
be tested again.

Of course, obtaining the all-zero measurement outcome is fortuitous, and the
possibility of obtaining a different measurement outcome would seem to be a 
potential problem.
Nevertheless, the indication that there is a possibility to obtain further
samples is encouraging.
As it turns out, the potential difficulty represented by the possibility to not
obtain the all-zero measurement outcome can be overcome by making a different
choice of the measurement on the ancillary register.
Rather than measuring with respect to the standard basis, we will measure with
respect to a binary-valued projective measurement having measurement operators
\begin{equation}
  \Delta_0 = \ket{0^m}\bra{0^m}
  \quad\text{and}\quad
  \Delta_1 = \I - \ket{0^m}\bra{0^m}.
\end{equation}

Let us describe the actual error reduction procedure in precise terms before
proceeding further with the discussion.
The procedure is described in Figure~\ref{fig:witness-preserving-procedure},
where $1\leq t\leq T$ are two arbitrary integers and it is to be assumed that
the register $\reg{X}$ initially contains the proof state $\ket{\psi}$ and
$\reg{Y}$ is initialized to the state $\ket{0^m}$.
\begin{figure}
  \noindent\hrulefill

  \begin{mylist}{6mm}
  \item[1.]
    Repeat the following for each $j = 1,\ldots,T$:\vspace{-2mm}
    
    \begin{mylist}{5mm}
    \item[a.] 
      Apply the unitary circuit $Q$ to $(\reg{X},\reg{Y})$,
      obtaining $(\reg{A},\reg{Z})$.
    \item[b.]
      Perform a standard basis measurement on $\reg{A}$, letting $a_j\in\{0,1\}$
      denote the outcome.
    \item[c.]
      Apply the unitary circuit $Q^{\ast}$ to $(\reg{A},\reg{Z})$, obtaining
      $(\reg{X},\reg{Y})$.
    \item[d.]
      Perform the projective measurement $\{\Delta_0,\Delta_1\}$ on $\reg{Y}$,
      where
      \begin{equation}
        \Delta_0 = \ket{0^m}\bra{0^m}
        \quad\text{and}\quad
        \Delta_1 = \I - \ket{0^m}\bra{0^m},
      \end{equation}
      letting $b_j\in\{0,1\}$ denote the outcome.
    \end{mylist}
    
  \item[2.]
    Define $c_1,\ldots,c_{2T}\in\{0,1\}$ as follows:
    \begin{equation}
      \begin{split}
        c_1 & = a_1,\\
        c_{2j} & = a_j \oplus b_j \quad (\text{for $j = 1,\ldots,T$})\\
        c_{2j-1} & = a_j \oplus b_{j-1} \quad (\text{for $j = 2,\ldots,T$}).
      \end{split}
    \end{equation}
    Accept if $c_1 + \cdots + c_{2T} \geq 2t$, reject otherwise.
  \end{mylist}
  \vspace{-3mm}
  \noindent\hrulefill
  \caption{Witness-preserving error reduction procedure for \class{QMA}.}
  \label{fig:witness-preserving-procedure}
\end{figure}
The procedure is illustrated in Figure~\ref{fig:strong-error-reduction}
for the case $T = 3$.
The box labelled $R$ represents the classical computation performed in step~2,
and it is to be assumed that all of the qubits aside from those included in
$\reg{X}$ are initialized to the $\ket{0}$ state before the circuit in the
figure is executed.
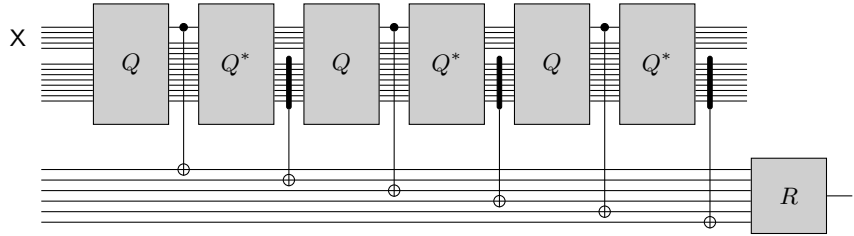
\begin{figure}
  \begin{center}
    \footnotesize
    \begin{tikzpicture}[scale=0.35, 
        circuit/.style={draw, minimum height=16mm, minimum width=10mm,
          fill = ChannelColor, text=ChannelTextColor},
        smallcircuit/.style={draw, minimum height=10mm, minimum width=10mm,
          fill = ChannelColor, text=ChannelTextColor},
        empty/.style={minimum width=4mm},
        initialized/.style={draw, fill=black, inner sep = 0mm, 
          minimum width = 0.6mm, minimum height = 7mm, 
          rounded corners=0.3mm},
        dot/.style={draw, fill=black, inner sep = 0mm, minimum size = 1mm,
          rounded corners = 0.5mm},
        circ/.style={draw, inner sep = 0mm, minimum size = 1.5mm,
          rounded corners = 0.75mm},
        >=latex]
      
      \node (R1) at (-14,0) [empty] {};
      \node (Q1) at (-10,0) [circuit] {$Q$};
      \node (R2) at (-6,0) [circuit] {$Q^{\ast}$};
      \node (Q2) at (-2,0) [circuit] {$Q$};
      \node (R3) at (2,0) [circuit] {$Q^{\ast}$};
      \node (Q3) at (6,0) [circuit] {$Q$};
      \node (R4) at (10,0) [circuit] {$Q^{\ast}$};
      \node (Q4) at (14,0) [empty] {};

      \node (R) at (15,-5) [smallcircuit] {$R$};
      \node (out) at (18,-5) [empty] {};
      \node (in) at (-14,-5) [empty] {};
      
      \draw (R.east) -- (out.west);

      \foreach \t in {1,2,3,4} {
      \foreach \y in {14,12,10,8,6,0,-2,-4,-6,-8,-10,-12,-14} {
        \draw ([yshift=\y mm]R\t.east) -- ([yshift=\y mm]Q\t.west);
      }
      }

      \foreach \y in {14,8,6,4,2,0,-2,-4,-6,-8,-10,-12,-14}{
        \draw ([yshift=\y mm]Q1.east) -- ([yshift=\y mm]R2.west);
      }

      \foreach \y in {14,8,6,4,2,0,-2,-4,-6,-8,-10,-12,-14} {
        \draw ([yshift=\y mm]Q2.east) -- ([yshift=\y mm]R3.west);
      }

      \foreach \y in {14,8,6,4,2,0,-2,-4,-6,-8,-10,-12,-14} {
        \draw ([yshift=\y mm]Q3.east) -- ([yshift=\y mm]R4.west);
      }

      \node (init1) at (-4,-0.7) [initialized] {};
      \node (init2) at (4,-0.7) [initialized] {};
      \node (init3) at (12,-0.7) [initialized] {};

      \node (measure1) at (-8,1.4) [dot] {};
      \node (measure2) at (0,1.4) [dot] {};
      \node (measure3) at (8,1.4) [dot] {};

      \foreach \y in {10,6,2,-2,-6,-10} {
        \draw ([yshift=\y mm]in.east) -- ([yshift=\y mm]R.west);
      }

      \node (target1) at (-8,-4) [circ] {};
      \node (target2) at (-4,-4.4) [circ] {};
      \node (target3) at (0,-4.8) [circ] {};
      \node (target4) at (4,-5.2) [circ] {};
      \node (target5) at (8,-5.6) [circ] {};
      \node (target6) at (12,-6.0) [circ] {};

      \draw (measure1) -- (target1.south);
      \draw (init1) -- (target2.south);
      \draw (measure2) -- (target3.south);
      \draw (init2) -- (target4.south);
      \draw (measure3) -- (target5.south);
      \draw (init3) -- (target6.south);

      \node at (-14.3,1) {$\reg{X}$};

    \end{tikzpicture}
  \end{center}
  \caption{Witness-preserving error reduction}
  \label{fig:strong-error-reduction}
\end{figure}

To analyze the procedure, we can extend the analysis that has been started
above: we have determined, under the assumption that $\ket{\psi}$ is an
eigenvector of the measurement operators $P_0$ and $P_1$, that
\begin{equation}
  Q^{\ast} \ket{0}\ket{\phi_0(\psi)} =
  \sqrt{p_0(\psi)} \ket{\psi}\ket{0^m} + \sqrt{p_1(\psi)} \ket{\gamma(\psi)},
\end{equation}
where $\ket{\gamma(\psi)}$ is defined as
\begin{equation}
  \ket{\gamma(\psi)}
  = \frac{1}{\sqrt{p_1(\psi)}}
  (\I \otimes \Delta_1) Q^{\ast} \ket{0}\ket{\phi_0(\psi)}.
\end{equation}
(We will ignore the possibility that either of $p_0(\psi)$ or $p_1(\psi)$ is
zero, which can be handled as a simple special case.)
Because $Q$ is unitary, we also have
\begin{equation}
  Q^{\ast} \biggl(
  \sqrt{p_0(\psi)}\ket{0}\ket{\phi_0(\psi)} 
  + \sqrt{p_1(\psi)}\ket{1}\ket{\phi_1(\psi)}\biggr) = \ket{\psi}\ket{0^m},
\end{equation}
which provides us with enough information to determine the state of the
registers in the procedure after every measurement, conditioned on every value.
In particular, we have
\begin{equation}
  \begin{split}
    Q^{\ast} \ket{0}\ket{\phi_0(\psi)}
    & = \sqrt{p_0(\psi)} \ket{\psi}\ket{0^m} +
    \sqrt{p_1(\psi)} \ket{\gamma(\psi)},\\
    Q^{\ast} \ket{1}\ket{\phi_1(\psi)}
    & = \sqrt{p_1(\psi)} \ket{\psi}\ket{0^m} -
    \sqrt{p_0(\psi)} \ket{\gamma(\psi)},
  \end{split}
\end{equation}
and
\begin{equation}
  \begin{split}
    Q \ket{\psi}\ket{0^m} & =
    \sqrt{p_0(\psi)} \ket{0} \ket{\phi_0(\psi)} + 
    \sqrt{p_1(\psi)} \ket{1} \ket{\phi_1(\psi)},\\
    Q \ket{\gamma(\psi)} & =
    \sqrt{p_1(\psi)} \ket{0} \ket{\phi_0(\psi)} - 
    \sqrt{p_0(\psi)} \ket{1} \ket{\phi_1(\psi)}.
  \end{split}
\end{equation}
From these equations it follows that
\begin{equation}
  \label{eq:QMA-error-reduction-probabilities}
  \begin{alignedat}{2}
    \op{Pr}\bigl( b_j = 0 \big| a_j = 0) & = p_0(\psi),
    & \quad
    \op{Pr}\bigl( a_{j+1} = 0 \big| b_j = 0) & = p_0(\psi),\\
    \op{Pr}\bigl( b_j = 1 \big| a_j = 0) & = p_1(\psi),
    & \quad
    \op{Pr}\bigl( a_{j+1} = 1 \big| b_j = 0) & = p_1(\psi), \\
    \op{Pr}\bigl( b_j = 0 \big| a_j = 1) & = p_1(\psi),
    & \quad
    \op{Pr}\bigl( a_{j+1} = 0 \big| b_j = 1) & = p_1(\psi), \\
    \op{Pr}\bigl( b_j = 1 \big| a_j = 1) & = p_0(\psi),
    & \quad
    \op{Pr}\bigl( a_{j+1} = 1 \big| b_j = 1) & = p_0(\psi),
  \end{alignedat}
\end{equation}
and therefore
\begin{equation}
  \op{Pr}\bigl(c_j = 0) = p_0(\psi)
  \quad\text{and}\quad
  \op{Pr}\bigl(c_j = 1) = p_1(\psi)
\end{equation}
for every $j = 1,\ldots, 2T$.
Figure~\ref{fig:QMA-error-reduction-probabilities} illustrates the transition
probabilities expressed by the equations
\eqref{eq:QMA-error-reduction-probabilities}.
A similar choice of $T$ and $t$ to the parallel error reduction procedure
yields an exponential reduction of error.

One must still consider the behavior of the witness-preserving error reduction
procedure when the given proof state $\ket{\psi}$ is not an eigenvector of
the original measurement operators $P_0$ and $P_1$, but this is easily done and
we will omit the details.
Intuitively speaking, the procedure operates independently on each eigenvector,
so that an arbitrary proof state must always behave as if it were a random
mixture of pure state eigenvectors.
Alternatively, one can argue that the measurement operators resulting from the
witness-preserving error reduction procedure share a common set of
eigenvectors with the original measurement operators $P_0$ and $P_1$, so there
is no loss of generality in considering the behavior of the procedure on
one of these eigenvectors.

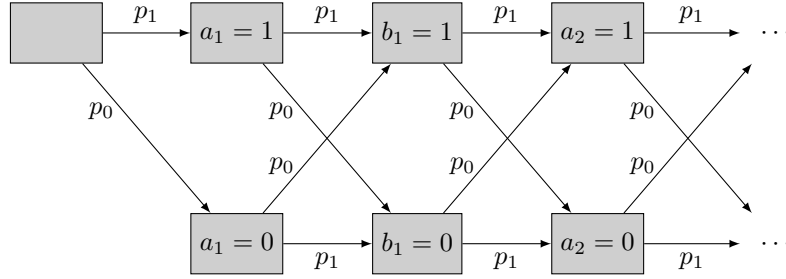
\begin{figure}
  \begin{center}\small
    \makebox[1in][c]{%
      \begin{tikzpicture}[scale=0.4,
          state/.style={draw, minimum height=8mm, 
            fill = ChannelColor, text=ChannelTextColor},
          invisible/.style={minimum height=8mm, minimum width=10mm},
          >=latex]
        
        \node at (-10,4) [invisible] {};
        \node (T1) at (-9,4) [state] {\phantom{$a_1=0$}};
        
        \node (T2) at (-3,4) [state] {$a_1=1$};
        \node (B2) at (-3,-3) [state] {$a_1=0$};
        
        \node (T3) at (3,4) [state] {$b_1=1$};
        \node (B3) at (3,-3) [state] {$b_1=0$};
        
        \node (T4) at (9,4) [state] {$a_2=1$};
        \node (B4) at (9,-3) [state] {$a_2=0$};
        
        \node (T5) at (15,4) [invisible] {$\cdots$};
        \node (B5) at (15,-3) [invisible] {$\cdots$};
        
        \draw[->] (T1) -- (T2) node [above, midway] {$p_1$};
        \draw[->] (T2) -- (T3) node [above, midway] {$p_1$};
        \draw[->] (T3) -- (T4) node [above, midway] {$p_1$};
        \draw[->] (T4) -- (T5) node [above, midway] {$p_1$};
        
        \draw[->] (B2) -- (B3) node [below, midway] {$p_1$};
        \draw[->] (B3) -- (B4) node [below, midway] {$p_1$};
        \draw[->] (B4) -- (B5) node [below, midway] {$p_1$};
        
        \draw[->] (T1) -- (B2) node [below=1mm, pos=0.15] {$p_0$};
        \draw[->] (T2) -- (B3) node [below=1mm, pos=0.15] {$p_0$};
        \draw[->] (T3) -- (B4) node [below=1mm, pos=0.15] {$p_0$};
        \draw[->] (T4) -- (B5) node [below=1mm, pos=0.15] {$p_0$};
        
        \draw[->] (B2) -- (T3) node [above=1mm, pos=0.15] {$p_0$};
        \draw[->] (B3) -- (T4) node [above=1mm, pos=0.15] {$p_0$};
        \draw[->] (B4) -- (T5) node [above=1mm, pos=0.15] {$p_0$};
      \end{tikzpicture}%
    }
  \end{center}
  \caption{Transition probabilities between different measurement outcomes
    in the witness-preserving error reduction procedure.}
  \label{fig:QMA-error-reduction-probabilities}
\end{figure}

\subsection{$\class{QMA}\subseteq\class{PP}$}
\label{sec:qma-pp}

Consider a problem $A=(A_{\yes},A_{\no})$ in \class{QMA} and consider a verifier for
$A$ for which the completeness and soundness probabilities have been amplified
to $a\geq 3/4$ and $b\leq 2^{-(k+2)}$ respectively, where $k$ is the number of
qubits of the proof. 
That this is possible follows from the witness-preserving error reduction
procedure described in the previous section. 
For a fixed input $x$, the verifier's maximum probability of producing the outcome 1 can be
expressed as the largest eigenvalue of the $k$-qubit measurement operator $P_1$
defined in~\eqref{eq:p0-p1-def}. 
Thus 
\begin{equation}
  \begin{aligned}
    x\in A_{\yes}&\implies \Tr(P_1)\geq \lambda_1(P_1) \,\geq\,\frac{3}{4},\\
    x\in A_{\no}&\implies\Tr(P_1)\leq 2^k\lambda_1(P_1)\,\leq\, \frac{1}{4}.
  \end{aligned}
\end{equation}

It follows that this problem can be decided by an unbounded-error quantum
polynomial-time procedure as follows. 
The procedure performs the same measurement as the \class{QMA} verifier but
replaces the witness by the completely mixed state on $k$ qubits, created, for
instance, as half of $k$ EPR pairs.\footnote{An EPR pair is a pair of qubits in
  the joint state $(\ket{00} + \ket{11})/\sqrt{2}$.}
This procedure will output $1$ (i.e., accept) with probability at least
$\frac{3}{4}2^{-k}$ in case $x\in A_{\yes}$, and accept with probability at
most $\frac{1}{4}2^{-k}$ in case $x\in A_{\no}$.

The above already shows $\class{QMA}\subseteq\class{PQP}$, the class of
problems that have  unbounded-error quantum polynomial-time algorithms. 
The inclusion can be re-stated as $\class{QMA}\subseteq\class{PP}$, the class 
of problems that have unbounded-error classical randomized polynomial-time
algorithms, because the equality $\class{PQP}=\class{PP}$ holds in general.
(Indeed, the two best-known proofs of the containment
$\class{BQP}\subseteq\class{PP}$, due to
Adleman, DeMarrais, and Huang \cite{AdlemanDH97} and Fortnow and Rogers
\cite{FortnowR99}, do not rely on the assumption of bounded error for
\class{BQP}.)

\section{Complete promise problems}
\label{sec:QMA-complete}

The class \class{QMA} has an interesting collection of complete promise
problems having connections to a range of problems motivated by quantum
information theory, condensed-matter physics, and quantum chemistry.
While it would be premature to compare the \class{QMA}-complete promise
problems to the rich collection of known \class{NP}-complete problems, either
with respect to the number of such problems or their broad relevance within
science and engineering, there is strong and active interest in the notion of
\class{QMA}-completeness, and the list of known \class{QMA}-complete problems
is growing steadily.
In this section we present just a couple of examples of \class{QMA}-complete
problems, referring the reader to the chapter notes for pointers to further
work on the subject.

Not surprisingly, it is possible to translate the definition of \class{QMA}
directly into a complete promise problem.
The following problem statement represents one way of doing this:

\begin{center}
  \underline{$(a,b)$-Quantum Circuit Satisfiability
    ($(a,b)$-QCS)}\\[2mm]
  \begin{tabular}{lp{0.8\textwidth}}
      \emph{Input:} & 
      A quantum circuit specifying a channel $\Phi$ with $k$ input qubits and
      $1$ output qubit.\\[1mm]
      \emph{Yes:} &
      There exists a $k$-qubit state $\rho$ such that 
      $\bra{1}\Phi(\rho)\ket{1}\geq a$.\\[1mm]
      \emph{No:} &
      For all $k$-qubit states $\rho$, $\bra{1}\Phi(\rho)\ket{1}\leq b$.
  \end{tabular}
\end{center}

\noindent
It follows immediately from the definition of $\QMA$ and the error
amplification procedure discussed in the previous section that $(a,b)$-QCS is
complete for $\QMA$ for any $a$ and $b$ satisfying
\begin{equation}
  2^{-\poly(n)}\leq b<a\leq 1-2^{-\poly(n)}
  \quad\text{and}\quad
  b-a \geq \poly^{-1}(n).
\end{equation}
Along similar lines, $(1,b)$-QCS is $\QMA_1$-complete for any $b$ satisfying
\begin{equation}
  2^{-\poly(n)}\leq b\leq 1-\poly^{-1}(n).
\end{equation}

Several of the classes of quantum interactive proofs to be discussed in later
chapters have complete problems of the following form: ``Given quantum channels
$\Phi_0$ and $\Phi_1$, determine whether or not $\Phi_0$ and $\Phi_1$ are
close.''
Problems of increasing complexity are obtained by varying the precise way in
which closeness of the channels is measured, as well as the type of channels that are
considered.
The following problem is $\QMA$-complete for the same range of parameters
$(a,b)$ as quantum circuit satisfiability.

\begin{center}
  \underline{$(a,b)$-Non-Identity Check ($(a,b)$-NIC)}\\[2mm]
  \begin{tabular}{lp{0.8\textwidth}}
      \emph{Input:} & 
      A unitary channel $\Phi:\rho\to U\rho U^*$ implemented by a quantum
      circuit on $k$ qubits.\\[1mm]
      \emph{Yes:} &
      $\frac{1}{2} \norm{\Phi-\I}_{\Diamond}\geq a$.\\[1mm]
      \emph{No:} &
      $\frac{1}{2} \norm{\Phi-\I}_{\Diamond}\leq b$.
  \end{tabular}
\end{center}

The proof that this problem lies in \class{QMA} for any $b-a>\poly^{-1}(n)$
follows by using the quantum phase estimation procedure, and QMA-completeness
can be shown by a simple reduction from
$(a,b)$-QCS~\cite{janzing2005non}.\footnote{The problem considered
  in~\cite{janzing2005non} is specified in a slightly different form, but
  essentially the same reduction can be used to show completeness of
  $(a,b)$-NIC.}

The problems described in the subsections that follow represent complete
problems for \class{QMA} that are, in some sense, more interesting than
QCS---the problems are interesting and well-motivated in their own right, and
they are not simply rephrasings of the definition of \class{QMA}.

\subsection{The local Hamiltonian problem}

The local Hamiltonian problem, introduced by Kitaev in the late 1990s, was
the first problem to be shown to be complete for \class{QMA}. 
Intuitively speaking, the proof that this problem is complete for \class{QMA} is
similar in spirit to the classical proof of the Cook--Levin theorem (but
different at a technical level for reasons to be discussed below).

The local Hamiltonian problem captures the notion of a quantum constraint
satisfaction problem, using the language of quantum many-body systems in
condensed-matter physics.
In this setting, a Hamiltonian $H$ is a Hermitian operator representing the 
total energy of a physical system. 
A typical Hamiltonian in classical mechanics is $H=p^2/2m$, where $p=mv$
is the momentum: this Hamiltonian characterizes the kinetic energy of the
system.
The eigenvectors of $H$ represent possible states of the system, and the
associated (real) eigenvalues specify the energy.
The problem of determining the smallest eigenvalue is of particular interest,
as it represents the energy of the equilibrium state at zero temperature.
The smallest eigenvalue is also called the \emph{ground state energy}, and the
associated eigenvector the \emph{ground state}.

Given a Hamiltonian $H$ acting on the Hilbert space
corresponding to $n$ qubits, it is said that $H$ is $k$-\emph{local} if it
admits a decomposition
\begin{equation}
  H=\sum_{i=1}^m H_i,
\end{equation}
where each $H_i$ can be written as
a tensor product of an operator $\tilde{H}_i$, acting on a subset 
$S_i \subseteq\{1,\ldots,n\}$ of at most $|S_i|\leq k$ qubits, with the
identity acting on the remaining qubits. 
Many physical systems can be characterized by local Hamiltonians: the locality
reflects the fact that the system's energy only depends (at least to some
approximation) on local interactions, such as particle-field interactions
(which are $1$-local) and pairwise particle-particle interactions (which are
$2$-local).
While a general Hamiltonian on $n$ qubits may require a number of bits
exponential in $n$ to be fully specified, a $k$-local Hamiltonian can be
described using $\poly(n,m,2^k)$
bits by listing the local terms $\tilde{H}_1,\ldots,\tilde{H}_m$ along with a
specification of the subsets $S_1,\ldots,S_m$ of qubits they act upon.

\pagebreak

\begin{center}
  \underline{$(k,a,b)$-Local Hamiltonian ($(k,a,b)$-LH)}\\[2mm]
  \begin{tabular}{lp{0.8\textwidth}}
    \emph{Input:} & 
    A $k$-local Hamiltonian $H=\sum_{i=1}^m H_i$ acting on $n$ qubits.
    For every $i\in\{1,\ldots, m\}$, $H_i$ acts nontrivially on at most $k$
    qubits and satisfies $0\leq H_i \leq \I$.\\[1mm]
    \emph{Yes:} &
    There exists a pure state $\ket{\psi}$ on $n$ qubits such that 
    $\frac{1}{m}\bra{\psi} H \ket{\psi} \leq 1-a$.\\[1mm]
    \emph{No:} &
    For all pure states $\ket{\psi}$ on $n$ qubits, 
    $\frac{1}{m}\bra{\psi} H \ket{\psi} \geq 1-b$.
  \end{tabular}
\end{center}

\noindent
The requirement that each Hamiltonian term satisfies $0\leq H_i \leq \I$
provides a convenient normalization of the problem.
The case of non-positive $H_i$ can always be reduced to $H_i\geq
0$ by adding a constant shift.

The problem $(k,a,b)$-LH is complete for $\QMA$ for any $k\geq 2$ and 
for $a$ and $b$ such that $ b < a \leq 1-\poly^{-1}(n)$ and 
$a-b=\Theta(\poly^{-1}(n))$.
We write $k$-LH when it is understood that $a,b$ are set to such values. 

A \class{QMA} verifier for an instance of $(k,a,b)$-LH can be constructed as
follows: the verifier selects one of the $m$ local terms $H_i$ uniformly at
random, performs the measurement $\{\I-H_i,H_i\}$ on the $n$-qubit quantum
proof, and outputs $1$ if and only if the first outcome is obtained. 
In the case of a yes-instance, setting the proof to be a state $\ket{\psi}$
such that $\frac{1}{m}\bra{\psi}H\ket{\psi} \leq 1-a$ leads to acceptance with
probability at least $a$. 
In the case of a no-instance, no proof can lead to the verifier outputting $1$
with probability larger than $b$. 
Under the assumption that $a-b>\poly^{-1}(n)$, the difference between the two
probabilities can be amplified to a constant, showing that the problem is in
$\QMA$.

Establishing completeness of $k$-LH for $\QMA$ requires more work. 
The original proof, due to Kitaev~\cite{KitaevSV02}, shows this for $k=5$, with
subsequent improvements bringing the locality down to 2-local terms.
	
To see the difficulty, consider first the analogous task in the classical
setting. 
In order to reduce from an instance of circuit-SAT to a local constraint
satisfaction problem, it is customary to introduce auxiliary variables
associated with each wire in the circuit, as well as local constraints that
enforce that the variables corresponding to the input and output wires for any
gate in the circuit are related as required by the gate.
A direct translation of this reduction does not work in the case of a circuit
acting on quantum states, even for the most trivial of circuits.
This is due in part to the fact that the equality of two pure quantum 
states cannot be checked locally.
For example, the states $\frac{1}{\sqrt{2}}(\ket{0\cdots 0} + \ket{1\cdots 1})$
and $\frac{1}{\sqrt{2}}(\ket{0\cdots 0} - \ket{1\cdots 1})$ are orthogonal, but
have the same reduced density operator as soon as just one of the qubits is
discarded.

It is therefore necessary to consider a different reduction.
Suppose we are given a quantum circuit $Q$ with $T$ gates acting on $n$ input
qubits. 
The idea is to require the proof $\ket{\psi}$ to be a uniform superposition of
``snapshot states'' $\ket{\psi_i}$, where for $i=0,\ldots,T$, $\ket{\psi_i}$ is
the state of all qubits in the circuit after the $i$-th gate has been
applied. 
More precisely, the expected quantum proof has the form
\begin{equation}\label{eq:history-state}
  \ket{\psi} = \frac{1}{\sqrt{T+1}} \sum_{i=0}^T \ket{i}\ket{\psi_i},
\end{equation}
where the first register is called the \emph{clock} register. 
The state $\ket{\psi}$ is sometimes called the \emph{history state} of the
computation. 

To check that an arbitrary state is of the form~\eqref{eq:history-state},
it is possible to define a local Hamiltonian term $H_i$ for each gate in
the circuit, acting only on the clock register and the qubits of $\ket{\psi}$ on
which the gate operates, such that $\bra{\psi}H_i\ket{\psi}=0$ if and only if
$\ket{\psi_{i+1}}=U_i\ket{\psi_i}$, where $U_i$ is the local unitary
implemented by the $i$-th gate in $Q$.
In addition, there should be terms to enforce that $\ket{\psi_0}$ is initialized
correctly (each ancilla qubit is set to $\ket{0}$), and $\ket{\psi_T}$ is an
accepting state (the output qubit is in state $\ket{1}$). 
An analysis of this construction shows $\QMA$-completeness of $(k,a,b)$-LH for
$k=O(\log n)$, the size of the clock register plus the locality of any gate in
the circuit, and $a,b$ inverse polynomial in $T$. 
A different implementation of the clock, relying on a unary, rather than
binary, encoding, can be used to devise a $5$-local Hamiltonian. 
Different techniques, such as ones based on the use of perturbation theory,
can be used to reduce the locality to $2$.

It is not known whether error amplification can be performed for $k$-LH: the
natural amplification procedure would replace $H$ by $H^t$ for a sufficiently
large positive integer $t$, which would lead to a corresponding increase
in the parameter $k$.
The \emph{Quantum PCP Conjecture} posits that the problem $(k,a,b)$-LH remains
\class{QMA}-complete for parameters $a$ and $b$ that are separated by a
constant. 
Aside from being a natural quantum analogue of the classical PCP theorem, this
conjecture has implications for the study of entanglement in low-energy
eigenstates of local Hamiltonians, and it has attracted interest from
theoretical computer scientists and condensed matter theorists alike.

\subsection{Quantum 3-SAT}

The special case of $3$-LH where all of the local terms $H_i$ are rank-one
projections is called \emph{quantum $3$-SAT}.
Each local term defines a $1$-dimensional subspace of invalid configurations
within the $8$-dimensional subspace on which $H_i$ acts non-trivially.

\begin{center}
  \underline{$b$-Quantum $3$-SAT
    ($b$-Q3SAT)}\\[2mm]
  \begin{tabular}{lp{0.8\textwidth}}
    \emph{Input:} & 
    A $3$-local Hamiltonian $H=\sum_{i=1}^m P_i$ on $n$ qubits, where each
    $P_i$ is a rank-one projection acting non-trivially on at most $3$ qubits.\\[1mm]
    \emph{Yes:} &
    There exists an $n$-qubit state $\ket{\psi}$ such that $\bra{\psi} H
    \ket{\psi} =0$.\\[1mm]
    \emph{No:} &
    For all states $\ket{\psi}$ on $n$ qubits, $\frac{1}{m}\bra{\psi} H
    \ket{\psi} \geq 1-b$.
  \end{tabular}
\end{center}

For any $b=1-\Theta(\poly^{-1}(n))$ the problem $b$-Quantum $3$-SAT is complete
for $\QMA_1$.
Containment in $\QMA_1$ follows as for $k$-LH, except that in order to
guarantee completeness $1$ it is important that the measurement
$\{\I-P_i,P_i\}$ can be implemented perfectly by the $\QMA$ verifier.
Depending on how the operator $P_i$ is specified, and which universal gate set
is allowed for the verifier's circuit, this may not be the case. 
For our purposes it is sufficient to point out that the required compatibility
can always be ensured by a careful choice of definitions.

The proof that $b$-Q3SAT is $\QMA_1$-hard follows the same general outline as
sketched previously for $3$-LH, but additional difficulties arise from the
requirement of perfect completeness---highly nontrivial modifications to the
construction of the clock register are required for the proof.

\subsection{Consistency of Density Operators}

We end this section with a problem indicative of the diversity of problems that
have been shown complete for \class{QMA}.

\begin{center}
  \underline{$(k,b)$-Consistency of Density Operators
  }\\[2mm]
  \begin{tabular}{lp{0.8\textwidth}}
    \emph{Input:} & 
    Density operators $\sigma_1,\ldots,\sigma_m$ on at most $k$ qubits each,
    together with subsets $S_1,\ldots,S_m\subseteq\{1,\ldots,n\}$ such that
    $|S_i|\leq k$ for each $i$.\\[1mm]
    \emph{Yes:} &
    There exists an $n$-qubit density operator $\rho$
    such that $\Tr_{\{1,\ldots,n\}\backslash S_i}(\rho)=\sigma_i$ for every
    $i\in\{1,\ldots,m\}$.\\[1mm]
    \emph{No:} &
    For every $n$-qubit density operator $\rho$,
    it holds that
    \[
    \bignorm{\Tr_{\{1,\ldots,n\}\backslash S_i}(\rho)-\sigma_i}_1\geq b
    \]
    for at least one choice of $i\in\{1,\ldots,m\}$.
  \end{tabular}
\end{center}

This problem is known to be \class{QMA}-complete with respect to Cook
reductions (i.e., polynomial-time Turing reductions).
Whether or not the problem is also \class{QMA}-complete with respect to
Karp reductions is an interesting open problem.

\section{Variations on \class{QMA}}

In this section we discuss a few noteworthy variants of \class{QMA}. 
In each case we consider changes to the definition of \class{QMA} that
either enhance the power of the verifier or restrict the allowable
set of quantum proof states it may receive.
In some cases these changes are superficial, leading to an equivalent
definition of the class \class{QMA}, and in other cases it appears that
the changes have a substantial effect. 
The following variants will be discussed:
\begin{mylist}{\parindent}
\item[1.] \emph{Super-verifiers}. 
  Super-verifiers are granted the ability to estimate the probability
  $\bra{1}V(\rho)\ket{1}$ with good accuracy, for any chosen $\QMA$ verifier
  $V$, and for $\rho$ being the proof received from the prover.
  This definition leads to a class $\class{QMA}_+$, which is equal to
  \class{QMA}---a fact that is sometimes useful for proving problems to be
  contained in \class{QMA}.
\item[2.]
  \emph{Subset-state proofs}.
  Here one restricts the completeness condition of \class{QMA} so that
  a valid quantum proof on yes-inputs to the problem must be a uniform
  superposition over a subset of computational basis states, disallowing
  any phase differences between the coefficients of these basis states.
  This leads to a class \class{SQMA}, which is also equal to \class{QMA}.
\item[3.]
  \emph{Trusted advice}.
  The notion of quantum advice is similar in spirit to quantum proofs, except
  that the advice state can be \emph{trusted}---but it must also be the same
  state for all input strings of a given length.
  This leads to the class $\class{BQP}/\class{qpoly}$.
  This is a non-uniform class and is therefore different from \class{QMA}, but
  the precise relationship of these classes is not clear.
  It is known, however, that
  $\class{BQP}/\class{qpoly}\subseteq\class{QMA}/\class{poly}$; 
  trusted quantum advice can be simulated by trusted classical advice
  together with an untrusted quantum proof.
\item[4.]
  \emph{Unentangled quantum proofs.}
  Here the verifier is granted the promise that the quantum proof
  splits into two unentangled parts, 
  $\ket{\psi}=\ket{\psi_1}\otimes \ket{\psi_2}$, with respect to a fixed
  bipartition of the qubits comprising the proof.
  This leads to the class $\class{QMA}(2)$.
  While it is an open question whether or not $\class{QMA}(2)$ and
  $\class{QMA}$ are equal, some evidence exists to suggest that
  $\class{QMA}(2)$ is larger than \class{QMA}. 
\item[5.] 
  \emph{Classical proofs.}
  Finally we consider the case in which a classical proof is supplied to a
  quantum verifier, which leads to the class \class{QCMA}. 
  There are arguments both in favor of and against this class being strictly
  smaller than \class{QMA}.
\end{mylist}

\subsection{Variations equivalent to \class{QMA}}
\label{sec:super-verifiers}

By definition, a promise problem is in \class{QMA} if yes-inputs have quantum
proofs that convince the verifier $V$ to accept with high probability, whereas
for no-inputs there is no such witness.
Thus, it is assumed that yes- and no-inputs can be distinguished by the
\emph{maximum} probability $\omega(V)$ with which the verifier can be made to
output $1$ over all possible input states.

In this section two independent modifications to this scenario are considered
that lead to alternative ways of defining \class{QMA}.
In the first modification, it is assumed that a prover aims not necessarily to
maximize the probability that the verifier accepts, but to obtain a certain
target probability $p = \bra{1} V(\rho) \ket{1}$ for a given verifier $V$.
In the second modification, the proof is restricted to the subset-state form
suggested above for yes-instances of the problem.
The fact that these modifications yield equivalent definitions of \class{QMA}
has a positive interpretation, in the sense that they show that \class{QMA}
verifiers are more powerful than immediately apparent.
It may also be helpful to rely on either promise when reasoning about the
class \class{QMA}.

\subsubsection*{Arbitrary probabilities}

The error reduction procedure for \class{QMA} described in 
Section~\ref{sec:witness-preserving-error-reduction} shows that, for a given
verifier $V$, it is possible to design a verifier $V'$ that uses $V$ to obtain
an accurate estimate of the maximum probability $\omega(V)$ of $V$ to 
output $1$, given an optimal proof state for $V$.
This procedure crucially relies on $\omega(V)$ being defined as the
maximum probability for $V$ to accept.
(More generally, a similar process would work for any eigenvalue of the
measurement operators~\eqref{eq:p0-p1-def} associated with
the verifier, given a corresponding eigenvector for that eigenvalue.)
Suppose instead that a new verifier $V'$ is granted the power to estimate, to
within inverse polynomial precision, the probability $\bra{1}V(\rho)\ket{1}$
for a verifier $V$ to accept an arbitrary state $\rho$ provided as the proof,
for $V$ being any \class{QMA} verifier selected by $V'$.
Thus, $V'$ can decide, for any target probability $p$, whether there exists a
state $\rho$ such that $\bra{1}V(\rho)\ket{1} \approx p$.
It is natural to question whether this ability allow $V'$ to decide problems
beyond those contained in \class{QMA}.

More formally, one defines a \emph{super-verifier} as a triple $(V,p,\eta)$ 
consisting of a polynomial-time mapping $V$ from strings $x\in\Sigma^{\ast}$ to
$\QMA$ verifiers $V(x)$, along with polynomial-time computable functions
$p:\Sigma^{\ast}\to[0,1]$ and $\eta:\natural\to (0,1]$ for which $\eta$ is
larger than the inverse of some polynomially bounded function.

\begin{definition} 
  A promise problem $(A_{\yes},A_{\no})$ is contained in $\class{QMA}_+$ if
  there exists a polynomial $q$ and a super-verifier $(V,p,\eta)$ such that the
  following conditions hold:
  \begin{mylist}{\parindent}
  \item[1.] For every $x\in A_{\yes}$, there exists a state $\rho$ such that 
    \begin{equation}
      |\bra{1}V(x)(\rho)\ket{1}-p(x)|\leq \eta(\abs{x}).
    \end{equation}
  \item[2.] For every $x\in A_{\no}$, and for every state $\rho$,
    \begin{equation}
      |\bra{1}V(x)(\rho)\ket{1}-p(|x|)|\geq \eta(\abs{x})+1/q(|x|).
    \end{equation}
  \end{mylist}
\end{definition}

It is evident that $\class{QMA}\subseteq\class{QMA}_+$: given a
$\QMA_{\frac{3}{4},\frac{1}{4}}$ verifier $V$ we can define an equivalent
super-verifier $(V,3/4,1/4)$. 
More interesting is that the reverse inclusion also holds, so that
\begin{equation}
  \label{eq:qmaplus}
  \class{QMA}_+\,=\,\class{QMA}.
\end{equation}
To see that this is so, let $(V,p,\eta)$ be a given super-verifier. 
For some sufficiently large (but polynomially bounded) number $k$,
consider a verifier $V'$ that expects $k$ copies of a quantum proof for $V$,
measures each copy independently according to the binary-valued measurement
associated with $V$, and accepts if and only if the fraction $r$ of outcomes
$1$ obtained satisfies $|r-p|\leq \eta+1/(2q)$.
Provided that $k$ is chosen as a sufficiently large multiple of $q$,
it follows from a Chernoff-type bound that every input $x\in A_{\yes}$ has
a proof that is accepted by $V'$ with probability exponentially close to $1$.
To establish soundness, suppose that the $\QMA_+$ super-verifier $V$ is such
that the inequality $|\bra{1}V(\rho)\ket{1}-p|\geq \eta+1/q$ holds for all
states $\rho$. 
Let $\sigma$ be an arbitrary witness for $V'$ and define a state $\rho$ by
taking the average of the $k$ reduced density operators of $\sigma$ on each of
the registers on which $V'$ executes the circuit specified by $V$. 
From this definition it follows that $\bra{1}V(\rho)\ket{1}$ coincides with the
expectation of $r$, the fraction of acceptances that $V'$ witnesses when making
the $k$ successive measurements of $V$ on the corresponding registers
of~$\sigma$.
There are two cases:
\begin{mylist}{8mm}
\item[(i)] 
  $\bra{1}V(\rho)\ket{1}\leq p-\eta-1/q$. By Markov's inequality the
  probability that $r \geq p-\eta-1/(2q)$ is at most $1-1/(2q)$.
\item[(ii)] $\bra{1}V(\rho)\ket{1}\geq p+\eta+1/q$. 
  In this case Markov's inequality applied to $1-r$ shows that  the probability
  that $r \leq p+\eta+1/(2q)$ is at most $1-1/(2q)$.
\end{mylist}
In both cases, $V'$ rejects with probability at least $1/(2q)$.
Thus, the soundness probability is bounded away from $1$ by an inverse
polynomial, and the gap between the completeness and soundness probabilities is
large enough that it can be amplified using the methods described in
Section~\ref{section:QMA-error-reduction}.

\subsubsection*{Subset-state witnesses}

Given an integer $n$ and a nonempty set $S\subseteq\Sigma^n$, define the
$n$-qubit \emph{subset state} $\ket{S}$ as
\begin{equation}\label{eq:subset-state}
  \ket{S} = \frac{1}{\sqrt{|S|} }\sum_{x\in S}\ket{x}.
\end{equation}
A promise problem $(A_{\yes},A_{\no})$ is said to lie in the class \class{SQMA}
if for every $x\in A_{\yes}$ there is a state of the
form~\eqref{eq:subset-state} that convinces the verifier to accept with
probability at least $2/3$, while for $x\in A_{\no}$ no state (of any form)
will convince the verifier to accept with probability more than $1/3$.

With this definition it is clear that $\class{SQMA}\subseteq \class{QMA}$,
because the witness is restricted to have a special form only in the case of a
yes-instance. 
It is perhaps surprising that the equality $\class{SQMA}=\class{QMA}$ holds. 
The main observation required to show this is that subset states are
sufficiently dense in the set of all states.

More precisely, it can be shown that for any $n$-qubit unit vector
$\ket{\psi}$, there exists a subset $S\subseteq\Sigma^n$ such that
$|\bra{\psi}S\rangle| = \Omega(n^{-1/2})$. 
To see that this overlap is sufficient to conclude that
$\class{SQMA}=\class{QMA}$, recall that the error reduction
procedure described in Section~\ref{sec:parallel-error-reduction}
shows that any problem in \class{QMA} has a verifier $V$ with completeness and
soundness parameters exponentially close to $1$ and $0$ respectively.
In particular the soundness error can be made smaller than any polynomial in
the number of qubits of the witness.
Whenever there exists a witness $\ket{\psi}$ accepted by $V$ with
probability exponentially close to $1$, the subset state with maximal overlap on
$\ket{\psi}$ will convince the verifier to accept with inverse polynomial
probability.
Thus, for any problem in \class{QMA}, it is possible to construct a
\class{SQMA} verifier with completeness and soundness parameters separated by
an inverse polynomial. 
The completeness can be amplified to at least $2/3$ by performing parallel
error reduction, which preserves the property that there exists a good witness
that has the form of a subset state.
A similar argument can be made for any restriction on the proof that forces it
to belong to a set that remains dense enough in the unit sphere. 

The fact just described demonstrates that the strength of quantum proofs does
not lie in the possibility to use signed, or complex, amplitudes.
Rather, the strength appears to lie in the ability to use superpositions in
various ways, such as in the case of the group non-membership problem discussed
in Section~\ref{sec:gnm}.

\subsection{Quantum advice}
\label{section:QMA-advice}
	
The definition of the class $\QMA$ specifies that the quantum proof provided by
the prover to the verifier is \emph{untrusted}: the prover is assumed to always
attempt to maximize its chances of convincing the verifier to accept, requiring
the verifier to carefully check the information provided by the prover.
One may envision a less paranoid situation in which the prover is 
\emph{trusted}, and always attempts to convince the verifier to make
the right decision: accept yes-inputs and reject no-inputs. 
Without any further restrictions, such a prover would immediately allow the
verifier to decide all problems, as a single bit of advice suffices to inform
the verifier of whether $x\in A_{\yes}$ or $x\in A_{\no}$. 

In the setting of \emph{computational advice}, the following restriction is
considered: the quantum state provided by the prover is trusted, but is only
allowed to depend on the input length $n=|x|$; so that the same advice is to be
provided for all inputs of the same length.
In greater detail, the class $\class{BQP}/\class{qpoly}$ is defined as follows:

\begin{definition}
  A promise problem $A = (A_{\yes},A_{\no})$ is contained in
  $\class{BQP}/\class{qpoly}$ if there exists
  polynomial-time computable functions $V$ and $p$ possessing the following
  properties:
  \begin{mylist}{\parindent}
  \item[1.]
    For every string $x\in A_{\yes}\cup A_{\no}$, $V(x)$ is an encoding of a
    quantum circuit implementing a channel $\Phi_x$ with $p(\abs{x})$ input
    qubits and 1 output qubit.
  \item[2.] For every integer $n$ there exists a state $\rho_n$ on $p(n)$
    qubits such that the following two conditions hold:
    \begin{mylist}{8mm}
    \item[(a)] \emph{Completeness}. 
      If $x\in\Sigma^n\cap A_{\yes}$ then $\bra{1}\Phi_x(\rho_n)\ket{1}\geq a$.
    \item[(b)] \emph{Soundness}.
      If $x\in \Sigma^n\cap A_{\no}$ then $\bra{1}\Phi_x(\rho_n)\ket{1}\leq b$. 
    \end{mylist}
  \end{mylist}
\end{definition}

In the case of classical advice, the restriction that an advice string can
only depend on the input length immediately rules out the sort of strategy
suggested above for a prover to trivially help the verifier in deciding
arbitrary problems, as there are exponentially many inputs of a given
length and only polynomially many bits provided as advice.
In the quantum setting, the situation is not quite as clear, but quantum
information-theoretic arguments (namely Holevo's theorem and Nayak's bound)
similarly rule out the possibility that a quantum state on polynomially many
qubits could encode the answers to an exponential number of problem instances
in a way that would be accessible by a valid quantum measurement.

Although it is therefore evident that there are limitations on the power of
quantum advice, it is not at all obvious how one can obtain interesting
complexity-theoretic upper bounds on the power of quantum advice.
One striking upper bound that is known 
directly relates quantum advice with quantum proofs.
It is the containment
\begin{equation}\label{eq:bqp-qpoly}
  \class{BQP}/\class{qpoly}\subseteq \class{QMA}/\class{poly},
\end{equation}
which demonstrates that a trusted quantum state is no more useful to a
polynomial-time quantum verifier than an untrusted quantum state, complemented
with a trusted classical advice string of polynomial length.
The class $\class{QMA}/\class{poly}$ is defined in a similar way to
$\class{BQP}/\class{qpoly}$, where in addition to the untrusted quantum proof 
from the $\class{QMA}$ prover the verifier receives polynomially many classical
bits of advice that are only allowed to depend on the input length $n$.

The proof of the inclusion~\eqref{eq:bqp-qpoly} is rather involved. 
Given a verifier $V$ for a problem in $\class{BQP}/\class{qpoly}$, a 
$\class{QMA}/\class{poly}$ verifier $V'$ that decides the same problem is
constructed. 
To accomplish this task, the polynomial number of bits of classical advice 
specify a polynomial-size quantum circuit enabling $V'$ to verify that the 
untrusted witness $\ket{\psi}$ matches, ``for all practical purposes,'' the 
trusted advice state $\rho_n$ expected by $V$. 
A simple counting argument shows that no polynomial-size circuit could certify 
closeness of $\ket{\psi}$ to an arbitrary state $\rho_n$ in trace distance to 
within any reasonable accuracy, but such a strong guarantee is not necessary.
The key observation is that it is sufficient to guarantee that $\ket{\psi}$ 
reproduces approximately the same statistics as $\rho_n$, not with respect to
the outcome of any measurement (which would lead to an approximation in trace
distance), but only with respect to polynomial-size quantum circuits of the
form that can be executed by~$V$.

There are still exponentially many such circuits, and the fact that such
a verification procedure can be specified using only polynomially many bits,
and implemented efficiently by $V'$, constitutes most of the work in
establishing~\eqref{eq:bqp-qpoly}. 
The proof provides a method, based on a tool called the
\emph{majority-certificates lemma}~\cite{aaronson2014full}, to achieve this. 
Given any state $\rho_n$, the lemma specifies that there exists a polynomial
number of tests, each of which can be specified by a polynomial-size quantum
circuit, that the verifier $V'$ can perform on the untrusted $\ket{\psi}$ such
that, provided $\ket{\psi}$ passes all tests, it is guaranteed that
$\ket{\psi}$ will also approximately reproduce the same outcome as $\rho_n$
with respect to all polynomial-size quantum circuits. 

\subsection{Two unentangled proofs}
\label{section:QMA2}

Are two proofs more useful than one? 
Unless one imposes very strict length requirements the answer to this question
for the case of classical proofs is uninspiring: two classical proof strings of
a given length are equivalent to a single proof string of twice that length.
In the case of quantum proofs, however, the situation is more subtle.
The question is studied by introducing the class $\class{QMA}(t)$.
  
\begin{definition}
  Let $t:\natural\to \natural$ be a polynomially bounded function. 
  A promise problem $A = (A_{\yes},A_{\no})$ is contained in
  $\class{QMA}_{a,b}(t)$ if there exists a polynomial-time
  computable function $V$ possessing the following properties:
  \begin{mylist}{\parindent}
  \item[1.]
    For every string $x\in A_{\yes}\cup A_{\no}$, $V(x)$ is an encoding of a 
    quantum circuit implementing a channel $\Phi_x$ having $t\cdot k$ input
    qubits and 1 output qubit, for some choice of $k$ and for $t = t(\abs{x})$.
  \item[2.] 
    \emph{Completeness}. 
    If $x\in A_{\yes}$, then there exist $t$ states
    $\ket{\psi_1},\ldots,\ket{\psi_t}$, on $k$ qubits each, such that
    \begin{equation}
      \bra{1} \Phi_x(\ket{\psi_1}\bra{\psi_1}\otimes\cdots\otimes 
      \ket{\psi_t}\bra{\psi_t})\ket{1}\geq a.
    \end{equation}
  \item[3.]
    \emph{Soundness}.
    If $x\in A_{\no}$, then for all choices of $k$-qubit states 
    $\ket{\psi_1},\ldots,\ket{\psi_t}$, it holds that
    \begin{equation}
      \bra{1}\Phi_x(\ket{\psi_1}\bra{\psi_1} \otimes\cdots\otimes 
      \ket{\psi_t}\bra{\psi_t})\ket{1}\leq b.
    \end{equation}
  \end{mylist}
\end{definition}

It is known that, for any polynomially bounded number of witnesses $t\geq 2$,
the equality $\class{QMA}(t)=\class{QMA}(2)$ holds.
Furthermore, a strong error reduction is possible:
\begin{equation}
  \class{QMA}_{a,b}(2) = \class{QMA}_{1-2^{-p(n)},2^{-p(n)}}(2)
\end{equation}
for every polynomial $p$, provided
\begin{equation}
  a(n)-b(n)\geq \frac{1}{q(n)}
\end{equation}
for some polynomial $q$.
Both equalities require a corresponding increase in the witness length, by a
factor $t$ for the first transformation and $O(p\cdot q)$ for the second. 

The proofs rely on the following \emph{product test} that attempts to determine
whether a state is unentangled across $t$ registers, by being given access to
\emph{two} unentangled copies of the state.
\begin{center}
  \underline{Product test}\\[2mm]
  \begin{tabular}{lp{0.76\textwidth}}
    \emph{Input:} & 
    Pure states $\ket{\phi_1}$, $\ket{\phi_2}$ on $t$ registers of $k$ qubits
    each.\\[1mm]
    \emph{Procedure:} & 
    Perform the SWAP test on each of the $t$ pairs of $k$-qubit
    registers. Accept if and only if all tests succeed. \\[1mm]
    \emph{Guarantee:} 
    & (i) If 
    $\ket{\phi_1}=\ket{\phi_2}=\ket{\psi_1}\otimes\cdots\otimes\ket{\psi_t}$,
    then the test always accepts.\\
    & (ii) If the test accepts with probability $1-\eps$ then there exist
    states $\ket{\psi_1},\ldots,\ket{\psi_t}$ such that
    \[
    \min\big\{|\bra{\phi_1}\psi_1,\ldots,\psi_t\rangle|^2,|\bra{\phi_2}
    \psi_1,\ldots,\psi_t\rangle|^2\big\} = 1-O(\eps).
    \]
  \end{tabular}
\end{center}
\vspace{-6mm}

The SWAP test is a special case of the controlled-unitary test described
in Figure~\ref{fig:control-u}:

\begin{center}
  \underline{SWAP test}\\[2mm]
  \begin{tabular}{lp{0.78\textwidth}}
      \emph{Given:} & Pure states $\ket{\psi}$, $\ket{\varphi}$ on $k$ qubits
      each.\\[1mm]
      \emph{Outcome:} & A classical bit that is $0$ with probability
      \begin{displaymath}
        p = \frac{1+|\bra{\psi}\varphi\rangle|^2}{2}
      \end{displaymath}
      and $1$ with probability $1-p$.\\[1mm]
      \emph{Procedure:} & Perform the controlled-unitary test using the
      $2k$-qubit state $\ket{\psi}\ket{\varphi}$ as input, and the
      $2k$-qubit unitary $S$ that permutes its two sets of $k$ input qubits.
  \end{tabular}
\end{center}

\noindent
If $S$ denotes the unitary that implements the permutation used in the SWAP
test, then it holds that
\begin{equation}
  \Re(\bra{\varphi,\psi} S \ket{\varphi,\psi}) = |\bra{\varphi}\psi\rangle|^2,
\end{equation}
and the analysis of the SWAP test follows immediately from that of the
controlled-unitary test.

The product test requires \emph{two} copies of the state to be tested, and it
allows for a reduction of the number of required unentangled proof states
from any polynomial $t$ to $2$.
It is open if two proofs are more powerful than one, but there is some evidence
pointing in the direction of a positive answer.

As discussed in Section~\ref{sec:qma-pp}, it is known that
$\class{QMA}\subseteq \class{PP}$, but the best upper bound known on
$\class{QMA}(2)$ is the trivial bound of $\class{NEXP}$ obtained by guessing
exponential-size vectors for the two witnesses. 
The problem of devising better upper bounds on $\class{QMA}(2)$ directly
relates to that of optimizing over the set
\begin{equation}
  \textsc{SEP} = 
  \textrm{Conv}\big\{ \rho \otimes \sigma,\, \rho,\sigma 
  \in \Density(\complex^d)\big\}
\end{equation}
of separable states. 
Although this set is convex, there is no efficient membership oracle known.
In fact, deciding weak membership is known to be \class{NP}-hard for precision
up to inverse polynomial in $d$.

Perhaps the strongest evidence known that suggests $\class{QMA}(2)$ may be a
strictly larger class than $\class{QMA}$ is the following: there is an
efficient $\class{QMA}(2)$ verifier $V$ for the satisfiability of $3$-SAT
formulas on $n$ variables, with completeness and soundness parameters separated
by a constant, given access to two unentangled quantum proofs of
$O(\sqrt{n}\,\mathrm{polylog}(n))$ qubits each.
The existence of such a procedure with a single quantum proof of the same size
(or even twice the size) would imply 
$\text{3-SAT}\in\class{DTIME}(\text{exp}(\sqrt{n}\,\mathrm{polylog}(n)))$, 
thereby violating the exponential time hypothesis.
Thus, the ability to receive unentangled witnesses can at least provide a 
quadratic improvement on the minimum witness length, which may be seen as
evidence in favor of $\class{QMA} \neq \class{QMA}(2)$.
	
\subsection{Classical certificates}
\label{section:QCMA}
	
The class \class{QMA} differs from \class{NP} and \class{MA} in two important
ways: the verifier is able to apply a quantum circuit, and the proof that is
provided to the verifier may be a quantum state. 
Of course, a quantum proof requires a quantum verifier and it does not make
sense to consider the latter without the former---but it is possible to ask
about the power of a \emph{quantum} verifier given access to a \emph{classical}
proof. 
This question can be studied by introducing the class \class{QCMA}, defined as
\class{QMA} except the proof is restricted to be a classical
polynomial-length string.

It holds that $\class{MA}\subseteq\class{QCMA}\subseteq\class{QMA}$, and none
of these containments is known to be strict.
One may conjecture that \class{QMA} properly contains \class{QCMA}, but little
evidence to support this conjecture is known.

\subsection*{A quantum oracle separation}

One piece of evidence to suggest that \class{QCMA} is properly contained in
\class{QMA} is a \emph{quantum oracle separation} as follows.
Black-box access to a unitary transformation $U$ is made available, and it is
promised that either (i) there exists a state $\ket{\psi}$ such that 
$U\ket{\psi}= -\ket{\psi}$, and $U\ket{\phi}=\ket{\phi}$ for all 
$\ket{\phi}\perp \ket{\psi}$, or (ii) it holds that
$U\ket{\phi}=\ket{\phi}$ for all states $\ket{\phi}$.
The yes-instances are those for which the property (i) holds.

As is to be expected, one can prove that this problem is contained in
\class{QMA} by taking $\ket{\psi}$ to be given as a quantum proof.
It is possible to prove that this problem is not contained in \class{QCMA} under
the assumption that only black-box access to $U$ is permitted.
More precisely, at least $\Omega(2^{k/2}/\sqrt{m})$ queries to $U$ are required
to decide between the two possibilities, given a classical proof of length $m$.
Intuitively speaking, the reason why this is so is that the state $\ket{\psi}$
could be any $k$-qubit state, so the best strategy for a classical prover
is to fix a net over the space of all such states, and to provide the
verifier with a classical description of the element of the net that is closest
to $\ket{\psi}$. 
Using a measure-theoretic argument, it is possible to show that for any
partition of the set of all $k$-qubit pure states into at most $2^m$ regions,
there will exist a region $S$ with the property that, for every state
$\ket{\phi}$, a state $\ket{\psi}$ chosen uniformly from $S$ will have expected
overlap
\begin{equation}
  \textup{E}_{\ket{\psi}\in S}
  |\bra{\phi} \psi\rangle|^2 = O(m 2^{-k})
\end{equation}
with $\ket{\phi}$.
This means that, if the $m$-bit classical proof is interpreted as the
description of such a region, in the worst case the proof will only let the
verifier reconstruct a state $\ket{\phi}$ whose overlap with a randomly
selected state $\ket{\psi}$ satisfies
$|\bra{\phi} \psi\rangle|^2 = O(m 2^{-k})$. 
Using this state as a starting point, and implementing a procedure based on
amplitude amplification, the quantum verifier can find $\ket{\psi}$ using the
gate $U$ as a black box with $\Omega(2^{k/2}/\sqrt{m})$ queries, which can be
shown to be optimal.

It is reasonable to conjecture that the problem suggested above is contained
in \class{QMA}, but not \class{QCMA}, when the unitary $U$ is specified as a
quantum circuit, rather being given as a black box---but naturally the proof
suggested above is not sufficient to establish that this is so, as it does not
rule out the possibility that an analysis of a quantum circuit's structure
could lead to the problem being contained in \class{QCMA}.

\subsection*{Verifying Group Non-Membership using a classical witness}

In Section~\ref{sec:gnm} we introduced the group non-membership (GNM) problem
as a promise problem having a natural \class{QMA} verification procedure. 
As GNM is not known to be \class{QMA}-complete, it is a natural target problem
to put in \class{QCMA}.

The honest witness for GNM has a specific form~\eqref{eq:subgroup-witness},
which is the uniform superposition over all elements in a subgroup. 
As discussed in Section~\ref{sec:super-verifiers}, however, \emph{every}
language in \class{QMA} has a verifier for which there is a witness that has a
similar ``subset state'' form. 
Thus, the form of the GNM witness~\eqref{eq:subgroup-witness} is not directly
indicative of a problem that should be easier than \class{QMA}-complete
problems.

Nevertheless, it can be shown that with the help of a classical polynomial-size
witness, a quantum verifier can decide any instance of GNM using only a
polynomial number of queries to the group oracle.
The catch is that this verifier may require an exponential amount of ``side''
computation---operations that do not involve the group $G$ in question. 
The idea is to use a classical witness to specify a certain ``model group''
$\Gamma$, as well as an injective homomorphism $f:G\to\Gamma$. 
Both can be specified with a polynomial number of bits using an appropriate set
of generators for $G$ and $\Gamma$. 
The verifier can compute the image in $\Gamma$ of $G$, the subgroup $H$, and
$x$, and verify that $f(x)\notin f(H)$ using only polynomially many operations
in $G$ (to decompose $x$ and generators of $H$ on the generators of $G$
provided by the classical witness) and exponentially many operations in
$\Gamma$ (to verify $f(x)\notin f(H)$).

The difficulty of this approach is to verify that the witness has the correct
form, i.e., that the map $f$ is (close to) an injective homomorphism. 
Checking that $f$ is close to a homomorphism can be done efficiently in
randomized polynomial time by verifying the identity $f(g_1g_2)=f(g_1)f(g_2)$
for sufficiently many random pairs of group elements $(g_1,g_2)$.
Thus, the main step is to check injectivity of $f$. 
But this is an instance of the hidden subgroup problem (HSP) and can be solved
by a quantum circuit making polynomially many group operations (and possibly
exponentially many classical operations not involving $G$).

\section{Chapter notes}

Quantum proofs were evidently first discussed by
Knill~\cite{Knill96}, and formalized as a complexity class called
$\class{BQNP}$ by Kitaev around 1999 (and later published
in~\cite{KitaevSV02}).
The name $\class{QMA}$ first appears in~\cite{Watrous00group}, where it was
also shown that group non-membership is in $\class{QMA}$.
Parallel error reduction is analyzed in~\cite{KitaevSV02}, and the procedure
for witness-preserving error reduction is due to~\cite{MarriottW05}, where
the containment $\class{QMA}\subseteq\class{PP}$ was proved but attributed to
unpublished work of Kitaev and Watrous.

The local Hamiltonian problem was introduced by Kitaev in his original work
on \class{QMA} and shown to be $\class{QMA}$-complete for $k=5$. 
This was improved to $k=3$ in~\cite{kempe20033} and $k=2$ 
in~\cite{KempeKR06lh}. 
Physical motivations suggest the consideration of restricted families of 
interaction graphs and types of local Hamiltonians.
For instance, \cite{TerhalO082d} proved $\class{QMA}$-hardness for instances
whose interaction graph is restricted to a two-dimensional grid, and this is
extended to $2$-local Hamiltonians on a line in~\cite{AharonovGIK091d}. 
The paper~\cite{cubitt2014complexity} established a classification of
restricted classes of $2$-local Hamiltonians in terms of their hardness,
providing a quantum analogue of Schaefer's dichotomy theorem for Boolean
constraint satisfaction problems.
The Quantum PCP conjecture was first formulated
in~\cite{aharonov2002quantum,aharonov2009detectability}. 
For further background on the conjecture we refer to the 
survey~\cite{aharonov2013guest}. 

Quantum $k$-SAT was introduced in~\cite{Bravyi11qsat}, where it was shown that
the problem is in $\class{P}$ for $k=2$ and $\class{QMA}_1$-complete for 
$k\geq 4$; $\class{QMA}_1$-completeness for $k=3$ is due 
to~\cite{gosset2013quantum}.
The Consistency of Density Operators problem was shown to be
$\class{QMA}$-complete with respect to Cook reductions
in~\cite{Liu06consistency}.
We refer to the survey~\cite{bookatz2014qma} for a more extensive list of
$\class{QMA}$-complete problems.

The class $\class{SQMA}$ was introduced in~\cite{grilo2014qma}, where an 
analogue for $\class{QMA}(2)$ was also introduced, and the equality 
$\class{SQMA}(2)=\class{QMA}(2)$ was proved along the same lines as 
$\class{SQMA}=\class{QMA}$. 
The class $\class{QMA}_+$ was introduced in~\cite{AharonovR05}, where it was
used to show that $\textrm{coGapSVP}_{\sqrt{n}}$, a gapped version of the 
shortest vector problem in lattices, lies in $\class{QMA}_+=\class{QMA}$. 
At the time it was not known if this problem was contained in $\class{NP}$, but
the authors later ``de-quantized'' their result to obtain this
fact~\cite{AharonovR05}.
The class $\class{BQP}/\class{qpoly}$ was first considered 
in~\cite{nishimura2004polynomial}, and the inclusion 
$\class{BQP}/\class{qpoly}\subseteq\class{QMA}/\class{poly}$ was proved 
in~\cite{aaronson2014full}. 
The class $\QMA(t)$ for $t\geq 2$ was introduced 
in~\cite{KobayashiMY03multiplemerlins}. 
The product test, its analysis, and the equality 
$\class{QMA}(t)=\class{QMA(2)}$ for $t\geq 2$ are due 
to~\cite{harrow2013testing}. 
The existence of a protocol for verifying a $3$-SAT formula on $n$ variables 
using unentangled proofs totalling $O(\sqrt{n}\poly\log n)$ qubits was first
shown in~\cite{AaronsonBDFS09}. 
A different protocol, this time for $3$-coloring~\cite{blier2009all}, is 
analyzed in~\cite{chiesa2013improved} where a trade-off between the number of
proofs used and the is given. 
NP-hardness of the weak membership problem for the set of separable states is
shown in~\cite{Gurvits03} for exponential accuracies, and improved to inverse
polynomial accuracies in~\cite{Gharibian10}. 
A consequence of the product test from~\cite{harrow2013testing} is that weak
membership for constant accuracy cannot be decided in polynomial time unless
$3$-SAT$\in\class{DTIME}(\text{exp}(\sqrt{n}\poly\log n))$.
Better upper bounds than \class{NEXP} are known on $\class{QMA}(2)$ when
additional restrictions on the verifier are
imposed~\cite{brandao2011quasipolynomial}.

The class $\class{QCMA}$ was defined in~\cite{AaronsonK06qcma}, where a quantum
oracle separation with \class{QMA} is proved.
Few complete problems for $\class{QCMA}$ are known; an interesting example is
the Ground State Connectivity (GSCON) problem considered
in~\cite{gharibian2014ground}.

A number of additional variations of the class $\class{QMA}$ have been
considered. 
The class $\class{UQMA}$, or unique $\class{QMA}$, 
corresponds to those problems for which in the yes-case there is a 
one-dimensional subspace of witnesses that convince the verifier to accept with
high probability, while any state in the orthogonal subspace will lead to a
success probability that is smaller by at least a fixed inverse polynomial.
The class $\class{FewQMA}$ is defined similarly, replacing the one-dimensional
subspace of convincing witnesses by a subspace having dimension at most
polynomial in the input size.
These classes, which are analogues of variants of $\class{NP}$ considered by
Valiant and Vazirani~\cite{ValiantV86unique}, were introduced
in~\cite{aharonov2008pursuit}.
It is known that $\class{FewQMA}=\class{UQMA}$~\cite{Jain12unique}, but it is
still open whether these classes equal $\class{QMA}$.

The class $\class{DQMA}$, introduced in~\cite{ambainis2014physical} by analogy 
with a similar extension of $\class{NP}$ called $\class{DP}$
\cite{PapadimitriouY84facets}, consists of all 
those problems whose difference is in $\class{QMA}$:
$(A_{\yes},A_{\no})\in\class{DQMA}$ if there exists $(B_{\yes},B_{\no})$ and
$(C_{\yes},C_{\no})\in\class{QMA}$ such that $x\in A_{\yes}$ implies $x\in
B_{\yes}\cap C_{\no}$ and $x\in A_{\no}$ implies $x\in B_{\no}\cup C_{\yes}$,
as well as $x\in (B_{\yes}\cup B_{\no}) \cap (C_{\yes}\cup C_{\no})$.
Complete problems for $\class{DQMA}$ that do not appear to lie in $\class{QMA}$
are given in~\cite{ambainis2014physical}; these include the problem of deciding
whether the ground state energy of a local Hamiltonian lies in a certain
interval, or is outside of that interval.

Quantum analogues of classes higher in the polynomial hierarchy such as
$\Sigma_2^p$ and $\Pi_2^p$ are introduced in~\cite{GharibianK12hierarchy},
where complete problems for these classes are given.

\chapter{Single-Prover Quantum Interactive Proofs}
\label{chapter:single-prover}

This chapter introduces a quantum computational analogue of the most standard
interactive proof system model, in which a verifier interacts with a single
prover, and surveys several known results concerning this model.
A few highlights of the results to be discussed in the chapter are as follows:

\begin{mylist}{\parindent}
\item[1.]
  Quantum interactive proof systems can be \emph{parallelized} to three-turn
  interactive proof systems having strong error bounds.
  More precisely, any promise problem having a bounded-error, polynomial-turn,
  single-prover quantum interactive proof system must also have a three-turn,
  single-prover quantum interactive proof system with perfect completeness and
  exponentially small soundness error.

\item[2.]
  The problem of optimizing the probability for a verifier in a single-prover
  quantum interactive proof system to accept a given string can be represented
  as a semidefinite program in a fairly simple and direct way.
  This representation provides a useful tool for reasoning about single-prover
  quantum interactive proof systems.

\item[3.]
  The class of promise problems having bounded-error, single-prover quantum
  interactive proof systems coincides with \class{PSPACE}.
\end{mylist}

\section{Definitions of quantum interactive proof systems}
\label{section:definitions}

In the classical setting, interactive proof systems have historically been
defined through variants of the probabilistic Turing machine model, modified in
such a way as to allow for interactions with an external entity (such as
another Turing machine).
Quantum computation, on the other hand, is more conveniently modeled by
quantum circuits, as was suggested in Chapter~\ref{chapter:preliminaries}, and
for this reason our definitions of quantum interactive proof systems will be
based on circuits rather than Turing machines.

The precise definitions of quantum interactive proof systems that we adopt in
this survey are essentially the same as ones considered in prior work on the
subject, although we will place a somewhat greater emphasis on the fixed-size
interactions that are induced by interactive proof systems on fixed input
strings.
The notion of an \emph{interactive game}, to be introduced shortly, is intended
to be an abstraction of this sort of interaction.

\subsection{Interactive games}

The first part of the definition of quantum interactive proof systems involves
the introduction of the general notion of an \emph{interactive game}.
Intuitively speaking, interactive games are abstractions of fixed-size
interactions, in which there is no notion of a shared input string to the
participants.
Although the focus of the current chapter will be on quantum interactions,
the concept of an interactive game is not inherently quantum---the basic
concept can be adapted to the classical setting is a straightforward manner.

The notion of an interactive game is, in fact, sufficiently general that one
may adapt it to formulate definitions of other interactive proof system
variants (such as interactive proof systems with multiple provers, either
cooperating or competing, and zero-knowledge interactive proofs), as well as
cryptographic interactions such as coin-flipping.
Indeed, it will likely be quite evident from the discussion that follows
that the notion of an interactive game can be generalized in numerous ways,
allowing for three or more participants, outputs for any subset of the
participants, and so on.
Although we will consider some such adaptations and generalizations in other
parts of this survey, we will not attempt to emphasize the generality of the
notion at this stage; our focus here will be limited to interactive games that
are representative of interactions between two entities, playing the roles of
prover and verifier.

With this focus in mind, an interactive game describes a situation in which
two participants, a \emph{prover} and a \emph{verifier}, exchange fixed-size
quantum registers for a fixed number of steps.
At the end of the interaction, the verifier produces a single classical bit
as output.
Figure~\ref{fig:interactive-game-1} illustrates an interactive game of this
sort, in the particular case in which five register exchanges (or
\emph{turns}) occur during the interaction.\footnote{%
  One may alternatively count \emph{messages} rather than \emph{turns}.
  In the single-prover setting, these notions are equivalent, but in the
  multi-prover setting (to be considered in 
  Chapter~\ref{chapter:multiple-provers}) it is convenient to consider that a
  turn may involve a collection of messages being either sent or received
  in parallel by a single entity.}
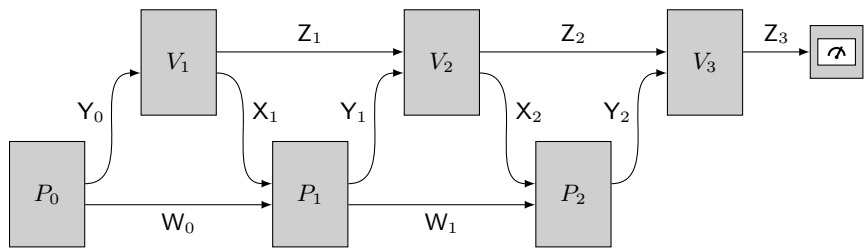
\begin{figure}
  \begin{center}
    \footnotesize
    \begin{tikzpicture}[scale=0.35, 
        turn/.style={draw, minimum height=14mm, minimum width=10mm,
          fill = ChannelColor, text=ChannelTextColor},
        measure/.style={draw, minimum width=7mm, minimum height=7mm,
          fill = ChannelColor},
        >=latex]
      
      \node (V1) at (-8,4) [turn] {$V_1$};
      \node (V2) at (2,4) [turn] {$V_2$};
      \node (V3) at (12,4) [turn] {$V_3$};
      
      \node (M) at (17,4.4) [measure] {};
     
      \node (P0) at (-13,-1) [turn] {$P_0$};
      \node (P1) at (-3,-1) [turn] {$P_1$};
      \node (P2) at (7,-1) [turn] {$P_2$};

      \node[draw, minimum width=5mm, minimum height=3.5mm, fill=ReadoutColor]
      (readout) at (M) {};
      
      \draw[thick] ($(M)+(0.3,-0.15)$) arc (0:180:3mm);
      \draw[thick] ($(M)+(0.2,0.2)$) -- ($(M)+(0,-0.2)$);
      \draw[fill] ($(M)+(0,-0.2)$) circle (0.5mm);
      
      \draw[->] ([yshift=4mm]V1.east) -- ([yshift=4mm]V2.west)
      node [above, midway] {$\reg{Z}_1$};
      
      \draw[->] ([yshift=4mm]V2.east) -- ([yshift=4mm]V3.west)
      node [above, midway] {$\reg{Z}_2$};
      
      \draw[->] ([yshift=-4mm]V1.east) .. controls +(right:20mm) and 
      +(left:20mm) .. ([yshift=4mm]P1.west) node [right, pos=0.4] {$\reg{X}_1$};
      
      \draw[->] ([yshift=-4mm]V2.east) .. controls +(right:20mm) and 
      +(left:20mm) .. ([yshift=4mm]P2.west) node [right, pos=0.4] {$\reg{X}_2$};
      
      \draw[->] ([yshift=4mm]P0.east) .. controls +(right:20mm) and 
      +(left:20mm) .. ([yshift=-4mm]V1.west) node [left, pos=0.6] {$\reg{Y}_0$};
      
      \draw[->] ([yshift=4mm]P1.east) .. controls +(right:20mm) and 
      +(left:20mm) .. ([yshift=-4mm]V2.west) node [left, pos=0.6] {$\reg{Y}_1$};
      
      \draw[->] ([yshift=4mm]P2.east) .. controls +(right:20mm) and 
      +(left:20mm) .. ([yshift=-4mm]V3.west) node [left, pos=0.6] {$\reg{Y}_2$};
      
      \draw[->] ([yshift=-4mm]P0.east) -- ([yshift=-4mm]P1.west)
      node [below, midway] {$\reg{W}_0$};
      
      \draw[->] ([yshift=-4mm]P1.east) -- ([yshift=-4mm]P2.west)
      node [below, midway] {$\reg{W}_1$};
      
      \draw[->] ([yshift=4mm]V3.east) -- (M.west)
      node [above, midway] {$\reg{Z}_3$};
      
    \end{tikzpicture}
  \end{center}
  \caption{A five-turn interactive game involving a prover and a verifier.
    The prover's actions are represented by the boxes labeled $P_0$, $P_1$, and
    $P_2$, while the verifier's actions are represented by the boxes labeled
    $V_1$, $V_2$, and $V_3$.
    It is a five-turn interactive game because five registers are exchanged
    during the interaction that is illustrated:
    in the first turn the prover sends the register $\reg{Y}_0$ to the
    verifier, in the second turn the verifier sends $\reg{X}_1$ to the
    prover, and so on.
    The registers $\reg{Z}_1$ and $\reg{Z}_2$ represent the verifier's memory
    registers, while $\reg{W}_0$ and $\reg{W}_1$ represent the prover's memory
    registers.
    The rightmost box represents a measurement that produces a classical output
    bit.
  }
  \label{fig:interactive-game-1}
\end{figure}
The actions performed by a prover and verifier at each step of an interactive
game must be valid physical operations, at least in the idealized sense that is
modeled by the theory of quantum information, and must therefore be described
by quantum channels.

When it is convenient, we will refer to an \emph{$m$-turn verifier}, an
\emph{$m$-turn prover}, or an \emph{$m$-turn interactive game} to
indicate that $m$ register exchanges between the verifier and prover take
place, for the verifier, prover, or interactive game being considered.
A prover and verifier in an interactive game must naturally be
\emph{compatible}, in the sense that they agree on both the number and timing
of the register exchanges and on the sizes of these registers.
Hereafter, we will take it as an implicit assumption that such an agreement is
in place, as there is little to be said about interactions between a
prover and verifier that are not compatible.

The pattern represented by the labeling of the registers and transformations
in Figure~\ref{fig:interactive-game-1} will be mimicked throughout this survey
to the extent that it is possible.
In general, the verifier's actions will be described by channels
$V_1,\ldots,V_n$, with each channel $V_k$ transforming the register pair
$(\reg{Z}_{k-1},\reg{Y}_{k-1})$ to $(\reg{Z}_k,\reg{X}_k)$, as suggested by
Figure~\ref{fig:interactive-game-verifier-action}, while the prover's 
actions are described by channels $P_0,\ldots,P_{n-1}$ (in case the number of
turns is odd) or $P_1,\ldots,P_{n-1}$ (in case the number of turns is
even), with each channel $P_k$ transforming $(\reg{X}_k,\reg{W}_{k-1})$ to
$(\reg{Y}_k,\reg{W}_k)$, as suggested by 
Figure~\ref{fig:interactive-game-prover-action}.
It will be possible to avoid the need to handle special cases at the beginning
and end of interactions by adopting the convention that ``absent'' registers,
such as $\reg{Z}_0$, $\reg{X}_n$, $\reg{X}_0$, and $\reg{W}_{-1}$, and possibly
$\reg{Y}_0$, are identified with trivial registers comprising zero
qubits.\footnote{
  Registers having zero qubits are legitimate quantum systems that have a
  single classical state, and a corresponding Hilbert space equal to
  $\complex$.
  We do not consider that the transmission of trivial registers contributes to
  the number of turns in an interaction.}

\begin{figure}
  \begin{center}
  \begin{minipage}[b]{0.45\textwidth}
    \begin{center}
      \begin{tikzpicture}[scale=0.5, 
          turn/.style={draw, minimum height=14mm, minimum width=10mm,
            fill = ChannelColor, text=Black},
          measure/.style={draw, minimum width=7mm, minimum height=7mm,
            fill = ChannelColor},
          >=latex]
        
        \node (V) at (0,4) [turn] {$V_k$};
        
        \node (previousV) at (-5,4) {};
        \node (previousP) at (-4,0) {};
        
        \node (nextV) at (5,4) {};
        \node (nextP) at (4,0) {};
        
        \draw[->] ([yshift=4mm]previousV.east) -- ([yshift=4mm]V.west)
        node [above, midway] {$\reg{Z}_{k-1}$};
        
        \draw[->] ([yshift=4mm]V.east) -- ([yshift=4mm]nextV.west)
        node [above, midway] {$\reg{Z}_{k}$};
        
        \draw[->] ([yshift=-4mm]V.east) .. controls +(right:20mm) and 
        +(left:20mm) .. ([yshift=4mm]nextP.west) node [right, pos=0.5]
        {$\reg{X}_k$};
        
        \draw[->] ([yshift=4mm]previousP.east) .. controls +(right:20mm) and 
        +(left:20mm) .. ([yshift=-4mm]V.west) node [left, pos=0.5]
        {$\reg{Y}_{k-1}$};
        
      \end{tikzpicture}
    \end{center}
    \caption{A general verifier action.}
    \label{fig:interactive-game-verifier-action}
  \end{minipage}\hspace{6mm}
  \begin{minipage}[b]{0.45\textwidth}
    \begin{center}
      \begin{tikzpicture}[scale=0.5, 
          turn/.style={draw, minimum height=14mm, minimum width=10mm,
            fill = ChannelColor, text=Black},
          measure/.style={draw, minimum width=7mm, minimum height=7mm,
            fill = ChannelColor},
          >=latex]
        
        \node (P) at (0,0) [turn] {$P_k$};
        
        \node (previousV) at (-4,4) {};
        \node (previousP) at (-5,0) {};
        
        \node (nextV) at (4,4) {};
        \node (nextP) at (5,0) {};
        
        \draw[->] ([yshift=-4mm]previousV.east) 
        .. controls +(right:20mm) and +(left:20mm) ..
        ([yshift=4mm]P.west) node [right, pos=0.4] {$\reg{X}_k$};
        
        \draw[->] ([yshift=4mm]P.east) 
        .. controls +(right:20mm) and +(left:20mm) ..
        ([yshift=-4mm]nextV.west) node [left, pos=0.6] {$\reg{Y}_k$};
        
        \draw[->] ([yshift=-4mm]P.east) --
        ([yshift=-4mm]nextP.west) node [below, pos=0.5] {$\reg{W}_k$};
        
        \draw[->] ([yshift=-4mm]previousP.east) --
        ([yshift=-4mm]P.west) node [below, pos=0.5] {$\reg{W}_{k-1}$};
      
      \end{tikzpicture}
    \end{center}
    \caption{A general prover action.}
    \label{fig:interactive-game-prover-action}
  \end{minipage}
  \end{center}
\end{figure}

Our primary interest is in the situation in which the specification of a
verifier is fixed, and it is to be viewed that one's goal is to optimize the
actions of the prover so as to cause the verifier to produce the output 1
(representing \emph{acceptance} in the terminology of interactive proofs).
Generalizing the notation and terminology used in the previous chapter,
we write $\omega(V)$ to denote the \emph{value} of a given verifier $V$ in an
interactive game, which is defined as the maximum
probability with which a prover (compatible with $V$) can cause $V$ to
output 1.\footnote{%
  Formally speaking, the value of a verifier $V$ is more naturally defined as
  the \emph{supremum} probability with which a compatible prover can cause $V$
  to output~1, as it is not immediate that this supremum value is achieved by a
  single prover.
  The supremum \emph{is} achieved, however, so one is justified in considering
  the value as the maximum probability with which a prover can cause $V$ to
  output 1.}

\begin{example}
  The following simple example is intended to illustrate the basic concept of
  an interactive game.
  Define a two-turn verifier $V$ as follows:
  \begin{mylist}{\parindent}
  \item[1.]
    The verifier's first action represents the creation of a maximally
    entangled state
    \begin{equation}
      \ket{\psi} = \frac{1}{\sqrt{2}}\ket{0}\ket{0} +
      \frac{1}{\sqrt{2}}\ket{1}\ket{1}
    \end{equation}
    of a pair of single-qubit registers $(\reg{Z}_1,\reg{X}_1)$.
    The qubit $\reg{X}_1$ is sent to the prover.
  \item[2.]
    The verifier's second action represents a binary-valued measurement
    performed on the pair $(\reg{Z}_1,\reg{Y}_1)$, for $\reg{Y}_1$ being a
    single-qubit register received from the prover.
    The measurement operator corresponding to the outcome 1 is defined as
    $\Pi_1 = \ket{\phi} \bra{\phi}$ for
    \begin{equation}
      \ket{\phi} = \cos(\pi/8)\ket{0}\ket{0} + \sin(\pi/8)\ket{1}\ket{1},
    \end{equation}
    while the measurement operator corresponding to the outcome 0 is
    $\Pi_0 = \I - \ket{\phi}\bra{\phi}$.
  \end{mylist}
  The optimal probability with which a prover can cause this verifier to output
  1 is equal to
  \begin{equation}
    \omega(V) = \cos^2(\pi/8) \approx 0.85.
  \end{equation}
  It is easy to see that this probability is achievable:
  a prover may simply return the register $\reg{X}_1$ to the verifier,
  renaming it $\reg{Y}_1$ but otherwise leaving it unchanged, which causes the
  verifier to output 1 with probability
  \begin{equation}
    \bigabs{\langle \phi | \psi \rangle}^2 = 
    \Biggabs{\frac{\cos(\pi/8) + \sin(\pi/8)}{\sqrt{2}}}^2 = \cos^2(\pi/8).
  \end{equation}
  The fact that this probability is optimal follows from the fact that the
  probability with which the verifier accepts, for any choice of a prover, is
  given by
  \begin{equation}
    \tr \bigl( \Pi_1 \rho\bigr) = \fid\bigl(\ket{\phi}\bra{\phi},\rho\bigr)^2
  \end{equation}
  for some two-qubit state $\rho$ whose first qubit, when viewed in isolation,
  is completely mixed.
  By the fact that the fidelity is non-decreasing under partial tracing, one
  finds that the probability of acceptance is at most
  \begin{equation}
    \fid \Biggl(
    \begin{pmatrix}
      1/2 & 0\\
      0 & 1/2
    \end{pmatrix},
    \begin{pmatrix}
      \cos^2(\pi/8) & 0\\
      0 & \sin^2(\pi/8)
    \end{pmatrix}
    \Biggr)^2 = \cos^{2}(\pi/8),
  \end{equation}
  which establishes the optimality of this acceptance probability.

  This example illustrates an important theme in the analysis of quantum
  interactive games, which is that a prover's possible actions exactly
  correspond to those transformations that leave the reduced state of the
  verifier's register unchanged (under the assumption that the joint state of
  the prover and verifier is pure).
  The analysis made here will re-appear in the proof of the perfect
  completeness property in Section~\ref{sec:qip-completeness}, and the idea is
  the key to the formulation of the value of an interactive game as the optimum
  of a semidefinite program to be described in Section~\ref{sec:qip-sdp}.
\end{example}

In the previous example, we have not specified the verifier's actions as
quantum channels, at least in a formal sense.
It is, however, possible to do this.
In particular, the verifier's first action corresponds to a channel $V_1$ that
takes no input (or, equivalently, takes a pair of trivial registers
$(\reg{Z}_0,\reg{Y}_0)$ as input) and outputs the state $\ket{\psi}\bra{\psi}$
contained in the pair $(\reg{Z}_1,\reg{X}_1)$.
More formally speaking, this channel corresponds to the mapping
\begin{equation}
  V_1(\alpha) = \alpha \ket{\psi}\bra{\psi}
\end{equation}
for all $\alpha \in \complex$, which is a completely positive and
trace-preserving map.
The second action may be expressed as a channel as well, in this case
transforming the pair of registers $(\reg{Z}_1,\reg{Y}_1)$ into a single-qubit
register $\reg{Z}_2$ (or, equivalently, into a pair $(\reg{Z}_2,\reg{X}_2)$
where $\reg{X}_2$ is trivial) in the manner described by the mapping
\begin{equation}
  V_2 (X) = \ip{\Pi_0}{X} \ket{0}\bra{0} + \ip{\Pi_1}{X} \ket{1}\bra{1}
\end{equation}
for all $X\in\Lin(\Z_1\otimes\Y_1)$.

For other interactive games described in this chapter, we will generally omit
the sorts of details that have been given in the previous paragraph---it is
usually a routine exercise to fill in such details.

\subsection{Descriptions and encodings of interactive games}

There are two natural ways in which one may describe either or both of the 
participants in an interactive game: one is by \emph{quantum circuits}, and the
other is by \emph{explicit matrix representations} of the participants'
actions.

\begin{mylist}{\parindent}
\item[1.] \emph{Quantum circuit representations.}
  An $m$-turn verifier is determined by an $n$-tuple $V = (V_1,\ldots,V_n)$,
  for $n = \lfloor m/2 + 1 \rfloor$, where each $V_k$ is a quantum channel
  transforming a pair of registers $(\reg{Z}_{k-1},\reg{Y}_{k-1})$ to a pair
  of registers $(\reg{Z}_k,\reg{X}_k)$.
  A quantum circuit description of such a verifier is simply an $n$-tuple of
  quantum circuits, each implementing one of these channels, along with a
  specification of which input and output qubits of each circuit are to be
  associated with the two registers in each pair.
  Each individual circuit may be encoded following the general principles
  outlined in Section~\ref{sec:circuits}.
  Provers can be described in an analogous manner
  (although we are typically not concerned with the efficiency of provers,
  making circuit descriptions of them generally less useful).

\item[2.] \emph{Explicit matrix representations.}
  As above, an $m$-turn verifier is determined by an $n$-tuple of quantum
  channels $V = (V_1,\ldots,V_n)$, for $n = \lfloor m/2 + 1 \rfloor$.
  Along with a specification of which input and output qubits of each channel
  are to be associated with the two registers in each register pair, one may
  describe each channel $V_k$ by an explicit matrix representation
  (such as a Stinespring representation).
  Again, provers may be represented in an analogous way.
\end{mylist}

\subsection{Quantum interactive proof systems}

Having defined interactive games, we are now prepared to define various
complexity classes based on the concept of quantum interactive proof systems.
Intuitively speaking, we view a quantum interactive proof system to be the
specification of a quantum interactive game for each possible input string to
the problem being considered.

As is typical for interactive proof system models, we will constrain verifiers
in interactive game representations of quantum interactive proof systems to be
computationally bounded.
To be more precise, we require the verifier's actions to be represented by
quantum circuits whose descriptions can be generated in polynomial time from
the problem input.
The fact that the prover is not computationally bounded is manifested in the
requirement that the maximum acceptance probability of a given verifier places
no computational restrictions on the prover's actions.

\begin{definition}
  A promise problem $A = (A_{\yes},A_{\no})$ is contained in the complexity
  class $\class{QIP}_{a,b}(m)$ if there exists a polynomial-time computable
  function $V$ that possesses the following properties:
  \begin{mylist}{\parindent}
  \item[1.] For every string $x\in A_{\yes} \cup A_{\no}$, one has that $V(x)$
    is an encoding of a quantum circuit description of an $m$-turn verifier in
    an interactive game.
  \item[2.] \emph{Completeness.} 
    For every string $x\in A_{\yes}$, it holds that $\omega(V(x)) \geq a$.
  \item[3.] \emph{Soundness.}
    For every string $x\in A_{\no}$, it holds that $\omega(V(x)) \leq b$.
  \end{mylist}
\end{definition}

In this definition, one may take $m$, $a$, and $b$ to be constants or functions
of the length of $x$.
When $a$ and $b$ are omitted, it is to be understood that $a = 2/3$ and
$b = 1/3$, so that
\begin{equation}
  \class{QIP}(m) = \class{QIP}_{2/3,1/3}(m).
\end{equation}
We also write $\class{QIP}$, without specifying a number of turns, to refer to
the class of promise problems $A$ for which there exists a polynomially bounded
function $m$ such that $A\in\class{QIP}(m)$.
Given that a polynomial-time computable function $V$ representing a verifier
would not be capable of outputting an $n$-tuple of quantum circuit descriptions
with $n$ being super-polynomial in $\abs{x}$, this is equivalent to placing no
restrictions on the number of turns.

It should be noted that any classical verifier in an interactive game can be
viewed as a restricted type of quantum verifier.
It is not difficult to prove that quantum prover strategies cannot gain an
advantage over optimal classical prover strategies against classical
verifiers, and based on this observation one may verify that
\begin{equation}
  \class{IP}_{a,b}(m)  \subseteq \class{QIP}_{a,b}(m)
\end{equation}
for all choices of $a$, $b$, and $m$.\footnote{The situation is not nearly so
  simple in the multi-prover setting, as will be discussed in
  Chapter~\ref{chapter:multiple-provers}.}

\subsection{Purifications of interactive games}
\label{sec:qip-purification}

Interactive games are, in some situations, easier to analyze when the joint
state of all of the co-existing registers at each instant is a pure state, as
opposed to being an arbitrary mixed state.
Indeed, for some of the proof techniques we will use later in this chapter,
this assumption of purity is essential.

Fortunately, there is no generality lost in restricting one's attention to
interactive games with this property.
This follows from the fact that each individual channel performed by either of
the participants in an interactive game may be purified in the manner described
in Section~\ref{sec:channel-purification}, so that each action is represented
by a linear isometry.
The additional output qubits produced by this process must be considered
private memory qubits for whichever player performs that particular channel.
Assuming that these additional qubits are not touched again during the
interactive game (so that subsequent actions of the player act trivially on
these qubits), the effect is identical to the original channel.
This process is illustrated in Figure~\ref{fig:purify-three-messages}.
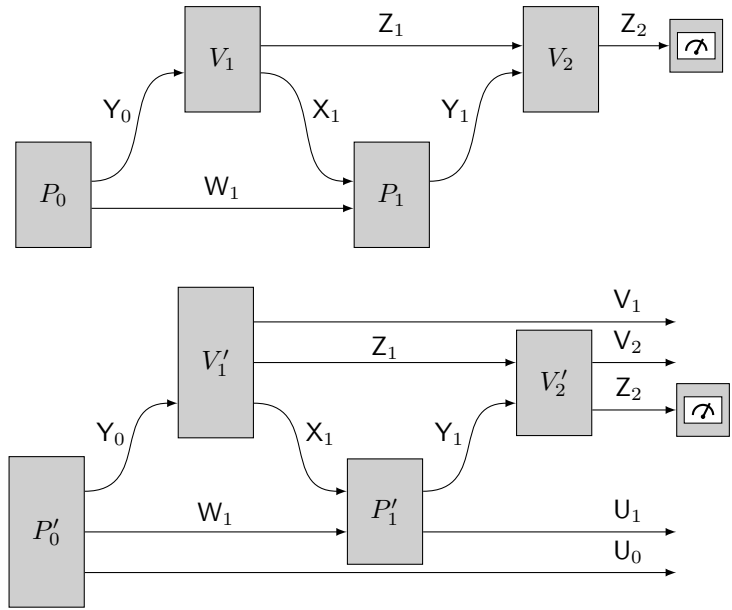
\begin{figure}[t]
  \begin{center}\small
    \begin{tikzpicture}[scale=0.45,
        turn/.style={draw, minimum height=14mm, minimum width=10mm,
          fill = ChannelColor, text=ChannelTextColor},
        measure/.style={draw, minimum width=7mm, minimum height=7mm,
          fill = ChannelColor},
        >=latex]
      
      \node (V1) at (-8,4) [turn] {$V_1$};
      \node (V2) at (2,4) [turn] {$V_2$};
      
      \node (M) at (6,4.4) [measure] {};
      
      \node[draw, minimum width=5mm, minimum height=3.5mm, fill=ReadoutColor]
      (readout) at (M) {};
      
      \draw[thick] ($(M)+(0.3,-0.15)$) arc (0:180:3mm);
      \draw[thick] ($(M)+(0.2,0.2)$) -- ($(M)+(0,-0.2)$);
      \draw[fill] ($(M)+(0,-0.2)$) circle (0.5mm);
      
      \node (P0) at (-13,0) [turn] {$P_0$};
      \node (P1) at (-3,0) [turn] {$P_1$};
      
      \draw[->] ([yshift=4mm]V1.east) -- ([yshift=4mm]V2.west)
      node [above, midway] {$\reg{Z}_1$};
      
      \draw[->] ([yshift=-4mm]V1.east) .. controls +(right:20mm) and 
      +(left:20mm) .. ([yshift=4mm]P1.west) node [right, pos=0.4] {$\reg{X}_1$};
      
      \draw[->] ([yshift=4mm]P0.east) .. controls +(right:20mm) and 
      +(left:20mm) .. ([yshift=-4mm]V1.west) node [left, pos=0.6] {$\reg{Y}_0$};
      
      \draw[->] ([yshift=4mm]P1.east) .. controls +(right:20mm) and 
      +(left:20mm) .. ([yshift=-4mm]V2.west) node [left, pos=0.6] {$\reg{Y}_1$};
      
      \draw[->] ([yshift=-4mm]P0.east) -- ([yshift=-4mm]P1.west)
      node [above, midway] {$\reg{W}_1$};
      
      \draw[->] ([yshift=4mm]V2.east) -- (M.west) node [above, midway]
           {$\reg{Z}_2$};
    \end{tikzpicture}\\[4mm]
    \begin{tikzpicture}[scale=0.45,
        turn/.style={draw, minimum height=14mm, minimum width=10mm,
          fill = ChannelColor, text=ChannelTextColor},
        bigturn/.style={draw, minimum height=20mm, minimum width=10mm,
          fill = ChannelColor, text=ChannelTextColor},
        measure/.style={draw, minimum width=7mm, minimum height=7mm,
          fill = ChannelColor},
        >=latex]
      
      \node (V1) at (-8,4.6) [bigturn] {$V_1'$};
      \node (V2) at (2,4) [turn] {$V_2'$};
      \node (M) at (6.4,3.2) [measure] {};
      
      \node[draw, minimum width=5mm, minimum height=3.5mm, fill=ReadoutColor]
      (readout) at (M) {};
      
      \draw[thick] ($(M)+(0.3,-0.15)$) arc (0:180:3mm);
      \draw[thick] ($(M)+(0.2,0.2)$) -- ($(M)+(0,-0.2)$);
      \draw[fill] ($(M)+(0,-0.2)$) circle (0.5mm);
      
      \node (P0) at (-13,-0.4) [bigturn] {$P_0'$};
      \node (P1) at (-3,0.2) [turn] {$P_1'$};
      
      \node [minimum width=7mm] (Pout0) at (6.4,-1.6) {};
      \node [minimum width=7mm] (Pout1) at (6.4,-0.4) {};

      \node [minimum width=7mm] (Vout1) at (6.4,5.8) {};
      \node [minimum width=7mm] (Vout2) at (6.4,4.6) {};

      \draw[->] ([yshift=0mm]V1.east) -- ([yshift=6mm]V2.west)
      node [above=-2pt, midway] {$\reg{Z}_1$};
      
      \draw[->] ([yshift=-12mm]V1.east) .. controls +(right:20mm) and 
      +(left:20mm) .. ([yshift=6mm]P1.west) node [right, pos=0.4] 
      {$\reg{X}_1$};
      
      \draw[->] ([yshift=12mm]P0.east) .. controls +(right:20mm) and 
      +(left:20mm) .. ([yshift=-12mm]V1.west) node [left, pos=0.6] 
      {$\reg{Y}_0$};
      
      \draw[->] ([yshift=6mm]P1.east) .. controls +(right:20mm) and 
      +(left:20mm) .. ([yshift=-6mm]V2.west) node [left, pos=0.6] 
      {$\reg{Y}_1$};
      
      \draw[->] ([yshift=0mm]P0.east) -- ([yshift=-6mm]P1.west) node
           [above=-1pt, midway] {$\reg{W}_1$};
      
      \draw[->] ([yshift=-8mm]V2.east) -- (M.west)
           node [above=8pt, left=9pt] {$\reg{Z}_2$};

      \draw[->] ([yshift=-6mm]P1.east) -- (Pout1.west) node 
           [above=8pt, left=9pt] {$\reg{U}_1$};

      \draw[->] ([yshift=-12mm]P0.east) -- (Pout0.west) node 
           [above=8pt, left=9pt] {$\reg{U}_0$};

      \draw[->] ([yshift=12mm]V1.east) -- (Vout1.west) node 
           [above=8pt, left=9pt] {$\reg{V}_1$};

      \draw[->] ([yshift=6mm]V2.east) -- (Vout2.west) node 
           [above=8pt, left=9pt] {$\reg{V}_2$};
     
    \end{tikzpicture}
  \end{center}
  \caption{A three-message interactive game and its purification.}
  \label{fig:purify-three-messages}
\end{figure}
In the situation in which the efficiency of a verifier is concerned,
this transformation is done gate-by-gate rather than turn-by-turn, as
discussed in Section~\ref{sec:circuit-purification}.

Naturally, it may also be assumed that the actions performed by each
participant correspond to unitary transformations, as opposed to
transformations described by linear isometries, provided that sufficiently many
ancillary input qubits are provided to each of these unitary tranformations.
Again, the initialized qubits must be understood to be included in a given
participant's private memory, so that the other participant may not tamper with
them to influence the output of the game.

\section{Perfect completeness and parallelization}
\label{sec:QIP-perfect-completeness-and-parallelization}

Two basic facts concerning single-prover interactive proof systems will be
discussed in the present section.
These facts may be proved through direct constructions, which efficiently
transform quantum interactive games in ways that allow one to conclude that the
facts hold.

The first construction establishes the relation
\begin{equation}
  \class{QIP}_{a,b}(m) \subseteq \class{QIP}_{1,c}(m+2),
\end{equation}
for any number of turns $m$, and for $c$ being bounded away from 1 as a
function of the gap $a-b$ between the completeness and soundness parameters.
This relation implies that quantum interactive proof systems do not lose
any power when restricted to having \emph{perfect completeness}, as long as one
is willing to increase the number of turns by two.

The second construction establishes the relation
\begin{equation}
  \class{QIP}_{1,c}(m) \subseteq \class{QIP}_{1,d}(3)
\end{equation}
for any polynomially bounded function $m$, for $d$ being bounded away from 1
when the same is true of $c$.
This implies that quantum interactive proof systems can be \emph{parallelized}
to a high degree, which is an important property that distinguishes them from
classical interactive proof systems.

\subsection{Perfect completeness}
\label{sec:qip-completeness}

Given a quantum circuit description of a verifier $V$ in a quantum interactive
game, as well as a target threshold $\alpha$ for its value $\omega(V)$, it is
possible to efficiently construct a new verifier $V'$ in such a way that the
following properties are in place:
\begin{mylist}{\parindent}
\item[1.]
  If $V$ is an $m$-turn verifier, then $V'$ is an $(m+2)$-turn verifier.
\item[2.]
  If it is the case that $\omega(V) \geq \alpha$, then $\omega(V') = 1$.
\item[3.]
  If it is the case that $\omega(V) < \alpha$, then 
  $\omega(V') < 1 - (\alpha - \omega(V))^2$.
\end{mylist}
We make the assumption that $\alpha$ is a dyadic rational in the construction,
which is to be explained shortly.
The construction has the following implication to quantum interactive proof
system classes.

\begin{theorem}
  \label{theorem:QIP-perfect-completeness}
  For any choice of polynomial-time computable functions
  $a,b:\natural\rightarrow(0,1)$ and $m:\natural\rightarrow\natural$ for which
  $a(n) < b(n)$ and $m(n) \geq 1$ for every $n\in\natural$, it holds that
  \begin{equation}
    \class{QIP}_{a,b}(m) \subseteq \class{QIP}_{1,c}(m+2)
  \end{equation}
  for
  \begin{equation}
    c = 1 - \frac{1}{2}(a-b)^2.
  \end{equation}
\end{theorem}

Thus, quantum interactive proof systems with \emph{perfect completeness}
are at least as powerful as those having a nonzero completeness error, provided
that one allows the quantum interactive proof system with perfect completeness
to have two additional turns and a somewhat larger soundness error.
(As we will soon see, a reduction in the number of turns and in the soundness
error is possible through the use of other methods, so these are not major
concessions to make for the property of perfect completeness.)

To describe the essential idea behind the construction, it will be convenient
to first consider the case in which $\alpha = 1/2$, delaying the discussion of
how one may handle other values of $\alpha$ until later.
It will be assumed that the given verifier $V$ is in a purified form, so that
its actions are described by an $n$-tuple of isometries $(V_1,\ldots,V_n)$.
Moreover, it will be assumed that if $\omega(V) \geq 1/2$, then there exists a
prover $P$ that is capable of making $V$ output 1 with probability
\emph{exactly} 1/2.
If these assumptions were not in place, it would be straightforward to
preprocess the description of the given verifier to ensure that these
assumptions are met;
the process of purifying quantum channels has already been discussed, and the
second assumption can be imposed by wiring the verifier with an additional
qubit that allows a prover to force the output 0 if it chooses, effectively
throwing the game with any desired probability.
Finally, it will be assumed, without loss of generality, that at the end of the
interaction, the first of the verifier's qubits is considered to be the output
qubit, which is measured with respect to the computational basis to produce the
verifier's output bit.

With these assumptions in place, the construction of the $(m+2)$-turn
verifier $V'$ from $V$ is as follows:
\begin{mylist}{\parindent}
\item[1.]
  The verifier $V'$ behaves precisely as $V$ does for $m$ turns, up to
  but not including the final measurement of the output qubit of~$V$.

\item[2.]
  The verifier $V'$ then makes a \emph{pseudo-copy} of the output qubit of $V$
  by performing the isometry defined by
  \begin{equation}
    \ket{00}\bra{0} + \ket{11}\bra{1}
  \end{equation}
  on this output qubit.
  In turn number $m+1$ of the interaction, $V'$ sends all of its qubits to the
  prover, aside from one of the two qubits produced by the pseudo-copy
  operation.

\item[3.]
  In the final turn, the verifier $V'$ receives a single qubit.
  It then measures the two qubits it holds (one from the pseudo-copy operation
  and the other received from the prover) against the two-qubit state
  \begin{equation}
    \label{eq:gamma-for-alpha=1/2}
    \ket{\gamma} = \frac{1}{\sqrt{2}} \ket{00} + \frac{1}{\sqrt{2}} \ket{11}
  \end{equation}
  (i.e., with respect to a measurement having measurement operators
  $\ket{\gamma}\bra{\gamma}$ and $\I - \ket{\gamma}\bra{\gamma}$).
  If the measurement outcome is consistent with this target state, it outputs
  1, and otherwise it outputs 0.
\end{mylist}

Consider the case in which there exists a prover $P$ that causes $V$ to accept
with probability exactly 1/2.
It may be assumed that $P$ has been purified, as the purification process
has no effect on the probability of acceptance.
As the acceptance probability of $V$ is 1/2, the final state of the interactive
game, immediately before a measurement of the verifier's output qubit takes
place, must take the following form:
\begin{equation}
  \frac{1}{\sqrt{2}}\ket{0}\ket{\phi_0}+\frac{1}{\sqrt{2}}\ket{1}\ket{\phi_1}.
\end{equation}
Here, the first qubit represents the verifier's output qubit, and the vectors
$\ket{\phi_0}$ and $\ket{\phi_1}$ represent all of the other qubits, including
all of the verifier's private qubits aside from the output qubit and all of the
prover's private qubits.
It need not be the case that $\ket{\phi_0}$ and $\ket{\phi_1}$ are orthogonal,
but they are necessarily unit vectors.

One may now define a new prover $P'$ that causes $V'$ to output 1 with
certainty.
The following description of $P'$ will achieve this goal:
\begin{mylist}{\parindent}
\item[1.]
  The new prover $P'$ behaves precisely as $P$ does for $m$ turns.

\item[2.]
  On turn number $m+1$, the new prover $P'$ receives a collection of
  qubits from $V'$, and is expected to return a single qubit on turn number
  $m+2$.
  The required transformation for $P'$ is as follows:
  \begin{equation}
    \label{eq:last-prover-transformation-perfect-completeness}
    \ket{0}\ket{\phi_0} \mapsto \ket{0}\ket{\psi}
    \qquad\text{and}\qquad
    \ket{1}\ket{\phi_1} \mapsto \ket{1}\ket{\psi},
  \end{equation}
  where $\ket{\psi}$ is any fixed unit vector, and where $\ket{\phi_0}$ and
  $\ket{\phi_1}$ are as above.
  The first qubit, which came from the verifier's pseudo-copy operation, is
  the qubit to be returned on the last turn.
  The sets
  \begin{equation}
    \bigl\{
    \ket{0}\ket{\phi_0},
    \ket{1}\ket{\phi_1}
    \bigr\}
    \qquad\text{and}\qquad
    \bigl\{
    \ket{0}\ket{\psi},
    \ket{1}\ket{\psi}
    \bigr\}
  \end{equation}
  are both orthonormal sets, for any choice of a unit vector $\ket{\psi}$,
  so under the assumption that $\ket{\psi}$ represents the same number of
  qubits as $\ket{\phi_0}$ and $\ket{\phi_1}$, one can extend the
  transformation \eqref{eq:last-prover-transformation-perfect-completeness} to
  a unitary operation.
\end{mylist}

\noindent
For $P'$ being defined in this way, the final state of the interaction
immediately before the measurement $V'$ performs on its last step is
\begin{equation}
  \biggl(\frac{1}{\sqrt{2}} \ket{00} + \frac{1}{\sqrt{2}} \ket{11}\biggr)
  \ket{\psi},
\end{equation}
which causes $V'$ to output 1 with certainty.

It remains to consider the case in which the maximum acceptance probability of
$V$ is less than 1/2.
Let us assume, more precisely, that $V'$ interacts with a prover $P'$ that,
after $m$ turns, would have led $V$ to accept with probability
$1/2 - \varepsilon$ for some choice of $\varepsilon > 0$.
When the verifier $V'$ creates the pseudo-copy of what would have been the
output qubit of $V$ and sends everything to $P'$ aside from one of the qubits
resulting from the pseudo-copy, it is evident that the reduced state of this
single qubit is given by
\begin{equation}
  \rho = \begin{pmatrix}
    \frac{1}{2} + \varepsilon & 0\\
    0 & \frac{1}{2} - \varepsilon 
  \end{pmatrix}.
\end{equation}
Regardless of the actions of $P'$, the acceptance probability of $V'$ must be
given by 
$\fid( \ket{\gamma}\bra{\gamma},\sigma)^2 = \bra{\gamma} \sigma \ket{\gamma}$,
for $\sigma$ being a density operator that represents the state of the two
qubits held by the verifier at the beginning of the last turn of the
protocol. 
Because the prover's actions are unitary it must be the case that $\sigma$
extends $\rho$. 
As the fidelity is monotonically increasing under partial tracing, it is not
possible that the quantity above exceeds
\begin{equation}
  \fid\Biggl(
  \begin{pmatrix}
    \frac{1}{2} & 0\\
    0 & \frac{1}{2}
  \end{pmatrix},
  \begin{pmatrix}
    \frac{1}{2} + \varepsilon & 0\\
    0 & \frac{1}{2} - \varepsilon 
  \end{pmatrix}
  \Biggr)^2 = \frac{1}{2} + \frac{1}{2} \sqrt{1 - 4 \varepsilon^2} 
  < 1 - \varepsilon^2.
\end{equation}

\noindent
The required properties of $V'$, as they relate to $V$, have therefore been
verified.

In the case in which a different value of $\alpha$ is to be considered, the
construction of $V'$ from $V$ is identical aside from the substitution of
\begin{equation}
  \label{eq:perfect-completeness-general-gamma}
  \ket{\gamma} = \sqrt{1-\alpha} \ket{00} + \sqrt{\alpha} \ket{11}
\end{equation}
in place of \eqref{eq:gamma-for-alpha=1/2}.
The analysis is the same, except that one must obtain an upper bound
on the value
\begin{equation}
  \label{eq:squared-fidelity-alpha}
  \fid\Biggl(
  \begin{pmatrix}
    1-\alpha & 0\\
    0 & \alpha
  \end{pmatrix},
  \begin{pmatrix}
    (1-\alpha) + \varepsilon & 0\\
    0 & \alpha - \varepsilon
  \end{pmatrix}\Biggr)^2.
\end{equation}
Writing $\beta = \alpha - \varepsilon$, one may bound the value
\eqref{eq:squared-fidelity-alpha} using the arithmetic-geometric mean
inequality as follows:
\begin{equation}
  \begin{split}
    \fid\Biggl(
    \begin{pmatrix}
      1-\alpha & 0\\
      0 & \alpha
    \end{pmatrix},
    \begin{pmatrix}
      1-\beta & 0\\
      0 & \beta
    \end{pmatrix}\Biggr)^2 \hspace{-4cm}\\
    & = \biggl( \sqrt{\alpha\beta} + \sqrt{(1-\alpha)(1-\beta)} \biggr)^2\\
    & = \alpha \beta + (1-\alpha)(1-\beta) +
    2\sqrt{\alpha\beta(1-\alpha)(1-\beta)}\\
    & \leq 
    \alpha \beta + (1-\alpha)(1-\beta) + \alpha (1-\alpha) + \beta (1-\beta)\\
    & = 1 - (\alpha - \beta)^2\\
    & = 1 - \varepsilon^2.
  \end{split}
\end{equation}

We may now connect the construction described above to the statement of
Theorem~\ref{theorem:QIP-perfect-completeness} in a fairly straightforward way.
For a given promise problem $A\in\class{QIP}_{a,b}(m)$, we have a
polynomial-time computable function $V$ witnessing this inclusion.
The functions (or constants) $a$ and $b$ are polynomial-time computable, so on
a given input string $x$ of length $n$, one may compute a 
dyadic rational number $\alpha$ satisfying the inequalities
\begin{equation}
  \frac{3a + b}{4} \leq \alpha \leq a.
\end{equation}
Intuitively speaking, what this is doing is to truncate a binary representation
of $a$ to obtain $\alpha$, taking sufficiently many bits to leave a reasonably
large gap between $\alpha$ and $b$.
The requirement that $\alpha$ is a dyadic rational guarantees that one may
perform a measurement against a state $\ket{\gamma}$ of the form 
\eqref{eq:perfect-completeness-general-gamma} without error, using the gates
from the universal gate set described in Section~\ref{sec:circuits}.
(If a different set of gates were selected, a suitable choice of $\alpha$ could
be substituted to allow for an error-free computation in this step, provided
that the gate set is a reasonable one.)
Finally, one may take $V'$ to be the function that outputs a description of the
verifier derived from $V$ from the construction above, for the choice of
$\alpha$ that has just been specified.
This is a polynomial-time computable function witnessing the inclusion
$A\in\class{QIP}_{1,c}(m+2)$ for
\begin{equation}
  c \leq 1 - \Biggl( \frac{3(a - b)}{4}\Biggr)^2 \leq
  1 - \frac{1}{2}(a - b)^2,
\end{equation}
as required.

\subsection{Parallelization to three turns}
\label{sec:qip-parallelization}

One of the most striking complexity-theoretic properties of quantum interactive
proof systems, at least insofar as they compare with classical interactive
proof systems, is that they may be \emph{parallelized} to a constant number of
turns without diminishing their computational power.
To be more precise, one has the following theorem.
\begin{theorem}
  \label{theorem:QIP=QIP(3)}
  $\textup{QIP} = \textup{QIP}(3)$.
\end{theorem}

\noindent
That is, any promise problem having a polynomial-turn quantum interactive proof
system necessarily has a three-turn quantum interactive proof system as well.
It is an open question whether an analogous statement holds classically, but if
such a statement were true it would imply the collapse of the polynomial-time
hierarchy \cite{Babai85,GoldwasserS89}---and so it is viewed by many as
being unlikely, and is closely connected to the most central open problems of
computational complexity.

There are essentially three steps required to prove
Theorem~\ref{theorem:QIP=QIP(3)}, only one of which is directly concerned with
the parallelization process itself.
The first step involves the transformation of a given quantum interactive proof
system to one having perfect completeness, as was discussed in the previous
subsection; the second step is the parallelization step, which will be
discussed in the present subsection; and the final step is concerned with error
reduction, which will be discussed in the section following this one.
Each of these steps may be represented by an efficient transformation applied
to a quantum circuit description of a verifier in an interactive game, and by
combining them in the most natural way the relationship
$\textup{QIP} = \textup{QIP}(3)$ is obtained.

For the remainder of the present subsection, we will consider an efficient
transformation that operates as follows.
It is assumed that an $m$-turn verifier $V$ in a quantum interactive game
is given, where $m$ may be arbitrary.
From this verifier $V$, a new verifier $V'$ is constructed that has the
following properties:
\begin{mylist}{\parindent}
\item[1.]
  $V'$ is a three-turn verifier.
\item[2.]
  If it is the case that $\omega(V) = 1$, then $\omega(V') = 1$ as well.
\item[3.]
  If it is the case that $\omega(V) \leq 1 - \varepsilon$, then
  $\omega(V') \leq 1 - \varepsilon/m^2$.
\end{mylist}

There are, in fact, multiple constructions known to parallelize quantum
interactive proof systems in this way.
We will describe a particularly simple construction of Kempe, Kobayashi,
Matsumoto, and Vidick \cite{KempeKMV09}, which is well-suited to the
presentation of quantum interactive proof systems through the interactive games
framework that has been adopted in this survey.
The essential idea of the construction is to iteratively transform a verifier
in a quantum interactive game into a new verifier having roughly half as many
turns, using a cut-and-choose style argument.
Each iteration will result in at most a constant-factor increase in the
size of the verifier descriptions, so the transformation may be iterated
logarithmically many times to reduce the number of turns to a constant.
The method cannot be used to reduce the number of turns below three.

With such an iterative approach in mind, suppose that $V$ is an $m$-turn
verifier, for $m$ taking the form
\begin{equation}
  m = 2^{r+1} + 1
\end{equation}
for some positive integer $r$.
For cases in which $m$ does not take this form, one may simply add dummy
turns that have no influence on the output of $V$.
(In general, the addition of such dummy turns will slightly less than double
the number of turns, and does not need to be
iterated---it is only done once at the beginning of the iterative process.)

Under the assumption that $m = 2^{r+1}+1$, actions of $V$ are specified by an
$n$-tuple $(V_1,\ldots,V_n)$ for $n = (m+1)/2 = 2^r + 1$.
It will be assumed that each $V_k$ is a unitary operator of the form
\begin{equation}
  V_k \in \Unitary(\Z_{k-1}\otimes\Y_{k-1},\Z_k\otimes\X_k);
\end{equation}
if this is not the case, then the purification procedure described in
Section~\ref{sec:qip-purification} may be applied.
It will also be assumed that every one of the registers
$\reg{X}_1,\reg{Y}_1,\ldots,\reg{X}_{n-1},\reg{Y}_{n-1}$ comprises exactly the
same number of qubits, which is a constraint that is easily met by adding dummy
qubits to registers as needed.
Now consider the following verifier that is derived from $(V_1,\ldots,V_n)$.

\begin{mylist}{\parindent}
\item[1.]
  Receive the pair of registers $(\reg{Z}_{t-1},\reg{Y}_{t-1})$ from the
  prover, where $t = 2^{r-1}+1$.

\item[2.]
  Choose a bit $a\in\{0,1\}$ uniformly at random.
  If $a = 0$, the original interactive game will be run forward in time,
  while if $a=1$, the original interactive game will be run backward in time.
  In either case, the bit $a$ is concatenated to the first message to be sent
  by the verifier to the prover (so that the prover knows which direction in
  time the game will be run).

\item
  {\emph{Forward} ($a = 0$):}
  Operate precisely at the original verifier $V$ operates, as if the register
  $\reg{Y}_{t-1}$ has just been received from the prover.
  The messages exchanged in the remainder of the interaction therefore
  correspond to the registers $\reg{X}_t$, $\reg{Y}_t$, \ldots,
  $\reg{X}_{n-1}$, $\reg{Y}_{n-1}$.
  The acceptance condition for $V'$ is the same as that of $V$.

\item
  {\emph{Backward} ($a = 1$):}
  Send $\reg{Y}_{t-1}$ back to the prover.
  Each subsequent action of $V'$ is the inverse of an action of $V$, and the
  actions are taken in the reverse order.
  In the turn immediately after $\reg{Y}_{t-1}$ is sent back to the prover,
  the verifier $V'$ expects to receive $\reg{X}_{t-1}$, it applies
  $V_{t-1}^{-1}$ to $(\reg{Z}_{t-1},\reg{X}_{t-1})$, obtaining
  $(\reg{Z}_{t-2},\reg{Y}_{t-2})$, and sends $\reg{Y}_{t-2}$ to the prover.
  This pattern continues until the verifier receives the register $\reg{X}_1$.
  The overall sequence of messages exchanged in this case therefore corresponds
  to $\reg{Y}_{t-1}$, $\reg{X}_{t-1}$, \ldots, $\reg{Y}_1$, $\reg{X}_1$.
  The verifier applies $V_1^{-1}$ to the pair $(\reg{Z}_1,\reg{X}_1)$,
  obtaining $(\reg{Z}_0,\reg{Y}_0)$, and outputs 1 (accept) if and only if a
  measurement of each qubit of $\reg{Z}_0$ in the standard basis yields 0.

\end{mylist}

In the case that the random bit $a$ is equal to 0, the total number of turns in
the protocol is $2(n-t) + 1$ while the number of turns is $2(t-1) + 1$ in case
$a = 1$.
One has $2(n - t) + 1 = 2^r + 1 = 2(t-1) + 1$, and therefore $V'$ is a 
$(2^r + 1)$-turn verifier.

It remains to consider the relationship between $\omega(V)$ and $\omega(V')$.
It is evident that, if $\omega(V) = 1$, then $\omega(V') = 1$ as well, for if
there exists a prover $P$ that causes $V$ to output 1 with certainty, then one
may obtain a prover $P'$ causing $V'$ to output 1 with certainty by adapting a
unitary purification of $P$ in the most straightforward way.
That is, $P'$ initially prepares the registers
$(\reg{Z}_{t-1},\reg{Y}_{t-1},\reg{W}_{t-1})$ in the pure state in which they
would have been, had the unitary purification of $P$ interacted with $V$ up to
this point in the game; and then $P'$ runs the unitary purification of $P$
forward or backward appropriately.

In the case that $\omega(V)$ is smaller than 1, we may obtain an upper bound
on $\omega(V')$ by focusing on the possible states of the register
$\reg{Z}_{t-1}$, over all possible choices of a prover interacting with $V$.
To be clear, we are considering the possible states of $\reg{Z}_{t-1}$ viewed
in isolation, which will generally be mixed states; $\reg{Y}_{t-1}$ and
$\reg{W}_{t-1}$ are to be viewed as having been traced out.
Let us, in particular, consider two sets of states
$\C_0,\C_1\subseteq\Density(\Z_{t-1})$ of the register $\reg{Z}_{t-1}$.
The set $\C_0$ represents all possible states of this register that could be
reached by some prover interacting with $V$, while $\C_1$ represents all
possible states of this register that could, under the actions of some possibly
different prover, lead to $V$ outputting 1 with certainty.
A fairly direct application of Uhlmann's theorem
(Theorem~\ref{theorem:Uhlmann}) reveals the expression
\begin{equation}\label{eq:qip-parallelization-fidelity}
  \omega(V) = \max\Bigl\{\fid(\sigma_0,\sigma_1)^2\,:\,
  \sigma_0\in\C_0,\:\sigma_1\in\C_1\Bigr\}.
\end{equation}
Now, a prover interacting with $V'$ must make an initial choice for the state
of the register $\reg{Z}_{t-1}$, and a similar reasoning reveals that the
value $\omega(V')$ of $V'$ is given by the expression
\begin{equation}
  \max\Biggl\{
  \frac{\fid(\sigma_0,\rho)^2 + \fid(\sigma_1,\rho)^2}{2}
  \,:\,
  \sigma_0\in\C_0,\,\sigma_1\in\C_1,\,\rho\in\Density(\Z_{t-1})\Bigr\}.
\end{equation}
Maximizing over $\rho$, one obtains
\begin{equation}
  \begin{split}
  \omega(V') & = \frac{1}{2} + \frac{1}{2}
  \max\Bigl\{\fid(\sigma_0,\sigma_1)\,:\,
  \sigma_0\in\C_0,\:\sigma_1\in\C_1\Bigr\}\\
  & = \frac{1}{2} + \frac{1}{2}\sqrt{\omega(V)}
  \end{split}
\end{equation}
by a sum-of-squares relationship for the fidelity function,
\begin{equation}
  \max_{\rho} \Bigl(\fid(\sigma_0,\rho)^2 + \fid(\sigma_1,\rho)^2\Bigr)
  = 1 + \fid(\rho_0,\rho_1),
\end{equation}
due to Spekkens and Rudolph \cite{SpekkensR01}.
If it is the case that $\omega(V) \leq 1 - \varepsilon$ for some choice of
$\varepsilon > 0$, then it follows that
\begin{equation}
  \omega(V') \leq \frac{1}{2} + \frac{1}{2} \sqrt{1 - \varepsilon}
  \leq 1 - \frac{\varepsilon}{4}.
\end{equation}

When this method is applied iteratively $r$ times, a three-turn verifier $V'$
is obtained that satisfies
\begin{equation}
  \omega(V') \leq 1 - \frac{\varepsilon}{4^r} \leq 1 - \frac{\varepsilon}{m^2}.
\end{equation}
(The second inequality also accounts for the possibility that dummy turns were
initially added to $V$.)
As each iteration of the procedure described above results in at most a
constant factor increase in the size of the description of the verifier,
iterating it $r$ times gives a polynomial-time procedure.

\begin{theorem}
  For every polynomially bounded function $m$ and every function
  $\varepsilon:\natural\rightarrow[0,1]$, it holds that
  \begin{equation}
    \class{QIP}_{1,1-\varepsilon}(m) \subseteq
    \class{QIP}_{1,1-\delta}(3)
  \end{equation}
  for $\delta = \varepsilon/m^2$.
\end{theorem}

The transformation from $V$ to $V'$ described above can be applied to a
three-turn verifier $V$.
In the interactive game that results, the verifier $V'$ will again be a
three-turn verifier, but it will also have the interesting property that the
only message it sends to the prover is the single random bit~$a$.
The prover sends the register $\reg{Z}_1$ to $V'$ as its first message, and the
register $\reg{Y}_1$ (in case $a=0$) or $\reg{X}_1$ (in case $a=1$) is sent in
the third turn.
The result is that any problem in $\class{QIP}$ also has a three-turn
\emph{public-coin} quantum interactive proof system:
\begin{equation}
  \class{QIP} = \class{QMAM},
\end{equation}
where $\class{QMAM}$ is the class of promise problems having three-turn
public-coin quantum interactive proofs.
We will return to this transformation in the context of multi-prover
interactive games in Chapter~\ref{chapter:multiple-provers}.

\section{SDPs for interactive games and parallel repetition}
\label{sec:qip-sdp}

This section explains how the optimization over all prover strategies in an
interactive game may be expressed as a semidefinite program.
When a generic semidefinite programming algorithm (such as the ellipsoid
algorithm) is applied to such a semidefinite programming formulation, the
relation
\begin{equation}
  \label{eq:QIP-in-EXP}
  \class{QIP}\subseteq\class{EXP}
\end{equation}
is easily obtained.
(The semidefinite programs obtained from a given problem $A\in\class{QIP}$
are of size exponential in the input to $A$, so a polynomial-time algorithm for
solving semidefinite programs gives an exponential time algorithm for $A$.)
The relation \eqref{eq:QIP-in-EXP} can be improved to
\begin{equation}
  \class{QIP} = \class{PSPACE},
\end{equation}
as will be explained in the next section---but the underlying ideas behind this
result are closely connected with the semidefinite programming formulation to
be discussed below.

A different fact that emerges from this semidefinite programming formulation 
is that single-prover quantum interactive proof systems
having perfect completeness possess the property of perfect parallel
repetition.
This is proved through semidefinite programming duality and will be described
later in this section.

\subsection*{Semidefinite programs for optimizing over prover strategies}

Consider an arbitrary verifier in an interactive game having any number of
turns.
One may consider an optimization over all possible prover strategies against
this verifier---we take the probability that the verifier outputs 1 as
the objective function to be maximized, so that the optimal value is
$\omega(V)$.

There are two distinct formulations of this optimization problem as a
semidefinite program that are known.
We will focus on just one of these formulations, in which the variables of a
semidefinite program represent states of the various registers of the
interactive game at different moments in time.
(The other formulation uses the variables of a semidefinite program to
represent the prover's actions through the use of the Choi representation of
channels.)

Assume hereafter that an $m$-turn verifier $V$ has been fixed.
It will be assumed that this verifier has been purified, as discussed
previously, so that $V$ is described by an $n$-tuple $(V_1,\ldots,V_n)$
of linear isometries, for $n = \floor{m/2 + 1}$, and where each $V_k$ takes the
form
\begin{equation}
  V_k\in \Unitary(\Z_{k-1}\otimes\Y_{k-1},\Z_k\otimes\X_k).
\end{equation}
In accordance with our default assumption, the registers 
$\reg{Z}_0$ and $\reg{X}_n$ are taken to be trivial---the verifier starts with
no memory at the beginning of the interaction and sends no message to the
prover immediately before making its final decision---so that
$\Z_0 = \complex$ and $\X_n = \complex$.
It is irrelevant to the present discussion whether the number of turns $m$
is even or odd; in the interest of generality, one may assume that an arbitrary
prover $P$ that interacts with $V$ is described by channels $(P_0,\ldots,P_n)$,
with $\reg{Y}_0$ being a trivial register in case the number of turns
happens to be even.

Now consider the possible states of the registers
\begin{equation}
  \reg{Z}_1,\ldots,\reg{Z}_n, \quad
  \reg{Y}_0,\ldots,\reg{Y}_{n-1}, \quad \text{and} \quad
  \reg{X}_1,\ldots,\reg{X}_{n-1},
\end{equation}
taken in groups of one or two for which the registers are co-existing, in an
interaction between $V$ and an arbitrary prover $P$.
For instance, in Figure~\ref{fig:interactive-game-blue-prover}, one may consider
the states of 
$\reg{Y}_0$, $(\reg{Z}_1,\reg{X}_1)$, $(\reg{Z}_1,\reg{Y}_1)$, 
$(\reg{Z}_2,\reg{X}_2)$, $(\reg{Z}_2,\reg{Y}_2)$, and $\reg{Z}_3$ in isolation.
\begin{figure}
  \begin{center} \small
    \begin{tikzpicture}[scale=0.35, 
        turn/.style={draw, minimum height=14mm, minimum width=10mm,
          fill = ChannelColor, text=ChannelTextColor},
        emptyturn/.style={draw, densely dotted, minimum height=14mm, 
          minimum width=10mm, fill=black!4},
        measure/.style={draw, minimum width=7mm, minimum height=7mm,
          fill = ChannelColor},
        >=latex]
      
      \node (V1) at (-8,4) [turn] {$V_1$};
      \node (V2) at (2,4) [turn] {$V_2$};
      \node (V3) at (12,4) [turn] {$V_3$};
      
      \node (M) at (17,4.4) [measure] {};
      
      \node (P0) at (-13,-1) [emptyturn] {$P_0$};
      \node (P1) at (-3,-1) [emptyturn] {$P_1$};
      \node (P2) at (7,-1) [emptyturn] {$P_2$};
      
      \node[draw, minimum width=5mm, minimum height=3.5mm, fill=ReadoutColor]
      (readout) at (M) {};
      
      \draw[thick] ($(M)+(0.3,-0.15)$) arc (0:180:3mm);
      \draw[thick] ($(M)+(0.2,0.2)$) -- ($(M)+(0,-0.2)$);
      \draw[fill] ($(M)+(0,-0.2)$) circle (0.5mm);
      
      \draw[->] ([yshift=4mm]V1.east) -- ([yshift=4mm]V2.west)
      node [above, midway] {$\reg{Z}_1$};
      
      \draw[->] ([yshift=4mm]V2.east) -- ([yshift=4mm]V3.west)
      node [above, midway] {$\reg{Z}_2$};
      
      \draw[->] ([yshift=-4mm]V1.east) .. controls +(right:20mm) and 
      +(left:20mm) .. ([yshift=4mm]P1.west) node [right, pos=0.4] {$\reg{X}_1$};
      
      \draw[->] ([yshift=-4mm]V2.east) .. controls +(right:20mm) and 
      +(left:20mm) .. ([yshift=4mm]P2.west) node [right, pos=0.4] {$\reg{X}_2$};
      
      \draw[->] ([yshift=4mm]P0.east) .. controls +(right:20mm) and 
      +(left:20mm) .. ([yshift=-4mm]V1.west) node [left, pos=0.6] {$\reg{Y}_0$};
      
      \draw[->] ([yshift=4mm]P1.east) .. controls +(right:20mm) and 
      +(left:20mm) .. ([yshift=-4mm]V2.west) node [left, pos=0.6] {$\reg{Y}_1$};
      
      \draw[->] ([yshift=4mm]P2.east) .. controls +(right:20mm) and 
      +(left:20mm) .. ([yshift=-4mm]V3.west) node [left, pos=0.6] {$\reg{Y}_2$};
      
      \draw[->,densely dotted] 
      ([yshift=-4mm]P0.east) -- ([yshift=-4mm]P1.west)
      node [below, midway] {$\reg{W}_0$};
      
      \draw[->,densely dotted] 
      ([yshift=-4mm]P1.east) -- ([yshift=-4mm]P2.west)
      node [below, midway] {$\reg{W}_1$};
      
      \draw[->] ([yshift=4mm]V3.east) -- (M.west)
      node [above, midway] {$\reg{Z}_3$};
      
    \end{tikzpicture}
  \end{center}
  \caption{For a fixed verifier in an interactive game, one may consider
    optimizing over all prover strategies against that verifier.
    The parts of the figure represented by dotted rectangles and arrows could
    be optimized in this example.}
  \label{fig:interactive-game-blue-prover}
\end{figure}
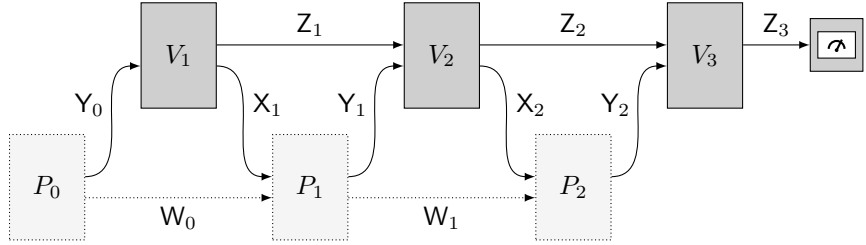
Let us choose names to represent these possible states as follows:
the states immediately prior to verifier actions will correspond to density
operators
\begin{equation}
  \label{eq:sigma-states}
  \sigma_0\in\Density(\Y_0),\,\sigma_1\in\Density(\Z_1\otimes\Y_1),\,
  \ldots,\,\sigma_{n-1}\in\Density(\Z_{n-1}\otimes\Y_{n-1}),
\end{equation}
representing states of 
$\reg{Y}_0,\,(\reg{Z}_1,\reg{Y}_1),\,\ldots,\,(\reg{Z}_{n-1},\reg{Y}_{n-1})$,
while the states immediately after verifier actions will correspond to
density operators
\begin{equation}
  \label{eq:rho-states}
  \rho_1\in\Density(\Z_1\otimes\X_1),\,
  \ldots,\,
  \rho_{n-1}\in\Density(\Z_{n-1}\otimes\X_{n-1}),\,
  \rho_n\in\Density(\Z_n),
\end{equation}
representing states of
$(\reg{Z}_1,\reg{X}_1),\,\ldots,\,(\reg{Z}_{n-1},\reg{X}_{n-1}),\,\reg{Z}_n$.

If it is the case that the states \eqref{eq:sigma-states} and
\eqref{eq:rho-states} truly arise from an interaction between $V$ and a
legitimate prover $P$, then these states must evidently obey certain simple
constraints.
There are two essential types of constraints, which are as follows:
\begin{mylist}{\parindent}
\item[1.]
  For every choice of $k\in\{1,\ldots,n-1\}$, the state of the register
  $\reg{Z}_k$, when it is viewed in isolation, must be the same with respect to
  both $\rho_k$ and $\sigma_k$.
  This is so because a prover interacting with $V$ cannot touch these
  registers.
  It must therefore hold that
  \begin{equation}
    \label{eq:SDP-constraint-type-1}
    \tr_{\X_k}(\rho_k) = \tr_{\Y_k}(\sigma_k).
  \end{equation}

\item[2.]
  For every choice of $k\in\{1,\ldots,n\}$, it must hold that
  \begin{equation}
    \label{eq:SDP-constraint-type-2}
    \rho_k = V_k \sigma_{k-1} V_k^{\ast}.
  \end{equation}
  This is so because the transition from $\sigma_k$ to $\rho_k$ is completely
  determined by the verifier's action at the corresponding position in the
  interaction.
\end{mylist}

\noindent
It is quite straightforward to see that these two types of constraints
must necessarily hold when the states \eqref{eq:sigma-states} and
\eqref{eq:rho-states} arise from an interaction between $V$ and some prover
$P$.
What is more remarkable is that these constraints are not only necessary but
sufficient in this regard.
That is, for any collection of states $\sigma_0,\ldots,\sigma_{n-1}$ and
$\rho_1,\ldots,\rho_n$, having the forms \eqref{eq:sigma-states} and
\eqref{eq:rho-states} and satisfying the constraints
\eqref{eq:SDP-constraint-type-1} and \eqref{eq:SDP-constraint-type-2}
for all of the possible values of $k$ indicated above, it must hold that there
exists a prover $P$ that causes these states to occur at their corresponding
positions in an interaction with $V$.\footnote{
  The reader is cautioned this statement is very much reliant on the assumption
  that $V$ has been purified: $V_1,\ldots,V_n$ are isometries and not general
  channels.}

That the constraints above are indeed sufficient in the respect just described
follows from the unitary equivalence of purifications
(Theorem~\ref{thm:unitary-equivalence}).
Under the assumption that a prover holds a purification of the state
of $\reg{Z}_k$ in registers $(\reg{X}_k,\reg{W}_{k-1})$, meaning that the state
of $(\reg{Z}_k,\reg{X}_k,\reg{W}_{k-1})$ is pure, it is free to transform the
state of these registers into any pure state of
$(\reg{Z}_k,\reg{Y}_k,\reg{W}_k)$ whatsoever, provided that the
state of $\reg{Z}_k$, when viewed in isolation, does not change.
In particular, the prover may transform these registers in such a way that
the state of the pair $(\reg{Z}_k,\reg{Y}_k)$ has been transformed to
$\sigma_k$, by virtue of the fact that the reduced state of $\reg{Z}_k$ is the
same for both of these density operators.
In doing this, the prover holds a purification of $\sigma_k$ in the register
$\reg{W}_k$, and is ready to perform the transformation corresponding to the
next step in the interaction.

It is now evident that a maximization of the probability for a prover to cause
the verifier to output 1 can be represented as a semidefinite program.
The probability that the verifier outputs 1 is given by a linear function 
$\ip{\Pi}{\rho_n}$ of the state of $\reg{Z}_n$, for $\Pi$ being a measurement
operator that corresponds to the verifier outputting 1.
One could then formulate a semidefinite program that maximizes this
value over all choices of density operators
$\sigma_0, \ldots, \sigma_{n-1}, \rho_1,\ldots,\rho_n$ of the forms
\eqref{eq:sigma-states} and \eqref{eq:rho-states} satisfying the constraints
described above; density operators must be positive semidefinite and trace 1,
and all of the constraints described above are linear, which allows for such a
semidefinite program.

In the interests of simplicity, one may omit the variables corresponding to the
states $\rho_1,\ldots,\rho_n$, as these states are determined by
$\sigma_0,\ldots,\sigma_{n-1}$.
We obtain the semidefinite program whose primal form is given in
Figure~\ref{fig:SDP-primal}.
One may compute that the corresponding dual form of this semidefinite program
is as given in Figure~\ref{fig:SDP-dual}.

\begin{figure}
  \begin{align*}
    \text{maximize:} \quad & \bigip{V_n^{\ast} \Pi V_n}{\sigma_{n-1}}\\
    \text{subject to:}\quad 
    & \tr(\sigma_0) = 1,\\
    & \tr_{\Y_1}(\sigma_1) = \tr_{\X_1}\bigl(V_1 \sigma_0 V_1^{\ast}\bigr),\\
    & \qquad\vdots\\
    & \tr_{\Y_{n-1}}(\sigma_{n-1}) 
    = \tr_{\X_{n-1}}\bigl(V_{n-1} \sigma_{n-2} V_{n-1}^{\ast}\bigr),\\
    & \sigma_0\in\Pos(\Y_0)\\
    & \sigma_1\in\Pos(\Z_1\otimes\Y_1),\\
    & \qquad\vdots\\
    & \sigma_{n-1}\in\Pos(\Z_{n-1}\otimes\Y_{n-1}).
  \end{align*}
  \caption{Primal form of a semidefinite program for computing $\omega(V)$.}
  \label{fig:SDP-primal}
\end{figure}

\begin{figure}
  \begin{align*}
    \text{minimize:} \quad & \lambda\\
    \text{subject to:}\quad 
    & \lambda \I_{\Y_0} \geq V_1^{\ast}(Z_1\otimes\I_{\X_1})V_1,\\
    & Z_1 \otimes \I_{\Y_1} \geq V_2^{\ast}(Z_2\otimes\I_{\X_2})V_2,\\
    & \qquad\vdots\\
    & Z_{n-2} \otimes \I_{\Y_{n-2}} \geq
    V_{n-1}^{\ast}(Z_{n-1}\otimes\I_{\X_{n-1}})V_{n-1},\\
    & Z_{n-1} \otimes \I_{\Y_{n-1}} \geq V_n^{\ast}\Pi V_n,\\
    & \lambda\in\real\\
    & Z_1\in\Herm(\Z_1),\\
    & \qquad\vdots\\
    & Z_{n-1}\in\Herm(\Z_{n-1}).
  \end{align*}
  \caption{Dual form of a semidefinite program for computing $\omega(V)$.}
  \label{fig:SDP-dual}
\end{figure}

It is evident that strong duality holds for this semidefinite program;
by choosing the dual variables $\lambda,Z_1,\ldots,Z_{n-1}$ to be suitably
large scalar multiples of the identity, a strictly feasible dual solution is
obtained, which leads to strong duality by Slater's theorem.

\subsection*{Parallel repetition}

Consider the situation in which a prover plays two independent interactive
games, as suggested by Figure~\ref{fig:parallel-repetition}.
The property of \emph{parallel repetition} concerns the optimal probability
that a prover may win both games simultaneously, and how this optimal
probability compares with the optimal probabilities with which the two games
may be won individually.
One may also consider the situation in which three or more games are played
in parallel, but once the behavior is understood for two independent games, the
general case will follow (either by induction or by generalizing the
methodology in the most straightforward way).
\begin{figure}
  \begin{center} \small
    \begin{tikzpicture}[scale=0.35, 
        turn/.style={draw, minimum height=14mm, minimum width=10mm,
          fill = ChannelColor, text=ChannelTextColor},
        emptyturn/.style={draw, densely dotted, minimum height=14mm, 
          minimum width=10mm, fill=black!6},
        measure/.style={draw, minimum width=7mm, minimum height=7mm,
          fill = ChannelColor},
        >=latex]
      
      \node (V1) at (-8,4) [turn] {$V_1^1$};
      \node (V2) at (2,4) [turn] {$V_2^1$};
      \node (V3) at (12,4) [turn] {$V_3^1$};
      
      \node (W1) at (-8,-6) [turn] {$V_1^2$};
      \node (W2) at (2,-6) [turn] {$V_2^2$};
      \node (W3) at (12,-6) [turn] {$V_3^2$};
      
      \node (M) at (17,4.4) [measure] {};

      \node (N) at (17,-6.4) [measure] {};
      
      \node (P0) at (-13,-1) [emptyturn] {$P_0$};
      \node (P1) at (-3,-1) [emptyturn] {$P_1$};
      \node (P2) at (7,-1) [emptyturn] {$P_2$};
      
      \node[draw, minimum width=5mm, minimum height=3.5mm, fill=ReadoutColor]
      (readout) at (M) {};
      
      \draw[thick] ($(M)+(0.3,-0.15)$) arc (0:180:3mm);
      \draw[thick] ($(M)+(0.2,0.2)$) -- ($(M)+(0,-0.2)$);
      \draw[fill] ($(M)+(0,-0.2)$) circle (0.5mm);
      
      \node[draw, minimum width=5mm, minimum height=3.5mm, fill=ReadoutColor]
      (readout) at (N) {};
      
      \draw[thick] ($(N)+(0.3,-0.15)$) arc (0:180:3mm);
      \draw[thick] ($(N)+(0.2,0.2)$) -- ($(N)+(0,-0.2)$);
      \draw[fill] ($(N)+(0,-0.2)$) circle (0.5mm);
      
      \draw[->] ([yshift=4mm]V1.east) -- ([yshift=4mm]V2.west)
      node [above, midway] {$\reg{Z}_1^1$};
      
      \draw[->] ([yshift=4mm]V2.east) -- ([yshift=4mm]V3.west)
      node [above, midway] {$\reg{Z}_2^1$};
      
      \draw[->] ([yshift=-4mm]V1.east) .. controls +(right:20mm) and 
      +(left:20mm) .. ([yshift=6mm]P1.west) node [right, pos=0.4] 
      {$\reg{X}_1^1$};
      
      \draw[->] ([yshift=-4mm]V2.east) .. controls +(right:20mm) and 
      +(left:20mm) .. ([yshift=6mm]P2.west) node [right, pos=0.4] 
      {$\reg{X}_2^1$};
      
      \draw[->] ([yshift=6mm]P0.east) .. controls +(right:20mm) and 
      +(left:20mm) .. ([yshift=-4mm]V1.west) node [left, pos=0.6] 
      {$\reg{Y}_0^1$};
      
      \draw[->] ([yshift=6mm]P1.east) .. controls +(right:20mm) and 
      +(left:20mm) .. ([yshift=-4mm]V2.west) node [left, pos=0.6] 
      {$\reg{Y}_1^1$};
      
      \draw[->] ([yshift=6mm]P2.east) .. controls +(right:20mm) and 
      +(left:20mm) .. ([yshift=-4mm]V3.west) node [left, pos=0.6]
      {$\reg{Y}_2^1$};
     
      \draw[->] ([yshift=4mm]V3.east) -- (M.west)
      node [above, midway] {$\reg{Z}_3^1$};
      
      \draw[->] ([yshift=-4mm]W1.east) -- ([yshift=-4mm]W2.west)
      node [below, midway] {$\reg{Z}_1^2$};
      
      \draw[->] ([yshift=-4mm]W2.east) -- ([yshift=-4mm]W3.west)
      node [below, midway] {$\reg{Z}_2^2$};
      
      \draw[->] ([yshift=4mm]W1.east) .. controls +(right:20mm) and 
      +(left:20mm) .. ([yshift=-6mm]P1.west) node [right,pos=0.4] 
      {$\reg{X}_1^2$};
      
      \draw[->] ([yshift=4mm]W2.east) .. controls +(right:20mm) and 
      +(left:20mm) .. ([yshift=-6mm]P2.west) node [right,pos=0.4]
      {$\reg{X}_2^2$};
      
      \draw[->] ([yshift=-6mm]P0.east) .. controls +(right:20mm) and 
      +(left:20mm) .. ([yshift=4mm]W1.west) node [left, pos=0.6] 
      {$\reg{Y}_0^2$};
      
      \draw[->] ([yshift=-6mm]P1.east) .. controls +(right:20mm) and 
      +(left:20mm) .. ([yshift=4mm]W2.west) node [left, pos=0.6] 
      {$\reg{Y}_1^2$};
      
      \draw[->] ([yshift=-6mm]P2.east) .. controls +(right:20mm) and 
      +(left:20mm) .. ([yshift=4mm]W3.west) node [left, pos=0.6]
      {$\reg{Y}_2^2$};
      
      \draw[->] ([yshift=-4mm]W3.east) -- (N.west)
      node [above, midway] {$\reg{Z}_3^2$};

      \draw[->,densely dotted] 
      ([yshift=0mm]P0.east) -- ([yshift=0mm]P1.west)
      node [below, midway] {$\reg{W}_0$};
      
      \draw[->,densely dotted] 
      ([yshift=0mm]P1.east) -- ([yshift=0mm]P2.west)
      node [below, midway] {$\reg{W}_1$};
      
    \end{tikzpicture}
  \end{center}
  \caption{A prover $P$, described by channels $P_0$, $P_1$, and $P_2$
    plays two interactive games simultaneously, one against a verifier
    described by $(V_1^1,V_2^1,V_3^1)$ and the other described by
    $(V_1^2,V_2^2,V_3^2)$.
    While the two verifiers behave independently, there is nothing that forces
    the prover $P$ to treat the two games independently.}
  \label{fig:parallel-repetition}
\end{figure}
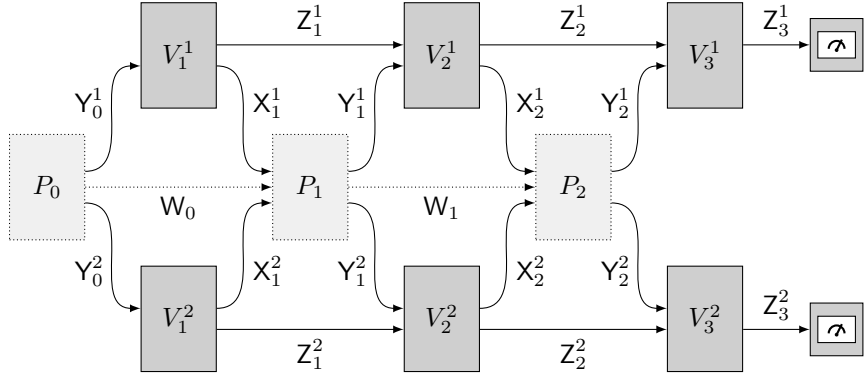

\pagebreak

In more precise terms, suppose that two verifiers
$V^1 = (V^1_1,\ldots,V^1_n)$ and $V^2 = (V^2_1,\ldots,V^2_n)$ 
for interactive games are given.
It is assumed that the verifiers agree on the number of turns they exchange
with a prover for simplicity---and if the two verifiers did not agree on the
number of turns, there would be nothing lost in the discussion that follows
by adding dummy turns to one of them in order to meet this condition.
Let us write\footnote{
  The notation $V^1\otimes V^2$ is slightly abusive, as it is not immediate how
  one tensors two verifiers in a formal sense, but it is nevertheless a
  reasonably natural notation;
  each of the actions performed by this combined verifier corresponds to a
  tensor product of channels, and similarly for the measurement operator
  corresponding to acceptance.}
$V^1\otimes V^2$ to denote the verifier obtained by running $V^1$
and $V^2$ in parallel, and defining the output bit of the combined prover to be
the AND of the output bits of $V^1$ and $V^2$.

It is evident that the optimal acceptance probability of the verifiers
$V^1$, $V^2$, and $V^1\otimes V^2$ satisfies
\begin{equation}
  \omega\bigl(V^1 \otimes V^2\bigr) \geq \omega\bigl(V^1\bigr)
  \omega\bigl(V^2\bigr),
\end{equation}
as a prover may achieve the acceptance probability $\omega(V^1)\omega(V^2)$
against $V^1 \otimes V^2$ simply by playing optimally and independently
against $V^1$ and $V^2$.
Given that a prover need not treat the two games independently, it is not
obvious that this inequality can be replaced by an equality in general.
This is, however, the case: for every choice of verifiers $V^1$ and $V^2$, one
has
\begin{equation}
  \omega\bigl(V^1 \otimes V^2\bigr) = \omega\bigl(V^1\bigr)
  \omega\bigl(V^2\bigr).
\end{equation}
More generally, for $V^1 \otimes \cdots \otimes V^k$ being a combined verifier
defined from any choice of verifiers $V^1,\ldots,V^k$ in the most natural way,
one has
\begin{equation}
  \omega\bigl(V^1 \otimes \cdots \otimes V^k\bigr) 
  = \omega\bigl(V^1\bigr) \cdots \omega\bigl(V^k\bigr).
\end{equation}

One way to prove that this is so is to use semidefinite programming duality.
Consider the dual form of the semidefinite program for the optimal acceptance
probabilities $\omega(V^1)$, $\omega(V^2)$, and $\omega(V^1\otimes V^2)$.
For any choice of dual-feasible solutions to the semidefinite programs
representing $\omega(V^1)$ and $\omega(V^2)$, which may be denoted
\begin{equation}
  \label{eq:parallel-repetition-dual-points}
  \bigl(\lambda^1,Z_1^1,\ldots,Z_{n-1}^1\bigr)
  \quad\text{and}\quad
  \bigl(\lambda^2,Z_1^2,\ldots,Z_{n-1}^2\bigr)
\end{equation}
(where superscripts represent indices, not exponents), one obtains
a dual-feasible solution to the semidefinite program for $\omega(V^1\otimes
V^2)$ by taking
\begin{equation}
  \lambda = \lambda^1 \lambda^2,\;
  Z_1 = Z_1^1 \otimes Z_1^2,\;\ldots,\;
  Z_{n-1} = Z_{n-1}^1 \otimes Z_{n-1}^2.
\end{equation}
The dual-feasibility of the solution $(\lambda,Z_1,\ldots,Z_{n-1})$ defined in
this way follows from the observation that each of the operators
$Z_1^1,\ldots,Z_{n-1}^1$ and $Z_1^2,\ldots,Z_{n-1}^2$ must be positive
semidefinite, together with the operator inequality
\begin{equation}
  Q_1 \otimes Q_2 \geq R_1 \otimes R_2,
\end{equation}
which holds provided that
\begin{equation}
  Q_1 \geq R_1 \geq 0
  \quad\text{and}\quad
  Q_2 \geq R_2 \geq 0.
\end{equation}
The dual objective value achieved by $(\lambda,Z_1,\ldots,Z_{n-1})$ is
precisely $\lambda^1\lambda^2$, and by optimizing over 
all dual feasible solutions \eqref{eq:parallel-repetition-dual-points} and
considering that strong duality holds, one obtains
\begin{equation}
  \omega\bigl(V^1 \otimes \cdots \otimes V^k\bigr) 
  = \omega\bigl(V^1\bigr) \cdots \omega\bigl(V^k\bigr).
\end{equation}

Based on the fact just described, the following theorem may be obtained.

\begin{theorem}
  \label{theorem:QIP=QIP(3)-with-bounds}
  Let $a,b:\natural\rightarrow[0,1]$ and $m:\natural\rightarrow\natural$
  be functions.
  For every choice of a positive, polynomially bounded function
  $p:\natural\rightarrow\natural$, it holds that
  \begin{equation}
    \class{QIP}_{a,b}(m) \subseteq \class{QIP}_{a^p,b^p}(m).
  \end{equation}
  In particular, for every choice of positive, polynomially bounded functions
  $r$ and $q$, one has
  \begin{equation}
    \class{QIP}_{1,1-1/r}(m) \subseteq \class{QIP}_{1,2^{-q}}(m).
  \end{equation}
\end{theorem}

\begin{cor}
  \label{cor:QIP=QIP(3)}
  It holds that
  $\class{QIP} = \class{QIP}_{1,2^{-r}}(3)$
  for every positive, polynomially bounded function
  $r:\natural\rightarrow\natural$.
\end{cor}

\pagebreak

\section{QIP = PSPACE} \label{sec:QIP=PSPACE}

The final section of the chapter concerns the proof of the following
theorem.

\begin{theorem}
  \label{theorem:QIP=PSPACE}
  $\class{QIP} = \class{PSPACE}$.
\end{theorem}

\noindent
As a result of this theorem, together with the well-known result
$\class{IP} = \class{PSPACE}$ of Shamir \cite{Shamir92}, based partly on the
work of Lund, Fortnow, Karloff, and Nisan \cite{LundFKN92}, one finds that
single-prover quantum interactive proof systems have precisely the same
computational power as single-prover classical interactive proof systems.

Indeed, because $\class{IP} = \class{PSPACE}$ and
$\class{IP}\subseteq\class{QIP}$, one of the containments required to prove
Theorem~\ref{theorem:QIP=PSPACE} follows immediately, namely
$\class{PSPACE} \subseteq \class{QIP}$.
The main focus of the present section is on the reverse containment, which is
$\class{QIP}\subseteq\class{PSPACE}$.

Theorem~\ref{theorem:QIP=PSPACE} was first proved by Jain, Ji, Upadhyay, and
Watrous \cite{JainJUW11}.
The proof we present below makes use of a simplification due to Wu
\cite{Wu10}.
The overall structure of the two proofs are the same, but at a technical level
Wu's formulation has a significant advantage, in that it replaces a more
complicated multiplicative weights update algorithm for a semidefinite program
with a simpler one for a min-max problem.

\subsection{Reduction of \class{QIP} to a min-max value computation}
\label{sec:QIP-reduction}

The first step in the proof that $\class{QIP} = \class{PSPACE}$ concerns the
relationship between interactive games and a certain type of min-max problem to
be solved by a polynomial-space algorithm.
By Corollary~\ref{cor:QIP=QIP(3)}, which implies that
$\class{QIP} = \class{QIP}_{1,\delta}(3)$ for any choice of a constant 
$\delta\in (0,1)$, we may restrict our attention to three-turn interactive
games.

Consider an arbitrary three-turn verifier, given by its quantum circuit
encoding.
Through the purification process described earlier, one may efficiently process
the description of such a verifier's circuits to obtain a unitary circuit
description of a verifier $V = (V_1,V_2)$ as suggested by
Figure~\ref{fig:unitary-3-turn-verifier}.
Here, it is assumed that $\reg{Z}_0$ comprises all of the ancillary qubits
needed by the verifier's computations, while $\reg{Z}_2$ represents a single
qubit, which is the output qubit of the interactive game to be measured with
respect to the standard basis.
Ordinarily there is no need for the register $\reg{X}_2$ in a three-turn
interactive game, and in the present case this register need not to be
interpreted as a fourth message register---it is simply a register representing
all of the qubits held by the verifier, aside from the output qubit, at the end
of the game.
Although this register has no influence on the outcome of an interaction
between $V$ and a prover $P$, it will play an important role in the reduction
that follows.

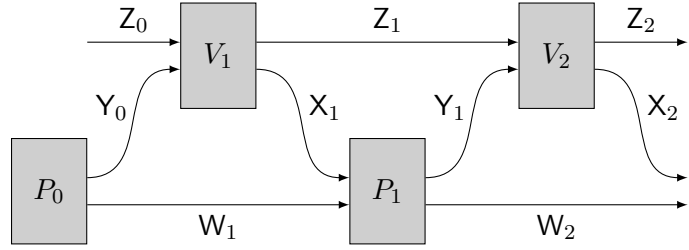
\begin{figure}
  \begin{center}
    \begin{tikzpicture}[scale=0.45, 
        turn/.style={draw, minimum height=14mm, minimum width=10mm,
          fill = ChannelColor, text=ChannelTextColor},
        emptyturn/.style={minimum height=14mm, minimum width=10mm},
        >=latex]

      \node (V0) at (-13,4) [emptyturn] {};      
      \node (V1) at (-8,4) [turn] {$V_1$};
      \node (V2) at (2,4) [turn] {$V_2$};
      \node (M) at (7,4.4) [emptyturn] {};
      \node (P0) at (-13,0) [turn] {$P_0$};
      \node (P1) at (-3,0) [turn] {$P_1$};
      \node (P2) at (7,0) [emptyturn] {};
      
      \node at (0,6.15) {};

      \draw[->] ([yshift=4mm]V1.east) -- ([yshift=4mm]V2.west)
      node [above, midway] {$\reg{Z}_1$};
      
      \draw[->] ([yshift=4mm]V0.east) -- ([yshift=4mm]V1.west)
      node [above, midway] {$\reg{Z}_0$};
      
      \draw[->] ([yshift=-4mm]V1.east) .. controls +(right:20mm) and 
      +(left:20mm) .. ([yshift=4mm]P1.west) node [right, pos=0.4] {$\reg{X}_1$};
      
      \draw[->] ([yshift=-4mm]V2.east) .. controls +(right:20mm) and 
      +(left:20mm) .. ([yshift=4mm]P2.west) node [right, pos=0.4] {$\reg{X}_2$};
      
      \draw[->] ([yshift=4mm]P0.east) .. controls +(right:20mm) and 
      +(left:20mm) .. ([yshift=-4mm]V1.west) node [left, pos=0.6] {$\reg{Y}_0$};
      
      \draw[->] ([yshift=4mm]P1.east) .. controls +(right:20mm) and 
      +(left:20mm) .. ([yshift=-4mm]V2.west) node [left, pos=0.6] {$\reg{Y}_1$};
      
      \draw[->] ([yshift=-4mm]P0.east) -- ([yshift=-4mm]P1.west)
      node [below, midway] {$\reg{W}_1$};
      
      \draw[->] ([yshift=-4mm]P1.east) -- ([yshift=-4mm]P2.west)
      node [below, midway] {$\reg{W}_2$};
      
      \draw[->] ([yshift=4mm]V2.east) -- (M.west) node [above, midway]
           {$\reg{Z}_2$};
      
    \end{tikzpicture}
  \end{center}
  \caption{A three-turn interactive game with a unitary verifier.}
  \label{fig:unitary-3-turn-verifier}
\end{figure}

Next, define two channels, $\Phi_1\in\Channel(\Y_0,\Z_1)$ and
$\Phi_2\in\Channel(\X_2,\Z_1)$, as 
\begin{equation}
  \begin{split}
    \Phi_1(\rho_1) 
    & = \tr_{\X_1} \bigl(V_1 (\ket{0\cdots 0}\bra{0\cdots 0} \otimes \rho_1)
    V_1^{\ast}\bigr)\\
    \Phi_2(\rho_2) & = \tr_{\Y_1} \bigl(V_2^{\ast} (\ket{1}\bra{1} 
    \otimes \rho_2) V_2\bigr)
  \end{split}
\end{equation}
for every $\rho_1\in\Density(\Y_0)$ and $\rho_2\in\Density(\X_2)$.
The action of these channels is illustrated in
Figure~\ref{fig:QIP=PSPACE-channels}.
As was already encountered in~\eqref{eq:qip-parallelization-fidelity}, the
maximum acceptance probability of $V$ is given by
\begin{equation}
  \label{eq:max-fidelity}
  \omega(V) = \max_{\rho_1,\rho_2} \fid(\Phi_1(\rho_1),\Phi_2(\rho_2))^2,
\end{equation}
where the maximum is over all states $\rho_1\in\Density(\Y_0)$ and
$\rho_2\in\Density(\X_2)$.
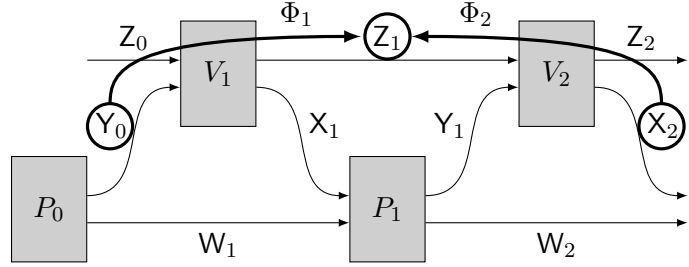
\begin{figure}
  \begin{center}
    \begin{tikzpicture}[scale=0.45, 
        turn/.style={draw, minimum height=14mm, minimum width=10mm,
          fill = ChannelColor, text=ChannelTextColor},
        emptyturn/.style={minimum height=14mm, minimum width=10mm},
        >=latex]

      \node (V0) at (-13,4) [emptyturn] {};      
      \node (V1) at (-8,4) [turn] {$V_1$};
      \node (V2) at (2,4) [turn] {$V_2$};
      \node (M) at (7,4.4) [emptyturn] {};
      \node (P0) at (-13,0) [turn] {$P_0$};
      \node (P1) at (-3,0) [turn] {$P_1$};
      \node (P2) at (7,0) [emptyturn] {};
      
      \draw[->] ([yshift=4mm]V1.east) -- ([yshift=4mm]V2.west)
      node [above, midway] {$\reg{Z}_1$};
      
      \draw[->] ([yshift=4mm]V0.east) -- ([yshift=4mm]V1.west)
      node [above, midway] {$\reg{Z}_0$};
      
      \draw[->] ([yshift=-4mm]V1.east) .. controls +(right:20mm) and 
      +(left:20mm) .. ([yshift=4mm]P1.west) node [right, pos=0.4] {$\reg{X}_1$};
      
      \draw[->] ([yshift=-4mm]V2.east) .. controls +(right:20mm) and 
      +(left:20mm) .. ([yshift=4mm]P2.west) node [right, pos=0.4] {$\reg{X}_2$};
      
      \draw[->] ([yshift=4mm]P0.east) .. controls +(right:20mm) and 
      +(left:20mm) .. ([yshift=-4mm]V1.west) node [left, pos=0.6] {$\reg{Y}_0$};
      
      \draw[->] ([yshift=4mm]P1.east) .. controls +(right:20mm) and 
      +(left:20mm) .. ([yshift=-4mm]V2.west) node [left, pos=0.6] {$\reg{Y}_1$};
      
      \draw[->] ([yshift=-4mm]P0.east) -- ([yshift=-4mm]P1.west)
      node [below, midway] {$\reg{W}_1$};
      
      \draw[->] ([yshift=-4mm]P1.east) -- ([yshift=-4mm]P2.west)
      node [below, midway] {$\reg{W}_2$};
      
      \draw[->] ([yshift=4mm]V2.east) -- (M.west) node [above, midway]
           {$\reg{Z}_2$};
      
      \draw[very thick, ->] (-11.2,3.1) 
      .. controls +(up:10mm) and +(left:70mm) .. (-3.78,5.1)
      node [above, pos = 0.9] {$\Phi_1$};
      
      \draw[very thick, ->] (5.1,3.1) 
      .. controls +(up:10mm) and +(right:70mm) .. (-2.29,5.1)
      node [above, pos = 0.9] {$\Phi_2$};

      \node[draw, very thick, circle, minimum size=6mm] at (5.1,2.45) {};
      \node[draw, very thick, circle, minimum size=6mm] at (-11.2,2.45) {};
      \node[draw, very thick, circle, minimum size=6mm] at (-3,5.1) {};

    \end{tikzpicture}
  \end{center}
  \caption{The action of the channels $\Phi_1$ and $\Phi_2$.}
  \label{fig:QIP=PSPACE-channels}
\end{figure}
Finally, define channels $\Psi_1,\Psi_2\in\Channel(\Y_0\otimes\X_2,\Z_1)$ as
\begin{equation}
  \Psi_1 = \Phi_1 \otimes \textup{Tr}
  \quad\text{and}\quad
  \Psi_2 = \textup{Tr} \otimes \Phi_2,
\end{equation}
where the traces are defined on $\X_2$ and $\Y_0$, respectively.

Now, for $\Xi = \Psi_1 - \Psi_2$ being the difference between these channels,
one may consider the min-max quantity
\begin{equation}
  \label{eq:min-max-value}
  \eta = \min_{\rho\in\Density(\Y_0\otimes\X_2)} \max_{\Pi\in\Proj(\Z_1)}
  \bigip{\Pi}{\Xi(\rho)}.
\end{equation}
There are two relevant cases to consider.

\begin{trivlist}
\item \emph{Case 1:} $\omega(V) = 1$.
  In this case, one may set $\rho = \rho_1\otimes\rho_2$, for
  $\rho_1\in\Density(\Y_0)$ and $\rho_2\in\Density(\X_2)$ maximizing the
  expression on the right-hand side of \eqref{eq:max-fidelity}.
  As $\fid(\Phi_1(\rho_1),\Phi_2(\rho_2)) = 1$, it follows that
  $\Phi_1(\rho_1) = \Phi_2(\rho_2)$, so $\Xi(\rho) = 0$, and therefore
  $\eta = 0$.

\item \emph{Case 2:} $\omega(V) \leq \delta$ for $\delta \in (0,1)$.
  In this case, for every choice of $\rho\in\Density(\Y_0\otimes\X_2)$ one has
  \begin{equation}
    \norm{\Xi(\rho)}_1 = \bignorm{\Phi_1(\rho_1) - \Phi_2(\rho_2)}_1
    \geq 2 - 2 \fid\bigl(\Phi_1(\rho_1), \Phi_2(\rho_2)\bigr),
  \end{equation}
  for $\rho_1 = \tr_{\X_2}(\rho)$ and $\rho_2 = \tr_{\Y_0}(\rho)$.
  The fidelity 
  $\fid\bigl(\Phi_1(\rho_1), \Phi_2(\rho_2)\bigr)$ can be no larger than
  $\sqrt{\delta}$ by \eqref{eq:max-fidelity}, so that
  \begin{equation}
    \norm{\Xi(\rho)}_1 \geq 2 - 2 \sqrt{\delta},
  \end{equation}
  and thus, for an optimal choice of $\Pi$, one has
  $\bigip{\Pi}{\Xi(\rho)} \geq 1 - \sqrt{\delta}$.
  Therefore, one has $\eta \geq 1 - \sqrt{\delta}$ in this case.
\end{trivlist}

An implication of the case analysis performed above may be stated as a lemma as
follows.

\begin{lemma}\label{lem:qip-pspace-eta}
  Let $V$ be the verifier in a three-turn interactive game, and let $\eta$ be
  the min-max quantity defined from $V$ in~\eqref{eq:min-max-value}. 
  The following implications hold: 
  \begin{equation}
    \bigl[\omega(V) = 1\bigr]\, \Rightarrow\, \bigl[ \eta = 0 \bigr]
    \quad\text{and}\quad
    \biggl[\omega(V) \leq \frac{1}{4}\biggr]\, \Rightarrow \,
    \biggl[\eta \geq \frac{1}{2}\biggr]
  \end{equation}
\end{lemma}

\subsection{Matrix multiplicative weights update method}

At the heart of the proof that QIP $\subseteq$ PSPACE is an algorithmic method
known as the \emph{matrix multiplicative weights update method}.
There are a variety of known algorithms that make use of this method,
which is well-suited to certain specialized forms of convex optimization
problems.
This subsection will present one algorithm in this family that will allow the
containment QIP $\subseteq$ PSPACE to be proved, based on the min-max problem
described in the previous subsection.
The specific formulation of this algorithm as a PSPACE algorithm will be
discussed in the next subsection.

The input to the problem is an explicit matrix description of a mapping
$\Xi = \Psi_1 - \Psi_2$, which is assumed to be the difference of two channels
of the form $\Psi_1,\Psi_2\in\Channel\bigl(\complex^N,\complex^M\bigr)$.
The goal of the algorithm will be to approximate the min-max value
\begin{equation}\label{eq:qip-pspace-def-eta}
  \eta = \min_{\rho} \max_{P} \,\ip{P}{\Xi(\rho)},
\end{equation}
where the minimum is over all density operators $\rho\in\Density(\complex^N)$
and the maximum is over all measurement operators $P\in\Pos(\complex^M)$
(meaning that $P$ satisfies $0\leq P \leq \I$).
As the function
\begin{equation}
  (\rho,P)\mapsto\ip{P}{\Xi(\rho)}
\end{equation}
is bilinear and the minimum and maximum are over convex and compact sets, one
may freely change the order of the minimum and maximum by Sion's min-max
theorem.
It is convenient to observe that
\begin{equation}
  \eta = \min_{\rho} \max_{\Pi} \,\ip{\Pi}{\Xi(\rho)},
\end{equation}
where the maximum is over all projection operators $\Pi\in\Proj(\complex^M)$
(as it was in the min-max problem described in the previous subsection);
when the min-max expression is viewed as a game and $\rho$ is played first,
there is nothing lost in restricting the maximization to the extreme points of
the set of measurement operators (i.e., the projection operators).
Along similar lines, if $P$ were played first, one could
then choose a minimizing $\rho$ among the pure states in
$\Density(\complex^N)$.

The accuracy of the algorithm is determined by the setting of an accuracy
parameter $\varepsilon > 0$, which is left as an indeterminate value for the
time being; a choice of $\varepsilon$ that is appropriate to the application of
the algorithm in proving QIP = PSPACE is made in the subsection following this
one.
(To obtain a cleaner expression on the accuracy of the algorithm, it is
convenient to make the assumption that $\varepsilon < 1/4$.)
The algorithm is described in Figure~\ref{fig:MMWU-algorithm}.

\begin{figure}
  \hrulefill
  \begin{center}
    \begin{minipage}{\textwidth}
      \begin{mylist}{6mm}
      \item[1.]
        Set $X_1=\I_N$ and $T=\bigl\lceil 2\ln(N)/\varepsilon^2\bigr\rceil$.
      \item[2.]
        For each $t=1,\ldots,T$, let
        \begin{equation}
          \rho_t = \frac{X_t}{\tr(X_t)},
        \end{equation}
        let $\Pi_t\in\Proj(\complex^M)$ be the projection operator
        corresponding to the positive eigenspace of $\Xi(\rho_t)$, and let
        \begin{equation}
          X_{t+1} = \exp\bigl(-\varepsilon\,\Xi^{\ast} 
          \bigl(\Pi_1 + \cdots + \Pi_t\bigr)\bigr).
        \end{equation}

      \item[3.]
        Output
        \begin{equation}
          \frac{1}{T}\sum_{t = 1}^T\bigip{\Pi_t}{\Xi(\rho_t)}.
        \end{equation}
      \end{mylist}
    \end{minipage}
  \end{center}
  \hrulefill
  \caption{Matrix multiplicative weights update algorithm for approximating
    the min-max value $\eta$ defined in~\eqref{eq:qip-pspace-def-eta}.}
  \label{fig:MMWU-algorithm}
\end{figure}

To analyze this algorithm, one may begin with the following technical lemma.
We omit the proof, which is based on elementary and routine calculus and matrix
theory.

\begin{lemma}
  \label{lemma:matrix-inner-product-bound}
  Let $N$ be a positive integer, let $H\in \Herm(\complex^N)$ be a Hermitian
  operator satisfying $\norm{H}\leq 1$, and let $\rho\in\Density(\complex^N)$ be
  a density operator.
  For every positive real number $\varepsilon > 0$, it holds that
  \begin{equation}
    \bigip{\rho}{\exp\bigl(-\varepsilon H\bigr)}
    \leq \exp\bigl(-\varepsilon\exp(-2 \varepsilon)\ip{\rho}{H}\bigr)
    \exp\bigl(2\varepsilon\sinh(2\varepsilon)\bigr).
  \end{equation}
\end{lemma}

The next lemma forms the mathematical backbone of the algorithm analysis.

\begin{lemma} \label{lemma:MMWU}
  Let $T$ and $N$ be positive integers, let 
  $H_1,\ldots,H_T \in \Herm(\complex^N)$ be Hermitian operators
  satisfying $\norm{H_t}\leq 1$ for each $t\in\{1,\ldots,T\}$,
  and let $\varepsilon > 0$.
  Define
  \begin{equation}
    X_1 = \I,\;\;
    X_{t+1} = \exp\bigl(-\varepsilon (H_1 + \cdots + H_t)\bigr),
    \;\;\text{and}\;\;
    \rho_t = \frac{X_t}{\tr(X_t)}
  \end{equation}
  for each $t\in\{1,\ldots,T\}$.
  It holds that
  \begin{equation}
    \label{eq:MMW-lemma-bound}
    \begin{multlined}
      \lambda_{\textup{min}}\bigl(H_1 + \cdots + H_T\bigr)\\
      \geq \exp(-2 \varepsilon) \sum_{t = 1}^T \ip{\rho_t}{H_t}
      - \frac{\ln(N)}{\varepsilon} - 2T\sinh(2\varepsilon).
    \end{multlined}
  \end{equation}
\end{lemma}

\begin{proof}
  For each $t\in\{1,\ldots,T\}$, one has
  \begin{equation}
    \tr(X_{t+1}) \leq \tr \bigl( X_t
    \exp\bigl(-\varepsilon H_t\bigr)\bigr)
    = \tr(X_t) \bigip{\rho_t}{\exp\bigl(-\varepsilon H_t\bigr)}
  \end{equation}
  by a matrix inequality known as the \emph{Golden--Thompson inequality}, which
  states that $\tr(\exp(A+B))\leq\tr(\exp(A)\exp(B))$
  for every choice of Hermitian operators $A$ and $B$.
  By applying this inequality repeatedly, and noting that $\tr(X_1) = N$, one
  finds that
  \begin{equation}
    \label{eq:MMWU-1}
    \tr(X_{T+1})
    \leq N\, \prod_{t = 1}^T
    \bigip{\rho_t}{\exp\bigl(-\varepsilon H_t\bigr)}.
  \end{equation}
  Because the trace of a positive semidefinite operator is at least as large as
  its largest eigenvalue, it follows that
  \begin{equation}
    \label{eq:MMWU-2}
    \tr(X_{T+1})
    \geq \lambda_{\textup{max}}(X_{T+1})
    = \exp\bigl(-\varepsilon \lambda_{\textup{min}}(H_1 + \cdots + H_T)\bigr).
  \end{equation}
  The required bound \eqref{eq:MMW-lemma-bound} is obtained by combining
  \eqref{eq:MMWU-1} and \eqref{eq:MMWU-2} with
  Lemma~\ref{lemma:matrix-inner-product-bound}.
\end{proof}

Now let us compare the output value of the algorithm to the min-max quantity
$\eta$.
It is evident that
\begin{equation}
  \eta \leq \frac{1}{T}\sum_{t = 1}^T\bigip{\Pi_t}{\Xi(\rho_t)};
\end{equation}
as each $\Pi_t$ is selected by the algorithm so that the quantity
$\ip{\Pi_t}{\Xi(\rho_t)}$ is maximized, it follows that
$\eta\leq\ip{\Pi_t}{\Xi(\rho_t)}$ for each $t\in\{1,\ldots,T\}$.
An upper bound on the output value is obtained from Lemma~\ref{lemma:MMWU}
by setting $H_t = \Xi^{\ast}(\Pi_t)$ for each $t\in\{1,\ldots,T\}$.
One has that $\norm{H_t} \leq 1$ by virtue of the fact that $\Xi$ is a
difference between two channels, and so it follows that
\begin{equation}
  \begin{split}
    \frac{1}{T}\sum_{t = 1}^T\bigip{\Pi_t}{\Xi(\rho_t)} 
    & \leq
    \exp(2\varepsilon) \,
    \lambda_{\textup{min}}\biggl(\Xi^{\ast} \biggl(
    \frac{\Pi_1 + \cdots + \Pi_T}{T} \biggr)\biggr)\\
    & \qquad
    + \frac{\ln(N)\exp(2\varepsilon)}{\varepsilon T}
    + \bigl(\exp(4\varepsilon) - 1\bigr).
  \end{split}
\end{equation}
One may observe that
\begin{equation}
  \begin{multlined}
    \lambda_{\textup{min}}\biggl(\Xi^{\ast} \biggl(
    \frac{\Pi_1 + \cdots + \Pi_T}{T} \biggr)\biggr)\\
    = \min_\rho \,\biggip{\Xi(\rho)}{\frac{\Pi_1 + \cdots + \Pi_T}{T}}
    \leq \eta,
  \end{multlined}
\end{equation}
and therefore
\begin{equation}
  \frac{1}{T}\sum_{t = 1}^T\bigip{\Pi_t}{\Xi(\rho_t)} 
  \leq \exp(2\varepsilon) \eta + 
  \frac{\varepsilon \exp(2\varepsilon)}{2}
  + \bigl(\exp(4\varepsilon) - 1\bigr).
\end{equation}
For any choice of $\delta \leq 1$, it holds that
$\exp(\delta) - 1 \leq 2\delta$, and by combining this bound with the 
observation that $\eta \leq 1$, one obtains
\begin{equation}
  \frac{1}{T}\sum_{t = 1}^T\bigip{\Pi_t}{\Xi(\rho_t)} 
  \leq \eta + 16 \varepsilon.
\end{equation}

\subsection{PSPACE and bounded-depth circuits}

The final step in the proof that $\class{QIP} = \class{PSPACE}$ is to analyze
the complexity of an algorithm that has been suggested by the previous two
subsections.
In more explicit terms, it is to be assumed that $A = (A_{\yes},A_{\no})$ is an
arbitrary promise problem in $\class{QIP}$, and our goal is to prove
$A\in\class{PSPACE}$.
There are two primary steps in the following algorithm, which will soon be
shown to be implementable as a $\class{PSPACE}$ algorithm for~$A$.

\pagebreak

\begin{mylist}{\parindent}
\item[1.]
  It holds that $\class{QIP} = \class{QIP}_{1,1/4}(3)$, and therefore, for a
  given input string $x$, there exists a corresponding three-turn verifier
  $V$ such that
  \begin{equation}
    x\in A_{\yes} \Rightarrow \omega(V) = 1
    \quad\text{and}\quad
    x\in A_{\no} \Rightarrow \omega(V) \leq \frac{1}{4}.
  \end{equation}
  Compute an explicit description of the mapping $\Xi$ corresponding to $V$,
  as described in Section~\ref{sec:QIP-reduction}.

\item[2.]
  Run the matrix multiplicative weights update algorithm described in
  Figure~\ref{fig:MMWU-algorithm} for approximating the min-max value $\eta$
  associated with $\Xi$, with sufficient precision to distinguish the cases
  $\eta = 0$ and $\eta \geq 1/2$ (e.g., $\varepsilon = 1/64$).
  \emph{Accept} if $\eta = 0$ and \emph{reject} if $\eta\geq 1/2$.
\end{mylist}

Of course, an explicit description of the mapping $\Xi$ will generally have
size exponential in $\abs{x}$, so one will not obtain a \class{PSPACE}
algorithm by applying the two steps above in the most straightforward way.
Instead, the steps are to be implemented by bounded-depth Boolean circuits,
and the implementation of the resulting algorithm as a \class{PSPACE} algorithm
will follow from a circuit complexity result due to
Borodin~\cite{Borodin77}.

In greater detail, consider the complexity class $\class{NC}$, which we will
take to include all functions, including predicates representing decision
problems, computable by logarithmic-space uniform Boolean circuit families of
polylogarithmic depth.
The requirement that such a family is logarithmic-space uniform implies that
its circuits are polynomial in size, and therefore represent polynomial-time
computations.

We also consider a ``scaled-up'' variant of $\class{NC}$, to be denoted
$\class{NC}(\mathit{poly})$, that consists of all functions computable by
\emph{polynomial-space uniform} families of Boolean circuits having
\emph{polynomial-depth}.
(The notation $\smash{\class{NC}(2^{\mathit{poly}})}$ has also previously been
used for this class \cite{BorodinCP83}.)
A family of circuits meeting these requirements could potentially have
exponential size, and therefore does not necessarily represent an efficient
computation.
However, the polynomial bound on the depth of these circuits does represent a
significant computational restriction.
In particular, restricting our attention to decision problems, we have
\begin{equation}
  \class{NC}(\mathit{poly}) = \class{PSPACE},
\end{equation}
which is the result of Borodin \cite{Borodin77} suggested above.
An appeal of this reformulation is that it allows one to make use of known
parallel algorithms for performing various computational tasks when designing
$\class{PSPACE}$ algorithms.

Given that $\class{NC}(\mathit{poly}) = \class{PSPACE}$, we may consider that
the task at hand is to prove $A\in\class{NC}(\mathit{poly})$.
When doing this we will make use of a property of $\class{NC}$ and
$\class{NC}(\mathit{poly})$, which is that functions in these classes
compose well.
Specifically, if $F$ is a function in $\class{NC}(\mathit{poly})$ and $G$ is a
function in $\class{NC}$, then the composition $G\circ F$ is also in
$\class{NC}(\mathit{poly})$.
This follows from the most straightforward way of composing the families of
circuits that compute $F$ and $G$.
Consequently, if it is proved that the first step of the algorithm above can be
implemented as an $\class{NC}(\mathit{poly})$ computation, and the second step
can be implemented as an $\class{NC}$ computation, then their composition
represents an $\class{NC}(\mathit{poly})$ computation.

Implementing the first step of the algorithm above as an
$\class{NC}(\mathit{poly})$ computation turns out to be straightforward,
through the use of elementary facts about quantum computations and
bounded-depth circuits.
For a given input string $x$, there are at most polynomially many quantum gates
in the circuit description of the verifier $V$ on this input, and they may be
listed in the order they are to be applied in polynomial time.
Expanding this list of gates into their explicit matrix representations
yields a polynomial-length list of exponential-size matrices, and producing
such a list is possible using an $\class{NC}(\mathit{poly})$ computation.
From such a list of matrices, one may obtain an explicit description of $\Xi$
through elementary matrix operations performed on the explicit matrix
representations of the gates.
As sums, differences, and iterated products of matrices can be implemented as
$\class{NC}$ computations, the fact that step 1 can be implemented as an
$\class{NC}(\mathit{poly})$ computation follows.

The implementation of the second step of the algorithm as an $\class{NC}$
computation is more involved, but the key observation underlying such an
implementation is that the main loop in step 2 of the algorithm described in
Figure~\ref{fig:MMWU-algorithm} runs for a number of iterations that is
logarithmic in $N$.
Based on this observation, the fact that the entire algorithm can be
implemented as an $\class{NC}$ computation follows, provided that
(i) each individual step of the algorithm can be implemented as an $\class{NC}$
computation, and (ii) the storage requirements of the algorithm grow at most
linearly with each iteration.
By fixing ahead of time the number of bits of precision to which each matrix
entry is stored, the second requirement is evidently met.
The first requirement follows from these observations:
\begin{mylist}{\parindent}
\item[1.] Elementary matrix computations can be performed in $\class{NC}$.
\item[2.] Matrix exponentials can be approximated with high accuracy in
  $\class{NC}$ (simply by truncating the power series representation
  of the exponential function to polynomially many terms).
\item[3.] Positive eigenspace computations for Hermitian operators can be
  approximated with high accuracy using the fact that polynomial root
  approximation is in \class{NC} \cite{BenOrFKT86,Neff94,BiniP98}.
\end{mylist}

It turns out to be somewhat tedious, although conceptually quite
straightforward, to verify that the approximations suggested above can be
performed in $\class{NC}$ with sufficient precision to yield a correct answer
to the problem of determining whether $\eta = 0$ or $\eta \geq 1/2$. Using
Lemma~\ref{lem:qip-pspace-eta} this finishes the proof of
Theorem~\ref{theorem:QIP=PSPACE}.
Interested readers may find further details in \cite{JainJUW11}.

\subsection{A complete problem for \class{QIP}}

As a simple byproduct of one part of the analysis described above in
Section~\ref{sec:QIP-reduction}, one finds that there is an interesting complete
promise problem for the class \class{QIP}, called the 
\emph{quantum circuit distinguishability} problem. 

Given that $\class{QIP} = \class{PSPACE}$, it holds that every
\class{PSPACE}-complete problem is \class{QIP}-complete (and by the same
reasoning the quantum circuit distinguishability problem to be described
shortly is also \class{PSPACE}-complete).
It is useful nevertheless to highlight this particular problem for two reasons.
One reason is that the quantum circuit distinguishability problem is a fairly
natural problem (within the setting of quantum information and computation)
having a more direct connection to the quantum interactive proof system
model than other known \class{PSPACE}-complete problems.
Indeed, the problem was known to be complete for \class{QIP} for quite some
time before it was known that \class{QIP} = \class{PSPACE}.
The second reason is that variants of this problem can, in some cases, be shown
to be complete for other complexity classes based on quantum interactive proof
systems.
One example, the Non-Identity Check problem, was already encountered in the
previous chapter.
Another example, concerning zero-knowledge quantum interactive proof systems,
will play an important role in the chapter following this one.

The quantum circuit distinguishability problem is parameterized by two real
numbers, $\alpha$ and $\beta$, satisfying $0\leq \beta < \alpha \leq 1$.
One may take $\alpha$ and $\beta$ to be fixed constants or functions of the
input length.
The statement of the problem is as follows.

\begin{center}
  \underline{$(\alpha,\beta)$-quantum circuit distinguishability
    ($(\alpha,\beta)$-QCD)}\nopagebreak\\[2mm]
  \begin{tabular}{lp{0.75\textwidth}}
      \emph{Input:} & 
      Quantum circuits $Q_0$ and $Q_1$, agreeing on both the number of input
      qubits they take and on the number of output qubits they produce.\\[2mm]
      \emph{Yes:} &
      $\frac{1}{2}\bignorm{Q_0 - Q_1}_{\Diamond} \geq \alpha$.\\[2mm]
      \emph{No:} &
      $\frac{1}{2}\bignorm{Q_0 - Q_1}_{\Diamond} \leq \beta$.
  \end{tabular}
\end{center}

The fact that the quantum circuit distinguishability problem is in \class{QIP}
(for a wide range of choices of $\alpha$ and $\beta$) may be shown directly
through a common sort of interactive proof system reminiscent of a ``blind
taste-test.''
In essence, the verifier chooses a bit $a\in\{0,1\}$ uniformly at random,
and agrees to perform for the prover one evaluation of the circuit $Q_a$,
for whichever value of $a$ was randomly selected.
The verifier accepts if and only if the prover successfully identifies the
value of $a$ after the circuit evaluation is performed on an input of the
prover's choice and the output qubits are returned to the prover.
The maximum acceptance probability is precisely
\begin{equation}
  \frac{1}{2} + \frac{1}{4} \bignorm{Q_0 - Q_1}_{\Diamond},
\end{equation}
yielding a quantum interactive proof system for the quantum circuit
distinguishability problem with completeness and soundness bounds
$(1 + \alpha)/2$ and $(1 + \beta)/2$, respectively.

One way to prove that the quantum circuit distinguishability problem is
complete for \class{QIP} is to first observe that the following 
\emph{close circuit images} problem is \class{QIP}-complete.

\begin{center}
  \underline{$(\alpha,\beta)$-close circuit images 
    ($(\alpha,\beta)$-CCI)}\\[2mm]
  \begin{tabular}{lp{0.75\textwidth}}
      \emph{Input:} &
      Quantum circuits $R_0$ and $R_1$, agreeing on both the number of input
      qubits they take and on the number of output qubits they produce.\\[2mm]
      \emph{Yes:} &
      $\fid(R_0(\rho_0),R_1(\rho_1))^2 \geq \alpha$ for some choice of
      input states $\rho_0$ and $\rho_1$.\\[2mm]
      \emph{No:} &
      $\fid(R_0(\rho_0),R_1(\rho_1))^2 \leq \beta$ for all choices of input
      states $\rho_0$ and $\rho_1$.
  \end{tabular}
\end{center}

\noindent
It follows from the analysis in Section~\ref{sec:QIP-reduction} that this
problem is \class{QIP}-complete for any choice of constants $\alpha$ and
$\beta$ satisfying $0 < \beta < \alpha \leq 1$.
It is also complete for $\alpha$ and $\beta$ being functions of the input
length, provided that $\alpha$ and $\beta$ can be computed in polynomial time,
that $\alpha$ and $\beta$ are separated by an inverse polynomial gap, and that
$\beta$ is at least inverse exponentially large.

Now, to prove that the quantum circuit distinguishability problem is
\class{QIP}-complete, it suffices to exhibit a reduction to it from the close
images problem.
The reduction itself is quite simple, and is illustrated in
Figure~\ref{fig:quantum-circuit-distinguishability}.
From an instance of the close images problem, one first creates a
controlled-unitary implementation of the two input circuits $R_0$ and $R_1$.
The circuits $Q_0$ and $Q_1$, representing an instance of the quantum circuit
distinguishability problem, take this control qubit as both an input and output,
but reverse the roles of the output qubits and ``garbage'' qubits of both $R_0$
and $R_1$.
The circuits $Q_0$ and $Q_1$ differ only by a
\begin{equation}
  Z = \begin{pmatrix}
    1 & 0\\
    0 & -1
  \end{pmatrix}
\end{equation}
operation being applied to the control qubit for one of the two circuits.

\begin{figure}
  \begin{center}
    \begin{minipage}[b]{0.55\textwidth}
      \begin{tikzpicture}[scale=0.5,
          channel/.style={draw, minimum height=10mm, minimum width=8mm,
            fill = ChannelColor, text=ChannelTextColor},
          unitary/.style={draw, minimum height=22mm, minimum width=10mm,
            fill = ChannelColor, text=ChannelTextColor},
          invisible/.style={minimum height=0mm, minimum width=8mm, 
            inner sep=0mm},
          gate/.style={draw, minimum height=4mm, minimum width=4mm,
            fill = ChannelColor, text=ChannelTextColor},
          control/.style={rounded corners = 1.5pt, draw, fill = Black, 
            inner sep = 0pt, minimum width = 3pt, minimum height = 3pt},
          >=latex]
        
        \node (R) at (0,0) [unitary] {$R_a$};
        \node (input) at (-5.4,1) [invisible] {$\rho$};
        \node (ancilla) at (-3,-1) [channel, minimum height=11mm] {\small
          $\ket{0^m}$};
        \node (garbage) at (3,-0.7) [channel, minimum height=14mm] {\small Tr};
        \node (output) at (5.4,1.4) [invisible] {$\;R_a(\rho)$};
        
        \draw ([yshift=8mm]input.east) -- ([yshift=18mm]R.west);
        \draw ([yshift=6mm]input.east) -- ([yshift=16mm]R.west);
        \draw ([yshift=4mm]input.east) -- ([yshift=14mm]R.west);      
        \draw ([yshift=2mm]input.east) -- ([yshift=12mm]R.west);
        \draw ([yshift=0mm]input.east) -- ([yshift=10mm]R.west);
        \draw ([yshift=-2mm]input.east) -- ([yshift=8mm]R.west);
        \draw ([yshift=-4mm]input.east) -- ([yshift=6mm]R.west);
        \draw ([yshift=-6mm]input.east) -- ([yshift=4mm]R.west);
        
        \draw ([yshift=8mm]ancilla.east) -- ([yshift=-2mm]R.west);
        \draw ([yshift=6mm]ancilla.east) -- ([yshift=-4mm]R.west);
        \draw ([yshift=4mm]ancilla.east) -- ([yshift=-6mm]R.west);
        \draw ([yshift=2mm]ancilla.east) -- ([yshift=-8mm]R.west);
        \draw ([yshift=0mm]ancilla.east) -- ([yshift=-10mm]R.west);
        \draw ([yshift=-2mm]ancilla.east) -- ([yshift=-12mm]R.west);
        \draw ([yshift=-4mm]ancilla.east) -- ([yshift=-14mm]R.west);
        \draw ([yshift=-6mm]ancilla.east) -- ([yshift=-16mm]R.west);
        \draw ([yshift=-8mm]ancilla.east) -- ([yshift=-18mm]R.west);
        
        \draw ([yshift=18mm]R.east) -- ([yshift=4mm]output.west);
        \draw ([yshift=16mm]R.east) -- ([yshift=2mm]output.west);
        \draw ([yshift=14mm]R.east) -- ([yshift=0mm]output.west);
        \draw ([yshift=12mm]R.east) -- ([yshift=-2mm]output.west);
        \draw ([yshift=10mm]R.east) -- ([yshift=-4mm]output.west);
        
        \draw ([yshift=4mm]R.east) -- ([yshift=11mm]garbage.west);
        \draw ([yshift=2mm]R.east) -- ([yshift=9mm]garbage.west);
        \draw ([yshift=0mm]R.east) -- ([yshift=7mm]garbage.west);
        \draw ([yshift=-2mm]R.east) -- ([yshift=5mm]garbage.west);
        \draw ([yshift=-4mm]R.east) -- ([yshift=3mm]garbage.west);
        \draw ([yshift=-6mm]R.east) -- ([yshift=1mm]garbage.west);
        \draw ([yshift=-8mm]R.east) -- ([yshift=-1mm]garbage.west);
        \draw ([yshift=-10mm]R.east) -- ([yshift=-3mm]garbage.west);
        \draw ([yshift=-12mm]R.east) -- ([yshift=-5mm]garbage.west);
        \draw ([yshift=-14mm]R.east) -- ([yshift=-7mm]garbage.west);
        \draw ([yshift=-16mm]R.east) -- ([yshift=-9mm]garbage.west);
        \draw ([yshift=-18mm]R.east) -- ([yshift=-11mm]garbage.west);
      \end{tikzpicture}
      \caption*{Unitary implementations of circuits $R_0$ and $R_1$ are assumed
        to act on the same number of qubits.}
    \end{minipage}\\[8mm]
    \begin{minipage}[b]{0.45\textwidth}
      \begin{tikzpicture}[scale=0.5,
          channel/.style={draw, minimum height=10mm, minimum width=8mm,
            fill = ChannelColor, text=ChannelTextColor},
          unitary/.style={draw, minimum height=22mm, minimum width=10mm,
            fill = ChannelColor, text=ChannelTextColor},
          invisible/.style={minimum height=0mm, minimum width=0mm, 
            inner sep=0mm},
          gate/.style={draw, minimum height=4mm, minimum width=4mm,
            fill = ChannelColor, text=ChannelTextColor},
          control/.style={rounded corners = 1.5pt, draw, fill = Black, 
            inner sep = 0pt, minimum width = 3pt, minimum height = 3pt},
          >=latex]
        
        \node (R) at (0,0) [unitary] {$R$};
        \node (input) at (-5,1) [invisible] {};
        \node (ancilla) at (-3,-1) [channel, minimum height=11mm] {\small
          $\ket{0^m}$};
        \node (garbage) at (5,-1) [invisible] {};
        \node (output) at (3,1.4) [channel, minimum height=7mm] {\small Tr};
        
        \node (cIn) at (-5,3.5) [invisible] {};
        \node (cOut) at (5,3.5) [invisible] {};
        \node (control) at (0,3.5) [control] {};
        
        \draw (cIn) -- (control) -- (cOut);
        
        \draw (control.south) -- (R.north);
        
        \draw ([yshift=8mm]input.east) -- ([yshift=18mm]R.west);
        \draw ([yshift=6mm]input.east) -- ([yshift=16mm]R.west);
        \draw ([yshift=4mm]input.east) -- ([yshift=14mm]R.west);      
        \draw ([yshift=2mm]input.east) -- ([yshift=12mm]R.west);
        \draw ([yshift=0mm]input.east) -- ([yshift=10mm]R.west);
        \draw ([yshift=-2mm]input.east) -- ([yshift=8mm]R.west);
        \draw ([yshift=-4mm]input.east) -- ([yshift=6mm]R.west);
        \draw ([yshift=-6mm]input.east) -- ([yshift=4mm]R.west);
        
        \draw ([yshift=8mm]ancilla.east) -- ([yshift=-2mm]R.west);
        \draw ([yshift=6mm]ancilla.east) -- ([yshift=-4mm]R.west);
        \draw ([yshift=4mm]ancilla.east) -- ([yshift=-6mm]R.west);
        \draw ([yshift=2mm]ancilla.east) -- ([yshift=-8mm]R.west);
        \draw ([yshift=0mm]ancilla.east) -- ([yshift=-10mm]R.west);
        \draw ([yshift=-2mm]ancilla.east) -- ([yshift=-12mm]R.west);
        \draw ([yshift=-4mm]ancilla.east) -- ([yshift=-14mm]R.west);
        \draw ([yshift=-6mm]ancilla.east) -- ([yshift=-16mm]R.west);
        \draw ([yshift=-8mm]ancilla.east) -- ([yshift=-18mm]R.west);
        
        \draw ([yshift=18mm]R.east) -- ([yshift=4mm]output.west);
        \draw ([yshift=16mm]R.east) -- ([yshift=2mm]output.west);
        \draw ([yshift=14mm]R.east) -- ([yshift=0mm]output.west);
        \draw ([yshift=12mm]R.east) -- ([yshift=-2mm]output.west);
        \draw ([yshift=10mm]R.east) -- ([yshift=-4mm]output.west);
        
        \draw ([yshift=4mm]R.east) -- ([yshift=14mm]garbage.west);
        \draw ([yshift=2mm]R.east) -- ([yshift=12mm]garbage.west);
        \draw ([yshift=0mm]R.east) -- ([yshift=10mm]garbage.west);
        \draw ([yshift=-2mm]R.east) -- ([yshift=8mm]garbage.west);
        \draw ([yshift=-4mm]R.east) -- ([yshift=6mm]garbage.west);
        \draw ([yshift=-6mm]R.east) -- ([yshift=4mm]garbage.west);
        \draw ([yshift=-8mm]R.east) -- ([yshift=2mm]garbage.west);
        \draw ([yshift=-10mm]R.east) -- ([yshift=0mm]garbage.west);
        \draw ([yshift=-12mm]R.east) -- ([yshift=-2mm]garbage.west);
        \draw ([yshift=-14mm]R.east) -- ([yshift=-4mm]garbage.west);
        \draw ([yshift=-16mm]R.east) -- ([yshift=-6mm]garbage.west);
        \draw ([yshift=-18mm]R.east) -- ([yshift=-8mm]garbage.west);
      \end{tikzpicture}
      \caption*{The circuit $Q_0$.}
    \end{minipage}
    \hfill
    \begin{minipage}[b]{0.45\textwidth}
      \begin{tikzpicture}[scale=0.5,
          channel/.style={draw, minimum height=10mm, minimum width=8mm,
            fill = ChannelColor, text=ChannelTextColor},
          unitary/.style={draw, minimum height=22mm, minimum width=10mm,
            fill = ChannelColor, text=ChannelTextColor},
          invisible/.style={minimum height=0mm, minimum width=0mm,
          inner sep=0mm},
          gate/.style={draw, minimum height=4mm, minimum width=4mm,
            fill = ChannelColor, text=ChannelTextColor},
          control/.style={rounded corners = 1.5pt, draw, fill = Black, 
            inner sep = 0pt, minimum width = 3pt, minimum height = 3pt},
          >=latex]
        
        \node (R) at (0,0) [unitary] {$R$};
        \node (input) at (-5,1) [invisible] {};
        \node (ancilla) at (-3,-1) [channel, minimum height=11mm] {\small
          $\ket{0^m}$};
        \node (garbage) at (5,-1) [invisible] {};
        \node (output) at (3,1.4) [channel, minimum height=7mm] {\small Tr};
        
        \node (cIn) at (-5,3.5) [invisible] {};
        \node (cOut) at (5,3.5) [invisible] {};
        \node (Z) at (3,3.5) [gate] {\scriptsize $Z$};
        \node (control) at (0,3.5) [control] {};
        
        \draw (cIn) -- (control) -- (Z) -- (cOut);
        
        \draw (control.south) -- (R.north);
        
        \draw ([yshift=8mm]input.east) -- ([yshift=18mm]R.west);
        \draw ([yshift=6mm]input.east) -- ([yshift=16mm]R.west);
        \draw ([yshift=4mm]input.east) -- ([yshift=14mm]R.west);      
        \draw ([yshift=2mm]input.east) -- ([yshift=12mm]R.west);
        \draw ([yshift=0mm]input.east) -- ([yshift=10mm]R.west);
        \draw ([yshift=-2mm]input.east) -- ([yshift=8mm]R.west);
        \draw ([yshift=-4mm]input.east) -- ([yshift=6mm]R.west);
        \draw ([yshift=-6mm]input.east) -- ([yshift=4mm]R.west);
        
        \draw ([yshift=8mm]ancilla.east) -- ([yshift=-2mm]R.west);
        \draw ([yshift=6mm]ancilla.east) -- ([yshift=-4mm]R.west);
        \draw ([yshift=4mm]ancilla.east) -- ([yshift=-6mm]R.west);
        \draw ([yshift=2mm]ancilla.east) -- ([yshift=-8mm]R.west);
        \draw ([yshift=0mm]ancilla.east) -- ([yshift=-10mm]R.west);
        \draw ([yshift=-2mm]ancilla.east) -- ([yshift=-12mm]R.west);
        \draw ([yshift=-4mm]ancilla.east) -- ([yshift=-14mm]R.west);
        \draw ([yshift=-6mm]ancilla.east) -- ([yshift=-16mm]R.west);
        \draw ([yshift=-8mm]ancilla.east) -- ([yshift=-18mm]R.west);
        
        \draw ([yshift=18mm]R.east) -- ([yshift=4mm]output.west);
        \draw ([yshift=16mm]R.east) -- ([yshift=2mm]output.west);
        \draw ([yshift=14mm]R.east) -- ([yshift=0mm]output.west);
        \draw ([yshift=12mm]R.east) -- ([yshift=-2mm]output.west);
        \draw ([yshift=10mm]R.east) -- ([yshift=-4mm]output.west);
        
        \draw ([yshift=4mm]R.east) -- ([yshift=14mm]garbage.west);
        \draw ([yshift=2mm]R.east) -- ([yshift=12mm]garbage.west);
        \draw ([yshift=0mm]R.east) -- ([yshift=10mm]garbage.west);
        \draw ([yshift=-2mm]R.east) -- ([yshift=8mm]garbage.west);
        \draw ([yshift=-4mm]R.east) -- ([yshift=6mm]garbage.west);
        \draw ([yshift=-6mm]R.east) -- ([yshift=4mm]garbage.west);
        \draw ([yshift=-8mm]R.east) -- ([yshift=2mm]garbage.west);
        \draw ([yshift=-10mm]R.east) -- ([yshift=0mm]garbage.west);
        \draw ([yshift=-12mm]R.east) -- ([yshift=-2mm]garbage.west);
        \draw ([yshift=-14mm]R.east) -- ([yshift=-4mm]garbage.west);
        \draw ([yshift=-16mm]R.east) -- ([yshift=-6mm]garbage.west);
        \draw ([yshift=-18mm]R.east) -- ([yshift=-8mm]garbage.west);
      \end{tikzpicture}
      \caption*{The circuit $Q_1$.}
    \end{minipage}
  \end{center}
 \caption{The circuit construction used to prove the \class{QIP}-completeness
   of the quantum circuit distinguishability problem.
   The circuits $Q_0$ and $Q_1$ are constructed from a controlled-unitary
   circuit implementation of $R_0/R_1$.
   The circuits $Q_0$ and $Q_1$ differ only in the $Z$ gate on the control
   qubit.}
 \label{fig:quantum-circuit-distinguishability}
\end{figure}

Intuitively speaking, $Q_0$ and $Q_1$ differ significantly precisely when
$R_0$ and $R_1$ can be forced to output similar states---for this is the
only way to produce a ``coherent'' superposition of $\ket{0}$ and $\ket{1}$ on
the control qubit, allowing the $Z$ gate to act nontrivially.
An analysis reveals that
\begin{equation}
  \frac{1}{2}\bignorm{Q_0 - Q_1}_{\Diamond}
  = \max_{\rho_0,\rho_1}\fid(R_0(\rho_0),R_1(\rho_1)),
\end{equation}
where the maximum is over all input states to $R_0$ and $R_1$
(and where the exponent 2 has intentionally been omitted on the right-hand
side expression), which directly translates to the required completeness and
soundness conditions being met.

\section{Chapter notes}

Quantum interactive proof systems were first proposed and studied in
\cite{Watrous99}, where it was proved that
$\class{PSPACE}\subseteq\class{QIP}(3)$.
(A journal version of this paper appeared later as \cite{Watrous03}.)
This result was subsumed shortly after by the results of
Kitaev and Watrous \cite{KitaevW00}, who proved the perfect completeness,
parallelization to three turns, and parallel repetition results described in
Sections~\ref{sec:QIP-perfect-completeness-and-parallelization}
and \ref{sec:qip-sdp}, as well as the upper-bound
$\class{QIP}\subseteq\class{EXP}$ through semidefinite programming.
The parallelization proof presented in this chapter is due to
Kempe, Kobayashi, Matsumoto, and Vidick \cite{KempeKMV09}, while the parallel
repetition proof presented was inspired by the work of Mittal and
Szegedy \cite{MittalS07} on product properties of semidefinite programs.
It is an open question whether the parallelization results could
be extended to interactive game with only two turns, and aside from the obvious
inclusions $\class{QMA}\subseteq \class{QIP(2)}\subseteq\class{QIP}(3)$ few
results are known concerning $\class{QIP}(2)$.

The public-coin variant of quantum interactive proof systems suggested in the
text, where Arthur (the verifier) generates random coin-flips and Merlin (the
prover) sends quantum information to Arthur, alternating in turns, and finally
Arthur performs a measurement on all of the quantum information sent by Merlin,
was considered by Marriott and Watrous \cite{MarriottW05} and shown to be
equivalent in power to the ordinary quantum interactive proof system model.

The equality \class{QIP} = \class{PSPACE} was proved by Jain, Ji, Upadhyay, and
Watrous \cite{JainJUW11}, following related but weaker results of Jain and 
Watrous \cite{JainW09} and Jain, Upadhyay, and Watrous \cite{JainUW09}.
In particular, the paper of Jain, Upadhyay, and Watrous proved the containment
$\class{QIP}(2)\subseteq\class{PSPACE}$. 
As stated in the text, the proof of \class{QIP} = \class{PSPACE} presented in
this chapter makes use of simplifications due to Wu \cite{Wu10}.
The specific formulation of the matrix multiplicative weights update method
upon which the proof that \class{QIP} = \class{PSPACE} relies was discovered
independently by Warmuth and Kuzmin \cite{WarmuthK06} and Arora and Kale
\cite{AroraK07}.
Readers interested in learning more about the matrix multiplicative weights
update method, as well as a bit of its history, are referred to Kale's PhD
thesis \cite{Kale07}.
Quantum interactive proof systems with unbounded error were considered in
\cite{ItoKW12} and shown to be equivalent in power to \class{EXP}, which
suggests (unless \class{PSPACE} = \class{EXP}) that a nonnegligible separation
between the completeness and soundness probability bounds in the definition of
\class{QIP} is essential to the proof of $\class{QIP} \subseteq \class{PSPACE}$.
(In contrast, the containment $\class{IP} \subseteq \class{PSPACE}$ does not
rely on a similar assumption.)

The quantum circuit distinguishability problem was proved to be complete for
\class{QIP} by Rosgen and Watrous \cite{RosgenW05}, and the proof of this fact
that has been summarized in this chapter makes use of a simplification due to
Kobayashi (through a personal communication).
Further work due to Rosgen on the computational hardness of distinguishing
between restricted classes of quantum channels appears in
\cite{Rosgen08b,Rosgen08a,Rosgen09}.
Another example of a complete promise problem for \class{QIP} relating to
quantum channels appears in \cite{HaydenMW14}.
Other problems of a similar nature are known to be complete for the class
\class{QSZK}, which is discussed in the chapter following this one.

Several interesting facts are known to hold concerning \emph{competing prover}
quantum interactive proof systems, where a \emph{yes-prover} and
\emph{no-prover} compete for the verifier's decision in the most natural way.
Proof systems of this sort were first considered by Gutoski and
Watrous \cite{GutoskiW05} and studied further in 
\cite{Gutoski05, Gutoski09, GutoskiW07, GutoskiW13}.
In particular, \cite{GutoskiW07} proved that the class \class{QRG} (short for
\emph{quantum refereed games}) representing problems for which a bounded-error
competing-prover quantum interactive proof system exists coincides with
\class{EXP};
while \cite{GutoskiW13} extended the machinery used to prove \class{QIP} =
\class{PSPACE} to obtain the result $\class{QRG}(2) = \class{PSPACE}$,
exactly characterizing the complexity of two-turn competing-prover quantum
interactive proof systems.

Quantum interactive proof systems in which only logarithmically many qubits
are exchanged between the prover and verifier were considered in
\cite{BeigiSW11} and shown to have only the power of \class{BQP}, thereby
offering no significant computational advantages over quantum computers that
do not interact with a prover.

\chapter{Quantum Zero-Knowledge} \label{chapter:QSZK}

Some interactive proof systems possess the property of being
\emph{zero-knowledge}, which means that the verifier can learn nothing (or
almost nothing) from an interaction with the prover, beyond the validity of the
statement being proved.
Although it might initially seem paradoxical, or perhaps impossible to fulfill,
there are many interesting interactive proof systems that possess this
property.

Perhaps the most well-known example is the Goldreich--Micali--Wigderson graph
isomorphism proof system described in Figure~\ref{fig:GMW-graph-isomorphism}.
\begin{figure}
\noindent\hrulefill
\begin{trivlist}
\item The input is a pair $(G_0,G_1)$ of simple, undirected $n$-vertex graphs.
  It is assumed that the prover knows a permutation $\sigma\in S_n$,
  for $S_n$ denoting the symmetric group on indices $\{1,\ldots,n\}$, that
  satisfies $\sigma(G_1) = G_0$ if $G_0$ and $G_1$ are isomorphic.
  \vspace{1mm}
\item {\bf Prover's step 1:}
  Choose a permutation $\pi\in S_n$ uniformly at random and send the graph $H =
  \pi(G_0)$ to the verifier.
  \vspace{1mm}
\item {\bf Verifier's step 1:}
  Choose $a\in\{0,1\}$ uniformly at random and send $a$ to the prover.
  (Implicitly, the verifier is challenging the prover to exhibit an
  isomorphism between $G_a$ and $H$.)
  \vspace{1mm}
\item {\bf Prover's step 2:}
  Set $\tau = \pi \sigma^a$ and send $\tau$ to the verifier.
  (If it is the case that $\sigma(G_1) = G_0$, then it must hold that
  $\tau(G_a) = H$.)
  \vspace{1mm}
\item {\bf Verifier's step 2:}
  Accept if $\tau(G_a) = H$, reject otherwise.
\end{trivlist}
\noindent\hrulefill
\caption{The Goldreich--Micali--Wigderson graph isomorphism proof system.}
\label{fig:GMW-graph-isomorphism}
\end{figure}
When one considers this proof system, it is intuitively clear that a verifier
can learn nothing from an interaction with the prover on a positive instance of
the problem, even if it ``cheats'' by deviating from the proof system---for all
it sees during an execution of the proof system is an isomorphism between
one of two input graphs and a randomly chosen permutation of that graph.
To prove in a more formal sense that this proof system is zero-knowledge, one
defines a polynomial-time \emph{simulator} that is capable of mimicking
anything that a \emph{cheating verifier} could compute by means of an
interaction with the prover (under the assumption that the input graphs are
isomorphic).
Because the simulator runs in polynomial time and (by assumption) does not
interact with the prover, one interprets that no ``knowledge'' about the input
graphs is revealed by the prover---whatever a cheating verifier could have
learned from the interaction could equally well have been computed efficiently
without the prover's help.

In the classical setting, the construction of such a simulator is fairly
straightforward in the case of the Goldreich--Micali--Wigderson graph
isomorphism proof system: it randomly guesses the cheating verifier's challenge
$a\in\{0,1\}$, computes the graph $H$ that satisfies $\tau(H) = G_a$ for a
random permutation $\tau$, and hopes the cheating verifier issues the challenge
$a$ when sent the graph $H$.
If it does, the simulator can output the prover's correct response, which is
$\tau$.
Otherwise the simulation has failed, but in this case one can simply
``rewind'' the simulator and try again, repeating the process with a new choice
of a random guess $a$, a random permutation $\tau$, and so on.
An analysis reveals that the simulator's guess must agree with the cheating
verifier's challenge with probability 1/2 (assuming a uniform selection of
$a$ and~$\tau$), and conditioned on a correct guess the output of the
simulation will be in perfect agreement with the cheating verifier/prover
interaction.

The quantum setting brings a new challenge to the subject of zero-knowledge,
as compared with the classical setting.
The algorithmic capabilities of quantum attackers (including
the ability to efficiently factor and compute discrete logarithms using Shor's
algorithms) represents one aspect of this challenge. 
More problematic from an analytic viewpoint is the failure of basic classical
techniques used in the study of zero-knowledge, such as the \emph{rewinding}
technique suggested above, to carry over directly to the quantum setting.
(In particular, if a cheating verifier begins its attack holding a potentially
useful quantum state, it could be that an unsuccessful simulation will have
irreparably damaged that state, ruling out a straightforward rewinding approach
like the one suggested above.)
Indeed, for some time even the simplest examples of classical zero-knowledge
proof systems, such as the Goldreich--Micali--Wigderson graph isomorphism
proof system, were not known to be zero-knowledge against quantum attacks,
leading some to question whether the zero-knowledge property could ever be
established in the quantum setting (aside from trivial cases in which the
zero-knowledge property holds vacuously).

The present chapter surveys some of the definitions and known facts concerning
quantum zero-knowledge proof systems.
Our focus will be on a variant of zero-knowledge known as \emph{statistical}
zero-knowledge, which demands an information-theoretic security condition as a
part of its definition.
The other fundamental variant is \emph{computational} zero-knowledge,
which relaxes the security condition in a natural complexity-theoretic way.
While the notion of quantum computational zero-knowledge is certainly
well-motivated (see, for instance, \cite{Kobayashi08} for work in this
direction) it is reasonable to claim that most of the uniquely quantum
aspects of zero-knowledge that are currently known are well-represented by the
study of statistical zero-knowledge.\footnote{
  There would also be a further overhead required to develop the requisite
  definitions and facts for a proper discussion of quantum computational
  zero-knowledge, so limiting our attention to quantum statistical
  zero-knowledge also serves to simplify the chapter in this respect.}

\pagebreak

Two highlights of the results to be described in this chapter are as follows:
\begin{mylist}{\parindent}
\item[1.]
  A fact known as the \emph{quantum rewinding lemma} may be used to prove that
  some simple interactive proof systems are zero-knowledge against general
  quantum verifier attacks.
  This is illustrated for the Goldreich--Micali--Wigderson graph isomorphism
  proof system introduced above.
  
\item[2.]
  It is proved that a promise problem called the \emph{close quantum states}
  problem is complete for the class \class{QSZK} of problems having quantum
  statistical zero-knowledge proof systems.
  The completeness of this problem allows one to conclude various facts,
  including the fact that every quantum statistical zero-knowledge proof system
  can be parallelized to just two turns (so that
  $\class{QSZK}\subseteq\class{QIP}(2)$).
\end{mylist}

\section{Definitions of zero-knowledge}

As suggested above, the two most fundamental variants of zero-knowledge are
\emph{statistical} zero-knowledge and \emph{computational} zero-knowledge.
In addition, there are two specific formulations of the zero-knowledge property
in both cases:
\emph{general verifier} zero-knowledge and \emph{honest verifier}
zero-knowledge.
Intuitively speaking, the distinction concerns how hard the verifier, now
thought of as an ``adversary'' to the proof system, might try to extract
knowledge from a prover.

General verifier zero-knowledge is the more cryptographically meaningful of the
two formulations, as one is required to prove security against arbitrary
attacks (made by computationally restricted adversaries) in this case.
Honest verifier zero-knowledge, on the other hand, only requires that a
verifier (or someone looking over the verifier's shoulder) learns nothing by
interacting with a prover, assuming that the verifier does not deviate from a
fixed specification of how it is supposed to behave.

The motivation for considering honest verifier zero-knowledge is that it is
generally much easier to reason about.
It may be viewed as a relaxation of the general verifier definition, so any
limitation that one proves on honest-verifier zero-knowledge proof systems must
also hold for general-verifier zero-knowledge proof systems.
In the particular case of statistical zero-knowledge (both in the classical and
quantum settings), it turns out that the honest verifier and general verifier
definitions lead to identical complexity classes.

\subsection{General verifier zero-knowledge}

Suppose that $P = (P_0,\ldots,P_{n-1})$ is a prover in an interactive game.
Under ordinary circumstances, it is to be expected that a fixed verifier $V$
engages in an interaction with $P$ and outputs a single binary value, indicating
acceptance or rejection.
In the context of zero-knowledge, one considers not just the interaction of
such a verifier $V$ with $P$, but of other entities (to be called
\emph{cheating verifiers}, as was done at the beginning of the chapter) that aim
to extract ``knowledge'' from $P$.
Figure~\ref{fig:cheating-verifier} illustrates a hypothetical cheating verifier
$W$ interacting with a prover $P = (P_0,P_1,P_2)$.
The cheating verifier $W$ does not necessarily output 0 or 1 like an ordinary
verifier, but instead aims to implement some channel $\Phi$ by means of an
interaction with $P$.
\begin{figure}
  \begin{center}
    \footnotesize
    \begin{tikzpicture}[scale=0.32, 
        turn/.style={draw, minimum height=14mm, minimum width=10mm,
          fill = ChannelColor, text=ChannelTextColor},
        invisible/.style={minimum width=7mm, minimum height=7mm},
        >=latex]
      
      \node (V0) at (-16,4.4) [invisible] {$\rho$};
      \node (V1) at (-8,4) [turn] {$W_1$};
      \node (V2) at (2,4) [turn] {$W_2$};
      \node (V3) at (12,4) [turn] {$W_3$};
      
      \node (M) at (17,4.4) [invisible] {$\Phi(\rho)$};
     
      \node (P0) at (-13,-1) [turn] {$P_0$};
      \node (P1) at (-3,-1) [turn] {$P_1$};
      \node (P2) at (7,-1) [turn] {$P_2$};
      
      \draw[->] ([yshift=4mm]V1.east) -- ([yshift=4mm]V2.west)
      node [above, midway] {$\reg{Z}_1$};
      
      \draw[->] (V0.east) -- ([yshift=4mm]V1.west)
      node [above, midway] {$\reg{Z}_0$};
      
      \draw[->] ([yshift=4mm]V2.east) -- ([yshift=4mm]V3.west)
      node [above, midway] {$\reg{Z}_2$};
      
      \draw[->] ([yshift=-4mm]V1.east) .. controls +(right:20mm) and 
      +(left:20mm) .. ([yshift=4mm]P1.west) node [right, pos=0.4] {$\reg{X}_1$};
      
      \draw[->] ([yshift=-4mm]V2.east) .. controls +(right:20mm) and 
      +(left:20mm) .. ([yshift=4mm]P2.west) node [right, pos=0.4] {$\reg{X}_2$};
      
      \draw[->] ([yshift=4mm]P0.east) .. controls +(right:20mm) and 
      +(left:20mm) .. ([yshift=-4mm]V1.west) node [left, pos=0.6] 
      {$\reg{Y}_0$};
      
      \draw[->] ([yshift=4mm]P1.east) .. controls +(right:20mm) and 
      +(left:20mm) .. ([yshift=-4mm]V2.west) node [left, pos=0.6] 
      {$\reg{Y}_1$};
      
      \draw[->] ([yshift=4mm]P2.east) .. controls +(right:20mm) and 
      +(left:20mm) .. ([yshift=-4mm]V3.west) node [left, pos=0.6]
      {$\reg{Y}_2$};
      
      \draw[->] ([yshift=-4mm]P0.east) -- ([yshift=-4mm]P1.west)
      node [below, midway] {$\reg{W}_0$};
      
      \draw[->] ([yshift=-4mm]P1.east) -- ([yshift=-4mm]P2.west)
      node [below, midway] {$\reg{W}_1$};
      
      \draw[->] ([yshift=4mm]V3.east) -- (M.west)
      node [above, midway] {$\reg{Z}_3$};

    \end{tikzpicture}
  \end{center}
  \caption{A cheating verifier $W = (W_1,\ldots,W_n)$ implements a channel
    $\Phi$ by means of an interaction with a prover
    $P = (P_0,\ldots,P_{n-1})$.} 
  \label{fig:cheating-verifier}
\end{figure}
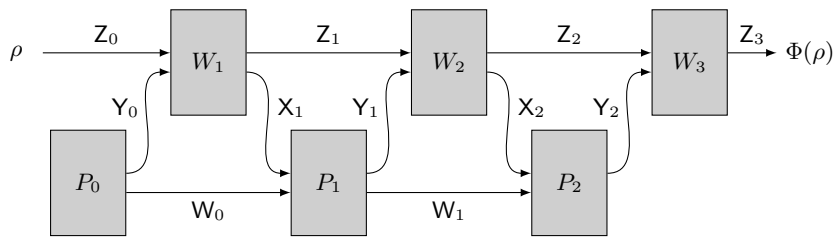

There is not a formal definition of ``knowledge'' in this situation, and the
term is placed in quotes to emphasize this point.
Rather, one considers the complexity-theoretic aspects of the class of channels
$\Phi$ that can be implemented through an interaction with $P$.
Intuitively speaking, if every such channel can be efficiently implemented
\emph{without} making use of an interaction with $P$, then it is to be
interpreted that no ``knowledge'' was revealed by $P$---for whatever the
cheating verifier might have learned from $P$ could equally well have been
computed efficiently without interacting with~$P$.
Although it is not necessary to use this terminology in the definitions to
follow, the term \emph{simulator} is typically used to refer to a hypothetical
entity that implements or approximates the channel implemented by a cheating
verifier/prover interaction in this way.

The following definitions aim to formalize the notions just suggested.
The first definition concerns the collection of channels that may be
implemented through an interaction with a given prover.

\begin{definition}
  Let $A\subseteq\Sigma^{\ast}$ be a set of binary strings, let $r$ and $s$ be
  polynomially bounded functions, and let $P(x)$ be a prover in an interactive
  game for each $x\in A$.
  Suppose further that
  \begin{equation}
    \label{eq:another-collection-of-channels}
    \{\Phi_x\,:\,x\in A\}
  \end{equation}
  is a collection of channels, where $\Phi_x$ transforms $r(\abs{x})$ qubits
  to $s(\abs{x})$ qubits, for each $x\in A$.
  The collection \eqref{eq:another-collection-of-channels} is
  \emph{efficiently implementable through an interaction with $P$} if 
  there exists a polynomial-time computable function $W$ possessing the
  following properties for each $x\in A$:
  \begin{mylist}{\parindent}
  \item[1.]
    $W(x)$ is an encoding of the quantum circuit description of a cheating
    verifier that is compatible with $P(x)$.
  \item[2.]
    In an interaction with $P(x)$, the cheating verifier encoded by $W(x)$
    takes $r(\abs{x})$ input qubits, produces $s(\abs{x})$ output qubits, and
    implements the channel $\Phi_x$ through its interaction with $P(x)$.
  \end{mylist}
\end{definition}

The next definition formalizes what it means for one collection of channels to
\emph{efficiently approximate} another.

\begin{definition}
  \label{definition:efficient-channel-approximation}
  Let $A\subseteq\Sigma^{\ast}$ be a set of binary strings, let $r$ and $s$ be
  polynomially bounded functions, and let
  $\varepsilon:\natural\rightarrow[0,1]$ be a function.
  Suppose further that
  \begin{equation}
    \label{eq:collection-of-channels}
    \{\Phi_x\,:\,x\in A\}
  \end{equation}
  is a collection of channels, where $\Phi_x$ transforms $r(\abs{x})$ qubits
  to $s(\abs{x})$ qubits for each $x\in A$.
  The collection \eqref{eq:collection-of-channels} is
  \emph{efficiently $\varepsilon$-approximable} if there exists a
  polynomial-time computable function $Q$ possessing the following properties
  for each $x\in A$:
  \begin{mylist}{\parindent}
  \item[1.]
    $Q(x)$ is an encoding of a quantum circuit transforming $r(\abs{x})$ qubits
    to $s(\abs{x})$ qubits.
  \item[2.]
    The circuit with encoding $Q(x)$ implements an
    $\varepsilon(\abs{x})$-approximation to the channel $\Phi_x$.
  \end{mylist}
\end{definition}

The property of a given prover being 
\emph{quantum statistical zero-knowledge} on a particular set of input strings
may now be defined.
Intuitively speaking, a prover has this property for a set of inputs if
every channel that can be implemented through an interaction with it on this
set of inputs can also be efficiently approximated without interacting with
it.

\begin{definition}
  Let $A\subseteq\Sigma^{\ast}$ be a set of binary strings and let $P(x)$ be a
  prover in an interactive game for each $x\in A$.
  It is said that $P$ is a \emph{quantum statistical zero-knowledge} prover
  on the set $A$ if, for every collection of channels
  \begin{equation}
    \label{eq:Phi-collection-of-channels-zero-knowledge}
    \{\Phi_x\,:\,x\in A\}
  \end{equation}
  that is efficiently implementable through an interaction with $P$,
  one has that the collection
  \eqref{eq:Phi-collection-of-channels-zero-knowledge} is also efficiently
  $\varepsilon$-approximable for some choice of a negligible\footnote{
    A function $\varepsilon:\natural\rightarrow[0,1]$ is typically defined to
    be \emph{negligible} if $\varepsilon(n) = o(n^{-c})$ for every
    positive constant $c > 0$.
    An alternative definition that serves equally well for the purposes of this
    survey is that $\varepsilon$ is negligible if and only if
    $\varepsilon(n) = O\bigl(2^{-n^c}\bigr)$ for some choice of a positive
    constant $c > 0$.}
  function $\varepsilon:\natural\rightarrow[0,1]$.
\end{definition}

\begin{definition}
  A promise problem $A = (A_{\textup{yes}},A_{\textup{no}})$ is contained in
  the complexity class $\class{QSZK}_{a,b}(m)$ if there exists
  a polynomial-time computable function $V$ and a function $P$ that possess the
  following properties:
  \begin{mylist}{\parindent}
  \item[1.]
    For every string $x\in A_{\textup{yes}} \cup A_{\textup{no}}$, one has that
    $V(x)$ is an encoding of a quantum circuit description of an $m$-message
    verifier in an interactive game, and $P(x)$ is a prover that is compatible
    with $V(x)$.
  \item[2.]
    \emph{Completeness.}
    For every string $x\in A_{\textup{yes}}$, the interaction between $V(x)$
    and $P(x)$ causes $V(x)$ to output 1 with probability at least $a$.
  \item[3.]
    \emph{Soundness.}
    For every string $x\in A_{\textup{no}}$, the value of the verifier
    $V(x)$ satisfies $\omega(V(x))\leq b$.
  \item[4.]
    \emph{Zero-knowledge.}
    $P$ is a quantum statistical zero-knowledge prover on the set
    $A_{\textup{yes}}$.
  \end{mylist}
\end{definition}

We will make use of shorthand conventions when referring to complexity classes
of the form $\class{QSZK}_{a,b}(m)$ that are similar to those used when
referring to $\class{QIP}_{a,b}(m)$.
In particular, it is to be understood that $a = 2/3$ and $b = 1/3$ when these
subscripts are omitted, while the omission of $m$ indicates a union over all
polynomially bounded functions $m$.

A few remarks concerning the previous definition are in order.
One is that the containment
\begin{equation}
  \class{QSZK}_{a,b}(m) \subseteq \class{QIP}_{a,b}(m)
\end{equation}
is immediate, for every choice of $m$, $a$, and $b$;
the verifier $V$ whose existence is required for a given promise problem $A$ to
be included in $\class{QSZK}_{a,b}(m)$ fulfills, without any modifications
being necessary, the requirements for $A$ to be contained in
$\class{QIP}_{a,b}(m)$.
The additional requirements for a promise problem to be included in
$\class{QSZK}_{a,b}(m)$, as compared with $\class{QIP}_{a,b}(m)$, therefore
concern the prover $P$ referred to in the definition.
First, it must be the case that this prover $P$ fulfills the requirements of
the completeness condition for $V$, and second, it must not be possible for
any cheating verifier to extract ``knowledge'' from $P(x)$, for any yes-input
string $x\in A_{\textup{yes}}$.

It is natural to ask why the zero-knowledge condition is required to hold only
on the set $A_{\textup{yes}}$, and not also on $A_{\textup{no}}$, for
instance.
Practically speaking, extending the zero-knowledge requirement from the set
$A_{\textup{yes}}$ to all of $A_{\textup{yes}}\cup A_{\textup{no}}$ would
immediately lead to the collapse $\class{QSZK}_{a,b}(m) = \class{BQP}$
(assuming $a$ and $b$ are representative of reasonable, bounded-error
probability bounds).
Indeed, even approximating the channel implemented by the \emph{honest}
verifier $V$ when interacting with $P$, which has no input qubits and one
output qubit, would provide a $\class{BQP}$ algorithm for $A$.
Requiring the zero-knowledge condition to hold for all input strings is
therefore too strong of a requirement to be interesting.
A more philosophical reason for requiring the zero-knowledge condition only
on $A_{\textup{yes}}$ is that there is no ``knowledge'' to be gained from a
prover $P$ on an input string $x\in A_{\textup{no}}$.
(Alternatively, one may view that it is not the role of cryptography to protect
the dishonest against one another.)

Another natural question regarding the definition above is why one considers
the possible \emph{channels} that can be implemented through an interaction
with a given prover, as opposed (for instance) to the states that may be
prepared through an interaction with a prover.
Again, there is a practical reason and a philosophical reason.
Practically speaking, the definition above that is based on channels rather
than state preparations is more robust, possessing various closure properties
that one would hope for such a concept to fulfill.
Philosophically speaking, the definition based on channels captures the idea
that a cheating verifier cannot \emph{increase} its knowledge through an
interaction with a given prover, while an analogous definition based on state
preparations would represent the idea that a verifier that knows nothing cannot
learn something through an interaction with a prover.
The channel-based definition is the more satisfying definition from a
cryptographic point of view for both reasons.

As suggested above, it is possible to formulate other natural variants of
zero-knowledge, such as \emph{quantum computational zero-knowledge}, by
modifying the definitions of quantum statistical zero-knowledge.
In particular, it is only the notion of one channel approximating another
(Definition~\ref{definition:channel-approximation}) that needs to be 
modified in a substantive way to lead to the notion of quantum computational
zero-knowledge.
Rather than requiring two channels to be close with respect to the diamond
norm distance, which represents the impossibility for the two channels to be
distinguished in an information-theoretic sense, quantum computational
zero-knowledge requires only that two channels cannot be
\emph{efficiently distinguished} (by polynomial-size quantum circuit families).

\subsection{Honest verifier zero-knowledge}

The zero-knowledge property of a prover $P$ on a collection of yes-inputs, for
a given promise problem $A = (A_{\textup{yes}},A_{\textup{no}})$, is sometimes
difficult to establish---one must prove the existence of an efficient
\mbox{$\varepsilon$-approximation} (i.e., a simulator) for \emph{every}
collection of 
channels that can be efficiently implemented through an interaction with $P$.
A simpler and less restrictive notion of zero-knowledge is known as
\emph{honest-verifier} zero-knowledge.
To define this notion, one must first define what is meant by the \emph{view}
of a given verifier in an interactive game.

\begin{definition}
  Suppose $V = (V_1,\ldots,V_n)$ is a verifier in an interactive game and
  $P = (P_0,\ldots,P_{n-1})$ is a prover compatible with $V$.
  The \emph{view} of $V$ when interacting with $P$ is the state
  \begin{equation}
    \op{view}(V,P) = \rho_0\otimes \rho_1 \otimes \cdots \otimes \rho_{n-1},
  \end{equation}
  where
  \begin{equation}
    \rho_0 \in\Density(\Y_0),\;
    \rho_1 \in \Density(\Z_1\otimes\Y_1),\;
    \ldots,\;
    \rho_{n-1} \in \Density(\Z_{n-1}\otimes\Y_{n-1})
  \end{equation}
  represent the states of the registers $\reg{Y}_0$, $(\reg{Z}_1,\reg{Y}_1)$,
  \ldots, $(\reg{Z}_{n-1},\reg{Y}_{n-1})$ at their corresponding times in the
  interaction between $V$ and $P$.
\end{definition}

\noindent
Figure~\ref{fig:verifier-view} illustrates the state representing the view of
a 5-message verifier interacting with a prover.
\begin{figure}
  \begin{center}
    \footnotesize
    \begin{tikzpicture}[scale=0.35, 
        turn/.style={draw, minimum height=14mm, minimum width=10mm,
          fill = ChannelColor, text=ChannelTextColor},
        measure/.style={draw, minimum width=7mm, minimum height=7mm,
          fill = ChannelColor},
        >=latex]
      
      \node (V1) at (-8,4) [turn] {$V_1$};
      \node (V2) at (2,4) [turn] {$V_2$};
      \node (V3) at (12,4) [turn] {$V_3$};
      
      \node (M) at (17,4.4) [measure] {};
     
      \node (P0) at (-13,-1) [turn] {$P_0$};
      \node (P1) at (-3,-1) [turn] {$P_1$};
      \node (P2) at (7,-1) [turn] {$P_2$};

      \node[draw, minimum width=5mm, minimum height=3.5mm, fill=ReadoutColor]
      (readout) at (M) {};
      
      \draw[thick] ($(M)+(0.3,-0.15)$) arc (0:180:3mm);
      \draw[thick] ($(M)+(0.2,0.2)$) -- ($(M)+(0,-0.2)$);
      \draw[fill] ($(M)+(0,-0.2)$) circle (0.5mm);
      
      \draw[->] ([yshift=4mm]V1.east) -- ([yshift=4mm]V2.west)
      node [above, midway] {$\reg{Z}_1$};
      
      \draw[->] ([yshift=4mm]V2.east) -- ([yshift=4mm]V3.west)
      node [above, midway] {$\reg{Z}_2$};
      
      \draw[->] ([yshift=-4mm]V1.east) .. controls +(right:20mm) and 
      +(left:20mm) .. ([yshift=4mm]P1.west) node [right, pos=0.4] {$\reg{X}_1$};
      
      \draw[->] ([yshift=-4mm]V2.east) .. controls +(right:20mm) and 
      +(left:20mm) .. ([yshift=4mm]P2.west) node [right, pos=0.4] {$\reg{X}_2$};
      
      \draw[->] ([yshift=4mm]P0.east) .. controls +(right:20mm) and 
      +(left:20mm) .. ([yshift=-4mm]V1.west) node [left, pos=0.6] 
      {$\reg{Y}_0$};
      
      \draw[->] ([yshift=4mm]P1.east) .. controls +(right:20mm) and 
      +(left:20mm) .. ([yshift=-4mm]V2.west) node [left, pos=0.6] 
      {$\reg{Y}_1$};
      
      \draw[->] ([yshift=4mm]P2.east) .. controls +(right:20mm) and 
      +(left:20mm) .. ([yshift=-4mm]V3.west) node [left, pos=0.6]
      {$\reg{Y}_2$};
      
      \draw[->] ([yshift=-4mm]P0.east) -- ([yshift=-4mm]P1.west)
      node [below, midway] {$\reg{W}_0$};
      
      \draw[->] ([yshift=-4mm]P1.east) -- ([yshift=-4mm]P2.west)
      node [below, midway] {$\reg{W}_1$};
      
      \draw[->] ([yshift=4mm]V3.east) -- (M.west)
      node [above, midway] {$\reg{Z}_3$};
      
      \node[rotate=31, draw, thick, ellipse, minimum width=6mm,
        minimum height=18mm] at (-2.2,3.6) {};
      
      \node[rotate=31, draw, thick, ellipse, minimum width=6mm,
        minimum height=18mm] at (7.8,3.6) {};
      
      \node[draw, circle, thick, minimum width=5.2mm, minimum
        height=5.2mm] at (-11.4,2.1) {};

    \end{tikzpicture}
  \end{center}
  \caption{The verifier's \emph{view} of an interaction is represented by
    the sequence of states corresponding to the registers it hold immediately
    prior to its actions.
    In the 5-message case, the view corresponds to the states of the
    registers $\reg{Y}_0$, $(\reg{Z}_1,\reg{Y}_1)$, and
    $(\reg{Z}_2,\reg{Y}_2)$.}
  \label{fig:verifier-view}
\end{figure}
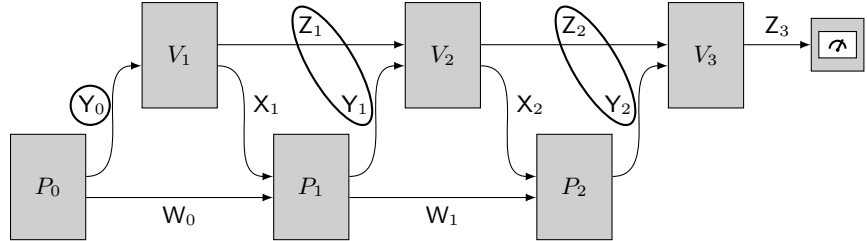

It must be noted that the state
$\rho_0\otimes\rho_1\otimes\cdots\otimes\rho_{n-1}$ representing the view of a
verifier interacting with a prover does not itself represent a state of a
collection of registers that ever co-exist at any moment.
Rather, one should view this state as a mathematical abstraction---it is not
possible to faithfully represent a transcript of a quantum interaction by a
single quantum state, but the state 
$\rho_0\otimes\rho_1\otimes\cdots\otimes\rho_{n-1}$ is ``good enough'' to
serve the needs of a study of the zero-knowledge property.


Next, we require a definition that formalizes what it means for a set of states
to be efficiently approximable.
This definition is analogous to
Definition~\ref{definition:efficient-channel-approximation}, but for states
rather than channels.

\begin{definition}
  Let $A\subseteq\Sigma^{\ast}$ be a set of binary strings, let $r$ be a
  polynomially bounded function, and let $\varepsilon:\natural\rightarrow[0,1]$
  be a function.
  Suppose further that
  \begin{equation}
    \label{eq:collection-of-states}
    \{\rho_x\,:\,x\in A\}
  \end{equation}
  is a collection of states on $r(\abs{x})$ qubits for each $x\in A$.
  The collection \eqref{eq:collection-of-states} is
  \emph{efficiently $\varepsilon$-approximable} if there exists a
  polynomial-time computable function $Q$ possessing the following properties
  for each $x\in A$:
  \begin{mylist}{\parindent}
  \item[1.]
    $Q(x)$ is an encoding of a quantum circuit taking no input qubits and
    outputting a state on $r(\abs{x})$ qubits.
  \item[2.]
    The state output by the circuit with encoding $Q(x)$ is an
    $\varepsilon(\abs{x})$-approximation to the state $\rho_x$.
  \end{mylist}
\end{definition}

The honest-verifier zero-knowledge property is concerned only with one specific
prover/verifier pair (for each yes-input to the problem under
consideration)---it states that the view of the verifier when
interacting with the prover must be close to a state that could be efficiently
prepared by a quantum circuit (without an interaction with a prover).
One critical requirement of the definition that must be highlighted is that
the verifier with respect to which it is defined must be represented by a
collection of isometric circuits, and not by general quantum circuits.
The reason for this requirement will be discussed shortly, but first we will
state the definition itself.

\begin{definition}[Honest-verifier quantum statistical zero-knowledge]
  \label{definition:HVQSZK}
  A promise problem $A = (A_{\textup{yes}},A_{\textup{no}})$ is contained
  in the complexity class $\class{HVQSZK}_{a,b}(m)$ if there exists
  a polynomial-time computable function $V$ and a function $P$ that possess the
  following properties:
  \begin{mylist}{\parindent}
  \item[1.]
    For every string $x\in A_{\textup{yes}} \cup A_{\textup{no}}$, one has that
    $V(x)$ is an encoding of an isometric quantum circuit description of an
    $m$-message verifier in an interactive game, and $P(x)$ is a prover
    compatible with $V(x)$.
  \item[2.]
    \emph{Completeness.}
    For every string $x\in A_{\textup{yes}}$, the interaction between $V(x)$
    and $P(x)$ causes $V(x)$ to output 1 with probability at least $a$.
  \item[3.]
    \emph{Soundness.}
    For every string $x\in A_{\textup{no}}$, it holds that 
    $\omega(V(x))\leq b$. 
  \item[4.]
    \emph{Zero-knowledge.}
    The collection of states
    \begin{equation}
      \bigl\{ \textup{view}(V(x),P(x))\,:\,x\in A_{\textup{yes}}\bigr\}
    \end{equation}
    is efficiently $\varepsilon$-approximable for some choice of a negligible
    function $\varepsilon:\natural\rightarrow[0,1]$.
  \end{mylist}
\end{definition}

It is appropriate to make a few comments regarding the above definition of
honest-verifier quantum statistical zero-knowledge.
One may reasonably view the definition to be representative of an easily
checked constraint on a prover/verifier pair that is inspired by the
zero-knowledge condition, and chosen so that a general zero-knowledge
prover/verifier pair is clearly honest-verifier zero-knowledge as well.
It is therefore evident that \class{QSZK} is contained in \class{HVQSZK}.
On the other hand, while the constraint represented by the definition of
\class{HVQSZK} is intended to be strong enough to allow nontrivial bounds on
the power of quantum statistical zero-knowledge interactive proofs to be
obtained, it should not be viewed that the definition is necessarily
cryptographically well-motivated or satisfying.
As it turns out, it is indeed the case that \class{HVQSZK} and \class{QSZK}
coincide, but this fact should perhaps be seen as good fortune rather than
something that should be expected.\footnote{
  It should be said, however, that it was not \emph{unexpected} that this
  equality would hold, as the analogous fact was known to hold in the classical
  setting prior to the formulation of the definition \cite{GoldreichSV98}.}

An important aspect of the above definition already alluded to is the
requirement that the verifier $V$ in the definition is constrained to be
isometric, meaning that each of the channels it performs must correspond to a
linear isometry.
By making this requirement of the verifier, one is potentially making the
zero-knowledge condition harder to satisfy, and this potential of added
difficulty must simply be accepted as an aspect of the definition.
The utility of this assumption is that a complementary relationship between
what the verifier sees and what the prover controls naturally emerges.

\section{Quantum rewinding}

This section discusses a fact known as the \emph{quantum rewinding lemma},
which allows for some interesting quantum interactive proof systems to be
proved to possess the (general verifier) zero-knowledge property.
The quantum rewinding lemma has a close connection to the strong error
reduction procedure for \class{QMA}, which was discussed in
Section~\ref{section:QMA-error-reduction}.

\subsection{Exact quantum rewinding lemma}

When introducing the quantum rewinding lemma, it is helpful to begin with an
exact version, which conveys the most relevant ideas of the lemma while
avoiding complications present in an approximate version.

Consider a unitary quantum circuit $Q$ acting on $k+m$ qubits, for positive
integers $k$ and $m$.
It is to be assumed that the input qubits comprise two registers: an $k$ qubit
register $\reg{X}$ and an $m$ qubit register $\reg{Y}$.
The first qubit output by $Q$ will be considered as a single-qubit register
$\reg{A}$, while the remaining $k + m - 1$ qubits form a register $\reg{Z}$.
Figure~\ref{fig:rewinding-circuit-1} provides an illustration of such a
circuit.

\begin{figure}
  \begin{center}
    \footnotesize
    \begin{tikzpicture}[scale=0.4, 
        circuit/.style={draw, minimum height=24mm, minimum width=20mm,
          fill = ChannelColor, text=ChannelTextColor},
        empty/.style={minimum width=10mm},
        >=latex]
      
      \node (Q) at (0,0) [circuit] {$Q$};
      \node (in) at (-6,0) [empty] {};
      \node (out) at (6,0) [empty] {};

      \foreach \y in {-24,-22,...,2} {
        \draw ([yshift=\y mm]in.east) -- ([yshift=\y mm]Q.west) {};
      }

      \foreach \y in {24,22,...,10} {
        \draw ([yshift=\y mm]in.east) -- ([yshift=\y mm]Q.west) {};
      }

      \foreach \y in {-24,-22,...,16} {
        \draw ([yshift=\y mm]Q.east) -- ([yshift=\y mm]out.west) {};
      }

      \draw ([yshift=24mm]Q.east) -- ([yshift=24mm]out.west) {};
      
      \node at (-7,1.5) {\begin{tabular}{c}
          $\reg{X}$\\[-1mm]
          \scriptsize{($k$ qubits)}
        \end{tabular}};
      
      \node at (-7,-1.1) {\begin{tabular}{c}
          $\reg{Y}$\\[-1mm]
          \scriptsize{($m$ qubits)}
        \end{tabular}};
      
      \node at (7,-0.5) {\begin{tabular}{c}
          $\reg{Z}$\\[-1mm]
          \scriptsize{($k+m-1$}\\[-1mm]
          \scriptsize{qubits)}
        \end{tabular}};
      
      \node at (7,2.5) {$\reg{A}$};

    \end{tikzpicture}
  \end{center}
  \caption{A unitary circuit $Q$ for the quantum rewinding lemma.}
  \label{fig:rewinding-circuit-1}
\end{figure}
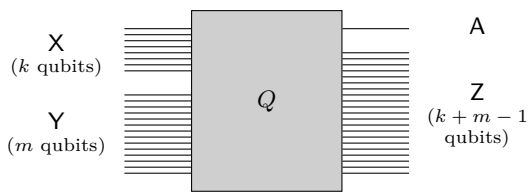

The situation that the quantum rewinding lemma concerns is that
$Q$ acts on a pure state of the form $\ket{\psi}\ket{0^m}$, resulting in a
state
\begin{equation}
  Q\ket{\psi}\ket{0^m}
  = \sqrt{p(\psi)} \ket{0} \ket{\phi_0(\psi)} + 
  \sqrt{1-p(\psi)} \ket{1} \ket{\phi_1(\psi)}
\end{equation}
for some real number $p(\psi) \in [0,1]$ and $(k+m-1)$-qubit unit vectors
$\ket{\phi_0(\psi)}$ and $\ket{\phi_1(\psi)}$.
If one measures the qubit $\reg{A}$ with respect to the standard basis, the
result will be 0 with probability $p(\psi)$, and conditioned on this outcome
the state of $\reg{Z}$ becomes $\ket{\phi_0(\psi)}$.
Otherwise, the measurement outcome is 1, and the state of $\reg{Z}$ becomes
$\ket{\phi_1(\psi)}$.
It is to be viewed that one desires to obtain the state $\ket{\phi_0(\psi)}$.
Evidently, if $Q$ is performed on $\ket{\psi}\ket{0^m}$, the register $\reg{A}$
is measured with respect to the standard basis, and the outcome $0$ is
obtained, then this goal is achieved---but this only happens with probability
$p(\psi)$, and one may hope to improve on this probability of success.

In general, there is little that can be done to 
successfully recover the state $\ket{\phi_0(\psi)}$ with probability greater
than $p(\psi)$, under the assumption that a single copy of the input state
$\ket{\psi}$ is made available.
What the quantum rewinding lemma provides is a condition under which it is
possible to efficiently recover the state $\ket{\phi_0(\psi)}$ with high
probability, from a single copy of $\ket{\psi}$.
The condition is that the probability $p(\psi)$ is nonzero and independent of
$\ket{\psi}$.
This is, in some sense, quite intuitive---if the probability $p(\psi)$ is
independent of $\ket{\psi}$, then a measurement of $\reg{A}$ reveals absolutely
no information about $\ket{\psi}$, so one might expect the quantum information
represented by $\ket{\psi}$ to be somehow contained in~$\reg{Z}$.
One can therefore hope to recover $\ket{\psi}$, and therefore have another try
at obtaining $\ket{\phi_0(\psi)}$ by running $Q$ and measuring $\reg{A}$.
The proof of the quantum rewinding lemma explains precisely how this may be
done.

\begin{lemma}[Exact quantum rewinding lemma]
  Let $Q$ be a unitary quantum circuit acting on $k+m$ qubits, and suppose that
  there exists a real number $p\in(0,1)$ for which the following statement
  holds: for every choice of an $k$-qubit unit vector $\ket{\psi}$, there exist
  $(k+m-1)$-qubit unit vectors $\ket{\phi_0(\psi)}$ and $\ket{\phi_1(\psi)}$
  such that
  \begin{equation}
    Q\ket{\psi}\ket{0^m}
    = \sqrt{p}\, \ket{0} \ket{\phi_0(\psi)} + 
    \sqrt{1-p}\, \ket{1} \ket{\phi_1(\psi)}.
  \end{equation}
  For every $\varepsilon \in (0,1/2)$, there exists a quantum circuit $R$
  taking $k$ input qubits, outputting $k+m-1$ qubits, having size
  \begin{equation}
    \label{eq:size-of-rewinder}
    \op{size}(R) = O\Biggl(\frac{\log(1/\varepsilon)\tinyspace\op{size}(Q)}{p
      (1-p)}\Biggr),
  \end{equation}
  and satisfying
  \begin{equation}
    \label{eq:rewinding-error-bound}
    \bignorm{\bigket{\phi_0(\psi)}\bigbra{\phi_0(\psi)}
      - R\bigl(\ket{\psi}\bra{\psi}\bigr)}_1 < \varepsilon
  \end{equation}
  for every $k$-qubit unit vector $\ket{\psi}$.
  Moreover, a description of the circuit $R$ can be generated from a
  description of $Q$ in polynomial time.
\end{lemma}

\begin{proof}
  Consider the procedure described in Figure~\ref{fig:rewinding-procedure}, and
  let $R$ be a quantum circuit implementing this procedure.
  \begin{figure}
    \hrulefill
    \vspace{-1mm}
    \begin{tabbing}
      xxx\=xxx\=xxx\=\kill
      \emph{Initial conditions:}\+\\[1mm]
      The register $\reg{X}$ contains a $k$-qubit quantum input $\ket{\psi}$.\\
      The register $\reg{Y}$ is initialized to the state $\ket{0^m}$.\\[1mm]
      \-\kill
      \emph{The procedure:}\+\\[1mm]
      Apply the circuit $Q$ to the pair $(\reg{X},\reg{Y})$ obtaining 
      $(\reg{A},\reg{Z})$.\\[1mm]
      Repeat $\lceil \log(1/\varepsilon)/(p(1-p))\rceil$ times:\\[1mm]
      \+\kill
      Measure $\reg{A}$ with respect to the standard basis.\\[1mm]
      If the outcome of the measurement is 1, do the following:\\[1mm]
      \+\kill
      Apply $Q^{-1}$ to $(\reg{A},\reg{Z})$, obtaining
      $(\reg{X},\reg{Y})$.\\[1mm]
      Apply the unitary operation $U = 2 \ket{0^m}\bra{0^m}-\I$ to
      $\reg{Y}$.\\[1mm]
      Apply $Q$ to the pair $(\reg{X},\reg{Y})$ obtaining $(\reg{A},\reg{Z})$.
      \\[1mm]
      \-\kill
      \-\kill
      Output the register $\reg{Y}$.
    \end{tabbing}
    \vspace{-3mm}
    \hrulefill
    \caption{The quantum rewinding procedure.}
    \label{fig:rewinding-procedure}
  \end{figure}
  It is evident from the description of the procedure that one may design $R$
  so that \eqref{eq:size-of-rewinder} holds, and moreover that a description of
  $R$ can be generated from a description of $Q$ in polynomial time.
  It remains to consider the output of $R$.

  The first step of the analysis is to consider the two operators
  \begin{equation}
    \begin{split}
      P_0 & = (\I\otimes\bra{0^m}) Q^{\ast} (\ket{0}\bra{0} \otimes \I) Q
      (\I\otimes\ket{0^m}),\\
      P_1 & = (\I\otimes\bra{0^m}) Q^{\ast} (\ket{1}\bra{1} \otimes \I) Q
      (\I\otimes\ket{0^m}),
    \end{split}
  \end{equation}
  where $Q$ is being regarded as a unitary operator in these equations.
  Because $P_0$ and $P_1$ are positive semidefinite and satisfy 
  $P_0 + P_1 = \I$, they describe a measurement $\{P_0,P_1\}$.
  This is the operator description of the measurement that is applied to
  $\reg{X}$ when $\reg{Y}$ is initialized to $\ket{0^m}$, $Q$ is applied to
  $(\reg{X},\reg{Y})$, yielding $(\reg{A},\reg{Z})$, and $\reg{A}$ is measured
  with respect to the standard basis.
  By the assumptions of the lemma, one has
  \begin{equation}
    \bra{\psi} P_0 \ket{\psi} = p
    \quad\text{and}\quad
    \bra{\psi} P_1 \ket{\psi} = 1-p
  \end{equation}
  for every choice of $\ket{\psi}$, which implies that
  \begin{equation}
    \label{eq:rewinding-measurement-trivial}
    P_0 = p \I \quad\text{and}\quad P_1 = (1-p) \I.
  \end{equation}
  This is a consequence of the fact that an operator $X$ acting on some complex
  vector space is uniquely determined by the function 
  $\ket{\psi}\mapsto\bra{\psi}X\ket{\psi}$, defined over the unit sphere of
  that space.
  
  Now, suppose $\ket{\psi}$ is a particular $k$-qubit unit vector, and the
  procedure described in Figure~\ref{fig:rewinding-procedure} is run when
  $\reg{X}$ is in the state $\ket{\psi}$ and $\reg{Y}$ is initialized to the
  state $\ket{0^m}$.
  The first step of the procedure is the application of $Q$, which leaves
  $(\reg{A},\reg{Z})$ in the state
  \begin{equation}
    Q\ket{\psi}\ket{0^m}
    = \sqrt{p}\,\ket{0}\ket{\phi_0(\psi)}+\sqrt{1-p}\,\ket{1}\ket{\phi_1(\psi)}.
  \end{equation}
  The loop is then iterated, beginning with a measurement of $\reg{A}$ with
  respect to the standard basis.
  If the measurement outcome is 0, the state of $\reg{Z}$ becomes
  $\ket{\phi_0(\psi)}$.
  An inspection of the procedure reveals that once the measurement outcome 0 is
  obtained, no further changes are made to $\reg{Z}$, so the state
  $\ket{\phi_0(\psi)}$ will indeed be the output of the procedure, conditioned
  on this first measurement resulting in 0.

  If the measurement result is 1, the state of the pair $(\reg{A},\reg{Z})$
  becomes $\ket{1}\ket{\phi_1(\psi)}$, and the conditional statement is
  performed.
  This results in $(\reg{A},\reg{Z})$ being left in the state
  \begin{equation}
    \begin{multlined}
    Q (\I\otimes U) Q^{\ast} \ket{1}\ket{\phi_1(\psi)}\\
    = -\ket{1}\ket{\phi_1(\psi)} + 2 
    Q (\I\otimes \ket{0^m}\bra{0^m}) Q^{\ast} \ket{1}\ket{\phi_1(\psi)}.
    \end{multlined}
  \end{equation}
  This expression can be simplified by making use of the equations
  \eqref{eq:rewinding-measurement-trivial}.
  In particular, one has
  \begin{equation}
    \ket{1}\ket{\phi_1(\psi)}
      = \frac{1}{\sqrt{1-p}} (\ket{1}\bra{1}\otimes\I) 
      Q (\I \otimes \ket{0^m}) \ket{\psi},
  \end{equation}
  and therefore
  \begin{equation}
    \begin{split}
      Q (\I\otimes \ket{0^m}\bra{0^m}) Q^{\ast} \ket{1}\ket{\phi_1(\psi)}
      \hspace{-4.6cm}\\
      & = \frac{1}{\sqrt{1-p}}
      Q (\I\otimes \ket{0^m}\bra{0^m}) Q^{\ast} (\ket{1}\bra{1}\otimes\I) 
      Q (\I \otimes \ket{0^m}) \ket{\psi}\\
      & = \frac{1}{\sqrt{1-p}} Q(\I\otimes\ket{0^m}) P_1 \ket{\psi}\\
      & = \sqrt{1-p} \, Q \ket{\psi}\ket{0^m}\\
      & = \sqrt{p(1-p)} \ket{0}\ket{\phi_0(\psi)} +
      (1-p)\ket{1}\ket{\phi_1(\psi)},
    \end{split}
  \end{equation}
  which implies that the state of $(\reg{A},\reg{Z})$ immediately after the
  conditional statement becomes
  \begin{equation}
    \label{eq:state-after-rewinding}
    \begin{multlined}
      Q (\I\otimes U) Q^{\ast} \ket{1}\ket{\phi_1(\psi)}\\
      = 2 \sqrt{p(1-p)} \ket{0}\ket{\phi_0(\psi)}
      + (1-2p)\ket{1}\ket{\phi_1(\psi)}.
    \end{multlined}
  \end{equation}

  On subsequent iterations of the loop, the analysis is the same.
  The measurement of $\reg{A}$ now gives the outcome 0 with probability
  $4p(1-p)$, and conditioned on this outcome the register $\reg{Z}$ is left
  in state $\ket{\phi_0(\psi)}$ until the end of the procedure.
  If the measurement outcome is 1, the state of $(\reg{A},\reg{Z})$ reverts
  back to $\ket{1}\ket{\phi_1(\psi)}$, and the conditional statement is
  applied precisely as before.

  As a result of this analysis, one obtains the following closed-form
  expression for the output of the procedure:
  \begin{equation}
    R(\ket{\psi}\bra{\psi}) = 
    (1 - \delta) \ket{\phi_0(\psi)}\bra{\phi_0(\psi)}
    + \delta \ket{\phi_1(\psi)}\bra{\phi_1(\psi)},
  \end{equation}
  where
  \begin{equation}
    \delta = (1-p) (1 - 2p)^{2T}
    \quad\text{and}\quad
    T = \biggl\lceil\frac{\log(1/\varepsilon)}{p(1-p)}\biggr\rceil.
  \end{equation}
  It holds that $\delta < \varepsilon/2$, from which the bound
  \eqref{eq:rewinding-error-bound} follows.
\end{proof}

%

\subsection{Example: graph isomorphism}
\label{sec:GMW-graph-isomorphism}

We will now illustrate the use of the quantum rewinding lemma to prove that
the Goldreich--Micali--Wigderson proof system for graph isomorphism (described
at the beginning of the chapter in Figure~\ref{fig:GMW-graph-isomorphism}) has
the zero-knowledge property against cheating quantum verifiers.


Let $(G_0,G_1)$ be a pair of isomorphic graphs, and $P$ the prover described by
the Goldreich--Micali--Wigderson graph isomorphism proof system on the input
$(G_0,G_1)$.
Suppose $W = (W_1,W_2)$ is a (computationally efficient) cheating verifier that
implements some channel $\Phi$ by means of an interaction with $P$.
To prove that the proof system is zero-knowledge against quantum attacks, it
must be shown that every channel $\Phi$ that can be implemented in this way can
be efficiently implemented without making use of an interaction with $P$.

Note that, by the assumption that the operations performed by the cheating
verifier $W$ are efficiently implementable, it suffices to focus on the state
contained in the registers $(\reg{Z}_1,\reg{Y}_1)$ after the prover's second
message---for if the state obtained by the cheating verifier $W$ at this point
can be efficiently computed, then by simply applying the channel $W_2$ to this
state, one obtains $\Phi$.

There is also no loss of generality in assuming that the cheating verifier is
such that its first action is given by an isometry
\begin{equation}
  W_1\in\Unitary(\Z_0\otimes\Y_0,\Z_1\otimes\X_1).
\end{equation}
It is convenient to express this isometry as
\begin{equation}
  W_1 = A_0\otimes\ket{0} + A_1\otimes\ket{1}
\end{equation}
for operators $A_0,A_1\in\Lin(\Z_0\otimes\Y_0,\Z_1)$, which is possible because
$\X_1$ corresponds to a single qubit (which will be seen as a single bit
response by the honest, classical prover in the Goldreich--Micali--Wigderson
proof system).
One may verify that the channel implemented by such a cheating verifier
when interacting with the honest prover is given by
\begin{equation}
  \label{eq:graph-isomorphism-cheater-channel}
  \Phi(\rho)
  = \frac{1}{n!} \sum_{\substack{\pi\in S_n\\a \in\{0,1\}}}
  A_a\bigl(\rho \otimes \ket{\pi(G_0)}\bra{\pi(G_0)}\bigr) A_a^{\ast} \otimes 
  \ket{\pi \sigma^a}\bra{\pi \sigma^a}.
\end{equation}
It is not immediately clear that the channel
\eqref{eq:graph-isomorphism-cheater-channel} can be implemented efficiently
because the permutation $\sigma$ may not be efficiently computable from
$(G_0,G_1)$ (unless the graph isomorphism problem is in \textup{BQP}).
We will need to make use of the quantum rewinding lemma to prove that it can be
efficiently implemented.

We note first that it is possible to efficiently implement a unitary
computation $Q$ transforming $\ket{\psi}\ket{0^m}$, for a suitable choice of
$m$, to the state
\begin{equation}
  \label{eq:graph-isomorphism-guess}
  Q \ket{\psi}\ket{0^m}
  = \frac{1}{\sqrt{2}} \ket{0} \ket{\phi_0(\psi)}
  + \frac{1}{\sqrt{2}} \ket{1} \ket{\phi_1(\psi)}
\end{equation}
for
\begin{equation}
  \begin{split}
    \ket{\phi_0(\psi)} & = \frac{1}{\sqrt{n!}}
    \sum_{\tau\in S_n} \sum_{b\in\{0,1\}} 
    A_b (\ket{\psi}\ket{\tau(G_b)}) \ket{b} \ket{\tau},\\
    \ket{\phi_1(\psi)} & = \frac{1}{\sqrt{n!}}
    \sum_{\tau\in S_n} \sum_{b\in\{0,1\}}
    A_{1-b}(\ket{\psi}\ket{\tau(G_b)})\ket{b} \ket{\tau}.
  \end{split}
\end{equation}
For instance, one may create the state
\begin{equation}
  \label{eq:graph-isomorphism-2}
  \frac{1}{\sqrt{2n!}} \sum_{\tau\in S_n} \sum_{b\in\{0,1\}}
  \ket{\tau(G_b)}\ket{b}\ket{\tau},
\end{equation}
apply $W_1$ to the first register of this state together with $\ket{\psi}$,
and then finally apply a controlled-NOT gate (and a permutation of the ordering
of the qubits) to obtain \eqref{eq:graph-isomorphism-guess}.

Given that the above transformation $Q$ can be performed efficiently, one may
recover $\ket{\phi_0(\psi)}$ using the quantum rewinding lemma.
By post-processing this state (by measuring the last register with respect to
the standard basis and discarding the qubit corresponding to $b$), the
outcome $\Phi\bigl(\ket{\psi}\bra{\psi}\bigr)$ is obtained.
As this is so for every state $\ket{\psi}$, the channel $\Phi$ has been
implemented efficiently.
This implies that the prover $P$ is a statistical zero-knowledge prover
on inputs $(G_0,G_1)$ for which $G_0$ and $G_1$ are isomorphic.

It may be noted that there is a fairly direct correspondence between the
classical argument (which was briefly summarized at the beginning of the
chapter) for proving that the Goldreich--Micali--Wigderson graph isomorphism
proof system is zero-knowledge against classical cheating verifiers
and the quantum case just discussed.
The state \eqref{eq:graph-isomorphism-2} is effectively a purified form of a
guess for the messages exchanged between the prover and verifier, and the
first qubit of \eqref{eq:graph-isomorphism-guess} indicates whether this guess
was correct: if this qubit is in the 0 state, the guess was correct and the
correct output can be produced, while if this qubit is in the 1 state, then
rewinding takes place, giving the simulation another chance for success.

\subsection{Approximate quantum rewinding lemma}

For an arbitrary quantum circuit $Q$ acting on $k+m$ qubits, as described in
the previous section, one may always write
\begin{equation}
  Q\ket{\psi}\ket{0^m}
  = \sqrt{p(\psi)} \ket{0} \ket{\phi_0(\psi)} + 
  \sqrt{1-p(\psi)} \ket{1} \ket{\phi_1(\psi)}
\end{equation}
for any given pure state $\ket{\psi}$, for some real number $p(\psi) \in [0,1]$
and $(k+m-1)$-qubit unit vectors $\ket{\phi_0(\psi)}$ and $\ket{\phi_1(\psi)}$.
The exact quantum rewinding lemma will generally not be applicable to such a
circuit, as $p(\psi)$ will generally depend on $\psi$.
Although one cannot hope that an approximate version of the quantum
rewinding lemma should hold in cases where $p(\psi)$ varies significantly as
$\ket{\psi}$ ranges over all $n$-qubit unit vectors, it is reasonable to expect
that small variations in $p(\psi)$ can be tolerated.
The following lemma states that this is indeed possible, 
provided that one accepts a small loss in the procedure's accuracy.

\begin{lemma}[Approximate quantum rewinding lemma]
  Let $Q$ be a unitary quantum circuit acting on $k+m$ qubits, and write
  \begin{equation}
    Q\ket{\psi}\ket{0^m}
    = \sqrt{p(\psi)}\, \ket{0} \ket{\phi_0(\psi)} + 
    \sqrt{1-p(\psi)}\, \ket{1} \ket{\phi_1(\psi)}
  \end{equation}
  for each $k$-qubit unit vector $\ket{\psi}$.
  Assume that there exist real numbers $\delta,q\in(0,1)$ such that
  $\abs{p(\psi) - q} < \delta$ for every $k$-qubit unit vector $\ket{\psi}$.
  For every $\varepsilon \in (0,1/2)$, there exists a quantum circuit $R$
  taking $k$ input qubits, outputting $k+m-1$ qubits, and satisfying
  $\op{size}(R) = O(T \op{size}(Q))$ and
  \begin{equation}
    \label{eq:approximate-rewinding-error-bound}
    \bignorm{
      \bigket{\phi_0(\psi)}\bigbra{\phi_0(\psi)}
      - R\bigl(\ket{\psi}\bra{\psi}\bigr)}_1 < \varepsilon + 
    (4 T + 2) \sqrt{2\delta}
  \end{equation}
  for every $k$-qubit unit vector $\ket{\psi}$, for
  \begin{equation}
    T = \Biggl\lceil \frac{\log(1/\varepsilon)}{q (1-q)}\Biggr\rceil.
  \end{equation}
  Moreover, a description of the circuit $R$ can be generated from a
  description of $Q$ in time polynomial in $\op{size}(Q)$ and $T$.
\end{lemma}

The principle behind the proof of this approximate version of the quantum
rewinding lemma is simple: one runs precisely the same rewinding procedure as
in the exact case, imagining that the action of $Q$ is given by a
unitary operator $U$ satisfying
\begin{equation}
  U\ket{\psi}\ket{0^m}
  = \sqrt{q}\, \ket{0} \ket{\phi_0(\psi)} + 
  \sqrt{1-q}\, \ket{1} \ket{\phi_1(\psi)}
\end{equation}
for all $\ket{\psi}$.
The analysis from the proof of the exact quantum rewinding lemma reveals what
the output of the procedure would be if $Q$ acted in this idealized way.
Finally, it is proved that there must exist a unitary operator acting as above
and satisfying $\norm{Q - U} \leq \sqrt{2\delta}$.
The bound in the statement of the approximate quantum rewinding lemma is then
recovered from the fact that errors accumulate at most additively in
compositions of quantum circuits, which may be proved through iterative
applications of the triangle inequality.

\section{Quantum statistical zero knowledge} \label{sec:QSZK-properties}

A collection of interesting properties of the complexity class \class{QSZK} are
known to hold.
For instance, \class{QSZK} is closed under complementation, contained
in \class{QIP}(2), and equal to \class{HVQSZK}.
Moreover, a complete promise problem is known for \class{QSZK} that is
both simple and fundamental from the viewpoint of quantum computation.
Each of these facts has a classical analogue that is also known to hold, and
the techniques for proving the quantum variants of these facts borrow heavily
from known proofs of their classical analogues.

\subsection{The close quantum states problem}

Much of the discussion of the complexity class \class{QSZK} to follow in the
present subsection is centered on the following promise problem, 
called the \emph{$(\alpha,\beta)$-close quantum states problem}.
This is the problem mentioned above that will soon be proved complete for
\class{QSZK}.

\begin{center}
  \begin{minipage}{0.95\textwidth}
    \begin{center}
      \underline{$(\alpha,\beta)$-close quantum states
        ($(\alpha,\beta)$-CQS)}\\[2mm]
      \begin{tabular}{lp{0.75\textwidth}}
        \emph{Input:} & 
        Quantum circuits $Q_0$ and $Q_1$ taking no input and outputting
        quantum states $\rho_0$ and $\rho_1$ on the same number of
        qubits.\\[1mm]
        \emph{Yes:} &
        $\frac{1}{2}\norm{\rho_0 - \rho_1}_1 < \alpha$.\\[1mm]
        \emph{No:} &
        $\frac{1}{2}\norm{\rho_0 - \rho_1}_1 > \beta$.
      \end{tabular}
    \end{center}
  \end{minipage}
\end{center}

\noindent
Formally speaking, each choice of a pair $(\alpha,\beta)$ defines a different
promise problem $(\alpha,\beta)$-CQS.
One may consider the situation in which $\alpha$ and $\beta$ are
constants or functions; and if they are functions, it should be
understood that they are functions of the number $n$ representing the size of
the input description of the pair $(Q_0,Q_1)$.

It is not surprising that the complexity of the $(\alpha,\beta)$-CQS problem
may depend on the choice of $\alpha$ and $\beta$.
As it turns out, there is a wide range of choices for these values that lead to
polynomial-time equivalent problems.

\begin{theorem}[Sahai--Vadhan]
  \label{theorem:Sahai-Vadhan}
  Let $p$ and $q$ be polynomially bounded functions and let
  $\alpha$ and $\beta$ be polynomial-time computable functions
  satisfying 
  \begin{equation}
    0 < \alpha(n) \leq \beta(n)^2 - \frac{1}{q(n)}
  \end{equation}
  for all sufficiently large positive integers $n$.
  It holds that
  \begin{equation}
    \class{$(\alpha,\beta)$-CQS} \leq_m \class{$(\delta,1-\delta)$-CQS}
  \end{equation}
  for $\delta(n) = 2^{-p(n)}$.
\end{theorem}

\noindent
This theorem may be attributed to Sahai and Vadhan \cite{SahaiV03}, who proved
a classical analogue to this fact---but the extension to quantum states is
direct, requiring nothing more than a verification that the techniques and
bounds indeed extend to quantum states \cite{Watrous02}.

There are two main observations required to prove this theorem.
The first observation is that, for any choice of quantum states $\rho_0$ and
$\rho_1$ (on the same number of qubits) and a positive integer $m$, and for
states
\begin{equation}
  \begin{split}
    \xi_0 = \frac{1}{2^{m-1}}
    \sum_{\substack{a_1,\ldots,a_m\in\{0,1\}\\
        a_1+ \cdots + a_m \equiv 0 \;(\op{mod} 2)}}
    \rho_{a_1} \otimes \cdots \otimes \rho_{a_m},\\
    \xi_1 = \frac{1}{2^{m-1}}
    \sum_{\substack{a_1,\ldots,a_m\in\{0,1\}\\
        a_1+ \cdots + a_m \equiv 1 \;(\op{mod} 2)}}
    \rho_{a_1} \otimes \cdots \otimes \rho_{a_m},
  \end{split}
\end{equation}
one has
\begin{equation}
  \frac{1}{2} \norm{\xi_0 - \xi_1}_1 =
  \Biggl(\frac{1}{2}\norm{\rho_0 - \rho_1}_1\Biggr)^m.
\end{equation}
The second observation is that, again for any choice of quantum states $\rho_0$
and $\rho_1$ and a positive integer $m$, one has
\begin{equation}
  1 - \exp\biggl(-\frac{m \norm{\rho_0 - \rho_1}_1^2}{8}\biggr)
  \leq \frac{1}{2} \bignorm{\rho_0^{\otimes m} - \rho_1^{\otimes m}}_1
  \leq \frac{m}{2}\norm{\rho_0 - \rho_1}_1.
\end{equation}

These two facts suggest constructions that may be applied to two circuits $Q_0$
and $Q_1$ given as input to the close quantum states problem.
The first construction produces states (corresponding to $\xi_0$ and $\xi_1$)
that are generally much closer to one another than the states $\rho_0$ and
$\rho_1$ produced by $Q_0$ and $Q_1$---but the rate at which $\xi_0$ and
$\xi_1$ approach one another is much faster (as a function of $m$) in the case
that $\rho_0$ and $\rho_1$ are close.
The second construction simply repeats a state many times, causing the
resulting states to move away from one another, this time more rapidly when the
original states are farther apart.
By alternating the two constructions in a suitable fashion
(first the first construction, then the second, then the first again), the
input circuits are transformed so as to produce output states that are either
very close or very far, depending on the initial closeness of the states.
The requirement $\alpha(n) \leq \beta(n)^2 - 1/ q(n)$ is needed to guarantee
that the correct behavior results from these constructions.

\subsection{HVQSZK-completeness of the close quantum states problem}

There are two main steps required to prove that the close quantum states
problem is complete for \class{QSZK}, for any choice of $\alpha$ and $\beta$
for which Theorem~\ref{theorem:Sahai-Vadhan} holds.
The first step is to prove that this problem is complete for \class{HVQSZK},
and the second is to prove that the problem itself is contained in
\class{QSZK}.
Because $\class{QSZK}$ is closed under Karp reductions and
$\class{QSZK} \subseteq \class{HVQSZK}$, the \class{QSZK}-completeness of the
problem follows---and in the process one finds that
$\class{QSZK} = \class{HVQSZK}$.
The fact that the close quantum states problem is contained in \class{QSZK}
will make use of the quantum rewinding lemma, and is discussed in the
subsection following this one.
The present section concerns the fact that this problem is complete for
\class{HVQSZK}.

The \class{HVQSZK}-completeness of the close quantum states problem can be
proved in multiple ways.
Here we will sketch a proof that is somewhat different from the original
proof \cite{Watrous02}, although the essential ideas are similar.

Suppose that $A$ is any promise problem for which $A\in\class{HVQSZK}$.
There must therefore exist a prover $P(x)$ and verifier $V(x)$ satisfying the
conditions of Definition~\ref{definition:HVQSZK}, for $a = 2/3$, $b=1/3$, and
$m$ being a polynomially bounded function.
By considering the three transformations associated with the perfect
completeness, parallelization, and error reduction by parallel repetition
constructions discussed in the previous chapter, one may transform both the
verifier $V$ and the prover $P$ into a new, three-message prover/verifier pair
$(P',V')$ for $A$ having perfect completeness and exponentially small soundness
error.
Moreover, it may be assumed that the new verifier $V'$ is an isometric
verifier, and that the honest-verifier zero-knowledge condition is still in
place for the pair $(P',V')$.
The fact that these assumptions are justified follows from a consideration
of each of the three constructions:
in each case, the most natural and direct way of defining the actions of the 
honest prover and the corresponding view of the verifier maintains the
honest-verifier zero-knowledge condition.
\begin{figure}
  \begin{center}
    \begin{tikzpicture}[scale=0.45, 
        turn/.style={draw, minimum height=14mm, minimum width=10mm,
          fill = ChannelColor, text=ChannelTextColor},
        emptyturn/.style={minimum height=14mm, minimum width=10mm},
        >=latex]

      \node (V0) at (-14,4) [emptyturn] {};      
      \node (V1) at (-8,4) [turn] {$V_1$};
      \node (V2) at (2,4) [turn] {$V_2$};
      \node (M) at (7,4.4) [emptyturn] {};
      \node (P0) at (-13,0) [turn] {$P_0$};
      \node (P1) at (-3,0) [turn] {$P_1$};
      \node (P2) at (7,0) [emptyturn] {};
      
      \draw[->] ([yshift=4mm]V1.east) -- ([yshift=4mm]V2.west)
      node [above, pos=0.74] {$\reg{Z}_1$};
      
      \draw[->] ([yshift=4mm]V0.east) -- ([yshift=4mm]V1.west)
      node [above, pos=0.46] {$\reg{Z}_0$};
      
      \draw[->] ([yshift=-4mm]V1.east) .. controls +(right:20mm) and 
      +(left:20mm) .. ([yshift=4mm]P1.west) node [right, pos=0.4] {$\reg{X}_1$};
      
      \draw[->] ([yshift=-4mm]V2.east) .. controls +(right:20mm) and 
      +(left:20mm) .. ([yshift=4mm]P2.west) node [right, pos=0.4] {$\reg{X}_2$};
      
      \draw[->] ([yshift=4mm]P0.east) .. controls +(right:20mm) and 
      +(left:20mm) .. ([yshift=-4mm]V1.west) node [left, pos=0.6] 
      {$\reg{Y}_0$};
      
      \draw[->] ([yshift=4mm]P1.east) .. controls +(right:20mm) and 
      +(left:20mm) .. ([yshift=-4mm]V2.west) node [left, pos=0.6]
      {$\reg{Y}_1$};
      
      \draw[->] ([yshift=-4mm]P0.east) -- ([yshift=-4mm]P1.west)
      node [below, midway] {$\reg{W}_1$};
      
      \draw[->] ([yshift=-4mm]P1.east) -- ([yshift=-4mm]P2.west)
      node [below, midway] {$\reg{W}_2$};
      
      \draw[->] ([yshift=4mm]V2.east) -- (M.west) node [above, midway]
           {$\reg{Z}_2$};
      
      \node[draw, thick, ellipse, minimum width=6.5mm,
        minimum height=20mm] at (-1.18,3.8) {};
      
      \node[draw, thick, ellipse, minimum width=6.5mm,
        minimum height=20mm] at (-11.18,3.8) {};
      
    \end{tikzpicture}
  \end{center}
  \caption{A unitary verifier's view in a three-message interactive game.}
  \label{fig:unitary-3-message-HVQSZK-verifier}
\end{figure}
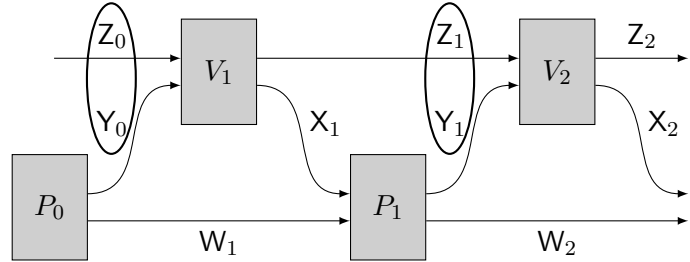

Given a pair $(P',V')$ as just described, one may output an instance of the
close quantum states problem for each input 
$x\in A_{\textup{yes}}\cup A_{\textup{no}}$ in the following way.
First, one may compute circuits $R_0$ and $R_1$ that take no input qubits and
output states of the registers $(\reg{Z}_0,\reg{Y}_0)$ and
$(\reg{Z}_1,\reg{Y}_1)$, respectively, with the property that these states are
close to the true states of these registers when $V'$ interacts with $P'$
(as illustrated in Figure~\ref{fig:unitary-3-message-HVQSZK-verifier}),
assuming the input is a yes-input for $A$.
This is possible by the honest-verifier statistical zero-knowledge property of
$(P',V')$.
The circuits $Q_0$ and $Q_1$, representing the instance of the close quantum
states problem being produced, are derived from $R_0$ and $R_1$ in the
following way:
\begin{enumerate}
\item[1.]
  $Q_0$ is the circuit obtained by first running $R_0$ to obtain a state of
  $(\reg{Z}_0,\reg{Y}_0)$, replacing each qubit of $\reg{Z}_0$ by an
  initialized qubit in the $\ket{0}$ state, applying $V_1$ to the pair
  $(\reg{Z}_0,\reg{Y}_0)$ to obtain $(\reg{Z}_1,\reg{X}_1)$, and then
  outputting just the register $\reg{Z}_1$.
\item[2.]
  $Q_1$ is the circuit obtained by first running $R_1$ to obtain a state of
  $(\reg{Z}_1,\reg{Y}_1)$, applying $V_2$ to obtain a state of
  $(\reg{Z}_2,\reg{X}_2)$, replacing the single-qubit register $\reg{Z}_2$ with
  a qubit in the $\ket{1}$ state, applying $V_2^{-1}$ to
  $(\reg{Z}_2,\reg{X}_2)$ to obtain $(\reg{Z}_1,\reg{Y}_1)$, and then 
  outputting just the register $\reg{Z}_1$.
\end{enumerate}
If $P'$ causes $V'$ to accept with probability 1, then the states output by
$Q_0$ and $Q_1$ will necessarily be close, by virtue of the honest-verifier
statistical zero-knowledge property: on yes-instances of the problem $A$, the
circuits $R_0$ and $R_1$ output close approximations of the true states of the
registers $(\reg{Z}_0,\reg{Y}_0)$ and $(\reg{Z}_1,\reg{Y}_1)$, which leads to
$Q_0$ and $Q_1$ outputting nearly identical states.
On the other hand, in the case that $\omega(V')$ is small, the states output by
$Q_0$ and $Q_1$ must necessarily be far apart, irrespective of the
zero-knowledge properties of $(P',V')$, as follows from a similar argument to
the one described in the previous chapter concerning the close circuit images
problem.

\subsection{Close quantum states in QSZK}

We will now sketch an argument demonstrating that the close quantum states
problem is contained in \class{QSZK}.
By Theorem~\ref{theorem:Sahai-Vadhan}, it suffices to consider the problem
$(\delta,1-\delta)$-CQS for $\delta(n)$ being exponentially small.

For a given input $(Q_0,Q_1)$ to the $(\delta,1-\delta)$-CQS problem, consider
the interactive game described in Figure~\ref{fig:CQS-public-coin}.
\begin{figure}[t]
\noindent\hrulefill
\begin{trivlist}
\item The input is a pair $(Q_0,Q_1)$ of quantum circuits that take no input
  and output states $\rho_0$ and $\rho_1$ of a $k$-qubit register $\reg{X}$.
  By purifying $Q_0$ and $Q_1$, and padding one of the circuits with
  extra qubits if necessary, unitary quantum circuits $R_0$ and $R_1$
  acting on registers $(\reg{X},\reg{Y})$, for $\reg{Y}$ being an $m$-qubit
  register, are obtained.
  Initializing $(\reg{X},\reg{Y})$ to $\ket{0^k}\ket{0^m}$, applying
  $R_a$, and tracing out $\reg{Y}$ leaves $\reg{X}$ in the state $\rho_a$.
  \vspace{1mm}
\item {\bf Prover's step 1:}
  Initialize $(\reg{X},\reg{Y})$ to $\ket{0^k}\ket{0^m}$, apply $R_0$ to the
  pair $(\reg{X},\reg{Y})$, and send $\reg{X}$ to the verifier.
  \vspace{1mm}
\item {\bf Verifier's step 1:}
  Choose $a\in\{0,1\}$ uniformly at random and send $a$ to the prover.
  \vspace{1mm}
\item {\bf Prover's step 2:}
  If $a = 0$, send $\reg{Y}$ to the verifier.
  If $a = 1$, apply $U$ to $\reg{Y}$ and then send $\reg{Y}$ to the verifier,
  where $U$ is an $m$-qubit unitary satisfying
  \begin{equation}
    \bignorm{
    (\I \otimes U) R_0 \ket{0^{k+m}} - R_1 \ket{0^{k+m}}}_1 \leq \varepsilon.
  \end{equation}
\item {\bf Verifier's step 2:}
  Apply $R_a^{-1}$ to $(\reg{X},\reg{Y})$ and measure all of the qubits in
  these registers with respect to the standard basis.
  Accept if all measurement results are 0, reject otherwise.
\end{trivlist}
\noindent\hrulefill
\caption{Public-coin proof system for close quantum states.}
\label{fig:CQS-public-coin}
\end{figure}
It is evident that this game is both complete and sound as an interactive game.
In particular, if $Q_0$ and $Q_1$ output close quantum states, the verifier
will accept with probability exponentially close to 1 when interacting with the
honest prover described in the figure.
On the other hand, if $Q_0$ and $Q_1$ output quantum states that are almost
perfectly distinguishable, there cannot exist a state of $\reg{X}$ that
is consistent with both of the pure states output by $R_0$ and $R_1$ when run
on the all-zero input state---and an analysis reveals that the verifier
rejects with probability exponentially close to 1/2 in this case.

It remains to consider the zero-knowledge property, which may be established
using the (approximate) quantum rewinding lemma, through a similar methodology
to the proof that the Goldreich--Micali--Wigderson graph isomorphism 
proof system is zero-knowledge against quantum attacks, as described in
Section~\ref{sec:GMW-graph-isomorphism}.
As in that proof, it suffices to consider a cheating verifier whose first
action is given by an isometry
$W_1\in\Unitary(\Z_0\otimes\Y_0,\Z_1\otimes\X_1)$, expressed as
\begin{equation}
  W_1 = A_0 \otimes \ket{0} + A_1 \otimes \ket{1}
\end{equation}
for $A_0,A_1\in\Lin(\Z_0\otimes\Y_0,\Z_1)$.
In the present case, the channel implemented by such a cheating verifier
(disregarding the application of $W_2$) through an interaction with the honest
prover $P$ is given by
\begin{equation}
  \Phi(\rho) = \sum_{a\in\{0,1\}}
  \bigl(A_a \otimes \I_{\Y_1}\bigr)
  \bigl(\rho\otimes\ket{\gamma_a}\bra{\gamma_a}\bigr)
  \bigl(A_a \otimes \I_{\Y_1}\bigr)^{\ast}
\end{equation}
for
\begin{equation}
  \ket{\gamma_0} = R_0\ket{0^{k+m}} \quad\text{and}\quad
  \ket{\gamma_1} = (\I\otimes U)R_0\ket{0^{k+m}}.
\end{equation}
This time, it is the fact that $U$ may not be efficiently implementable
that represents the main obstacle to efficiently implementing $\Phi$ in the
most obvious way.

It is possible, however, to efficiently implement a unitary quantum circuit
$Q$ that operates as follows:
for a given pure state $\ket{\psi}$ and a suitably chosen value $r$, the
circuit $Q$ transforms $\ket{\psi}\ket{0^r}$ as follows:
\begin{equation}
  Q \ket{\psi} \ket{0^r} = 
  \frac{1}{\sqrt{2}} \ket{0} \ket{\phi_0(\psi)}
  + \frac{1}{\sqrt{2}} \ket{1} \ket{\phi_1(\psi)}
\end{equation}
for
\begin{equation}
  \begin{split}
    \ket{\phi_0(\psi)}
    & = \sum_{b\in\{0,1\}} \bigl( A_b \otimes \I) \ket{\psi} R_b 
    \ket{0^{k+m}} \ket{b},\\
    \ket{\phi_1(\psi)}
    & = \sum_{b\in\{0,1\}} \bigl( A_{1-b} \otimes \I) \ket{\psi} R_b 
    \ket{0^{k+m}}\ket{b}.
  \end{split}
\end{equation}
The vectors $\ket{\phi_0(\psi)}$ and $\ket{\phi_1(\psi)}$ are not necessarily
unit vectors in this case, but they are exponentially close to unit vectors
under the assumption that $Q_0$ and $Q_1$ output states that are exponentially
close to one another.
Applying the approximate quantum rewinding lemma allows for the recovery of
$\ket{\phi_0(\psi)}$ with high probability, which may then be processed to
implement a close approximation to $\Phi$.

\begin{theorem}
  Let $q$ be a polynomially bounded function and let
  $\alpha$ and $\beta$ be polynomial-time computable functions
  satisfying 
  \begin{equation}
    0 < \alpha(n) \leq \beta(n)^2 - \frac{1}{q(n)}
  \end{equation}
  for all sufficiently large positive integers $n$.
  It holds that $\class{$(\alpha,\beta)$-CQS} \in \class{QSZK}$.
\end{theorem}

\begin{cor}
  \class{QSZK} = \class{HVQSZK}.
\end{cor}

\subsection{Closure of QSZK under complementation}

Finally, we may observe that the class \class{QSZK} is closed under
complementation, meaning that a promise problem 
$A = (A_{\textup{yes}}, A_{\textup{no}})$ is contained in \class{QSZK} if and
only if the same is true of the promise problem
$\overline{A} = (A_{\textup{no}}, A_{\textup{yes}})$.
By virtue of the fact that $\class{QSZK} = \class{HVQSZK}$ and the
close quantum states problem is complete for \class{QSZK}, it suffices to prove
that the complement of $(\delta,1-\delta)$-CQS is contained in
$\class{HVQSZK}$ for $\delta(n) = 2^{-n}$.

To prove that this is so, one may consider the very simple interactive game
described in Figure~\ref{fig:CQS-complement}.
\begin{figure}
\noindent\hrulefill
\begin{trivlist}
\item The input is a pair $(Q_0,Q_1)$ of quantum circuits that take no input
  and output states $\rho_0$ and $\rho_1$ of an $n$-qubit register $\reg{X}$.
\item {\bf Verifier's step 1:}
  Choose $a\in\{0,1\}$ uniformly at random, prepare the state
  $\rho_a$ produced by $Q_a$ in a register $\reg{X}$, and send $\reg{X}$ to the
  prover.
  \vspace{1mm}
\item {\bf Prover's step 1:}
  Perform an optimal measurement to distinguish the states
  $\rho_0$ and $\rho_1$.
  If the measurement indicates that the state is $\rho_0$, send $b = 0$ to the
  verifier, and otherwise send $b = 1$ to the verifier.
\item {\bf Verifier's step 2:}
  Accept if $a = b$, reject otherwise.
\end{trivlist}
\noindent\hrulefill
\caption{A proof system for the complement of the close quantum states
  problem.}
\label{fig:CQS-complement}
\end{figure}
The value of the verifier $V$ for the input $(Q_0,Q_1)$ described in the
figure is given by the expression
\begin{equation}
  \omega(V) = \frac{1}{2} + \frac{1}{4}\bignorm{\rho_0 - \rho_1}_1,
\end{equation}
and the prover $P$ described in the figure achieves this optimal winning
probability.

It remains to prove that, for a purified form of the verifier described in the
figure, one has that the prover and verifier pair $(P,V)$ possesses the
honest-verifier zero-knowledge property on yes-inputs, which are those for
which $\frac{1}{2}\norm{\rho_0 - \rho_1}_1$ is exponentially close to 1
(i.e., $\rho_0$ and $\rho_1$ are almost perfectly distinguishable).
The view of $V$ in this case is simply the state of the pair 
$(\reg{Z}_1,\reg{Y}_1)$ immediately after the prover sends its only message in
the proof system.
It is nearly trivial to approximate this state efficiently---one may simply
perform the computation represented by the verifier's first action $V_1$,
then substitute the correct identification of $b$ for the prover's message.
Specifically, this may be done by discarding the register $\reg{X}_1$, 
measuring the qubit corresponding to the verifier's random choice of
$a\in\{0,1\}$, and then setting $\reg{Y}_1$ to contain this classical value.
Because the honest prover correctly determines the value of $a$ with a
negligible error probability, this nearly trivial approximation to the
verifier's view will deviate from the verifier's actual view on yes-inputs
by a negligible quantity.

\section{Chapter notes}

The notion of zero-knowledge was proposed by Goldwasser, Micali, and Rackoff
\cite{GoldwasserMR85,GoldwasserMR89}, and has been investigated by many
researchers in theoretical computer science and cryptography since then.
The survey of Goldreich \cite{Goldreich02} may be consulted by readers
interested in learning more about this topic of study in the classical
setting.
The topic of quantum zero-knowledge, as well as the problematic issue of
classical techniques for proving interactive proof systems to possess the
zero-knowledge property not carrying over to the quantum setting, was raised
by van~de~Graaf in his PhD thesis \cite{Graaf97}.

Honest-verifier quantum statistical zero-knowledge was defined and studied in
\cite{Watrous02}, wherein it was proved that \class{HVQSZK} is closed under
complementation, contained in $\class{QIP}(2)$, and has the complete promise
problems mentioned in Section~\ref{sec:QSZK-properties}.
As suggested in the main text, these facts have classical analogues---the paper
of Sahai and Vadhan \cite{SahaiV03} proves several facts along these lines.
Earlier papers, including ones of Fortnow \cite{Fortnow89},
Aiello and H{\aa}stad \cite{AielloH91}, and
Okamoto \cite{Okamoto00}, proved related results on the complexity-theoretic
aspects of statistical zero-knowledge proof systems.
Kobayashi \cite{Kobayashi03} proved several results of a similar nature
for a non-interactive variant of quantum statistical zero-knowledge.

The quantum rewinding lemma was proved in \cite{Watrous09}, along with its
application to the security of the Goldreich--Micali--Wigderson graph
isomorphism proof system against quantum attacks, the equality of
\class{QSZK} and \class{HVQSZK}, and to the security of a (computational)
zero-knowledge proof system for graph 3-coloring (also due to Goldreich,
Micali, and Wigderson).
Hallgren, Kolla, Sen, and Zhang \cite{HallgrenKSZ08} extended these results to
prove that a wide range of classical statistical zero-knowledge proof systems
remain zero-knowledge against quantum attacks, and Kobayashi \cite{Kobayashi08}
proved several results of a similar nature concerning quantum computational
zero-knowledge.
The quantum rewinding lemma has been applied in a couple of other settings
relating to quantum cryptography---see, for instance,
Damg{\aa}rd and Lunemann \cite{DamgaardL09} and
Hallgren, Smith, and Song \cite{HallgrenSS11}.
Unruh \cite{Unruh12} has considered quantum proofs of knowledge, making use of
the quantum rewinding lemma and a different rewinding technique
to prove interesting results concerning this notion.
Ambainis, Rosmanis, and Unruh \cite{AmbainisRU14} have proved limitations on
the applicability of these rewinding techniques, which indeed appear to be
rather limited when compared to their classical counterparts.

In addition to the close quantum states problem (and its complement), there are
a few other promise problems known to be complete for the class \class{QSZK}.
Ben-Aroya, Schwartz, and Ta-Shma \cite{Ben-AroyaST10} proved that a promise
problem based on deciding which of two quantum circuits produces a state with
greater von Neumann entropy is complete for \class{QSZK}, and
Gutoski, Hayden, Milner, and Wilde \cite{GutoskiHMW15} proved the
\class{QSZK}-completeness of a promise problem relating to separability
testing.
Gutoski, Hayden, Milner, and Wilde also prove the completeness of other
problems relating to separability testing for other complexity classes based on
quantum proofs.

\chapter{Multi-Prover Quantum Interactive Proofs}
\label{chapter:multiple-provers}

This chapter considers \emph{multi-prover} interactive proof systems, an
extension of the model of single-prover interactive proofs in which the
verifier interacts simultaneously with two or more provers. Each prover is
modeled as a separate participant in the game, and together the provers attempt
to maximize the verifier's probability of eventually outputting~1, meaning that
it accepts the interaction. 

What makes these games particularly interesting is that the provers, which will
always be assumed to cooperate in this chapter, are not permitted to exchange
messages with each other---the only communication is between the verifier and
each individual prover.
This restriction empowers the verifier, allowing for games in which the
prover's answers are checked against each other, akin to a detective attempting
to confound a suspected team of robbers by submitting them to isolated
interrogations and cross-checking their answers against one other. 

The first scenario that will be considered is the case of a quantum verifier
interacting with quantum provers that are restricted to applying local
transformations and do not have any further means of coordinating their
actions. 
After introducing the required definitions in Section~\ref{sec:qmip-nexp}, it
will be shown that the class $\QMIP$ of problems that can be decided in this
model is precisely equal to its classical counterpart $\MIP$, in which all
parties are classical: $\QMIP = \MIP$.

While the model evidently collapses to the single-prover model as soon as
communication between the provers is allowed, it is interesting to consider
provers having access to sources of correlations that do not require
communication.
For the case of classical provers, shared randomness may be considered, but it
does not affect the computational power of the model---any shared random string
used by the provers can always be replaced by a deterministic setting of the
shared string that maximizes the probability with which the verifier accepts.
For the case of quantum provers, quantum physics suggests that it may be
beneficial to the provers to share a quantum state that is entangled across the
registers associated with different provers.
The study of Bell inequalities demonstrates that by performing local
measurements on a shared entangled state (such as an EPR pair) the provers are
able to generate correlations that, although they do not imply communication,
are stronger than the correlations that can be generated by shared randomness
alone. 
Thus, the use of entanglement may enhance the provers' ability to coordinate
their answers, leading to a class $\QMIP^*$ of problems having entangled-prover
multi-prover interactive proof systems that is \emph{a priori} distinct from
$\QMIP$.

Most of this chapter is concerned with results on the class $\QMIP^*$. 
We begin by considering the effect that the use of entanglement can have on the
soundness of multi-prover interactive proof systems. In
Section~\ref{sec:oracularization} it will be shown that the important classical
technique of \emph{oracularization} fails in the presence of entanglement, and
in Section~\ref{sec:xor} we will see that entanglement leads to the collapse of
a certain restricted class of proof systems, namely two-prover XOR interactive
proofs.

Section~\ref{sec:qmip-structure} discusses structural results on $\QMIP^*$.
Many of these results have the interesting peculiarity that they are only known
to be achievable when honest provers make use of shared entanglement. 
These results include the parallelization of arbitrary multi-prover interactive 
proofs to ones having a single round of interaction, the transformation of
multi-prover interactive proofs into ones possessing the property of perfect
completeness, and the simulation of arbitrary multi-prover quantum interactive
proofs by ones in which the verifier is classical, leading to the equality
$\QMIP^*=\MIP^*$.

Section~\ref{sec:nexp-in-mipstar} is devoted to a proof that
$\NEXP\subseteq\QMIP^*$. 
This shows that, in spite of the possible use of entanglement by the provers,
the verifier in a quantum multi-prover interactive proof system has no less
verification power than that of classical multi-prover interactive proofs,
which is characterized as $\MIP=\NEXP$. 
The analysis will introduce a three-prover variant of oracularization and
discuss its relation to a phenomenon known as the 
\emph{monogamy of entanglement}. 

In the concluding Section~\ref{sec:qmip-further}, we discuss two important
topics in the study of quantum multi-prover interactive proof systems that
remain largely unsettled. 
The first topic is the question of parallel error amplification, and the second
is the problem of placing upper bounds on the class $\QMIP^*$.

\section{Definitions of multi-prover interactive proof systems}
\label{sec:def-multi}

As was suggested in Chapter~\ref{chapter:single-prover}, the interactive game
model through which single-prover quantum interactive proof systems were defined
may be extended in a straightforward way that allows a verifier to interact
with multiple provers.
A multi-prover interactive game is completely determined by the description of
the verifier, and we will always assume that messages are sent synchronously 
in \emph{turns}, consisting either of a set of messages from the verifier to 
each of the provers, or a set of messages from the provers to the verifier. 
A \emph{round} is made of two turns, the first consisting of messages from the
verifier to the provers and the second consisting of messages from the provers
back to the verifier. 
Except when stated otherwise, for notational convenience we will usually assume
that all interactive games have an integral number of rounds (and in particular
the first turn consists of a set of messages from the verifier to the provers).

We will use the same labeling convention for the registers corresponding to
different messages and private memories as in the single-prover case, 
introducing superscripts to distinguish registers associated with distinct
provers. 
For instance, $\rX_1^2$ denotes the register containing the first message sent
by the verifier to the second prover and $\reg{W}_0^1$ the register representing
the first prover's private memory at the start of the game. 
Figure~\ref{fig:multi-game-1} provides an illustration of a four-turn (or
two-round) interactive game between a verifier and two provers.
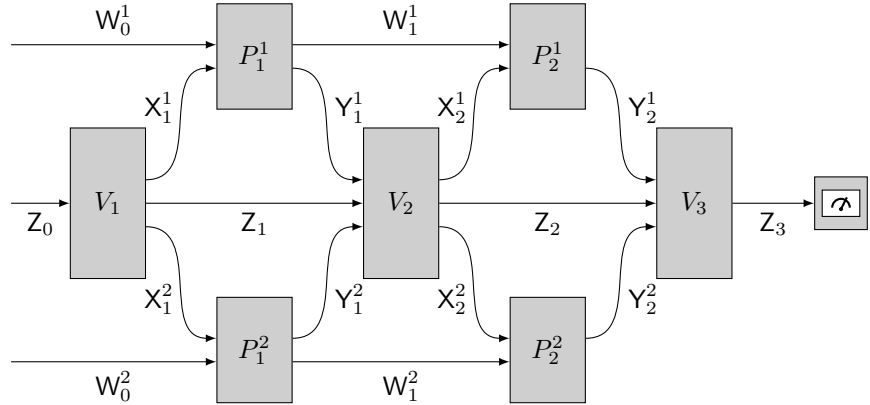
\begin{figure}
  \begin{center}
    \small
    \begin{tikzpicture}[scale=0.39, 
        turn/.style={draw, minimum height=14mm, minimum width=10mm,
          fill = ChannelColor, text=ChannelTextColor},
        invisibleturn/.style={minimum height=14mm, minimum width=10mm},
        bigturn/.style={draw, minimum height=20mm, minimum width=10mm,
          fill = ChannelColor, text=ChannelTextColor},
        measure/.style={draw, minimum width=7mm, minimum height=7mm,
          fill = ChannelColor},
        >=latex]
      
      \node (Vleft) at (-15,5) [invisibleturn] {};
      \node (V1) at (-8,5) [turn] {$P_1^1$};
      \node (V2) at (2,5) [turn] {$P_2^1$};      
     
      \node (Xleft) at (-15,-5) [invisibleturn] {};
      \node (X1) at (-8,-5) [turn] {$P_1^2$};
      \node (X2) at (2,-5) [turn] {$P_2^2$};
      
      \node (N) at (12,0) [measure] {};
      
      \node[draw, minimum width=5mm, minimum height=3.5mm, fill=ReadoutColor]
      (readout) at (N) {};
      
      \draw[thick] ($(N)+(0.3,-0.15)$) arc (0:180:3mm);
      \draw[thick] ($(N)+(0.2,0.2)$) -- ($(N)+(0,-0.2)$);
      \draw[fill] ($(N)+(0,-0.2)$) circle (0.5mm);
      
      \node (Pleft) at (-15,0) [invisibleturn] {};

      \node (P0) at (-13,0) [bigturn] {$V_1$};
      \node (P1) at (-3,0) [bigturn] {$V_2$};
      \node (P2) at (7,0) [bigturn] {$V_3$};
      
      \draw[->] ([yshift=4mm]V1.east) -- ([yshift=4mm]V2.west)
      node [above, midway] {$\reg{W}_1^1$};
			
      \draw[->] ([yshift=4mm]Vleft.west) -- ([yshift=4mm]V1.west)
      node [above, midway] {$\reg{W}_0^1$};
      
      \draw[->] ([yshift=-4mm]V1.east) .. controls +(right:20mm) and 
      +(left:20mm) .. ([yshift=8mm]P1.west) node [right, pos=0.4]
      {$\reg{Y}_1^1$};
      
      \draw[->] ([yshift=-4mm]V2.east) .. controls +(right:20mm) and 
      +(left:20mm) .. ([yshift=8mm]P2.west) node [right, pos=0.4] 
      {$\reg{Y}_2^1$};
      
      \draw[->] ([yshift=8mm]P0.east) .. controls +(right:20mm) and 
      +(left:20mm) .. ([yshift=-4mm]V1.west) node [left, pos=0.6]
      {$\reg{X}_1^1$};
      
      \draw[->] ([yshift=8mm]P1.east) .. controls +(right:20mm) and 
      +(left:20mm) .. ([yshift=-4mm]V2.west) node [left, pos=0.6]
      {$\reg{X}_2^1$};
      
      \draw[->] (Pleft.west) -- (P0.west) node [below, midway] {$\reg{Z}_0$};

      \draw[->] (P0.east) -- (P1.west) node [below, midway] {$\reg{Z}_1$};
      
      \draw[->] (P1.east) -- (P2.west) node [below, midway] {$\reg{Z}_2$};
      
      \draw[->] (P2.east) -- (N.west)
      node [below, midway] {$\reg{Z}_3$};
      
      \draw[->] ([yshift=-4mm]X1.east) -- ([yshift=-4mm]X2.west)
      node [below, midway] {$\reg{W}_1^2$};
      
      \draw[->] ([yshift=-4mm]Xleft.west) -- ([yshift=-4mm]X1.west)
      node [below, midway] {$\reg{W}_0^2$};
      
      \draw[->] ([yshift=4mm]X1.east) .. controls +(right:20mm) and 
      +(left:20mm) .. ([yshift=-8mm]P1.west) node [right, pos=0.4]
      {$\reg{Y}_1^2$};
      
      \draw[->] ([yshift=4mm]X2.east) .. controls +(right:20mm) and 
      +(left:20mm) .. ([yshift=-8mm]P2.west) node [right, pos=0.4] 
      {$\reg{Y}_2^2$};
      
      \draw[->] ([yshift=-8mm]P0.east) .. controls +(right:20mm) and 
      +(left:20mm) .. ([yshift=4mm]X1.west) node [left, pos=0.6]
      {$\reg{X}_1^2$};
      
      \draw[->] ([yshift=-8mm]P1.east) .. controls +(right:20mm) and 
      +(left:20mm) .. ([yshift=4mm]X2.west) node [left, pos=0.6]
      {$\reg{X}_2^2$};
      
    \end{tikzpicture}
  \end{center}
  \caption{A four-turn interactive game between a verifier and two
    provers. The provers do not exchange any messages between themselves. Their
    initial private registers $(\rW_0^1,\rW_0^2)$ may be initialized in an
    arbitrary joint entangled state.}
  \label{fig:multi-game-1}
\end{figure}
As in the single-prover case, one may assume without loss of generality that
all of the verifier and prover actions in a multi-prover interactive game are
represented by isometric channels acting on pure states, or unitary channels
provided that sufficiently large ancillary spaces are made available for each
participant at the start of the game.

The initial state of the $k$ provers' private registers 
$\reg{W}_0^1,\ldots,\reg{W}_0^k$ will play a particularly important role in
multi-prover interactive games.
While it is always possible in the case of single-prover games to assume,
without any loss of generality, that the prover's starting register $\reg{W}_0$
is initialized to the all-zero standard basis state (or to ignore the existence
of this register altogether), this is no longer the case for multiple provers.
An alternative preparation of a single prover's starting register could always
be incorporated into this prover's first action, but multiple provers might
benefit from shared starting states (especially \emph{entangled states}) that
cannot be prepared locally.
Provers whose private registers are initialized to a product state, which could
be prepared locally and independently by each prover, will be referred to as
\emph{unentangled} provers. 
General provers, on the other hand, are permitted to start the game with the
collection of registers $(\reg{W}_0^1,\ldots,\reg{W}_0^k)$ initialized in an
arbitrary quantum state.
Such provers will typically be called \emph{entangled} provers, and the shared
starting state will be
referred to as their \emph{prior shared entanglement}. 

The following example demonstrates that the set of \emph{entangled strategies},
or strategies that can be implemented by entangled provers having access to
prior shared entanglement, is strictly larger than the set of
\emph{unentangled strategies} associated with provers restricted to
initial product states. 
  
\begin{example}[Coherent state exchange game]\label{ex:dltw-1}
  Consider the following one-round two-prover verifier $V=(V_1,V_2)$. 
  Following a well-established convention the two provers will be given the
  names \emph{Alice} and \emph{Bob}.
  In the game defined by this verifier, the registers $\reg{Z}_1$,
  $\reg{Y}_1^1$, and $\reg{Y}_1^2$ are qubit registers while $\reg{X}_1^1$ and
  $\reg{X}_1^2$ are qutrit registers (having standard basis states
  $\ket{0}$, $\ket{1}$, and $\ket{2}$).
  \begin{enumerate}
  \item 
    The verifier prepares the registers $(\reg{Z}_1,\reg{X}_1^1,\reg{X}_1^2)$
    in the pure state
    \begin{equation}
      \frac{1}{\sqrt{2}}\ket{0}\ket{00} + \frac{1}{\sqrt{2}}\ket{1} \ket{\phi},
    \end{equation}
    for
    \begin{equation}
      \ket{\phi} =
      \frac{1}{\sqrt{2}}\ket{11} + \frac{1}{\sqrt{2}}\ket{22}.
    \end{equation}
    It sends $\rX_1^1$ to Alice and $\rX_1^2$ to Bob.
  \item 
    Alice and Bob respond with the registers $\rY_1^1$ and $\rY_1^2$,
    respectively.
  \item 
    The verifier measures the registers $(\rZ_1,\rY_1^1,\rY_1^2)$ using a binary
    projective measurement $\{\Pi_1,\I-\Pi_1\}$, with the outcome $1$ being 
    associated with the projector $\Pi_1=\ket{\gamma}\bra{\gamma}$, for
    \begin{equation}
      \ket{\gamma} = \frac{1}{\sqrt{2}}\ket{0}\ket{00} 
      + \frac{1}{\sqrt{2}}\ket{1}\ket{11}. 
    \end{equation}
  \end{enumerate}
  Intuitively speaking, the provers Alice in Bob are aiming to transform
  $(\reg{X}_1^1,\reg{X}_1^2)$ into $(\reg{Y}_1^1,\reg{Y}_1^2)$ in such a way
  that (i) the state $\ket{00}$ is transformed to $\ket{00}$, (ii) the state
  $\ket{\phi}$ is transformed to $\ket{11}$, and (iii) the ``superposition''
  between $\reg{Z}_1$ being in the states $\ket{0}$ and $\ket{1}$ is not
  disturbed.
  This is challenging for them because $\ket{\phi}$ is entangled while
  $\ket{00}$ is not.

  This game has the particularity that the provers' maximum probability of
  convincing the verifier to produce the outcome $1$ increases with the
  dimension of their initial private registers $\reg{W}_0^1$ and $\reg{W}_0^2$.
  Informally speaking, this is so because the entanglement present in the state
  $\ket{\phi}$ can be ``hidden'' within a vast reservoir of entanglement in such
  a way that the ``superposition'' between $\reg{Z}_1$ being in the states
  $\ket{0}$ and $\ket{1}$ is not disturbed.
  (The idea is essentially the reverse of the \emph{embezzling} of entanglement
  phenomenon of van Dam and Hayden \cite{vanDamH03}.)
  More quantitatively, as shown in~\cite{LeungTW13}, unentangled provers can
  achieve a success probability of at most $3/4$ in this game, but optimal
  entangled provers sharing a state of local dimension $d$ succeed with
  probability $1-\Theta(\log^{-2} d)$, which tends to $1$ as $d\to\infty$.

  In particular, a simple strategy achieving a success probability that
  approaches $1$ as the dimension of the provers' shared entangled state grows
  can be devised as follows. 
  Suppose that Alice and Bob share the entangled state
  \begin{equation}
    \frac{1}{\sqrt{N}} \sum_{i = 1}^N \ket{00}^{\otimes i} \otimes
    \ket{\phi}^{\otimes (N-i+1)}
  \end{equation}
  for a very large value of $N$, where each copy of $\ket{00}$ and $\ket{\phi}$
  represents the state of a pair of qutrits shared between Alice and Bob.
  Using this state as a resource, Alice and Bob can approximately convert
  $\ket{00}$ to $\ket{\phi}$ (or \emph{vice versa}) by the unitary process
  which performs a cyclic rotation of the $(N+1)$ registers in their
  possession; the term \emph{embezzlement} comes from the fact that this
  process will leave the entanglement almost unchanged (for large $N$).  
  When used as a subroutine, this process allows Alice and Bob to win the game
  described above with probability approaching 1 as $N$ goes to infinity.
 \end{example}

Example~\ref{ex:dltw-1} suggests the introduction of two distinct quantities to
measure the maximum acceptance probability of the verifier in a multi-prover
interactive game:
\begin{itemize}
\item The \emph{unentangled value}, denoted $\val(V)$,
  is the highest probability with which the verifier can be made to output $1$
  when interacting with provers whose private registers are all initialized to
  the all-zero product state $\ket{0\cdots 0}$.
\item The \emph{entangled value} $\omega^*(V)$ is defined as the supremum over
  all finite-dimensional Hilbert spaces $\W_0^1,\ldots,\W_0^k$, corresponding to
  the provers' initial private registers $\rW_0^1,\ldots,\rW_0^k$, and all 
  initial pure states\footnote{
    Similar to the classical setting in which shared randomness does not
    affect the power of multiple provers, there is no increase in power for
    multiple quantum provers when their initial private registers are
    in a mixed quantum state, as compared with a pure state.}
  $\ket{\psi}\in\W_0^1\otimes\cdots\otimes \W_0^k$
  of these registers, of the provers' maximum probability of
  causing the verifier to output $1$.
\end{itemize}  

These two values lead to potentially distinct classes of problems having
multi-prover interactive proof systems:
$\QMIP$ for the case of unentangled provers and $\QMIP^*$ when the provers are
allowed to share arbitrary entangled states.

\begin{definition}\label{def:mips}
A promise problem $A = (A_{\textup{yes}},A_{\textup{no}})$ is contained in the
complexity class $\class{QMIP}_{a,b}(k,m)$ if and only if there exists a
polynomial-time computable function $V$ that possesses the following properties:
\begin{mylist}{\parindent}
\item[1.] For every string $x\in A_{\textup{yes}} \cup A_{\textup{no}}$, one
  has that $V(x)$ is an encoding of a quantum circuit description of an
  $m$-turn verifier in an interactive game with $k$ provers.
\item[2.] \emph{Completeness.} 
  For every string $x\in A_{\textup{yes}}$, it holds that
  $\omega(V(x)) \geq a$.
\item[3.] \emph{Soundness.}
  For every string $x\in A_{\textup{no}}$, it holds that $\omega(V(x)) \leq b$.
\end{mylist}
The complexity class $\QMIP^*_{a,b}(k,m)$ is defined in the same way, except
that the quantity $\omega^*(V(x))$ replaces $\omega(V(x))$.
\end{definition}

Similar conventions to those in the single-prover setting will be used to refer
to the classes above.
For instance, we denote
\begin{equation}
  \begin{gathered}
  \QMIP(k,m) = \QMIP_{2/3,1/3}(k,m),\\
  \QMIP^*(k,m) = \QMIP^*_{2/3,1/3}(k,m).
  \end{gathered}
\end{equation}
We let $\QMIP_{a,b}$ and $\QMIP^{\ast}_{a,b}$ denote the classes of promise
problems $A$ for which $A \in \QMIP_{a,b}(k,m)$ or $A\in\QMIP^*_{a,b}(k,m)$,
respectively, for some choice of polynomially bounded functions $k$ and $m$,
and we denote $\QMIP = \QMIP_{2/3,1/3}$ and $\QMIP^* = \QMIP^{\ast}_{2/3,1/3}$.

In addition, the classes $\MIP_{a,b}(k,m)$ and $\MIP^*_{a,b}(k,m)$ are defined
in an analogous way, except that the verifier is classical (specified by a
classical Boolean circuit that may take a uniformly random bit string as an
auxiliary input).
All messages exchanged with the provers are restricted to being classical
strings in this case.

The fact that both the completeness and soundness parameters of the classes
$\QMIP_{a,b}(k,m)$ and $\QMIP^{\ast}_{a,b}(k,m)$ are defined with respect to
$\val$ and $\val^*$, respectively, makes their relationship non-obvious.
The inequality $\val^* \geq \val$ always holds, but it can have countervailing 
effects. 
First, it implies that a proof system sound against unentangled provers may no
longer be sound when the provers are allowed to share entanglement. 
Second, a proof system achieving a certain completeness parameter with entangled
provers may not have the same property when the provers are restricted to
unentangled strategies. 
Because both the soundness and completeness parameters are affected in possibly
different ways, it is not clear in which cases the presence of a gap between the
parameters (corresponding to the distinction between yes- and no-inputs) is
preserved. 
This phenomenon will be discussed in greater detail in subsequent sections.

As in the single-prover setting, the choice of completeness and soundness
parameters $a,b$ does not affect the class of problems that lie in
$\QMIP_{a,b}$ or $\QMIP^*_{a,b}$, so long as they are polynomially
separated---any inverse polynomial separation between $a$ and $b$ can be
amplified in a straightforward way, either by repeating the game sequentially
or with different sets of provers.
The following proposition states this fact in more precise terms.

\begin{prop}\label{prop:amplify}
  Let $V$ be a verifier in a $k$-prover $m$-turn interactive game and
  let $a, b\in [0,1]$ be real numbers such that $a>b$.
  For every positive integer $T$, there exists a verifier $V'$ in a $k$-prover,
  $T m$-turn interactive game (or, alternatively, a verifier $V'$ in a
  $T k$-prover, $m$-turn interactive game) for which the implications
  \begin{equation}
    \begin{aligned}
      \omega(V) \geq a & \;\Rightarrow\;
      \omega(V') \geq 1 - \exp\biggl(-\frac{(a-b)^2}{2}T\biggr)\\
      \omega(V) \leq b & \;\Rightarrow\;
      \omega(V') \leq \exp\biggl(-\frac{(a-b)^2}{2}T\biggr)
    \end{aligned}
  \end{equation}
  hold.
  Furthermore, the same implications hold for $\omega^*$ (under the same
  transformation).
\end{prop}

The fact that this procedure works as described when repeated in parallel with
$Tk$ entangled provers follows along the same lines as for sequential
repetition, as it can always be considered that the interaction with each group
of $k$ provers is performed in sequence.

\section{The importance of entanglement}
\label{sec:multiprover-entanglement}

The first part of this section, Section~\ref{sec:qmip-nexp}, is devoted to
proof systems with multiple unentangled provers.
It will be shown that quantum verifiers have exactly the same power as
classical verifiers in this setting: $\QMIP=\MIP=\NEXP$. 
The proof relies on the characterization $\MIP=\NEXP$, but is otherwise
not difficult.
Thus, in the absence of entanglement between the provers, quantum verifiers are
neither less nor more powerful than their classical counterparts.

In the second part of this section, Section~\ref{sec:oracularization}, it will
be argued that the situation is markedly different in the presence of entangled
provers.
In particular, the technique of \emph{oracularization},
which is central to establishing the soundness property of natural proof
systems for $\NEXP$-complete problems, is shown to fail for entangled provers
in its most standard form.
 
Not only does entanglement allow provers to break the soundness of simple
proof systems, but for certain restricted classes of verifiers it appears to be
impossible (under commonly conjectured complexity-theoretic assumptions) to
modify proof systems in such a way as to make them sound against entangled
provers.
This will be demonstrated in Section~\ref{sec:xor} for the special case of XOR
proof systems, for which the associated class with unentangled provers,
$\oplus\MIP$, equals $\NEXP$, but collapses to a subset of
$\PSPACE$ when the provers are allowed to share prior entanglement.

\subsection{Provers without prior shared entanglement: 
  \class{QMIP}=\class{NEXP}}
\label{sec:qmip-nexp}

This section considers interactive proof systems based on games in which the
provers do not share any prior entanglement. 
As will be shown, the class $\QMIP$ of promise problems 
that can be decided by such proof systems exactly coincides with the class
$\MIP = \NEXP$ of problems that can be decided by a classical verifier
interacting with multiple unentangled provers.
The situation in this case is therefore analogous to the single-prover setting,
where the equality $\QIP = \IP$ demonstrates that the ability to exchange 
quantum information does not affect the verification power of the verifier.

The proof that $\QMIP$ coincides with $\NEXP$ relies on two separate
inclusions.
The first inclusion is
\begin{equation}\label{eq:qmip-in-nexp}
\QMIP\subseteq \NEXP,
\end{equation}
which follows from the existence of a nondeterministic exponential-time
procedure for determining the unentangled value of a multi-prover interactive
game with high accuracy.
Second, the containment
\begin{equation}\label{eq:mip-in-qmip}
\MIP \subseteq \QMIP
\end{equation}
is easily seen to hold, as a quantum verifier can simulate a classical verifier
in a straightforward way by systematically measuring the provers' messages in
the standard basis. 
This leaves the value of the game unchanged, as unentangled provers gain no
advantage from using quantum information against a classical verifier.
The equality
\begin{equation}
  \class{QMIP}=\class{NEXP}
\end{equation}
follows by combining~\eqref{eq:qmip-in-nexp} and~\eqref{eq:mip-in-qmip}
together with the characterization $\class{MIP}=\class{NEXP}$, which is an
important classical result to which we will return in
Section~\ref{sec:nexp-in-mipstar}.
We are not aware of a direct proof of $\QMIP=\MIP$ that does not rely on this
characterization.

One consequence of the equality $\class{QMIP}=\class{MIP}$ is that various
results applying to classical multi-prover interactive proof systems immediately
extend to their quantum unentangled counterparts.
For instance, it is known that such proof systems can be given perfect
completeness and exponentially small soundness error, and can be parallelized
to a single round of interaction with just two provers.

\begin{theorem}
  \label{theorem:QMIP=QMIP(2,1)}
  For every positive polynomially bounded function $p$ it holds that
  \begin{equation}
    \class{QMIP} = \class{QMIP}_{1,2^{-p}}(2,2) = \class{NEXP}.
  \end{equation}
\end{theorem}

Assuming the known results on \class{MIP} just suggested (about which more will
be said when we discuss their entangled-prover counterparts in
Section~\ref{sec:qmip-structure}), one therefore has that
Theorem~\ref{theorem:QMIP=QMIP(2,1)} follows from the
inclusion~\eqref{eq:qmip-in-nexp}.

With the goal of proving \eqref{eq:qmip-in-nexp} in mind, consider the problem
of certifying the provers' maximum acceptance probability in a given
$k$-prover, $m$-turn interactive game.
An arbitrary strategy for the provers can be specified by an explicit
description of the $j$-th prover's isometry in the $i$-th round,
\begin{equation}
  \label{eq:porover-isometries}
  P_i^j\in \Unitary(\W_{i-1}^j\otimes \X_i^j, \W_i^j\otimes \Y_i^j),
\end{equation}
for all $j=1,\ldots,k$ and $i=1,\ldots,\lceil m/2\rceil$. 
Putting issues of precision aside, which can be handled by specifying rational
approximations with exponential accuracy to the real and imaginary part of each
complex matrix entry, the probability of the verifier outputting $1$ in the
corresponding interaction can be computed by performing the appropriate matrix
operations.

The inclusion $\QMIP\subseteq \NEXP$ will therefore follow once it is proved
that there exists an optimal prover strategy that can be specified by
isometries of dimension at most exponential in the description size of the
verifier.
Because the message registers necessarily satisfy such a bound, it will suffice
to bound the dimension of the private register $\reg{W}_i^j$ associated with
the $j$-th prover's isometry in the $i$-th round.
Such a bound can be obtained based on the following theorem (which represents a
very minor extension of Theorem~\ref{theorem:Stinespring-equivalence}).

\begin{theorem}
  \label{theorem:dim-reg-multi}
  Let $\X$, $\Y$, $\V$, and $\W$ be finite-dimensional Hilbert spaces with
  $\dim(\W) \geq \dim(\V) = \dim(\X\otimes\Y)$, and let
  $A \in \Unitary(\X,\W\otimes\Y)$ be an isometry.
  There exist isometries $B\in \Unitary(\X, \V\otimes\Y)$ and
  $C\in\Unitary(\V,\W)$ such that
  \begin{equation}
    A = (C \otimes \I_{\Y})B.
  \end{equation}
\end{theorem}

Figure~\ref{fig:prover-memory1} illustrates this theorem in the form of a
picture suggestive of a circuit diagram.
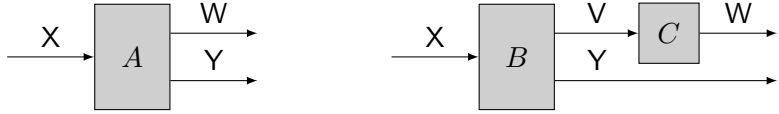
\begin{figure}
  \begin{center}
    \begin{tikzpicture}[scale=0.45, 
        isometry/.style={draw, minimum height=14mm, minimum width=10mm,
          fill = ChannelColor, text=ChannelTextColor},
        empty/.style={minimum height=14mm, minimum width=1mm},
        >=latex]

      \node (In) at (-4,0) [empty] {};      
      \node (A) at (0,0) [isometry] {$A$};
      \node (Out) at (4,0) [empty] {};
      
      \draw[->] (In.east) -- (A.west) node [above, midway] {$\reg{X}$};
      
      \draw[->] ([yshift=7mm]A.east) -- ([yshift=7mm]Out.west)
      node [above, midway] {$\reg{W}$};
      
      \draw[->] ([yshift=-7mm]A.east) -- ([yshift=-7mm]Out.west)
      node [above,midway] {$\reg{Y}$};
      
    \end{tikzpicture}
    \hspace*{10mm}
    \begin{tikzpicture}[scale=0.45, 
        isometry/.style={draw, minimum height=14mm, minimum width=10mm,
          fill = ChannelColor, text=ChannelTextColor},
        smallisometry/.style={draw, minimum height=8mm, minimum width=8mm,
          fill = ChannelColor, text=ChannelTextColor},
        empty/.style={minimum height=14mm, minimum width=1mm},
        >=latex]

      \node (In) at (-4,0) [empty] {};      
      \node (B) at (0,0) [isometry] {$B$};
      \node (C) at (4.5,0.7) [smallisometry] {$C$};
      \node (Out) at (8,0) [empty] {};
      
      \draw[->] (In.east) -- (B.west) node [above, midway] {$\reg{X}$};
      
      \draw[->] ([yshift=-7mm]B.east) -- ([yshift=-7mm]Out.west)
      node [above, pos=0.19] {$\reg{Y}$};
      
      \draw[->] ([yshift=7mm]B.east) -- (C.west)
      node [above,midway] {$\reg{V}$};
      
      \draw[->] (C.east) -- ([yshift=7mm]Out.west)
      node [above,midway] {$\reg{W}$};
      
    \end{tikzpicture}
  \end{center}
  \caption{An isometry $A$ transforming $\reg{X}$ to $(\reg{W},\reg{Y})$ is
    necessarily equivalent to an isometry $B$ transforming
    $\reg{X}$ to $(\reg{V},\reg{Y})$, followed by an isometry $C$ transforming
    $\reg{V}$ to $\reg{W}$, assuming $\reg{V}$ has the same size as
    $(\reg{X},\reg{Y})$ and $\reg{W}$ is at least this large.}
  \label{fig:prover-memory1}
\end{figure}

Through the use of this theorem, one may replace a given prover $P^j$ by an
equivalent prover $Q^j$ that substitutes a register $\reg{V}^j_i$, which
has dimension equal to the product of the dimensions of the message registers
\begin{equation}
  \reg{X}^j_1,\reg{Y}^j_1,\ldots, \reg{X}^j_i,\reg{Y}^j_i,
\end{equation}
for each register $\reg{W}^j_i$ used by $P^j$ (which we assume has been
specified by a collection of isometries $\{P^j_i\}$ as in
\eqref{eq:porover-isometries}).
The theorem is applied independently to each prover action, beginning with
$i = 1$ and increasing to $i = \lceil m/2 \rceil$.

In particular, one starts with $i = 1$, and takes $\X = \X^j_1$, $\Y = \Y^j_1$,
$\W = \W^j_1$, and $\V$ being a space with
\begin{equation}
  \dim(\V) = \dim(\X\otimes\Y) = \dim\bigl(\X^j_1 \otimes \Y^j_1\bigr),
\end{equation}
corresponding to a new private memory register $\reg{V}^j_1$ that will replace
the register $\reg{W}^j_1$.
The first isometry performed by the new prover $Q^j$ is the isometry
$Q^j_1\in\Unitary(\X^j_1,\V^j_1\otimes\Y^j_1)$ that is represented by $B$ in
the theorem.
The isometry $C$ from the theorem is composed with $P^j_2$, and the process is
repeated for $i = 2,\ldots, \lceil m/2 \rceil$.
In general, one applies the theorem with
$\X = \V^j_{i-1}\otimes \X^j_i$ (with $\V_0 = \complex$), 
$\Y = \Y^j_i$, $\W = \W^j_i$, and $\V$ being a space with
\begin{equation}
  \dim(\V) = \dim(\X\otimes\Y) 
  = \dim\bigl(\V^j_{i-1} \otimes \X^j_i \otimes \Y^j_i\bigr),
\end{equation}
corresponding to a new private memory register $\reg{V}^j_i$ that will replace
the register $\reg{W}^j_i$.
The size of the memory register $\reg{V}^j_i$ obtained in this way
therefore has the same size as the tuple of registers
\begin{equation}
  (\reg{X}^j_1,\reg{Y}^j_1,\ldots,\reg{X}^j_i,\reg{Y}^j_i).
\end{equation}
When the theorem is applied to each prover's final operation, the isometry
$C$ is simply discarded---as the verifier never touches the provers' private
memory registers, nothing is lost in disregarding this isometry.

It is worth noting that, in contrast to the single-prover case, it is not known
if an efficient optimization over strategies for the provers in a multi-prover
interactive game is possible (given an explicit matrix description of a
verifier).
One cannot accomplish such an optimization by considering only the local
properties of a sequence of reduced states of the verifier's private and
message registers at each turn of the interactive game in a manner similar to
the single-prover setting, as there is no known analogue of
Theorem~\ref{thm:unitary-equivalence} on the unitary equivalence of
purifications that would apply to the setting of multiple quantum provers.


For example, consider a setting in which three single-qubit registers
$(\reg{X}_1,\reg{Z},\reg{X}_2)$ are in the mixed state
\begin{equation}
  \label{eq:example-state-in}
  \ket{0}\bra{0} \otimes \frac{\I}{2} \otimes \ket{0}\bra{0}.
\end{equation}
Prover 1 is permitted to transform $\reg{X}_1$ into $\reg{Y}_1$ and prover 2
transforms $\reg{X}_2$ into $\reg{Y}_2$, where $\reg{Y}_1$ and $\reg{Y}_2$ are
also single-qubit registers.
One may ask if it is possible for the provers to transform the original state
\eqref{eq:example-state-in} into one of the three states
\begin{equation}
  \label{eq:example-states-out}
  \ket{\phi}\bra{\phi} \otimes \ket{0}\bra{0},
  \quad
  \ket{0}\bra{0} \otimes \ket{\phi}\bra{\phi},
  \quad\text{or}\quad
  \ket{\gamma}\bra{\gamma}
\end{equation}
of $(\reg{Y}_1,\reg{Z},\reg{Y}_2)$, where
\begin{equation}
  \begin{aligned}
    \ket{\phi} & = \frac{1}{\sqrt{2}}\ket{0}\ket{0} +
    \frac{1}{\sqrt{2}}\ket{1}\ket{1},\\
    \ket{\gamma} & = \frac{1}{\sqrt{2}}\ket{0}\ket{0}\ket{0} +
    \frac{1}{\sqrt{2}}\ket{1}\ket{1}\ket{1}.
  \end{aligned}
\end{equation}
All three transformations may or may not be possible, depending on the initial
correlations among the registers $\reg{X}_1$, $\reg{Z}$, and $\reg{X}_2$ and
two additional registers $\reg{W}_1$ and $\reg{W}_2$ representing the memories
of prover 1 and prover 2, respectively.
For instance, if $(\reg{W}_1,\reg{Z})$ is initially in the pure state
$\ket{\phi}$, then the provers are capable of transforming the original state
\eqref{eq:example-state-in} into the first state of
\eqref{eq:example-states-out}, but neither the second nor the third.
A transformation to either the second or third state is also possible assuming
different initial states of $(\reg{W}_1,\reg{X}_1,\reg{Z},\reg{X}_2,\reg{W}_2)$.
The ability of the provers to perform a particular transformation is therefore
not a function of the states in question, but also depends on the initial
state of the provers' memories.

\subsection{The failure of oracularization}
\label{sec:oracularization}

The equality $\class{QMIP}=\class{MIP}$ demonstrates that quantum multi-prover
interactive proof systems with unentangled provers are no more powerful than
their classical counterparts. 
As already discussed, allowing the use of prior shared entanglement for the
provers can raise the value of an interactive game, affecting both the
soundness and completeness parameters of a proof system. 
As a result, both inclusions on which the aforementioned equality are based, 
\begin{equation}
  \class{QMIP}\subseteq\class{NEXP}
  \quad\text{and}\quad
  \class{MIP}\subseteq\class{QMIP},
\end{equation}
may in principle fail for entangled provers.

The possible failure of the inclusion $\QMIP^*\opn\class{NEXP}$ is directly
related to the absence of an analogue of Theorem~\ref{theorem:dim-reg-multi}
for provers sharing prior entanglement of arbitrary dimension. 
This issue will be discussed in greater detail in
Section~\ref{sec:entanglement}.

The possible failure of the second inclusion, $\MIP\opn\class{QMIP}^*$, is
suggested by Example~\ref{ex:dltw-1}, which demonstrates that the entangled
value can be much larger than the unentangled value. 
To investigate how entanglement may allow the provers to break the soundness
property of a proof system, we study this effect in more detail in the context
of so-called \emph{oracularized games}. 
Oracularization is a technique frequently employed in the study of the class
$\MIP$---it allows for a reduction in both the numbers of rounds and provers
required by proof systems for problems in $\MIP$, and it plays an important
role in known proofs of the inclusion $\class{NEXP}\subseteq\class{MIP}$. 
The failure of this technique in the presence of entanglement between the
provers is a source of considerable difficulty in working with the class
$\class{QMIP}^*$.

Oracularization leverages the presence of multiple provers by using one of the
provers to check that the others provide answers in a non-adaptive manner. 
As an example demonstrating this technique, consider the one-round two-prover
\emph{clause-versus-variable} interactive game described in
Figure~\ref{figure:clause-versus-variables}.
In this game, the provers' goal is to convince the verifier of the
satisfiability of a set of constraints $\varphi$, where each constraint acts on
a subset of $n$ variables $x_1,\ldots,x_n$ taking values in some finite
alphabet $\Gamma$.
Similar ideas play an important role in the proof that 
$\MIP = \NEXP$.

\begin{figure}
  \noindent\hrulefill
  \begin{trivlist}
  \item
    The input is a collection of constrains 
    $\varphi = (C_1,\ldots,C_m)$ on variables $x_1,\ldots,x_n$
  \item
    {\bf Verifier's step 1:}
    Select a constraint $C_j$ uniformly at random, and send $C_j$ to Alice.
    Also select a variable $x_i$ on which $C_j$ acts, uniformly at random, and
    send $x_i$ to Bob.
  \item
    {\bf Provers' actions:}
    Alice replies with an assignment to all variables in $C_j$. 
    Bob replies with an assignment to $x_i$.
  \item
    {\bf Verifier's step 2:}
    Output $1$ if and only if the provers' assignments are consistent on $x_i$
    and satisfy the constraint $C_j$.
  \end{trivlist}
  \noindent\hrulefill
  \caption{Clause-versus-variable interactive game.}
  \label{figure:clause-versus-variables}
\end{figure}

Consider first the value of this game when the two provers, Alice and Bob, are
restricted to classical deterministic strategies. 
In this case, a strategy for Bob is a function mapping each variable $x_i$ to
an element of the alphabet $\Gamma$, so that the strategy coincides with a
complete assignment to the variables.
For each constraint $C$ that this assignment fails to satisfy, there is a
probability at least $1/\ell$ that the provers will fail, for $\ell$ being the
number of variables in the constraint $C$; either Alice's assignments fail to
satisfy $C$, or they must differ from the assignment represented by Bob's
strategy on at least one variable.
Consequently, for $V_{\varphi}$ being the verifier defined in the
clause-versus-variable game for $\varphi$, one has that
\begin{equation}\label{eq:clause-var-classical}
  \omega(\varphi)\,\leq\, \omega(V_\varphi)\,\leq\,
  1-\frac{1-\omega(\varphi)}{\ell},
\end{equation}
where $\omega(\varphi)$ is the maximum fraction of constraints that are
simultaneously satisfiable in $\varphi$, and where it has been assumed that
$\ell$ variables appear in each constraint.

Unfortunately this technique fails in the presence of shared entanglement
between the provers, as is demonstrated by the following example.

\begin{example}[The Magic Square game]\label{ex:ms-1}
  Consider a $3\times 3$ matrix of Boolean variables
  \begin{equation}
    \begin{pmatrix}
      X_1 & X_2 & X_3\\
      X_4 & X_5 & X_6\\
      X_7 & X_8 & X_9
    \end{pmatrix},
  \end{equation}
  and define a one-round two-prover interactive game as follows.
  \begin{mylist}{\parindent}
  \item[1.]
    The verifier first chooses either a row or a column in the $3\times 3$
    matrix of Boolean random variables, uniformly at random from the 6 possible
    choices, and sends these variables to the first prover Alice.
    The verifier also selects one of the three variables in the chosen row or
    column, uniformly at random from the 3 possibilities, and sends this
    variable to the second prover Bob.
  \item[2.]
    The provers must respond with Boolean assignments to the variables they
    were sent.
  \item[3.]
    The verifier outputs 1 (i.e., accepts) if and only if the following
    conditions hold:
    \begin{mylist}{8mm}
      \item[(a)]
        Both Alice and Bob give the same assignment to the one variable they
        were sent in common.
      \item[(b)]
        If the verifier initially selected a \emph{row} in the $3\times 3$
        matrix of Boolean random variables, then Alice's assignments
        to these variables must have \emph{even parity}.
      \item[(b)]
        If the verifier initially selected a \emph{column} in the $3\times 3$
        matrix of Boolean random variables, then Alice's assignments
        to these variables must have \emph{odd parity}.
    \end{mylist}
  \end{mylist}

  If the provers employ a classical strategy in this game, their probability of
  causing the verifier to accept is at most 17/18.
  This follows from the fact that any deterministic strategy for Bob must
  determine an assignment to the Boolean variables $X_1,\ldots,X_9$, and no
  assignment to these variables can satisfy all six of the parity constraints
  (because the parity of all 9 Boolean variables cannot be both even and
  odd).
  It is straightforward to see that there exists a deterministic strategy 
  for the provers that succeeds with probability exactly 17/18, which
  establishes that this upper-bound is achievable.
  (For instance, Alice may respond with assignments (0,0,0), (0,0,0), and
  (1,1,0) for rows 1, 2, and 3, respectively, and assignment (0,0,1) for all
  three columns;
  and Bob may answer in a manner consistent with the assigment
  $(X_1,\ldots,X_9) = (0,0,0,0,0,0,1,1,0)$.
  This strategy only loses in the case that the verifier asks Alice for an
  assignment of column 3 and Bob for an assignment to $X_9$.)

  In contrast, entangled provers have a perfect strategy for this game---they
  can win with certainty.
  One strategy for the provers that achieves this goal is based on the
  construction of nine $\pm 1$-\emph{observables} $H_1,\ldots,H_9\in
  \Herm(\complex^4)$, meaning that they are Hermitian operators whose
  eigenvalues are all either $1$ or $-1$, having the following properties: 
  if $H_i$ is placed in the $i$-th position of the $3\times 3$ magic square,
  then
  \begin{mylist}{8mm}
  \item[(i)]
    the operators appearing in the same row or in the same column must commute,
    and
  \item[(ii)]
    the product of the operators appearing in each row is $\I$, and the
    product of the operators appearing in each column is $-\I$.
  \end{mylist}
  Such operators can be constructed from the Pauli operators
  \begin{equation}
    \sigma_x =
    \begin{pmatrix}
      0 & 1\\
      1 & 0
    \end{pmatrix},
    \quad
    \sigma_y =
    \begin{pmatrix}
      0 & -i\\
      i & 0
    \end{pmatrix},
    \quad
    \sigma_z =
    \begin{pmatrix}
      1 & 0\\
      0 & -1
    \end{pmatrix}
  \end{equation}
  as follows:
  \begin{equation}
    \label{eq:magic-square-operators}
    \begin{pmatrix}
      H_1 & H_2 & H_3\\
      H_4 & H_5 & H_6\\
      H_7 & H_8 & H_9
    \end{pmatrix}
    =
    \begin{pmatrix}
      \sigma_x \otimes \sigma_y & \sigma_y \otimes \sigma_x &
      \sigma_z\otimes\sigma_z \\
      \sigma_y \otimes \sigma_z & \sigma_z \otimes \sigma_y &
      \sigma_x\otimes\sigma_x \\
      \sigma_z \otimes \sigma_x & \sigma_x \otimes \sigma_z &
      \sigma_y\otimes\sigma_y
    \end{pmatrix}.
  \end{equation}
  From each of these observables a projective measurement
  $\{\Pi^k_0,\Pi^k_1\}$, for $k = 1,\ldots,9$, can be defined with $\Pi^k_a =
  (\I\otimes\I + (-1)^a H_k)/2$, for $a\in\{0,1\}$, being the projector on the
  eigenspace of $H_k$ with associated eigenvalue $(-1)^a$.
 Suppose the provers share the entangled state
     \begin{equation}
    \label{eq:magic-square-shared-state}
    \ket{\psi} =
    \frac{1}{2} \ket{00}\ket{00}
    + \frac{1}{2} \ket{01}\ket{01}
    + \frac{1}{2} \ket{10}\ket{10}
    + \frac{1}{2} \ket{11}\ket{11},
  \end{equation}
    in which they both hold two qubits. This state has the property that
    \begin{equation}\label{eq:magic-square-me}
    \bra{\psi}A\otimes B\ket{\psi}=\frac{1}{4}\,\Tr\big(A B^{\t}\big)
    \end{equation}
    for any $A,B\in\Lin(\complex^4)$.
   When Alice (or Bob) receives the labels of some variables, she will perform
   the measurements described above, in sequence, to the pair of qubits she
   holds, using the measurement $\{\Pi^k_0,\Pi^k_1\}$ to determine the
   assignment she responds with for the variable $X_k$. 
   Property (i) of the $H_i$ ensures that measurements within any single row
   or column commute, so it does not matter which order Alice would choose to
   perform these measurements. 
   Using property (ii) and~\eqref{eq:magic-square-me}, it may be verified
   that the required parity conditions will always hold for the outcomes of
   these measurements, and that Alice and Bob will always produce the same
   assignment to the variable they both received.
\end{example} 

Based on Example~\ref{ex:ms-1} a simple $3$-SAT formula $\varphi$ with $9$
variables and $24$ clauses can be devised such that $\varphi$ is not
satisfiable but the clause-vs-variable interactive game defined from this
formula can be won with certainly by entangled provers.
In contrast, the unentangled value is strictly less than 1.

Consequently, one has that the oracularization technique does not extend
directly to the case of entangled provers. 
In the following section it is shown that this failure is not limited to
specific examples such as the Magic Square game, but can affect the
verification possibilities of broad classes of verifiers in multi-prover
interactive games.

\subsection{XOR games}
\label{sec:xor}

XOR games are a class of two-prover one-round interactive games in which the
verifier is restricted to have the following form. The verifier's message (also
called its \emph{question}) to each prover in the first turn is classical. 
Each prover's message (its \emph{answer}) to the verifier in the second turn is
classical and consists of a single bit. Finally the verifier decides on its
output bit based solely on the parity of the provers' answers. It is possible
to consider XOR games with any number of provers, but in this section we focus
on the case of two-prover XOR games. 

The class of promise problems that can be decided by verifiers having this
restricted form is denoted $\oplus\MIP_{a,b}(2,2)$ in case the provers are
unentangled, and $\oplus\MIP_{a,b}^*(2,2)$ with entangled provers. An important
result in the field of hardness of approximation states that the
unentangled-prover class is powerful enough to capture all problems in $\NEXP$,
meaning that the inclusion 
\begin{equation}\label{eq:xor-mip}
\NEXP\subseteq  \oplus\MIP_{a,b}(2,2)
\end{equation}
 holds for a specific choice of constants $0<b<a<1$. In contrast, allowing
 entanglement between the provers reduces the verifier's decision power (under
 the assumption that $\PSPACE$ is properly contained in $\NEXP$):
\begin{equation}\label{eq:xor-mip-star}
\oplus\MIP_{a,b}^*(2,2) \subseteq\PSPACE,
\end{equation}
which holds for any $0\leq b<a\leq 1$ separated by at least an inverse 
polynomial. 
Thus, the introduction of entanglement has the effect of collapsing the
verifier's ability to decide problems, from $\NEXP$ to $\PSPACE$. 

The inclusion~\eqref{eq:xor-mip-star} can be shown by giving a direct
simulation of any $\oplus\MIP^*(2,2)$ verifier by a $\class{QIP}(2)$ verifier,
concluding via the inclusion $\QIP\subseteq\PSPACE$ described in
Section~\ref{sec:QIP=PSPACE}. 
In the remainder of this section we will describe the weaker inclusion
$\oplus\MIP_{a,b}^*(2,2) \subseteq\EXP$, which has the advantage that it can be
proven by expressing the entangled value $\val^*$ of an XOR game  directly as
the optimum of a semidefinite program.

The verifier $V$ in an XOR game can be specified explicitly as a pair
$(\pi,V)$, consisting of a distribution $\pi$ on pairs of questions
$(x,y)\in X\times Y$ and a predicate $V(c\mathbin{|}x,y)$ that dictates the
parities $c = a\oplus b$ for the provers' answers $a$ and $b$ that cause the
verifier to accept.
A strategy for the provers consists of a choice of Hilbert spaces $\V$ and
$\W$, a pure state $\ket{\psi}\in \V\otimes\W$, and two families of
binary-valued measurements
\begin{equation}
  \{P^x_0,P^x_1\} \quad\text{and}\quad \{Q^y_0,Q^y_1\}
\end{equation}
on the spaces $\V$ and $\W$, respectively.
Upon receiving questions $(x,y)$, the probability that the provers return 
answers $(a,b)$ is
\begin{equation}
  \bra{\psi} P^x_a \otimes Q^y_b \ket{\psi}.
\end{equation}
Thus, $\omega^{\ast}(V)$ is equal to the supremum value of the expression
\begin{equation}
  \sum_{x,y} \pi(x,y) \sum_c V(c\mathbin{|}x,y) 
  \sum_{\substack{a,b\\a\oplus b = c}}
  \bra{\psi} P^x_a \otimes Q^y_b \ket{\psi},
\end{equation}
over all strategies for the provers, as described above. 

Up to an additive scaling of $\omega^*$, one may assume that for each pair
$(x,y)$, there is a unique $c\in\{0,1\}$ such that $V(c\mathbin{|}x,y)=1$.
For each question pair $(x,y)$, let $r(x,y)=(-1)^c \pi(x,y)$ for this unique
choice of $c$.
It holds that
\begin{equation}
  \omega^*(V)= \frac{1}{2} + \frac{1}{2} \beta^*(V),
\end{equation}
where the \emph{bias} $\beta^*(V)$ is defined as 
\begin{equation}\label{eq:opt-xor}
  \beta^*(V) = \sup \sum_{x,y} r(x,y) \bra{\psi} A_x \otimes B_y\ket{\psi},
\end{equation}
where $A_x = P^x_0 - P^x_1$ and $B_y = Q^y_0 - Q^y_1$, and the supremum is
over all strategies as before.
(The operators $A_x$ and $B_y$ are observables that represent the
binary-valued measurements $\{P^x_0, P^x_1\}$ and $\{Q^y_0,Q^y_1\}$.)

\begin{example}[CHSH game]\label{ex:chsh}
  A simple XOR game is the CHSH game, named after its inventors Clauser, Horne,
  Shimony and Holt~\cite{Clauser:69a}. 
  In this game, the verifier's questions consist of a single bit each, the
  distribution $\pi$ is uniform on $\{0,1\}\times\{0,1\}$, and the predicate
  representing the verifier's final decision is defined as
  \begin{equation}
    V(c\mathbin{|}x,y) =
    \begin{cases}
      1 & \text{if $c=x\wedge y$}\\
      0 & \text{if $c\not=x\wedge y$}.
    \end{cases}
  \end{equation}
  It therefore holds that $r(x,y)=(-1)^{x \wedge y}/4$.
  The bias $\beta^*(\CHSH)$ is given by the expression
  \begin{equation}
    \sup_{\ket{\psi},A_0,A_1,B_0,B_1}\frac{1}{4}
    \bra{\psi}\big(A_0\otimes B_0 + A_1\otimes B_0 + A_0\otimes B_1 
    - A_1 \otimes B_1\big)\ket{\psi},
  \end{equation}
  where the supremum is over all bipartite states $\ket{\psi}$ and observables
  $A_0,A_1,B_0,B_1$. 
  By considering the choices
  \begin{equation}
    \begin{gathered}
      \ket{\psi} = \frac{1}{\sqrt{2}}(\ket{00}+\ket{11}),\\
      A_0 = \sigma_x,\quad B_0 = (\sigma_x + \sigma_z)/\sqrt{2},\\
      A_1 = \sigma_z,\quad B_1 = (\sigma_x - \sigma_z)/\sqrt{2},
    \end{gathered}
  \end{equation}
  one finds that $\beta^*(\CHSH)\geq \sqrt{2}/2$. 
  That this holds with equality will be show below. 
  In contrast, unentangled provers are easily seen to achieve a bias at most
  $\beta(\CHSH)=1/2$.
\end{example}

There exists a natural semidefinite programming relaxation for the bias of a
given XOR game as follows.
First, let $R$ be a matrix indexed by the disjoint union $X \sqcup Y$ of the
question sets, defined as
\begin{equation}
  R(x,y) = R(y,x) = \frac{r(x,y)}{2}
  \quad\text{and}\quad
  R(x,x') = R(y,y') = 0
\end{equation}
for all $x,x'\in X$ and $y,y'\in Y$.
Next, for a given strategy, defined by a shared entangled state $\ket{\psi}$
and collections of $\pm 1$-observables $\{A_x\}$ and $\{B_y\}$, define unit
vectors
\begin{equation}
  u_x = (A_x\otimes\I)\ket{\psi}
  \quad\text{and}\quad
  v_y = (\I \otimes B_y)\ket{\psi},
\end{equation}
and observe that
\begin{equation}
  \bra{\psi} A_x \otimes B_y \ket{\psi} = \ip{u_x}{v_y}
\end{equation}
for every pair $(x,y)\in X\times Y$.
One finds that the bias obtained by this particular strategy is given by
\begin{equation}
  \sum_{x,y} r(x,y) \bra{\psi} A_x \otimes B_y\ket{\psi}
  = \ip{R}{Z},
\end{equation}
for $Z$ the Gram matrix of the collection 
$\{u_x\,:\,x\in X\}\cup\{v_y\,:\,y\in Y\}$, i.e.,
\begin{equation}
  \begin{alignedat}{2}
    Z(x,y) & = \ip{u_x}{v_y}, \quad & Z(x,x') & = \ip{u_x}{u_{x'}},\\
    Z(y,x) & = \ip{v_y}{u_x}, \quad & Z(y,y') & = \ip{v_y}{v_{y'}}.
  \end{alignedat}
\end{equation}
(In the present case, one has that each of the values $Z(x,y)$ is real
and satisfies $Z(x,y) = Z(y,x)$ because the value
$\bra{\psi} A_x \otimes B_y \ket{\psi}$ is a real number.)

It is therefore the case that
\begin{equation}
  \label{eq:sdp-xor}
  \beta^{\ast}(V) \leq \op{SDP}(V) = 
  \sup_Z \, \ip{R}{Z},
\end{equation}
where the supremum is over all positive semidefinite matrices $Z$ indexed by
$X\sqcup Y$ and satisfying $Z(x,x) = Z(y,y) = 1$ for each $x\in X$ and $y\in Y$
(which reflects the fact that the vectors $\{u_x\}\cup\{v_y\}$ are unit
vectors).

The relaxation~\eqref{eq:sdp-xor} is very useful to prove upper bounds on the
bias of two-prover XOR games. 
For the case of the CHSH game (as described in Example~\ref{ex:chsh}), the
matrix $R$ is a $4\times 4$ matrix with both its $2\times 2$ diagonal blocks
equal to $0$, and each off-diagonal block equal to
\begin{equation}
  \frac{1}{8}\begin{pmatrix} 1 & 1 \\ 1 & -1\end{pmatrix}.
\end{equation}
The dual problem to~\eqref{eq:sdp-xor} is
\begin{equation}
  \op{SDP}^*(\op{CHSH}) = \inf_H \Tr(H),
\end{equation}
where the infimum is over all Hermitian matrices $H$ such that $H\geq R$. 
Using the fact that $R$ squares to $(1/32)\I$, one finds that 
$H = (\sqrt{2}/8)\I$ provides a dual certificate with objective value
$\sqrt{2}/2$.
Because this value is achieved by the strategy described earlier, it follows 
by weak duality that $\beta^*(\op{CHSH})=\op{SDP}(\op{CHSH})=\sqrt{2}/2$.

This is not a coincidence: the equality $\beta^*(V)=\op{SDP}(V)$ always holds.
There is an explicit mapping, due to Tsirelson \cite{Tsirelson87}, that shows
how any feasible solution to the semidefinite program (corresponding to the
operator $Z$ above) can be transformed into a strategy for the provers (a state
$\ket{\psi}$ and binary-valued measurements $\{P^x_0,P^x_1\}$ and
$\{Q^y_0,Q^y_1\}$) achieving a bias equal to the objective value given by $Z$
in~\eqref{eq:sdp-xor}.

The characterization of the bias of two-prover XOR games as the optimum of a
semidefinite program has multiple consequences. 
First, it allows one to replace the supremum in~\eqref{eq:opt-xor} by an
efficiently computable quantity. 
The inclusion of $\MIP^*\subseteq\EXP$ follows, as an explicit representation
of the matrix $R$ specifying an XOR game can be computed in exponential time
from a description of a quantum circuit for the verifier, and the optimum of
the resulting exponential-size semidefinite program can be approximated to
within exponential precision in time polynomial in its size.

A second noteworthy consequence is a bound on the entanglement of optimal
strategies in XOR games. 
The optimum of~\eqref{eq:sdp-xor} is always achieved by a matrix of dimension
$N = \abs{X} + \abs{Y}$, whose Gram factorization involves vectors of the same
dimension.
Tsirelson's transformation can be used to map these vectors onto a state and
two collections of measurements in which each prover holds
$\lfloor N/2\rfloor$ qubits.
Thus, for every XOR game there exists an optimal strategy that uses a number of
qubits linear in the number of questions in the game. 
This is not true of more general interactive games, as demonstrated for
instance by the game from Example~\ref{ex:dltw-1}, for which the entangled
value is only achieved in the limit as the dimension of the provers' shared
entangled state goes to infinity.

\section{Using entanglement in multi-prover games}
\label{sec:qmip-structure}

This section is devoted to the presentation of structural results, such as
parallelization and perfect completeness, that apply to the class $\QMIP^*$. 
Some of these results parallel similar properties known to hold for classical
multi-prover interactive proof systems. 
Proofs of the latter type of results, however, usually rely on the technique of
oracularization, which was shown to fail in the presence of entangled provers
in the preceding section. 
Thus, a direct extension of the classical results to the entangled-prover
setting is not generally possible, and different proofs must be devised. 

The reductions established in this section will make crucial use of
entanglement between the provers---it will typically be the case that, even if
honest unentangled provers could win with high probability in a certain
interactive game, provers in the modified game will nevertheless still need to
make use of prior shared entanglement in order to win with high probability.
In some cases, entanglement will be used to achieve reductions unlikely to hold
in the classical setting, such as a reduction to public-coin systems. 
The following properties will be shown:  
\begin{mylist}{\parindent}
\item[1.] \emph{Perfect completeness.}
  Multi-prover quantum interactive proof systems can be made perfectly
  complete.
\item[2.] \emph{Parallelization.}
  Multi-prover quantum interactive proof systems can be parallelized to three
  turns of interaction.
  Moreover, any three-turn multi-prover quantum interactive proof system can be
  transformed into one that is \emph{public-coin}: the verifier's unique
  message is a single random bit broadcast to all provers.
  In addition, public-coin proof systems can be further parallelized to only two
  turns of interaction by introducing an additional prover.
\item[3.] \emph{Classical verifiers.}
  Any multi-prover quantum interactive proof system can be transformed into one
  in which the verifier is classical at the cost of considering two additional
  provers and a polynomial increase in the number of rounds of interaction.
\end{mylist}
Putting these properties together, any $k$-prover quantum interactive proof
system can be transformed into a one-round proof system with $k+1$ provers,
perfect completeness, and soundness bounded away from $1$ by an inverse
polynomial.

\begin{theorem}\label{thm:qmipstar-parallel}
For every polynomially bounded functions $k$ and $m$ it holds that
\begin{equation}\label{eq:qmipstar-reductions}
  \QMIP^*(k,m) \subseteq \QMIP^{*}_{1,1-1/p}(k+1,2),
\end{equation}
for some choice of a polynomially bounded function $p = O(m^{2})$.
\end{theorem}

It is not known whether the soundness parameter of multi-prover interactive
proof systems with quantum verifiers can be amplified in parallel with the same
set of provers, a problem that will be discussed in
Section~\ref{sec:qmip-par-rep}.
Thus, amplifying the inverse-polynomial gap in completeness and soundness
from~\eqref{eq:qmipstar-reductions} requires a polynomial increase in either
the number of provers or the number of rounds of interaction
(q.v. Proposition~\ref{prop:amplify}).

Allowing for a polynomial number of rounds of interaction, the verifier can
further be made classical. 

\begin{theorem}\label{thm:ruv}
For all polynomially bounded functions $k$, $m$ and $q$, it holds that
\begin{equation}
  \QMIP^*(k,m) \subseteq \MIP^{*}_{1,2^{-q}}(k+2, p), 
\end{equation}
for some choice of a polynomially bounded function $p = O(q\cdot m^{2})$.
\end{theorem}

When comparing these results with those known to hold for unentangled
provers, there is a significant gap: 
to determine whether or not the number of provers be reduced.
There is currently no compelling evidence in favor of
$\QMIP^*(k,\poly)$ being a larger class than $\QMIP^*(2,\poly)$, but also
there is no known transformation allowing a reduction of the number of
provers.

\subsection{Perfect completeness}
\label{sec:qmip-complete}

The standard transformation to achieve perfect completeness for the class
$\MIP$ proceeds as follows.
Given a verifier $V$, a modified verifier $V'$ is defined that executes the
same procedure as $V$ many times in parallel with a carefully chosen set of
distinct private random strings. 
The strings are chosen so as to guarantee that, provided the provers had
successful strategies for at least half of the possible choices of a random
string for $V'$, there will always be at least one string in the set for which
the provers can convince the verifier $V$ to output $1$ with certainty, when
$V$ is executed with this choice of randomness.

In the case of a quantum verifier, this sort of transformation is 
meaningless---there is no discrete set of ``random bits'' for the verifier that
parametrizes its verification procedure and can be easily manipulated.
For this reason a different transformation is required.
The reduction to be described will be similar in spirit to the one introduced
in Section~\ref{sec:qip-completeness} for the single-prover case, with an
important twist. 
Recall that, in that transformation, during the last round of interaction, the
prover is required to apply a certain unitary transformation on its private
register to disentangle it from the message register (q.v. 
Eq.~\ref{eq:last-prover-transformation-perfect-completeness}). 
If the corresponding register in the multi-prover setting is shared between
multiple provers, it may not be possible for them to implement such a unitary
transformation locally. 
A more complicated transformation, which requires that the number of turns
in the game increases from $m$ to $3m$, will allow the provers to achieve
the desired effect: they execute the entire game backward ($m$ extra
turns), and then forward again ($m$ extra turns).
It is interesting to note that even if the original verifier was classical, and
the provers could achieve their maximum success probability without using any
prior entanglement, the new verifier will make use of quantum messages and in
general the provers may need to use prior entanglement in order to achieve the
optimal success probability of $1$ in the modified game.

Suppose a verifier $V$ in a quantum multi-prover interactive game is given,
along with a target threshold $\alpha \geq 1/2$ for its maximum acceptance
probability $\omega^*(V)$. 
We will describe a transformation mapping $V$ to a new verifier $V'$ such that
the following properties hold:
\begin{mylist}{\parindent}
\item[1.]
  If $V$ is an $m$-turn verifier, then $V'$ is a $3m$-turn verifier.
\item[2.]
  If it is the case that $\omega^*(V) \geq \alpha$, then $\omega^*(V') = 1$.
\item[3.]
  It always holds that
  $\omega^*(V') \leq 1/2+2\sqrt{\omega^*(V)}+5\omega^*(V)/2$.
\end{mylist}
By this transformation one may conclude that the following theorem holds.

\begin{theorem}\label{lem:qmip-complete}
Let $a\geq 1/2$ and $b \leq 1/25$.
For every choice of $k$ and $m$ it holds that
\begin{equation}
  \QMIP^*_{a,b}(k,m) \subseteq \QMIP^*_{1,c}(k, 3m),
\end{equation}
for $c = 1/2+2\sqrt{b}+5b/2$.
\end{theorem}

To explain the idea behind the reduction, it will be convenient to replace the
assumption $\omega^*(V) \geq \alpha$ in item 2 by the more specific requirement
that there exists a fixed strategy for the provers with the property that the
optimal success probability of this strategy, when maximized over all possible
initial states of the provers' private registers, is exactly $1/2$. 
This is easily achieved by allowing the provers to force a rejection in order
to artificially lower their success probability, along the same lines as was
discussed in the single-prover setting. 
It will also be convenient to assume that the first turn of the game is
executed by the verifier, sending a message to each of the provers.

Assuming $V$ is given in purified form, the construction of the $3m$-turn
verifier $V'$ can be described as follows. 

\begin{mylist}{\parindent}
\item[1.] $V'$ simulates $V$ for the first $m$ turns, up to
  but not including the final measurement of the output qubit of~$V$.
\item[2.] $V'$ chooses a bit $b\in\{0,1\}$ uniformly at random. If $b=0$ it
  executes the \emph{rewinding test} described in step 3.
  If $b=1$ it performs the \emph{invertibility test} described in step 4.
\item[3.] \emph{Rewinding test:}
\begin{mylist}{8mm}
\item[(i)]\label{step:rewind-1} $V'$ measures the output qubit of $V$. 
  If the result is $1$ it stops the game and outputs $1$.
  If it is $0$, the original interactive game is executed backward in time for
  $m$ turns, interacting with the provers as needed. 
  At the last step $V'$ applies $V_1^{-1}$ to the registers
  $(\reg{Z}_1,\reg{X}_1^1,\ldots,\reg{X}_1^k)$, obtaining
  $(\reg{Z}_0,\reg{Y}_0^1,\ldots,\reg{Y}_0^k)$.
\item[(ii)] $V'$ performs a controlled-phase flip $Z$, multiplying the phase by
  $-1$ if all the qubits in $\reg{Z}_0$ are in state $\ket{0}$.
\item[(iii)] $V'$ executes the original interactive game forward in time for
  $m$ turns.
  It measures the output qubit of $V$ and returns the outcome.
\end{mylist}
\item[4.] \emph{Invertibility test:}
  \begin{mylist}{8mm}
    \item[(i)]
      $V'$ executes the original interactive game backward in time for $m$
      turns. 
    \item[(ii)]
      After applying $V_1^{-1}$ it applies the measurement
      $\{\Pi_{\mathit{init}},\I-\Pi_{\mathit{init}}\}$, which measures all
      qubits of register $\reg{Z}_0$ in the computational basis. 
      If the outcome associated with $\Pi_{\mathit{init}}$, corresponding to
      all qubits being in the $\ket{0}$ state, is obtained $V'$ returns the
      outcome $1$; otherwise it returns $0$.
  \end{mylist}
\end{mylist} 

Consider first the case where there exist $k$ provers $P^1,\ldots,P^k$ such
that, with the optimal choice of initial state of their private registers
$\reg{W}_0^1,\ldots,\reg{W}_0^k$, the provers cause $V$ to accept with
probability exactly $1/2$. Define new provers $R^1,\ldots,R^k$ who perform
precisely the same actions as the original provers when asked by the verifier,
including performing the reverse action when asked to do so. It is clear that
such provers will always cause the verifier to output $1$ with certainty in the
invertibility test. That they also cause the verifier to output $1$ with
certainty in the rewinding test follows from a similar analysis as was
performed in the single-prover case in Section~\ref{sec:qip-completeness}.

Now suppose the provers' maximum probability to convince $V$ to accept is less
than $1/25$.
Let $R^1,\ldots,R^k$ be arbitrary provers in the interactive game specified by
$V'$, and let $\ket{\psi}$ be the initial state of all parties' private
registers, including the provers' shared entanglement, at the beginning of the
game. 
We may introduce three unitary operators that capture the actions performed
jointly by the verifier and provers in the \emph{forward}, \emph{backward}, and
\emph{forward} phases of the game.
Unitary $U_1$ implements all parties' actions in the \emph{forward} phase,
including the verifier's last unitary operation $V_{m/2+1}$, but without
measuring the output qubit. 
Unitary $U_2$ implements all parties' actions in the \emph{backward} phase, 
starting with the verifier's application of $V_{m/2+1}^{\ast}$ and ending with
$V_1^{\ast}$.
Finally, unitary $U_3$ implements all parties' actions in the second
\emph{forward} phase. 
For the verifier, these are the same transformations that were used in the
definition of $U_1$, but in general the provers' actions may be different.

With respect to the operators $U_1$, $U_2$, and $U_3$ just defined, the
provers' success probability in the game may be characterized as follows. 
Let $p_1/2$, where $p_1 = \|\Pi_1 U_1 \ket{\psi}\|^2$, be the probability that
the verifier stops and accepts in step~3(i).
It holds that $p_1\leq \omega^*(V)$. 
The probability that the verifier stops and accepts in step 3(iii) is $p_2/2$,
where
\begin{equation}
  p_2 = \|\Pi_1 U_3ZU_2(\I-\Pi_1)U_1\ket{\psi}\|^2.
\end{equation}
Finally let $p_3=\|\Pi_{\mathit{init}} U_2U_1 \ket{\psi}\|^2$, so that the
probability that the verifier stops and accepts in step~4 is $p_3/2$. 

The value $\omega^*(V') = (p_1+p_2+p_3)/2$ is bounded by expressing a tradeoff
between $p_2$ and $p_3$. 
Either $U_2$ is such that the combined unitary $U_2U_1$ brings the state of all
registers into one that is consistent with a possible initial state of the game
specified by $V$. 
In this case the invertibility test will accept, but $p_2$ will be bounded by
$\omega^*(V)$. 
Alternatively, the provers' actions in the \emph{backwards} phase of the
game are such that the combined action $U_2U_1$ results in a state in which
the verifier's private register is not in the $\ket{0}$ state, in which case
the invertibility test will reject and $p_3$ will be small. 
This tradeoff can be expressed by applying the triangle inequality as follows.
\begin{equation}
  \begin{aligned}
    \sqrt{p_2} & = \|\Pi_1 U_3ZU_2(\I-\Pi_1)U_1\ket{\psi}\| \\
    & \leq \|\Pi_1 U_3ZU_2\Pi_1 U_1\ket{\psi}\| 
    + \|\Pi_1 U_3ZU_2 U_1\ket{\psi}\| \\
    & \leq \sqrt{\omega^*(V)} + \|\Pi_1 U_3Z (\I-\Pi_{\mathit{init}})U_2
    U_1\ket{\psi}\| \\
    & \qquad +\|\Pi_1 U_3Z \Pi_{\mathit{init}}U_2 U_1\ket{\psi}\|\\
    & \leq \sqrt{\omega^*(V)} + \sqrt{1-p_3} + \sqrt{\omega^*(V)},
  \end{aligned}
\end{equation}
where the last term is bounded by using the fact that $\Pi_{\mathit{init}}U_2
U_1\ket{\psi}$ is a valid initial state for the game specified by $V$, and can
therefore not lead to a higher acceptance probability than $\omega^*(V)$.

Putting everything together, one has
\begin{equation}
  \begin{aligned}
    \omega^*(V') & \leq \frac{1}{2}\Bigl( \omega^*(V) +
    \Bigl(2\sqrt{\omega^*(V)}+\sqrt{1-p_3}\Bigr)^2 + p_3\Bigr)\\
    & \leq \frac{1}{2} \Bigl( \omega^*(V) +
    \Bigl(1+4\sqrt{\omega^*(V)}+4\omega^*(V)-p_3\Bigr) + p_3\Bigr)\\
    & = \frac{1}{2} + 2\sqrt{\omega^*(V)} + \frac{5}{2}\omega^*(V),
  \end{aligned}
\end{equation}
as desired.

\subsection{Parallelization and public-coin systems}
\label{sec:qmip-pub}

Classical multi-prover interactive proof systems can be parallelized to a
single round of interaction by using the oracularization technique. 
Starting from a verifier $V$ in an $m$-turn interactive game, the two-turn
verifier $V'$ selects a random string $r$ representing the private random bits
of $V$ and asks a first prover to provide a complete transcript, including all
messages that would have been exchanged between the verifier and all provers,
for the execution of the $m$-turn game using the random string $r$. 
The other provers are used to check that the transcript is one that could
indeed have arisen in the original game, and in particular that messages from
the provers in a certain turn, as described in the transcript, do not depend on
messages sent by the verifier in subsequent turns. 
To check this condition, the provers are only given access to those random bits
that determine messages from $V$ sent in the first $t$ turns, where 
$1\leq t \leq m$ is randomly chosen; they are asked for a transcript of the
game until that round. 
The verifier $V'$ checks the transcripts received from the provers for
consistency. 
The same transformation allows for a reduction of the number of provers to
two.

For quantum interactive games the notion of a transcript is ill-defined, a
difficulty that was already encountered in the single-prover settings discussed
in Chapters~\ref{chapter:single-prover} and \ref{chapter:QSZK}.
As was also discussed previously, the technique of oracularization will in
general not apply even to classical verifiers in the presence of 
prior shared entanglement between the provers.
 
Fortunately, it turns out that the same transformation used to parallelize
single-prover quantum interactive proof systems does extend to the multi-prover
setting. 
Recall that this transformation requires the provers to start the interaction
in the state they would be in halfway through the original interaction,
proceeding either \emph{forward} or \emph{backward} depending on a coin-flip
made by the verifier. 
The same idea can be applied to multiple provers, who will have no more
latitude to cheat than in the single-prover setting. 
As for the transformation achieving perfect completeness, honest provers may be
required to use entanglement in order to succeed in the modified proof system,
irrespective of whether it is required by honest provers in the original proof
system.
This is because the joint state of their message registers halfway through the
original game may contain entanglement generated by the verifier's messages. 
The consequence of this transformation for interactive proof systems is stated
in the following theorem.

\begin{theorem}\label{lem:qmip-parallel}
  For all polynomially bounded functions $k$ and $m \geq 4$, and for every
  function $\eps:\natural\to[0,1]$, it holds that
  \begin{equation}
    \QMIP^*_{1,1-\eps}(k,m)\subseteq \QMIP^*_{1,1-\delta}(k, 3)
  \end{equation}
  for $\delta=\eps/m^2$.
\end{theorem}

Even if the verifier in the original $m$-turn game is classical, the same
transformation will require a quantum verifier to execute the three-turn game,
and it is not known if a similar transformation can be performed while keeping
the verifier classical.
If the soundness property of the original proof system is only known to hold
against unentangled provers (that is, the proof system is a $\QMIP$ proof
system), the reduction will not apply, as the provers must be able to use shared
entanglement in order to succeed even in the honest case. 
To handle this case, the quantum verifier would first have to be simulated by a
classical verifier through the circuitous route described in
Section~\ref{sec:qmip-nexp} (involving encoding the problem decided by the
$\QMIP$ verifier as an $\MIP$ problem, and going through the constructions
proving $\QMIP\subseteq\NEXP\subseteq\MIP$).
The resulting classical verifier can be parallelized to a single round using
oracularization as described earlier.
It is an open question whether $\QMIP$ systems can be directly parallelized to
a single round of interaction with a classical verifier and without requiring
the addition of a prover.

Quantum multi-prover interactive proof systems can be further parallelized to a
single round (two turns) of interaction by introducing an additional prover. 
The transformation proceeds in two steps, each of which is of interest in its
own right. 
The first step establishes that any three-turn verifier can be transformed into
one whose single message to each prover consists of a uniformly random bit
broadcast simultaneously to all provers. 
This public-coin form is unique to quantum multi-prover games, and it is
unlikely to be achievable for classical games: because the verifier's message
to all provers is publicly known, all provers receive the same information and
can coordinate their actions perfectly.
Thus $\MIP^{\textrm{pub}}$, the public-coin variant of $\MIP$, collapses to the
single-prover class $\IP$, and $\MIP=\MIP^{\textrm{pub}}$ would imply
$\PSPACE=\NEXP$.

What makes this result possible in the case of quantum provers is that, even
upon receiving the same message, the provers are still restricted to applying a
local transformation on their respective registers.
The verifier thus has the guarantee that the joint states of the message
registers that could be sent by the provers in the third turn
are related by the action of a quantum channel in tensor product form.
This guarantee turns out to be sufficient to establish soundness of the
public-coin proof system.

Given a three-turn verifier $V=(V_1,V_2)$, a three-turn public-coin verifier
$V'$ can be constructed as follows.
\begin{mylist}{\parindent}
\item[1.] $V'$ receives message register $\reg{Z}_1$ from the first prover,
  and nothing from the other provers.
\item[2.] $V'$ chooses $c\in\{0,1\}$ uniformly at random and broadcasts it to
  all provers.
\item[3.] $V'$ receives register $\reg{Y}^i_1$ from the $i$-th prover, for
  $i=1,\ldots,k$.
\begin{mylist}{8mm} 
 \item[(i)]
    If ${c=0}$, $V'$ applies $V_2$ to the qubits in
    $(\reg{Z}_1,\reg{Y}_1^1,\ldots,\reg{Y}_1^k)$, and outputs the outcome of
    the measurement performed by $V_2$.
  \item[(ii)]
    If ${c=1}$, $V'$ applies $V_1^{\ast}$ to the qubits in
    $(\reg{Z}_1,\reg{Y}_1^1,\ldots,\reg{Y}_1^k)$ and produces the output $1$
    if and only if all the qubits in $\reg{Z}_0$ are in state $\ket{0}$.
 \end{mylist}
\end{mylist}
The analysis of the completeness and soundness properties of $V'$ follows along
the same lines as the analysis of the min-max formulation of the value of a
$\QIP(3)$ interactive game given in Section~\ref{sec:QIP-reduction}. 
The consequence for interactive proofs is stated in the following theorem.

\begin{theorem}\label{lem:qmip-public}
  For every polynomially bounded function $k$ and every function
  $b:\natural\to[0,1]$, it holds that
  \begin{equation}
    \QMIP^*_{1,b}(k,3)\subseteq\QMIP^{*,\textrm{pub}}_{1,(1+\sqrt{b})/2}(k,3),
  \end{equation}
  where $\QMIP^{*,\textrm{pub}}_{a,b}(k, m )$ is the class of promise problems
  having quantum $k$-prover $m$-turn interactive proof systems in which all the
  verifier's messages to the provers are public coins.
\end{theorem}
  
In the second step of the parallelization procedure it is shown how any
three-turn public-coin verifier $V$ interacting with $k$ provers can be
transformed into a two-turn verifier $V'$ (no longer public-coin) interacting
with $k+1$ provers.
Because $V$ is public-coin, we may assume that its first action $V_1$ consists
of generating uniformly random bits and broadcasting them to the provers. 
At the end of the game the verifier applies a unitary $V_2$ to the joint
state formed by the provers' message registers $\reg{Y}_1^1,\ldots,\reg{Y}_1^2$
received in the first turn, $\reg{Y}_2^1,\ldots,\reg{Y}_2^k$ received in the
third turn, and its own private register. 
Define the new verifier $V'$ as follows:
\begin{mylist}{\parindent}
\item[1.] $V'$ broadcasts public coins to the first $k$ provers exactly as $V$
  would. No message is sent to the $(k+1)$-st prover.
\item[2.] $V'$ applies the unitary $V_2$ to the messages received, treating
  the first $k$ provers' answer registers as if they contained the provers'
  messages in the second turn of the original game, and the $(k+1)$-st prover's
  message as if it contained the joint state of all provers' messages in the
  first turn of the original game. 
  $V'$ then measures the output qubit and produces the outcome.
\end{mylist}

First we claim that for any strategy for the provers $P^1,\ldots,P^k$ in the
interactive game specified by $V$ there exists a strategy for the provers
$R^1,\ldots,R^{k+1}$ with the same probability of being accepted by $V'$. To
achieve this $R^1,\ldots,R^k$ can simulate the actions of $P^1,\ldots,P^k$ in
the first turn of their interaction with $V$, handing over their joint message
registers to $R^{k+1}$ before the interaction with $V'$ starts (which is
allowed as part of the provers' prior shared entanglement).
When the game specified by $V'$ is initiated, $R^{k+1}$ sends all its
registers to $V'$ and $R^1,\ldots,R^k$ continue as if they were
$P^1,\ldots,P^k$ interacting with $V$. 
The new provers' probability of being accepted by $V'$ is
identical to the original provers' probability of being accepted by $V$.

Conversely, fix a strategy for provers $R^1,\ldots,R^{k+1}$ in an interaction
with $V'$ and define a strategy for $P^1,\ldots,P^k$ that has the same
probability of being accepted by $V$ as follows. $P^1,\ldots,P^k$ initialize
their private registers exactly as $R^1,\ldots,R^k$ would, except that for each
$i\in\{1,\ldots,k\}$ prover $P^i$ is also given the register
$\reg{W}_1^{k+1,i}$ sent by $R^{k+1}$ to $V'$ that would have been interpreted
as message register $\reg{Y}_1^i$ by $V'$ in the game. In the first turn of
their interaction with $V$ each prover sends $\reg{W}_1^{k+1,i}$. In the second
turn they behave exactly as $R^1,\ldots,R^{k}$ would have in their interaction
with $V'$. 
Once again, the probability of $P^1,\ldots,P^k$ being accepted by $V$ is
identical to the probability of $R^1,\ldots,R^{k+1}$ being accepted by $V'$.

Through the transformation just described, one concludes the following theorem.

\begin{theorem}\label{lem:qmip-oneround}
  For every polynomially bounded function $k$ and all functions
  $a,b:\natural\to[0,1]$ such that $a>b$, it holds that
\begin{equation}
  \QMIP_{a,b}^{*,\textrm{pub}}(k,3) 
  \subseteq \QMIP^{*,\textrm{pub}}_{a,b}(k+1,2).
\end{equation}
\end{theorem}

Combining Theorem~\ref{lem:qmip-complete}, Theorem~\ref{lem:qmip-parallel}, Theorem~\ref{lem:qmip-public} and Theorem~\ref{lem:qmip-oneround} proves Theorem~\ref{thm:qmipstar-parallel}.

\subsection{Classical verifiers}
\label{sec:qmip-verifier}

In Section~\ref{sec:qmip-nexp} it is argued, albeit rather indirectly, that
quantum verifiers interacting with multiple unentangled provers are no more
powerful than their classical counterparts. In the presence of entangled
provers it may seem that the possibility for the verifier to exchange quantum
messages is essential, and indeed this is the case for some of the reductions
discussed in the preceding section. 
Nevertheless, it is still the case that any quantum multi-prover interactive
proof system with entangled provers can be transformed into one in which the
verifier is classical, provided the number of provers is allowed to increase by
two and the number of rounds to a polynomial.
This fact was stated as Theorem~\ref{thm:ruv} earlier in this section.

The reduction from quantum to classical verifiers that underlies the theorem
just mentioned is highly non-trivial. 
Its completeness requires honest provers to share polynomially many qubits of
entanglement, and its soundness rests on the property of 
\emph{entanglement rigidity}. 
Informally speaking, this property states that certain correlations generated
by the provers, as witnessed by a high success probability in certain
interactive games (such as the CHSH game, Example~\ref{ex:chsh}), are
\emph{rigid} in the sense that they can only be obtained by performing
measurements on a specific entangled state, up to local isometries that
could be performed by the provers.
(In the case of the CHSH game, this state is an EPR pair.)
Rigidity can be leveraged by the verifier to exert a tight control over the
provers' actions: by verifying that they are able to successfully play the CHSH
game, it is possible to assert that, up to local isometries acting on their
private registers, the provers share an EPR pair on which they apply specific
measurements.

Using additional ideas, it is possible to devise a proof system whereby a
classical verifier $V'$ is able to ``orchestrate'' $k+2$ provers
$R^1,\ldots,R^{k+2}$, using only classical messages, so as to reproduce any
polynomial-time interaction between a quantum verifier $V$ and $k$ provers.
In this orchestration, one of the additional provers, say $R^{k+1}$, plays the
role of $V$, and the other, $R^{k+2}$, is used to control the actions of
$R^{k+1}$ via a form of distributed process tomography. 
The original proof system may call for quantum messages to be exchanged between
$V$ and the provers.
In the new proof system, such messages are simulated via teleportation between
$R^{k+1}$ and the first $k$ provers, where $V'$ uses classical messages to
relegate the required correction bits between the provers. 
The EPR pairs used for teleportation are tested by executing a sufficiently
large numbers of CHSH games in sequence and verifying that the provers achieve
a success rate close to the optimal $\omega^*(\CHSH)$. 
This large number of CHSH games leads to a polynomial blow-up in the number of
rounds of interaction of $V'$, even if $V$ is single-round.

\section{\texorpdfstring{Containment of \class{NEXP} in 
    $\class{QMIP}^*$}{Containment of NEXP in QMIP*}}
\label{sec:nexp-in-mipstar}

As was previously discussed, the introduction of entanglement between provers
can sometimes give them a significant advantage in a multi-prover interactive
game.
As a result, it is not immediately clear that the complexity class
$\QMIP^*$ is larger than the single-prover class $\QIP$, as the soundness
property of the multi-prover interactive proof system constructions that
establish $\NEXP \subseteq \MIP$ could be compromised by entanglement between
the provers.
Indeed, as was mentioned in Section~\ref{sec:xor}, a collapse of this sort does
occur for the restricted case of XOR proof systems:
$\oplus\MIP^* \subseteq \PSPACE$, while $\oplus\MIP = \NEXP$.
The following theorem shows that this does not happen in the more general
setting (assuming $\PSPACE\not=\NEXP$), and more precisely that quantum
interactive proof systems with entangled provers are powerful enough to decide
all problems in nondeterministic exponential time.

\begin{theorem} \label{thm:nexp-mipstar}
  Every language in $\NEXP$ has a three-prover one-round interactive proof
  system in which completeness $a=1$ can be achieved by unentangled provers,
  and soundness $b=1/2$ holds against entangled provers. 
  In particualr, it holds that
  \begin{equation}
    \label{eq:nexp-in-mipstar}
    \NEXP\subseteq \QMIP^*.
  \end{equation}
\end{theorem}

Additional properties of verifiers that establish the containment 
\eqref{eq:nexp-in-mipstar} are also known.
For instance, the verifier can be taken to be classical, to send messages to 
two out of three provers chosen at random, and to receive a number of bits 
from each prover that scales as $O(\log(1/b))$, where $b$ is the desired
soundness parameter. 
If one is willing to relax the condition of perfect completeness, the inclusion
\begin{equation}
  \NEXP\subseteq\oplus\MIP^*_{1-\eps,1/2+\delta}(3,2)
\end{equation}
(three-prover one-round XOR games) is also known to hold for any choice of
constants $\delta,\eps>0$.

In this section we sketch some of the ingredients that go into the proof of
Theorem~\ref{thm:nexp-mipstar}. 
The starting point is the proof system introduced by Babai, Fortnow, and
Lund~\cite{BabaiFL91} in their proof of $\NEXP\subseteq\MIP$.
This proof system has two main components, both of which need to be made
``entanglement-resistant,'' meaning that their soundness guarantee can be
extended to hold against entangled provers.

The first component is the technique of oracularization. 
This is used in combination with the technique of arithmetization (introduced
for the proof of $\IP=\PSPACE$) to devise a basic two-prover proof system for 
a certain $\NEXP$-complete problem. 
As discussed earlier, oracularization fails in general with entangled provers. 
In Section~\ref{sec:three-prover} two workarounds are described that establish 
a weaker form of oracularization with entangled provers, first by using three,
and then two, provers.

The second component is an interactive game called the 
\emph{multilinearity test}. 
This test is used as a means to enforce that the provers' answers are
determined according to a multilinear function of the message received from the
verifier, which is interpreted as a point $x\in\field^m$ for some 
large finite field $\field$. 
The statement and analysis of a multilinearity test with entangled provers
requires care, and the main ideas are discussed in the simpler context of the
linearity test in Section~\ref{sec:linearity-test}. 
In Section~\ref{sec:nexp-in-qmipstar} the two components are combined into a
brief sketch of the proof of Theorem~\ref{thm:nexp-mipstar}.

\subsection{Games with three provers and monogamy of entanglement}
\label{sec:three-prover}

Consider the following (apparently trivial) modification of the
oracularization technique. 
Given a verifier $V$ specifying a two-prover one-round interactive game,
define a three-prover one-round verifier $V'$ as follows. 
At the start of the game, $V'$ selects a permutation of the three provers
uniformly at random and assigns them labels \emph{Alice}, \emph{Bob}, and
\emph{Charlie}.
$V'$ then plays the two-prover game specified by $V$ with the provers that were
designated as Alice and Bob, ignoring the prover designated as Charlie.
Each prover is assigned its name when it is sent its first message.

If the provers employ classical deterministic (or even randomized) strategies,
the presence of Charlie makes no difference whatsoever; and so it holds that
$\omega(V)=\omega(V')$. 
This equality no longer holds with entangled provers. 
The reason is related to a property of entanglement called 
\emph{entanglement monogamy}. 
Informally speaking, monogamy states that there exist strong correlations that
can be realized between two parties sharing entanglement that cannot be 
extended to three or more parties.
For instance, three qubits cannot be in a state in which each pair of qubits
forms an EPR pair.
Thus, while a random string shared between two parties can just as easily be
shared among three, a bipartite entangled state cannot in general be extended
to a tripartite state reproducing the bipartite correlations among any subset
of two out of three of the parties.

\begin{example}[Three-prover CHSH game $\CHSH_3$]\label{ex:3prover-chsh}
  The transformation described above can be applied to the CHSH game,
  which was presented in Example~\ref{ex:chsh}. 
  In the new, three-prover variant of this game, the verifier selects two
  provers at random to play the roles of Alice and Bob, and sends them
  questions as in the CHSH game.
  The third prover is ignored. 
  It is evident that the classical value of this game coincides with that of the
  two-prover variant: $\omega(\CHSH_3)=\omega(\CHSH) = 3/4$. 
  What is perhaps more surprising is that the entangled value is no larger:
  $\omega^*(\CHSH_3)=\omega(\CHSH_3)=3/4$.\footnote{
    Indeed, it holds that the so-called \emph{no-signaling} value of this game
    is 3/4, which can be proved through the use of linear programming.
    As the no-signaling value upper-bounds the entangled value, it follows that
    $\omega^*(\CHSH_3)\leq 3/4$.}
  This fact is representative of the monogamy of quantum correlations.
\end{example}

Figure~\ref{figure:three-clause-versus-variables} describes
a three-prover variant of the clause-versus-variable game introduced in
Section~\ref{sec:xor}, demonstrating further the monogamy phenomenon.
\begin{figure}
  \noindent\hrulefill
  \begin{trivlist}
  \item
    The input is a collection of constrains 
    $\varphi = (C_1,\ldots,C_m)$ on variables $x_1,\ldots,x_n$,
    where each constraint $C_j$ is an arbitrary constraint involving at most
    $\ell$ of the variables, each ranging over a finite alphabet $\Gamma$.
  \end{trivlist}
  \begin{mylist}{5mm}
  \item[1.]
    Select a random permutation of the three provers, and name the first 
    \emph{Alice}, the second \emph{Bob}, and the third \emph{Charlie}.
  \item[2.]
    Run the $2$-prover clause-versus-variable verifier $V_\varphi$ with
    Alice and Bob, ignoring Charlie.
  \end{mylist}
  \noindent\hrulefill
  \caption{$3$-prover clause-versus-variable verifier $T_\varphi$.}
  \label{figure:three-clause-versus-variables}
\end{figure}
The following analogue of~\eqref{eq:clause-var-classical} can be established
for the verifier $T_\varphi$ described in
Figure~\ref{figure:three-clause-versus-variables}:
there exists a constant $c>1$ such that, for all $\varphi$,
\begin{equation}\label{eq:3prover-quant}
 \omega(\varphi) \,\leq\,\omega^*(T_\varphi) \,\leq\, 
 1- \biggl(\frac{1-\omega(\varphi)}{n}\biggr)^c.
\end{equation}

The remainder of this section is devoted to a proof of the implication
$[\omega^*(T_\varphi)=1]\Rightarrow[\varphi\;\text{is satisfiable}]$.
(The converse implication is immediate.) 
Taking the contrapositive, this statement already implies that if
$\omega(\varphi)<1$ it must also be that $\omega^*(\varphi)<1$. 
The quantitative bound provided by the second inequality
in~\eqref{eq:3prover-quant} can be derived using the same proof outline, but
requires substantially more technical work to keep track of the losses incurred
in all inequalities.

Applying the bound~\eqref{eq:3prover-quant} to an exponential-sized family of
constraints determining membership in an $\NEXP$-complete language yields the
following complexity-theoretic consequence:
\begin{equation}
  \NEXP\subseteq\MIP^*_{1,1-2^{-\mathrm{poly}}}(3,2).
\end{equation}
Although the inclusion is non-trivial, the exponentially small gap between the
completeness and soundness parameters is too small to be amplified by any
efficient method.
This small gap is a consequence of the dependence on the number of variables
$n$ in the right-hand side of~\eqref{eq:3prover-quant}, which is exponential in
the input size for an $\NEXP$-complete language.
For the case of the unentangled value $\omega(T_\varphi)$, as seen
from~\eqref{eq:clause-var-classical}, there is no such dependence. 
For the entangled value it is not known if some dependence on $n$ is
necessary.

The analysis of the entangled value of $T_\varphi$ rests on a stand-alone 
\emph{consistency test}. 
Let $X$ and $\Gamma$ be finite sets and let $\pi$ be a distribution on $X$.
The test only requires two provers, but its analysis extends to the case where
it is played in the presence of additional (passive) provers.
 
\begin{center}
  \underline{Consistency test $\cons(X,\Gamma,\pi)$}
\end{center}
\noindent Given finite sets $X$ and $\Gamma$, and a distribution $\pi$
on $X$, perform the following steps:
\begin{mylist}{8mm}
\item[1.] 
  Choose $x\in X$ according to $\pi$, and send $x$ to two provers.
\item[2.]
  Receive answers $a,b\in\Gamma$ respectively. Accept if and only if $a=b$.
  \vspace{2mm}
\end{mylist}

A strategy for the provers in $\cons(X,\Gamma,\pi)$ can be described succinctly
by specifying an initial shared entangled state $\ket{\psi}\in\V\otimes\W$
and measurements $\{P^x_a\,:\,a\in\Gamma\}$ and $\{Q^x_b\,:\,b\in\Gamma\}$, 
for every $x\in X$, corresponding respectively to the first and second provers'
measurements upon receiving message $x$ from the verifier.
The properties of the test are summarized in the following lemma. 

\begin{lemma}
  \label{lem:constest} 
  Suppose a strategy for the provers, specified by measurements
  $\{P^x_a\,:\,a\in\Gamma\}$ and $\{Q^x_b\,:\,b\in\Gamma\}$ and a shared
  entangled state $\ket{\psi}\in\V\otimes\W$, succeeds with probability $1$ in
  the game~$\cons(X,\Gamma,\pi)$.
  For every $x\in X$ such that $\pi(x)>0$, and for all $a\in\Gamma$, it holds
  that
  \begin{equation}
    \label{eq:constest-0}
    \begin{multlined}
      \Bigl((P^x_a)^2 \otimes \I\Bigr)\ket{\psi}
      = \bigl(P^x_a\otimes\I\bigr)\ket{\psi}\\
      = \bigl(\I\otimes Q^x_a\bigr)\ket{\psi}
      = \Bigl(\I\otimes (Q^x_a)^2\Bigr)\ket{\psi}.
    \end{multlined}
  \end{equation}
\end{lemma}

\begin{proof}
  For each $x\in X$ satisfying $\pi(x) > 0$, define vectors
  \begin{equation}
    \begin{aligned}
      v_x & = \sum_{a\in\Gamma}\, \ket{a} \otimes
      \bigl(P^x_a\otimes\I\bigr)\ket{\psi},\\
      w_x & = \sum_{a\in\Gamma}\, \ket{a} \otimes
      \bigl(\I \otimes Q^x_a\bigr)\ket{\psi},
    \end{aligned}
  \end{equation}
  and observe that
  \begin{equation}
    \ip{w_x}{v_x} = \sum_{a\in \Gamma} 
    \bra{\psi} P^x_a \otimes Q^x_a \ket{\psi},
  \end{equation}
  which is a nonnegative real number in the interval $[0,1]$ for every choice
  of $x\in X$.
  One has that
  \begin{equation}
    \label{eq:v_x-norm-bound}
    \begin{multlined}
      \norm{v_x}^2 
      = \sum_{a\in\Gamma} \bigbra{\psi} (P^x_a)^2 \otimes \I \bigket{\psi}
      \leq \sum_{a\in\Gamma} \bigbra{\psi} P^x_a \otimes \I \bigket{\psi} = 1,
    \end{multlined}
  \end{equation}
  where the inequality holds by virtue of the fact that
  $0\leq P^x_a \leq \I$
  for each $x\in X$ and $a\in \Gamma$, and the second equality follows from the
  assumption that $\{P^x_a\}$ is a measurement.
  Along similar lines, one finds that $\norm{w_x}\leq 1$.

  Now, under the assumption that the strategy succeeds with probability 1, it
  must hold that
  \begin{equation}
    1 = \sum_{x\in X} \pi(x) \sum_{a\in\Gamma}
    \bra{\psi} P^x_a \otimes Q^x_a \ket{\psi}
    = \sum_{x\in X} \pi(x) \ip{w_x}{v_x},
  \end{equation}
  and therefore $\ip{w_x}{v_x} = 1$ for every $x\in X$ satisfying $\pi(x) > 0$.
  By the equality condition of the Cauchy--Schwarz inequality, one finds that
  $w_x = v_x$ for every $x\in X$ satisfying $\pi(x) > 0$, and moreover these
  vectors must all be unit vectors.
  Consequently
  \begin{equation}
    \bigl(P^x_a\otimes\I\bigr)\ket{\psi} = 
    \bigl(\I \otimes Q^x_a\bigr)\ket{\psi}
  \end{equation}
  for each $a\in \Gamma$, again for each $x\in X$ satisfying $\pi(x) > 0$.

  Finally, observing the equivalence of the statements
  \begin{trivlist}
  \item (i) $\bra{\psi} R \otimes \I \ket{\psi} = 
    \bra{\psi} S \otimes \I \ket{\psi}$ and
  \item (ii)
    $(R \otimes \I) \ket{\psi} = (S \otimes \I) \ket{\psi}$
  \end{trivlist}
  for all choices of positive semidefinite operators $0 \leq R \leq S$,
  along with the fact that the inequality in \eqref{eq:v_x-norm-bound} must be
  an equality, one may conclude that
  \begin{equation}
    \Bigl((P^x_a)^2 \otimes \I\Bigr)\ket{\psi}
    = \bigl(P^x_a\otimes\I\bigr)\ket{\psi}
  \end{equation}
  for every $a\in \Gamma$ and $x\in X$ satisfying $\pi(x) > 0$.
  The equality
  \begin{equation}
    \bigl(\I\otimes Q^x_a\bigr)\ket{\psi}
    = \Bigl(\I\otimes (Q^x_a)^2\Bigr)\ket{\psi}
  \end{equation}
  is proved through a similar methodology.
\end{proof}

One useful consequence of this lemma is that if the reduced density operator of
$\ket{\psi}$ on the first prover's subspace has full support, then for any $x$
such that $\pi(x)>0$ the measurement $\{P_a^x\}$ is a projective measurement.
A similar conclusion holds whenever the reduced density operator of
$\ket{\psi}$ on the second prover's subspace has full support.

Suppose now that $\omega^*(T_\varphi)=1$, so there is a strategy for the
provers in the clause-versus-variable game that succeeds with probability $1$.
The strategy is specified by a set of \emph{Alice-measurements} and
\emph{Bob-measurements},
\begin{equation}
  \bigl\{P^j_{b_1,\ldots,b_{\ell}}\,:\,b_1,\cdots b_{\ell}\in\Gamma\bigr\}
  \quad\text{and}\quad \bigl\{Q^i_a\,:\,a\in\Gamma\bigr\},
\end{equation}
respectively.
The Alice-measurements result in a variable setting $b_1,\ldots,b_{\ell}$ for
the $\ell$ variables appearing in each constraint $C_j$, while the
Bob-measurements result in a variable setting $a$ for each variable $x_i$.
Because the game treats all three provers symmetrically, it is possible to
argue that there is no loss of generality in taking $\ket{\psi}$ to be 
invariant under all permutations of its three registers, and one may also
assume that it has full support on each provers' register. 
Similarly, one may assume that each prover always performs the same measurement
when given the same name (i.e., Alice or Bob) by the verifier.
Finally, one may assume that all of the measurements are projective
measurements.

The goal is to show that $\varphi$ is satisfiable. 
Define a distribution on assignments to the $n$ variables as follows:
\begin{equation}\label{eq:multi-prob-def}
  p(a_1,\ldots,a_n) \,=\, 
  \Bignorm{\Bigl(\I\otimes Q^n_{a_n}\cdots Q^1_{a_1}\otimes\I\Bigr)
    \ket{\psi}}^2.
\end{equation}
This is the distribution one would obtain by sequentially applying the
Bob-measurements, in the order $i = 1,\ldots,n$, to the second prover's
register of $\ket{\psi}$.
This, of course, is not what the second prover does---the distribution $p$ is
only being defined in this way for the sake of the analysis.
It must also be stressed that because the Bob-measurements do not necessarily
commute, it is not immediate that this distribution is consistent with any
prover's single Bob-measurement for a selected variable (except for the
measurement  associated with the variable $x_1$, which is the first measurement
applied in~\eqref{eq:multi-prob-def}).

It will be shown that, if the strategy succeeds with probability $1$, any
assignment in the support of $p$ must satisfy all of the constraints. 
As there must be at least one assignment in the support of $p$, this will
imply that $\varphi$ is satisfiable.
Toward this goal, for each $j \in \{1,\ldots,m\}$, define 
$q_j:\Gamma^{\ell}\rightarrow[0,1]$ to be the marginal probability distribution
on the possible assignments to the $\ell$ variables appearing in the constraint
$C_j$ that is obtained from the distribution $p$.
We will prove that
\begin{equation}
  \label{eq:marginal-goal}
  q_j(b_1,\ldots,b_{\ell}) = \Bignorm{\bigl(P^j_{b_1,\ldots,b_{\ell}}\otimes
    \I\otimes\I\bigr)\ket{\psi}}^2.
\end{equation}
Note that this will suffice to complete the proof.
In greater detail, because the provers are assumed to cause the verifier to
accept with certainty, Alice's answers always satisfy the clause she was asked
about, and therefore \eqref{eq:marginal-goal} implies that every assignment in
the support of $p$ satisfies $C_j$.
As this is so for all $j$, it must hold that every assignment in the support
of $p$ satisfies all of the constraints, as required.

It therefore remains to prove \eqref{eq:marginal-goal}.
For each $j\in\{1,\ldots,m\}$ and $i\in\{1,\ldots,\ell\}$, define
measurements
\begin{equation}
  \{R^{j,i}_a\,:\,a\in\Gamma\} \quad\text{and}\quad
  \{S^{j,i}_a\,:\,a\in\Gamma\}
\end{equation}
as follows:
\begin{equation}
  R^{j,i}_a = \sum_{\substack{b_1,\ldots,b_{\ell}\in\Gamma\\b_i = a}}
  P^j_{b_1,\ldots,b_{\ell}}
  \quad\text{and}\quad
  S^{j,i}_a = Q^{k_i}_a,
\end{equation}
where $k_1 < \cdots < k_{\ell}$ are the indices of the variables appearing
in the constraint $C_j$.
In words, the first measurement is equivalent to performing the
measurement $\{P^j_{b_1,\ldots,b_{\ell}}\,:\,b_1,\ldots,b_{\ell}\in\Gamma\}$
and outputting just the assignment $a = b_i$ rather than the entire assignment
$(b_1,\ldots,b_{\ell})$, while the second measurement is equivalent to the
Bob-measurement for the $i$-th variable appearing in $C_j$.

Observe that the provers' success in the consistency check performed by the clause-versus-variable
verifier implies that the measurements $\{R^{j,i}_a\}$ and $\{S^{j,i}_a\}$
necessarily constitute a perfect strategy for the game $\cons(X,\Gamma,\pi)$,
where $X = \{1,\ldots,m\}\times\{1,\ldots,\ell\}$ and $\pi$ is the distribution
obtained by selecting a constraint and a variable appearing in that
constraint, both uniformly at random.
It therefore holds, for all choices of $j\in\{1,\ldots,m\}$,
$i\in\{1,\ldots,\ell\}$, and $a\in\Gamma$, that
\begin{equation}
  \label{eq:R-and-S-flip}
  \bigl(R^{j,i}_a \otimes \I \otimes \I\bigr)\ket{\psi}
  = \bigl(\I \otimes S^{j,i}_a \otimes \I\bigr)\ket{\psi},
\end{equation}
which is remarkable because it implies that the probability for the first
prover to assign the value $a$ to a given variable is constant over all choices
of the constraints in which that variable appears.
(Recall that $S^{j,i}_a$ only depends on the label of the $i$-th variable
appearing in $C_j$, but not on the clause $C_j$ itself.)
By using the permutation invariance of both $\ket{\psi}$ and the measurements
performed by the three provers, the 5 similar equations to
\eqref{eq:R-and-S-flip} obtained by permuting the 3 systems also hold.
It therefore follows that
\begin{equation}
  \label{eq:Q-flip}
  \bigl(Q^k_a \otimes \I \otimes \I\bigr)\ket{\psi}
  = \bigl(\I \otimes Q^k_a \otimes \I\bigr)\ket{\psi}
  = \bigl(\I \otimes \I \otimes Q^k_a\bigr)\ket{\psi}
\end{equation}
for every $k\in\{1,\ldots,n\}$ and $a\in\Gamma$.
For instance, if $k$ is the index of the $i$-th variable appearing in the
constraint $C_j$, then one has
\begin{equation}
  \begin{multlined}
    \bigl(Q^k_a \otimes \I \otimes \I\bigr)\ket{\psi}
    = \bigl(S^{j,i}_a \otimes \I \otimes \I\bigr)\ket{\psi}
    = \bigl(\I \otimes \I\otimes R^{j,i}_a\bigr)\ket{\psi}\\
    = \bigl(\I \otimes S^{j,i}_a\otimes\I\bigr)\ket{\psi}
    = \bigl(\I \otimes Q^k_a \otimes \I\bigr)\ket{\psi}.
  \end{multlined}
\end{equation}
Note that \eqref{eq:Q-flip} does not follow directly from the permutation
invariance of $\ket{\psi}$ and the provers' measurements, and it makes
essential use of all three systems.

Now, suppose that $C_j$ is a constraint acting on the subset of variables
indexed by $k_1< \cdots < k_{\ell}$, and let
$i_1< \cdots < i_{n - \ell}$ be the indices of the remaining $n - \ell$
variables.
Let us also take $(a_1,\ldots,a_n)$ to be an assignment to the variables
$(x_1,\ldots,x_n)$ that is obtained from the assignment
$(b_1,\ldots,b_{\ell})$ to the variables indexed by $(k_1,\ldots,k_{\ell})$
and an assignment $(c_1,\ldots,c_{n-\ell})$ to the variables indexed by
$(i_1,\ldots,i_{n-\ell})$.
By repeatedly applying \eqref{eq:Q-flip} one finds that
\begin{equation}
  \begin{multlined}
    \bigl(Q^n_{a_n}\cdots Q^1_{a_1} \otimes \I\otimes\I\bigr)\ket{\psi}\\
    = 
    \bigl(\I \otimes Q^{k_{\ell}}_{b_{\ell}}\cdots Q^{k_1}_{b_1} \otimes
    Q^{i_{n-\ell}}_{c_{n-\ell}}\cdots Q^{i_1}_{c_1} \bigr)\ket{\psi},
  \end{multlined}
\end{equation}
and therefore, by summing over all choices of $c_1,\ldots,c_{n-\ell}$,
one finds that
\begin{equation}
  q_j(b_1,\ldots,b_{\ell}) =
  \Bignorm{\bigl(\I \otimes Q^{k_{\ell}}_{b_{\ell}}\cdots Q^{k_1}_{b_1} \otimes
    \I \bigr)\ket{\psi}}^2.
\end{equation}
Finally, by applying \eqref{eq:R-and-S-flip} repeatedly, it follows that
\begin{equation}
  \begin{multlined}
    \bigl(\I \otimes Q^{k_{\ell}}_{b_{\ell}}\cdots Q^{k_1}_{b_1} \otimes
    \I \bigr)\ket{\psi}
    = \bigl(R^{j,1}_{b_1} \otimes Q^{k_{\ell}}_{b_{\ell}}\cdots Q^{k_2}_{b_2}
    \otimes \I \bigr)\ket{\psi}\\
    = \cdots = 
    \bigl(R^{j,1}_{b_1}\cdots R^{j,\ell}_{b_\ell} \otimes \I \otimes \I
    \bigr)\ket{\psi} = 
    \bigl(P^{j}_{b_1,\ldots,b_{\ell}}\otimes \I \otimes \I
    \bigr)\ket{\psi},
  \end{multlined}
\end{equation}
where the last equality has made use of the fact that
$\{P^j_{b_1,\ldots,b_{\ell}}\}$ is a projective measurement.
The equation \eqref{eq:marginal-goal} has therefore been proved, as required.

This concludes the analysis of the three-prover clause-versus-variable verifier
$T_\varphi$.
One of the provers, playing the role of Charlie, plays a completely passive
role. 
It is not asked any questions, but its presence is necessary for the
proof to go through, as Example~\ref{ex:ms-1} of the Magic Square game
demonstrates. 
At a more technical level, all three registers of $\ket{\psi}$ played an
essential role in the analysis above.

It is possible to define a different variant of the clause-versus-variable verifier that involves an interaction with only two provers and remains sound against entangled provers. In the three-prover game the intuition behind the role of the third prover is that it ``confuses'' the other provers into not knowing with which prover they should coordinate their answers. The idea for the two-prover variant is to induce the same type of confusion by sending an additional question, chosen uniformly at random, to the second prover in the clause-versus-variable game. The prover thus receives two variables: one taken from the first prover's constraint, and the other chosen uniformly at random. The fact that the prover is not told which variable it will be tested on makes its task harder.

\begin{center}
  \underline{$2$-prover confuse-SAT verifier $F_\varphi$}
\end{center}
\noindent
Given $\varphi=(C_1,\ldots,C_m)$, where each $C_i$ is a constraint acting on
$\ell$ out of $n$ variables $x_1,\ldots,x_n$, the verifier proceeds as
follows:
\begin{mylist}{8mm}
\item[1.]
  Select an index $j\in [m]$ uniformly at random. 
  Let $\{x_{i_1},\ldots,x_{i_\ell}\}$ be the variables on which constraint
  $C_j$ acts. 
  Select $t\in\{1,\ldots,\ell\}$ and $i_{t'}\in \{1,\ldots,n\}$ uniformly at
  random. 
  Send $C_j$ to the first prover, and the unordered pair $\{i_t,i_{t'}\}$ to
  the second prover.
\item[2.] 
  The first prover replies with an assignment
  $(a_{1},\ldots,a_{\ell})\in\Gamma^k$. 
  The second prover replies with an assignment $(b,b')\in\Gamma^2$.
\item[3.]
  Accept if and only if the first prover's answers satisfy clause $C_j$
  and the provers' answers are consistent on the variable they were both
  asked: $a_{t}=b$.
  \vspace{2mm}
\end{mylist}

The same result as shown for the three-prover clause-versus-variable verifier $T_\varphi$ holds: $\omega^*(F_\varphi)=1$ if and only if $\varphi$ is satisfiable, and a similar quantitative bound as the one in~\eqref{eq:3prover-quant} holds as well. The proof follows the same outline. The main additional ingredient is the analysis of the following ``confuse'' test, that complements the consistency test $\cons(X,\Gamma,\pi)$ introduced earlier.

\begin{center}
  \underline{Confuse test $\conf(X,\Gamma,\nu)$}
\end{center}
\noindent
Given finite sets $X$, $\Gamma$ and a distribution $\nu$ on $X\times X$, the
verifier proceeds as follows:
\begin{mylist}{8mm}
\item[1.]
  Select $(x,y)\in X\times X$ according to $\nu$, and send the unordered pair
  $\{x,y\}$ to the first prover and either $x$ (with probability $1/2$) or $y$
  (with probability $1/2$) to the second.
\item[2.] 
  Receive answers $(a,a')\in\Gamma^2$ and $b\in \Gamma$ respectively. 
  Accept if and only if the provers' answers are consistent.
  \vspace{2mm}
\end{mylist}

This game has the following soundness property: any strategy given by
$\{P^{x,y}_{a,a'}\}$, $\{Q^x_b\}$, and $\ket{\psi}$ that has success probability $1$ must be such that, for any $x,y\in X$ such that $\nu(x,y)>0$ and all $a,b\in\Gamma$,
\begin{equation}
  \label{eq:conftest-0}
  \bigl( P^{x,y}_{a,b} \otimes \I\bigr) \ket{\psi} 
  = \bigl(\I \otimes Q^x_a Q^y_b\bigr) \ket{\psi} 
  = \bigl(\I \otimes Q^y_b Q^x_a\bigr) \ket{\psi}.  
\end{equation}
Furthermore, if $\ket{\psi}$ has full support on the second prover's space,
then the measurement operators $Q^x_a$ and $Q^y_b$ must commute provided
that $\nu(x,y)>0$. 
These properties together with those that follow from the test
$\cons(X,\Gamma,\pi)$ suffice to prove soundness of the confuse-SAT verifier.

\subsection{Linearity testing with entangled provers}
\label{sec:linearity-test}

Originally introduced in the context of efficient program checking, the
\emph{linearity test} of Blum, Luby and Rubinfeld~\cite{BlumLR93} quickly found
one of its most important applications in the study of classical multi-prover
interactive proof systems. 
Based on this test alone, it is already possible to prove a weak form of the
PCP theorem, establishing that all languages in \class{NP} have proofs of
exponential length that can be verified with a constant gap between
completeness and soundness by querying only a constant number of bits from the
proof. 
This section is devoted to the formulation and analysis of the test when
executed with entangled provers.

The linearity test can be formulated as a one-round interactive game played
between a classical verifier and three provers. The verifier's messages to the
provers are elements of $\field_2^n$, for some integer $n$. The provers'
messages are elements of $\field_2$. The test is designed to certify that the
provers' answers are consistent with a linear function $f: \field_2^n \to
\field_2$, i.e., one that can be written as $f(x) = u\cdot x$ for some $u\in
\field_2^n$.

\begin{center}
  \underline{Linearity test}
\end{center}
\noindent
The verifier performs either of the following with probability $1/2$ each:
\begin{mylist}{\parindent}
\item[1.]
  \emph{(Consistency)} 
  Select $x \in \field_2^n$ uniformly at random and send $x$ to each prover. 
  Accept if and only if all provers provide the same answer $a\in\field_2$.
\item[2.]
  \emph{(Linearity)} 
  Select $x,y \in \field_2^n$ uniformly at random, and set $z = x+y$. 
  Send $x$ to the first prover, $y$ to the second, and $z$ to the third. 
  Accept if and only the provers' answers $a,b,c\in \field_2$ satisfy $a+b=c$.
\end{mylist}
 
The linearity test has perfect completeness: if the provers answer according to
the same linear function, they are accepted by the verifier with probability
$1$. Its soundness property against classical deterministic provers can be
stated as follows.

\begin{theorem}[BLR linearity test]\label{thm:BLR} Suppose that three classical
   deterministic provers succeed in the linearity test with probability
  $1-\eps$, and let $f_1,f_2,f_3:\field_2^n\to\field_2$ be the functions
  describing their respective strategies. 
  There exists a vector $u\in \field_2^n$ such that, for each $i\in\{1,2,3\}$,
  $f_i(x) = u\cdot x$ for all but a fraction at most $8\eps$ of 
  $x\in\field_2^n$. 
\end{theorem} 

Theorem~\ref{thm:BLR} makes the assumption that the three provers are deterministic. In general the provers may use private or shared randomness. The result is easily extended by ``fixing the randomness''; each private or shared random string that may be used by the provers corresponds to a deterministic strategy to which the theorem can be applied. 

In the case of entangled provers it is not possible to ``fix the quantumness''
present in the provers' strategy, which in general is specified by
families of binary-valued measurements $\{P^x_a\}$, $\{Q^y_b\}$, $\{R^z_c\}$ on
the three parts of a tripartite state $\ket{\psi}\in\X\otimes\Y\otimes\Z$. 
What does it mean for such a strategy to be linear? 
Without ``fixing the randomness,'' it is impossible to claim that the strategy is close to a single linear function $f$. A more reasonable statement would be that any strategy for the provers with a high probability of success in the test is \emph{almost indistinguishable} from a ``linear'' strategy of the following form:
\begin{mylist}{\parindent}
\item[1.] Each prover measures its share of the entangled state using the same
  measurement $\{S_u\,:\,u\in\field_2^n\}$, respectively obtaining outcomes
    $u_1,u_2,u_3\in\field_2^n$ such that $u_1=u_2=u_3$ with high probability.
\item[2.] Upon receiving the verifier's message $x$, $y$, or $z$, the first,
  second, or third prover answers with $u_1\cdot x$, $u_2\cdot x$, or 
  $u_3\cdot x$, respectively.
\end{mylist}

The strength of this statement resides in the existence of the measurement
$\{S_u\}$, its independence from the prover's question, and the claim that it
faithfully reproduces the original  strategy.

Whether the statement is meaningful or not rests on the precise quantification
of the claim that the original and oblivious strategies are ``almost
indistinguishable.''
In the case of classical provers (i.e., Theorem~\ref{thm:BLR}), this is the
statement that the oblivious strategy differs from the original one in a
fraction at most $8\eps$ of questions $x$.
The case of entangled provers is more subtle, as it requires the introduction
of a measure of distance between strategies that is
\begin{mylist}{8mm}
\item[(i)] strong enough that ``nearby'' strategies have a similar success
  probability, not only in the test itself but also in any proof system that
  would invoke the test as a sub-game; and
\item[(ii)] weak enough that it is possible to place bounds on the distance
  solely from the assumption that the provers have a high success in the
  linearity test.\footnote{This implies for instance that the operator norm
    between the provers' measurements (and even more so the diamond norm
    between the associated quantum channels) would not be appropriate, as
    success in the test does not put constraints directly on the provers'
    measurements themselves, but only on their probability of obtaining certain
    outcomes when applied on the specific entangled state $\ket{\psi}$.}
\end{mylist}
Such a distance measure can be defined as follows. 
Consider two strategies for the first prover, specified by measurements
$\big\{P^x_a\big\}$ and $\big\{\tilde{P}^x_a\big\}$.
Fixing a message $x$ from the verifier, the distance between the 
post-measurement states resulting from these two measurements can be expressed
as
\begin{equation}
  \label{eq:work_dist-1}
  \begin{multlined}
    \sum_a\, \biggnorm{\biggl(\sqrt{P^x_a} \otimes \I \biggr) \ket{\psi} - 
      \biggl(\sqrt{\tilde{P}^x_a}  \otimes \I\biggr) \ket{\psi}}^2\\
    = \sum_a \Tr\Biggl(\biggl(\sqrt{P^x_a} - \sqrt{\tilde{P}^x_a}\biggr)^2
    \rho\Biggr),
  \end{multlined}
\end{equation}
where $\rho$ is the reduced density operator of $\ket{\psi}$ on the first
prover's register.
Provided the quantity in~\eqref{eq:work_dist-1} is small, the provers' shared
state is almost the same after the first prover has measured its subsystem
using either of the two measurements and obtained an answer $a$. 
Thus, the joint distributions on outcomes obtained when the first prover
measures using either measurement, and the other provers perform any
measurement whatsoever, are close in statistical distance.

\begin{theorem}[Entangled-prover linearity test]\label{thm:quantumBLR} 
  Suppose three entangled provers succeed in the linearity test with
  probability at least $1-\eps$ using a strategy specified by measurements
  $\{P^x_a\}$, $\{Q^y_b\}$, $\{R^z_c\}$ on an entangled state $\ket{\psi}$. 
  There exists a measurement $\big\{S_u\,:\,u\in \field_2^n\big\}$, which
  is independent of $x$, such that, for
  \begin{equation}
    \tilde{P}^x_a = \sum_{\substack{u\in \field_2^n\\u\cdot x=a}} S_u,
  \end{equation}
  it holds that
  \begin{equation}\label{eq:work_entclose}
    \frac{1}{2^n}\sum_{x\in\field_2^n} \sum_{a\in\field_2} 
    \Tr\Biggl(\biggl(\sqrt{P^x_a} - \sqrt{\tilde{P}^x_a} \biggr)^2 
    \rho \Biggr) \, = \, O\bigl(\sqrt{\eps}\bigr),
\end{equation}
  where $\rho$ is the reduced density operator of $\ket{\psi}$ on the first
  prover's register.
  A similar statement holds for $\{Q^y_b\}$ and $\{R^z_c\}$.
\end{theorem}

The measurement $\{S_u\}$ whose existence is claimed in the previous theorem
has a simple definition. 
Given $\{P^x_a\}$,
 let its matrix-valued Fourier coefficient be defined as 
\begin{equation}
  \widehat{S}_u = 2^{-n}\sum_x (-1)^{x\cdot u} \big(P^x_0-P^x_1\big),
\end{equation}
and for every $u$ let
\begin{equation}\label{eq:work_budef}
  S_u \,=\, \bigl( \widehat{S}_u \bigr)^2.
\end{equation}
It follows from Parseval's formula that $\{S_u\}$ is a well-defined
measurement. 
For intuition on this definition it is useful to consider the special case of
classical provers who may use shared randomness, for which $\{S_u\}$
corresponds to the following definition of an oblivious strategy. 
Let $r$ be a value for the shared randomness. 
In the original strategy, $r$ points to a function $f_r$ according to which the
prover would determine its answer to the verifier's question. 
In the new strategy, the oblivious prover uses $r$ to sample a linear function
$g:x\mapsto u\cdot x$, where $u$ is chosen according to the distribution
suggested by the Fourier spectrum of $f_r$.\footnote{Letting $g_r =
  (-1)^{f_r}$, this is the distribution induced by the $|\widehat{g_r}(u)| ^2$.
  Parseval's identity shows that this is indeed a distribution.} 
Upon receiving question $x$ the prover answers with $g(x) = u\cdot x$. 

The fact that this is a  good strategy follows from the classical proof of
Theorem~\ref{thm:BLR}, which establishes that the Fourier coefficients of
$g=(-1)^f$ are sharply concentrated.
In the case of a randomized strategy, most functions $f_r$ will have a large
Fourier coefficient, and the strategy described above will with high
probability provide answers that are consistent and chosen according to a
linear function.
(Which linear function this is depends on the random string $r$ and may change
with each interaction with the verifier.)

With the definition of $S_u$ in hand, the proof of Theorem~\ref{thm:quantumBLR}
is not too difficult---it uses similar arguments to those employed in the
analysis of the three-prover clause-versus-variable game given in the previous
section, and we omit the details.

\subsection{$\NEXP\subseteq \QMIP^*$}
\label{sec:nexp-in-qmipstar}

The proof of Theorem~\ref{thm:nexp-mipstar} follows the same broad outline as
Babai, Fortnow, and Lund's proof that $\NEXP\subseteq\MIP$, with significant
modifications to the analysis needed to ensure that soundness is preserved
against entangled strategies.
(Completeness is straightforward, and does not require honest provers to use
any entanglement.)
In the following discussion we give a high-level overview of the soundness
analysis that emphasizes the most important modifications required and ties
them to the results from the preceding sections;
the reader is advised not to interpret the discussion too literally and is
referred to~\cite{BabaiFL91} for a technical exposition.

The starting point is an encoding of an $\NEXP$-complete problem, such as
SUCCINCT-$3$-COLORABILITY, as an instance of the following problem: given a
multilinear function $f$ in $\ell+3$ variables, do there exist functions
$A_1,A_2,A_3:\{0,1\}^\ell\to\{0,1\}$ such that $f(z,A_1(z),A_2(z),A_3(z))=0$
for every $z\in\{0,1\}^\ell$? Here the function $f$ encodes the structure of
the problem (such as the graph to be colored), the variable $z$ is used to
index all the required constraints (such as edges of the graph), and $A_1,A_2$
and $A_3$ provides an assignment (the coloring). 
The input size $n$ is the size of an arithmetic circuit specifying $f$, and in
particular the total number of variables $z$ can be exponential in $n$.

Using the technique of arithmetization~\cite{LundFKN92} it is possible to
devise a three-prover proof system for this problem with these attributes:
\begin{mylist}{\parindent}
\item[1.] 
  The verifier selects random elements $z,b_1,b_2\in \field^{\ell}$, where
  $\field$ is a field of size exponential in $n$ known to all parties.
\item[2.] Based on $z$, the verifier has a multiple-round interaction with the
  first prover.
\item[3.] The verifier asks two other provers for values $\tilde{A}_i(b_1)$ and
  $\tilde{A}_i(b_2)$, for all $i\in\{1,2,3\}$, respectively, where
  $\tilde{A}_i$ is the unique multilinear extension of $A_i$ to
  $\field^\ell$.\footnote{A multilinear function is one that is linear in each
  of its variables; $\tilde{A}_i$ extends $A_i$ in the sense that
  $\tilde{A}_i(z)=A_i(z)$ for every $z\in\{0,1\}^\ell\subseteq\field^\ell$.}
\item[4.] The verifier decides to accept or reject based on the results of its
  interaction with all three provers.
\end{mylist}
The most important step in the analysis of this proof system consists of
devising a test that can be used to guarantee that the functions
$\field^\ell\to\field$ used by the second and third prover in step 3 agree with
a multilinear function on a large fraction of inputs. 
With the promise that this is the case, the remainder of the proof system can
be proven sound in a very similar way whether the provers use entanglement or
not, as it mostly relies on analyzing the single-prover interaction performed
in step 2.

Thus, the analysis will be completed once it is shown how to test that a prover
answers queries $z\in\field^\ell$ with a function
$\tilde{A}:\field^\ell\to\field$ that is linear in each  variable. The test for
this is a natural extension of the linearity test described in the previous
section, and can be formulated as a one-round three-prover interactive game as
follows. 
The verifier selects a coordinate $i\in\{1,\ldots,\ell\}$, an $x\in
\field^\ell$ uniformly at random, and $y_i,z_i\in \field\backslash\{x_i\}$. Let
$y$ and $z$ be equal to $x$ with the $i$-th coordinate replaced by $y_i$ and
$z_i$, respectively.
The verifier asks the three provers for the value of $\tilde{A}$ at $x,y$ and
$z$ respectively. 
Upon receiving answers $a,b,c$ from each of them  the verifier checks the
identities
\begin{equation}
  \frac{c-a}{z-x}\,=\,\frac{b-a}{y-x}\,=\, \frac{c-b}{z-y}.
\end{equation}
The analysis proceeds by induction on $\ell$. 
The case $\ell=1$ is provided by the linearity test described in the previous
section. 
Unfortunately, as one proceeds through the induction, the error (as measured by
the fraction of points on which $\tilde{A}$ differs from a multilinear
function) increases rapidly: 
even in the most optimistic case, it will be multiplied by a factor $2$ in each
step, eventually yielding an unmanageable exponential blow-up.

To handle this difficulty, Babai, Fortnow, and Lund introduce a
``self-improvement lemma,'' which establishes the following: any function
$B:\field^\ell\to\field$ that is such that
\begin{mylist}{8mm}
\item[(i)] $B$ is very close to linear along ``lines'' (meaning that the
  functions $B_i:x_i\mapsto B(x_1,\ldots,x_n)$, for $i=1,\ldots,\ell$, are
  close to linear for most choices of $x_1,\ldots,x_{i-1},x_{i+1},\ldots,x_n$),
  and
\item[(ii)] $B$ is globally ``somewhat close'' to a multilinear function,
\end{mylist}
must in fact be very close to multilinear. 
That is, if the error in (i)~is some small $\eps>0$, and the error in (ii) is a
possibly much larger $\eta<1/4$, then in fact $\eta$ is automatically much
smaller than expected, of order $\ell\eps$. 
The proof of this claim is based on expansion properties of the hypercube, from
which the factor of $\ell$ loss in the error originates.

The case of entangled strategies is substantially more difficult. 
In particular, the conclusion of the lemma can only be reached by a suitable
modification of the provers' strategy: an additional step of active correction is required whereby the distance to a multilinear function (as measured by a statement of a form similar to the bound given in Theorem~\ref{thm:quantumBLR}) is reduced by leveraging the entangled-prover analogue of assumption (i) above. The ``improved'' measurement can be defined as the optimum of a semidefinite program and shown to satisfy similar (though quantitatively weaker) error bounds as those promised by the self-improvement lemma in the classical case, thus enabling the induction to carry through. 

Once the soundness of the multilinearity test against entangled provers has
been established, the remainder of the analysis of Babai, Fortnow, and Lund's
proof system for $\NEXP\subseteq \MIP$ goes through with minor modifications,
leading to a proof of the inclusion $\NEXP\subseteq \MIP^*$.
The main difference between the two end results is that, while the first is
known to hold with two provers, the proof of the second containment seems to 
require three provers: although it is possible to formulate the multilinearity
test as a two-prover game, only the three-prover variant is known to be sound
against entangled provers. 
It is an open question to determine whether $\NEXP\subseteq \QMIP^*(2,\poly)$.

\section{Further topics}
\label{sec:qmip-further}

This section is devoted to two topics on which relatively little is
known---both are currently active areas of research.
The first topic is error reduction through parallel repetition, and the second
is the problem of placing computational upper bounds on the classes $\MIP^*$
and $\QMIP^*$.

\subsection{Parallel repetition}
\label{sec:qmip-par-rep}

Given a multi-prover interactive proof system with completeness and soundness
parameters $a>b$, Proposition~\ref{prop:amplify} states that the difference
$a-b$ can be amplified efficiently, either by repeating the protocol
sequentially (which increases the number of turns), or with different sets of
provers. 
Both procedures work for all types of proof systems considered in this
chapter, with or without entanglement between the provers.

It is natural to ask if the same effect can be achieved without any increase in
the number of provers or turns.
The most standard approach to this question is to consider repeating the
protocol in parallel with the same set of provers. 
This type of repetition was already considered in the single-prover setting
(q.v. Section~\ref{sec:qip-sdp}), and it is defined analogously with multiple
provers.
Given a one-round two-prover\footnote{
  The transformation can be described more generally, but almost all known
  results apply to the one-round two-prover setting only.} 
verifier $V=(V_1,V_2)$ and a number of repetitions $\ell$, another one-round 
two-prover verifier ${V}^{(\ell)}$ is defined as follows: ${V}^{(\ell)}$ 
executes $\ell$ independent copies of the first transformation $V_1$ of $V$, 
generating $\ell$ independent pairs of registers 
$(\rX_1^{1,i},\rX_1^{2,i})_{i=1,\ldots,\ell}$.
${V}^{(\ell)}$ then sends $(\rX_1^{1,1},\cdots,\rX_1^{1,\ell})$ to the first 
prover and $ (\rX_1^{2,1},\cdots,\rX_1^{2,\ell})$ to the second. 
The provers are expected to return $\ell$ message registers each.
${V}^{(\ell)}$ measures each pair of answer registers independently according
to $V_2$, and accepts if and only if all measurements produce the outcome $1$.

Even though the repeated verifier $V^{(\ell)}$ has a simple product form, 
due to the fact that the provers receive all of their questions simultaneously
they may in general apply an arbitrary quantum channel, introducing 
correlations between their answers that make this type of repetition harder to
analyze. 
The following example demonstrates that, in contrast to the single-prover case,
parallel repetition cannot be expected to perfectly amplify either the
unentangled or entangled values of a multi-prover game.

\begin{example}[Fortnow--Feige--Lov\'asz game]
  \label{ex:par-rep-counter-ex}
  Consider the following one-round two-prover verifier $V$. 
  The verifier selects a pair of messages $(s, t)$ uniformly from 
  $\{(0, 0),(0, 1),(1, 0)\}$ and sends $s$ to the first prover and $t$ to the
  second. Each prover replies with a bit $a,b\in\{0,1\}$ respectively. 
  The verifier outputs $1$ if and only if $s\wedge a \neq t \wedge b$. 

  It may be verified that $\omega(V) = \omega(V^{(2)}) = 2/3$, as well as
  $\omega^*(V) = 2/3$.
  (Indeed, even the no-signaling value of $V$ is equal to 2/3.)
  Combined with the previous equalities, this implies $\omega^*(V^{(2)})=2/3$ 
  as well, so that neither the classical or entangled value of the game
  described by $V$ multiply under parallel repetition.
\end{example}

The only scenario for which there is a satisfactory understanding of the effect of parallel repetition is that of classical two-prover one-round games with unentangled provers, to which the following result applies.  

\begin{theorem}\label{thm:raz}
There exists a constant $C$ such that the following holds. Let $V$ be a one-round two-prover classical verifier, $r$ an upper bound on the number of bits of each prover's message register, and $\ell$ an integer. Then the $\ell$-th parallel repetition ${V}^{(\ell)}$ of $V$ satisfies
\begin{equation}\label{eq:mip-parrep}
\omega(V)^{(\ell)} \leq \omega({V}^{(\ell)}) \leq \Big(1- \big(1-\omega(V)\big)^3\Big)^{C\ell/r}.
\end{equation}
\end{theorem}

The theorem, proved in \cite{Raz98,Holenstein09}, only applies to games with
two provers. 
The situation for more than two provers is poorly understood and only very weak
results are known.

For the case of entangled provers only partial results, that apply to specific
classes of games with a classical verifier, are known; we briefly describe some
of these results below. There is no known parallel repetition theorem that
applies to fully quantum verifiers.\footnote{See the chapter notes for results
  that apply to restricted quantum verifiers that either send quantum messages
  but receive classical answers, or vice versa.}

Two classes of games for which good results are known are \emph{free games} and
\emph{projection games}. 
A free game is one in which the verifier's messages to the provers are chosen
according to a product distribution. 
An example of a free game is the CHSH game (Example~\ref{ex:chsh}).
Projection games are characterized by the form of the verifier's acceptance
criterion: for any pair of messages from the verifier, and any possible message
from the first prover, there should always be at most one message from the
second prover that will result in acceptance. 
An example of a projection game is the Magic Square game 
(Example~\ref{ex:ms-1}).
The CHSH game is also a projection game, but the Magic Square game is not a
free game, and the game described in Example~\ref{ex:par-rep-counter-ex} is
neither.
For both classes of games it is known that the entangled value decreases
exponentially fast with the number of repetitions, in a manner analogous
to~\eqref{eq:mip-parrep} except the exponent $3$ is replaced by $3/2$ for free
games and some other universal constant for projection games; for the latter
there is no dependence on the answer length $r$.

Free games and projection games encompass many interesting games, but not
all. For a general game it is possible to consider a modified form of
repetition. The simplest way to describe this transformation is as a
transformation of the game itself.

Given a one-round two-prover verifier $V$, consider a verifier $V_\perp$ that
performs the following:
\begin{mylist}{\parindent}
\item[1.] Generate a pair of messages $(x^1,x^2)$ as $V$ would. 
\item[2.] Independently and with probability $1/4$ each, replace the message $x^1$ with a ``dummy'' message $\perp$, and the message $x^2$ with $\perp$. 
\item[3.] Send the new messages to the provers. 
\item[4.] If either message was replaced by a $\perp$, accept any answer from the provers. If neither message was modified, accept the provers' answers if and only if the original verifier $V$ would have accepted them.
\end{mylist} 
It is clear that this transformation can only increase the value of the game;
in fact it is not hard to verify that
\begin{equation}
  \omega(V_\perp) = \frac{1}{4} + \frac{3}{4}\,\omega(V)
  \quad\text{and} \quad
  \omega^*(V_\perp) = \frac{1}{4} + \frac{3}{4}\,\omega^*(V),
\end{equation}
a simple affine scaling. 
In spite of its almost na\"ive simplicity, it turns out that this
transformation allows to prove strong parallel repetition results: it is known
that if $\omega^*(V)=1-\eps$ then $\omega^*(V_\perp^{(\ell)}) \leq
(1-\eps^c)^{\Omega(\ell)}$, for some constant $c>1$. 
(The same holds for the classical value.) 
Intuitively, the role of the ``dummy'' question $\perp$ is to limit the
provers' ability to exploit correlations present in their $\ell$ pairs of
questions in order to succeed with substantially higher probability than a
strategy which treats all repetitions independently; at a high level the
transformation makes the game closer to a free game.  

\subsection{Upper bounds on $\QMIP^*$}
\label{sec:entanglement}

As discussed in Section~\ref{sec:qmip-nexp}, the inclusion
$\QMIP\subseteq\NEXP$ follows from the observation that optimal
\emph{unentangled} strategies for the provers can always be implemented using
private registers that are of dimension at most exponential in the number of
qubits exchanged between each prover and the verifier. 
Example~\ref{ex:dltw-1} demonstrates that this fact no longer holds for
entangled strategies, which may in general benefit from arbitrarily large
amounts of prior shared entanglement.

The example suggests the introduction of a hierarchy of values 
\begin{equation}\label{eq:omega-d}
\omega(V)=\omega^{(1)}(V)\leq\cdots \leq \omega^{(d)}(V)\leq\cdots\leq\omega^*(V),
\end{equation}
where $\omega^{(d)}(V)$ is the maximum success probability of provers whose initial shared entangled state has local dimension at most $d$. Example~\ref{ex:dltw-1} provides a $V$ for which the chain of inequalities~\eqref{eq:omega-d} does not eventually collapse into a series of equalities for large enough $d$. For small values of $d$ explicit examples show that the first few inequalities can be strict; for instance the CHSH game (Example~\ref{ex:chsh}) is such that $\omega(\textrm{CHSH})=\omega^{(1)}(\textrm{CHSH}) < \omega^{(2)}(\textrm{CHSH})=\omega^{*}(\textrm{CHSH})$. 
We refer to the chapter notes for pointers to further examples. 

Such examples raise the question of whether the entangled value is even computable. It is possible to devise a procedure for approaching $\omega^*$ from below, with the implication that $\QMIP^*\subseteq \class{RE}$, the class of problems that are recursively enumerable. 
Because $\omega^*$ is given by the supremum over all $d$ of $\omega^{(d)}$, 
for every verifier $V$ and $\eps>0$ there is a dimension $d$ such that
$\omega^{(d)}(V) \geq \omega^*(V)-\eps$. 
Applying a union bound over the suitably discretized space of possible
verifiers of a given size we may deduce the existence of an integer
$d=d(\eps,s)$ depending only on $\eps$ and $s\in \natural$ such that, for all
verifiers of size at most $s$, there is an entangled strategy using private
registers of dimension at most $d$ whose success probability is at least
$\omega^*(V)-\eps$. 
No estimates are known on the dependence of $d(\eps,s)$ on either parameter. 

To prove that every problem in $\QMIP^*$ is decidable, it would suffice to 
devise a counterpart to the above procedure that approaches $\omega^*(V)$ from 
above. 
This problem, however, is not settled, even for the simplest case of classical
one-round two-prover games. 
For the remainder of this section we focus on that setting and introduce a
procedure that may, under plausible but difficult mathematical conjectures,
provide the required sequence of approximations.

The procedure is based on a hierarchy of semidefinite programs whose
optimum is always at least the entangled value. For the case of XOR games the
first level of the hierarchy coincides with the semidefinite program introduced
in Section~\ref{sec:xor}. 
Higher levels introduce variables not only for the inner products between
vectors $(P_a^x \otimes \I) \ket{\psi}$ and $(\I \otimes Q_b^y)\ket{\psi}$
associated with each provers' possible measurement operators, but also for
composite terms involving vectors such as $(P_a^x\otimes Q_b^y) \ket{\psi}$ (at
the second level), $(P_a^x P_{a'}^{x'} \otimes Q_b^y) \ket{\psi}$ (at the third
level), and so on.
All natural constraints that should hold for projective strategies, such as 
\begin{equation}
  \bigbra{\psi} \bigl(P_a^x  \otimes \I \bigr)
  \bigl(P_a^x  P_{a'}^{x'} \otimes Q_b^y\bigr) \bigket{\psi}
  = \bigbra{\psi} \bigl(\I\otimes Q_b^y\bigr) 
  \bigl(P_a^x P_{a'}^{x'}\otimes \I\bigr) \bigket{\psi},
\end{equation}
are enforced as constraints in the semidefinite program. 
This hierarchy, introduced in~\cite{NPA08NJP}, can be shown to converge
in the limit of infinitely many levels to a value called the
\emph{field-theoretic} value $\omega^{\textup{\scriptsize FV}}$ of the game. 
The field-theoretic value has an alternative definition as the supremum over
all \emph{commuting strategies} for the provers of the probability that the
verifier outputs $1$. 
Commuting strategies are a relaxation of entangled strategies in which the
provers share a single quantum system, not necessarily finite-dimensional and
initialized in to an arbitrary pure state, on which they perform their
measurements in sequence. 
The only restriction, which is required for the model to be well-defined, is
that the measurement operators associated with distinct provers must commute
pairwise.

If it is assumed that the provers' joint register is finite-dimensional,
commuting-prover strategies are equivalent to entangled strategies. 
This follows from a well-known (but nontrivial) result from the theory of
$C^*$-algebras showing that for finite-dimensional algebras, pairwise
commutation implies the existence of a decomposition of the algebra as a direct
sum of tensor products.
If the associated Hilbert space is allowed to be infinite dimensional, as seems
necessary for $\omega^{\textup{\scriptsize FV}}$ to coincide with the limiting
value of the hierarchy introduced above, then the equality 
$\omega^*(V) \stackrel{?}{=}\omega^{\textup{\scriptsize FV}}(V)$ forms the
content of \emph{Tsirelson's problem}.
Tsirelson's problem is equivalent to a range of long-standing conjectures in
the theory of $C^*$-algebras, including Connes' embedding conjecture and
Kirchberg's QWEP conjecture. An affirmative resolution of any of these
equivalent conjectures would imply that all problems in $\MIP^*(2,2)$ are
decidable. The converse is not necessarily true; the decidability of all
problems in $\MIP^*$ could hold even if the conjectures fail. 
In particular, deciding problems in $\MIP^*$ only requires one to compute a
constant-factor approximation, rather than an arbitrarily close one, of the
entangled value.
No such procedure is known, making the question of (un)decidability of
languages in $\MIP^*$ and $\QMIP^*$ one of the most intriguing open problems in
the area of quantum multi-prover interactive proofs.

\section{Chapter notes}

The class $\MIP$ was first considered in~\cite{BenOrGKW88}, where its
introduction was motivated by the study of zero-knowledge proof systems. 
Many results about $\MIP$ were discovered soon after, in the late 1980s and
early 1990s (see \cite{FortnowRS88,BabaiFL91,FeiLov92STOC}, for instance).
Modern accounts tend to emphasize a point of view on the class that results
from its connection to probabilisticaly checkable proofs and the PCP
theorem~\cite{AroraLMSS98,AroraS98}, with its many applications to the hardness
of approximation problems \cite{FeiGolLovSafSze96JACM}. 
We refer the interested reader to~\cite{donnellPCP} for a brief history of the
developments that led to the PCP theorem.

The origins of quantum multi-prover interactive games can be traced back to the
study of Bell inequalities in the foundations of quantum mechanics. 
This study was pioneered by Bell~\cite{Bell64}, with the simplest non-trivial
inequality, the so-called CHSH inequality, being explicitly written as such by
Clauser, Horne, Shimony, and Holt~\cite{Clauser:69a}. 
This inequality is equivalent to the bound $\omega(\CHSH)\leq 3/4$ discussed in
Example~\ref{ex:chsh}, but formulated in a different language. 
The perspective of Bell inequalities as interactive games came through the
work of Mermin~\cite{Mermin90}, Peres~\cite{Peres90}, Cabello~\cite{Cabello01},
Cleve, and others, and the connection to the soundness of interactive proof
systems was observed in \cite{CleveHTW04}.
The Magic Square game (Example~\ref{ex:ms-1}) is attributed to
Mermin~\cite{Mermin90} and Peres~\cite{Peres90}; our formulation can be found
in~\cite{Aravind02} and \cite{CleveHTW04}.

Multi-prover games with quantum messages and the associated complexity classes
$\class{QMIP}$ and $\class{QMIP}^*$ were introduced in~\cite{KobMat03JCSS},
where it was shown that $\QMIP=\NEXP$.
This equality was also extended to the case of provers who may use prior shared
entanglement on a number of qubits bounded by a polynomial in the input
length.

The class $\class{MIP}^*$ was first defined in~\cite{CleveHTW04}, whose focus
is on the special case of XOR proof systems, corresponding to the class
$\oplus\MIP^*(2,2)$. 
The inclusion $\NEXP\subseteq \oplus \MIP_{a,b}(2,2)$ for constants $0<b<a<1$
is due to Hast{\aa}d~\cite{Hastad01}. 
The inclusion $\oplus\MIP^*(2,2)\subseteq\EXP$ is implicit
in~\cite{CleveHTW04}. 
In~\cite{Wehner06} it was shown that the inclusion can be improved to
$\oplus\MIP_{a,b}^*(2)\subseteq\QIP_{a,b}(2)$ for any $a,b$ separated by an
inverse polynomial gap, which therefore implies 
$\oplus\MIP_{a,b}^*(2)\subseteq\PSPACE$.

In~\cite{CleveSUU08xor} it was shown that XOR games obey a perfect parallel
repetition theorem. 
Slofstra~\cite{Slofstra11xor} investigated the question of entanglement in XOR
games, and described a game for which the dimension of entanglement required
for an optimal strategy matches the upper bound obtained from Tsirelson's
construction.

An extension of XOR games that allows for quantum messages from the verifier to
the prover (``quantum XOR games'') was considered in~\cite{RegevV13}.
(See also~\cite{cooney2011rank} for a closely related model.)
It was proved in \cite{RegevV13} that Example~\ref{ex:dltw-1} of the coherent
state exchange game can be cast as a quantum XOR game with quantum messages for
which the optimal success probability can only be reached in the limit as the
dimension of the provers' shared entangled state goes to infinity.

The results on perfect completeness, parallelization and public-coin systems
presented in Section~\ref{sec:qmip-structure} appear in~\cite{KempeKMV09}. 
An alternative proof of the perfect completeness property appears
in~\cite{KobayashiLN13}. 
Ito~\cite{Ito14parallelization} showed that games with classical verifiers
could be parallelized to $4$ turns while keeping the verifier classical,
provided one is granted the promise that there exists a classical strategy for
the provers that achieves the completeness parameter and soundness holds
against entangled-prover strategies. 
(The corresponding proof systems are called \emph{entanglement-resistant}
$\MIP$ systems.)

An interesting variant on the multi-prover quantum interactive proof system
model, in which the verifier is quantum and the provers are permitted to
communicate classically (but not to share prior entanglement) was considered 
in \cite{Ben-OrHP08}.
The resulting complexity class was proved to contain $\NEXP$.

Theorem~\ref{thm:ruv} on simulating a quantum verifier by a classical one was
proved in~\cite{ReichardtUV12}.
(See also~\cite{ReichardtUV13} for a high-level exposition of the results.)
This result is based on the rigidity of sequential repetitions of the CHSH
game. Rigidity of a single repetition of the CHSH game is proved
in~\cite{McKagueYS12rigidity}.

Theorem~\ref{thm:nexp-mipstar} stating the inclusion $\NEXP\subseteq\QMIP^*$
was proved in~\cite{IV12}, where an analysis of Babai, Fortnow, and Lund's
multilinearity test with entangled provers was given, generalizing the
linearity test presented in Section~\ref{sec:linearity-test}. 
Letting $\QMA_{\textup{\scriptsize EXP}}$ denote the exponential-length proof
variant of $\QMA$ (defined in \cite{GottesmanI2009}), one has that 
$\QMA_{\textup{\scriptsize EXP}} \subseteq \QMIP^*_{a,b}(5,2)$ for some choice
of $a$ and $b$ such that $a-b>2^{-p(n)}$ for a polynomially bounded function
$p$ \cite{FitzsimonsV15}, suggesting that a stronger inclusion may be
achievable.
This result was improved to the inclusion 
$\QMA_{\textup{\scriptsize EXP}} \subseteq \MIP^*_{a,b}(4,1)$ 
by~\cite{Ji2015}, again for some choice of $a$ and $b$ separated by an inverse
exponential.
The exponential gap between completeness and soundness in both results is too
small to be amplified via standard techniques, and whether or not the inclusion
of $\QMA_{\textup{\scriptsize EXP}}$ in $\QMIP^*$ holds is an interesting open
problem.

The phenomenon of entanglement monogamy is pervasive in quantum information
theory. 
This terminology is generally attributed to Bennett; one of the first times it
appears in print is in~\cite{Terhal04}. 
The three-prover CHSH game described in Example~\ref{ex:3prover-chsh} was
introduced and analyzed in~\cite{toner2009monogamy}. 
A more general phenomenon is known to hold;
for any two-prover one-round game $G$ with classical messages of length at most
$t$, the associated $(1+2^t)$-prover game $G'$ (with the original game being
played with two randomly chosen provers) has value $\omega^*(G')=\omega(G)$. 
The $3$-prover clause-versus-variable verifier $T_\varphi$ appeared
in~\cite{KKMTV11}, and the two-prover variant is due to~\cite{ItoKM09}.

The procedure of parallel repetition for classical multi-prover interactive
proof systems was first suggested in~\cite{FortnowRS88}. 
We refer the interested reader to e.g.~\cite{DinurS14,BravermanG14} for recent
developments.
The fact that the no-signaling value of the Fortnow--Feige--Lov\'asz game is
2/3 was proved in \cite{Holenstein09}, and the game itself was considered by
Fortnow~\cite{fortnow1989complexity} and Feige and Lovasz~\cite{FeiLov92STOC},
who were the first to observe that it provides a counter-example to the perfect
exponentiation of the classical value of a multi-prover game under parallel
repetition.
Parallel repetition for two-prover one-round free multi-prover games with
classical messages and entangled provers was proved 
in~\cite{ChaillouxS13parallel,JainPY13parallel}. 
These results were extended to any number of provers and games with quantum
messages from the provers in~\cite{WuCY14parallel}. 
The case of projection games was considered in~\cite{DSV13full}. 
The transformation described at the end of Section~\ref{sec:qmip-par-rep} is
introduced in~\cite{bavarian2015anchoring}. 
The transformation is inspired by earlier work of Feige and 
Kilian~\cite{Feige2000} who considered a slightly more complicated
transformation in the setting of two-prover classical games. 
The transformation of Feige and Kilian, and its analysis, was extended to the
case of games with entangled provers in~\cite{KV11parallel}.

Brunner et al.~\cite{BrunnerPAGMS08} described a verifier $V$ such that
$\omega^{(2)}(V) < \omega^{(3)}(V)$.
Moreover there is a single one-round two-prover verifier $V$ that demonstrates
the same inequalities, but $V$ chooses its messages to the provers from a
continuous set of possibilities (the provers provide binary answers). 
Vertesi and Pal~\cite{VertesiP09} construct a two-message two-prover classical
verifier $V=V_d$ such that 
$\omega^{(d)}(V)>\omega^{(\lceil \log d\rceil -1)}(V)$.

Tsirelson's problem was shown to be
equivalent~\cite{JungeNPPSW11tsirelson,Fritz12tsirelson} to a range of
conjectures in the theory of $C^*$-algebras, including Connes' embedding 
conjecture~\cite{Connes76} and Kirchberg's QWEP 
conjecture~\cite{Kirchberg93qwep}.
(See the papers by Ozawa~\cite{Ozawa04qwep,Ozawa13connes} for surveys on the
equivalence between the latter two conjectures). 
Tsirelson's formulation of his problem is available as~\cite{TsirelsonProblem};
see also~\cite{ScholzW08tsirelson} for a discussion.

\begin{acknowledgements}
  \addcontentsline{toc}{chapter}{Acknowledgements}
  We thank Mark Wilde and an anonymous reviewer for helpful comments,
  corrections, and suggestions.
  Thomas Vidick acknowledges support from the IQIM, an NSF Physics Frontiers Center (NFS Grant PHY-1125565) with support of the Gordon and Betty Moore Foundation (GBMF-12500028).
  John Watrous acknowledges support from Canada's NSERC and the
  Canadian Institute for Advanced Research.
\end{acknowledgements}

\backmatter


\begin{thebibliography}{100}

\bibitem{AaronsonBDFS09}
S.~Aaronson, S.~Beigi, A.~Drucker, B.~Fefferman, and P.~Shor.
\newblock The power of unentanglement.
\newblock {\em Theory of Computing}, 5(1):1--42, 2009.

\bibitem{aaronson2014full}
S.~Aaronson and A.~Drucker.
\newblock A full characterization of quantum advice.
\newblock {\em SIAM Journal on Computing}, 43(3):1131--1183, 2014.

\bibitem{AaronsonK06qcma}
S.~Aaronson and G.~Kuperberg.
\newblock Quantum versus classical proofs and advice.
\newblock {\em Theory of Computing}, 3(7):129--157, 2007.

\bibitem{AdlemanDH97}
L.~Adleman, J.~DeMarrais, and M.~Huang.
\newblock Quantum computability.
\newblock {\em SIAM Journal on Computing}, 26(5):1524--1540, 1997.

\bibitem{aharonov2009detectability}
D.~Aharonov, I.~Arad, Z.~Landau, and U.~Vazirani.
\newblock The detectability lemma and quantum gap amplification.
\newblock In {\em Proceedings of the 41st Annual ACM Symposium on Theory of
  Computing}, pages 417--426, 2009.

\bibitem{aharonov2013guest}
D.~Aharonov, I.~Arad, and T.~Vidick.
\newblock Guest column: the quantum {PCP} conjecture.
\newblock {\em ACM SIGACT News}, 44(2):47--79, 2013.

\bibitem{aharonov2008pursuit}
D.~Aharonov, M.~Ben-Or, F.~Brand{\~a}o, and O.~Sattath.
\newblock The pursuit for uniqueness: extending {V}aliant--{V}azirani theorem
  to the probabilistic and quantum settings.
\newblock Available as arXiv e-Print 0810.4840, 2008.

\bibitem{AharonovGIK091d}
D.~Aharonov, D.~Gottesman, S.~Irani, and J.~Kempe.
\newblock The power of quantum systems on a line.
\newblock {\em Communications in Mathematical Physics}, 287(1):41--65, 2009.

\bibitem{AharonovKN98}
D.~Aharonov, A.~Kitaev, and N.~Nisan.
\newblock Quantum circuits with mixed states.
\newblock In {\em Proceedings of the 30th Annual ACM Symposium on Theory of
  Computing}, pages 20--30, 1998.

\bibitem{AharonovR05}
D.~Aharonov and O.~Regev.
\newblock Lattice problems in {NP} $\cap$ {coNP}.
\newblock {\em Journal of the ACM}, 52(5):749--765, 2005.

\bibitem{AielloH91}
W.~Aiello and J.~H{\aa}stad.
\newblock Statistical zero-knowledge languages can be recognized in two rounds.
\newblock {\em Journal of Computer and System Sciences}, 42(3):327--345, 1991.

\bibitem{ambainis2014physical}
A.~Ambainis.
\newblock On physical problems that are slightly more difficult than {QMA}.
\newblock In {\em Proceedings of the 29th Conference on Computational
  Complexity}, pages 32--43, 2014.

\bibitem{AmbainisRU14}
A.~Ambainis, A.~Rosmanis, and D.~Unruh.
\newblock Quantum attacks on classical proof systems: The hardness of quantum
  rewinding.
\newblock In {\em 55th Annual IEEE Symposium on Foundations of Computer
  Science}, pages 474--483, 2014.

\bibitem{Aravind02}
P.~Aravind.
\newblock A simple demonstration of {B}ell's theorem involving two observers
  and no probabilities or inequalities.
\newblock Available as arXiv.org e-Print quant-ph/0206070, 2002.

\bibitem{AroraB09}
S.~Arora and B.~Barak.
\newblock {\em Complexity Theory: A Modern Approach}.
\newblock Cambridge University Press, 2009.

\bibitem{AroraK07}
S.~Arora and S.~Kale.
\newblock A combinatorial, primal-dual approach to semidefinite programs.
\newblock In {\em Proceedings of the 39th Annual ACM Symposium on Theory of
  Computing}, pages 227--236, 2007.

\bibitem{AroraLMSS98}
S.~Arora, C.~Lund, R.~Motwani, M.~Sudan, and M.~Szegedy.
\newblock Proof verification and the hardness of approximation problems.
\newblock {\em Journal of the ACM}, 45(3):501--555, 1998.

\bibitem{AroraS98}
S.~Arora and S.~Safra.
\newblock Probabilistic checking of proofs: a new characterization of {NP}.
\newblock {\em Journal of the ACM}, 45(1):70--122, 1998.

\bibitem{Babai85}
L.~Babai.
\newblock Trading group theory for randomness.
\newblock In {\em Proceedings of the 17th Annual ACM Symposium on Theory of
  Computing}, pages 421--429, 1985.

\bibitem{Babai91}
L.~Babai.
\newblock Local expansion of vertex-transitive graphs and random generation in
  finite groups.
\newblock In {\em Proceedings of the 23rd Annual ACM Symposium on Theory of
  Computing}, pages 164--174, 1991.

\bibitem{BabaiFL91}
L.~Babai, L.~Fortnow, and C.~Lund.
\newblock Non-deterministic exponential time has two-prover interactive
  protocols.
\newblock {\em Computational Complexity}, 1(1):3--40, 1991.

\bibitem{BabaiM88}
L.~Babai and S.~Moran.
\newblock {A}rthur--{M}erlin games: a randomized proof system, and a hierarchy
  of complexity classes.
\newblock {\em Journal of Computer and System Sciences}, 36(2):254--276, 1988.

\bibitem{BabaiS84}
L.~Babai and E.~Szemer\'edi.
\newblock On the complexity of matrix group problems {I}.
\newblock In {\em Proceedings of the 25th Annual IEEE Symposium on Foundations
  of Computer Science}, pages 229--240, 1984.

\bibitem{bavarian2015anchoring}
M.~Bavarian, T.~Vidick, and H.~Yuen.
\newblock Anchoring games for parallel repetition.
\newblock Available as arXiv.org e-Print 1509.07466, 2015.

\bibitem{BeigiSW11}
S.~Beigi, P.~Shor, and J.~Watrous.
\newblock Quantum interactive proofs with short messages.
\newblock {\em Theory of Computing}, 7:101--117, 2011.

\bibitem{Bell64}
J.~Bell.
\newblock On the {E}instein--{P}odolsky--{R}osen paradox.
\newblock {\em Physics}, 1:195--200, 1964.

\bibitem{Ben-AroyaST10}
A.~Ben-Aroya, O.~Schwartz, and A.~Ta-Shma.
\newblock Quantum expanders: Motivation and construction.
\newblock {\em Theory of Computing}, 6(3):47--79, 2010.

\bibitem{BenOrFKT86}
M.~Ben-Or, E.~Feig, D.~Kozen, and P.~Tiwari.
\newblock A fast parallel algorithm for determining all roots of a polynomial
  with real roots.
\newblock In {\em Proceedings of the 18th Annual ACM Symposium on Theory of
  Computing}, pages 340--349, 1986.

\bibitem{BenOrGKW88}
M.~Ben-Or, S.~Goldwasser, J.~Kilian, and A.~Wigderson.
\newblock Multi-prover interactive proofs: how to remove intractability
  assumptions.
\newblock In {\em Proceedings of the 20th Annual ACM Symposium on Theory of
  Computing}, pages 113--131, 1988.

\bibitem{Ben-OrHP08}
M.~Ben-Or, A.~Hassidim, and H.~Pilpel.
\newblock Quantum multi prover interactive proofs with communicating provers.
\newblock In {\em Proceedings of the 49th Annual IEEE Symposium on Foundations
  of Computer Science}, pages 467--476, 2008.

\bibitem{BiniP98}
D.~Bini and V.~Pan.
\newblock Computing matrix eigenvalues and polynomial zeros where the output is
  real.
\newblock {\em SIAM Journal on Computing}, 27(4):1099--1115, 1998.

\bibitem{blier2009all}
H.~Blier and A.~Tapp.
\newblock All languages in {NP} have very short quantum proofs.
\newblock In {\em Proceedings of the 2009 3rd International Conference on
  Quantum, Nano and Micro Technologies}, pages 34--37, 2009.

\bibitem{BlumLR93}
M.~Blum, M.~Luby, and R.~Rubinfeld.
\newblock Self-testing/correcting with applications to numerical problems.
\newblock {\em Journal of Computer and System Sciences}, 47(3):549--595, 1993.

\bibitem{bookatz2014qma}
A.~Bookatz.
\newblock {QMA}-complete problems.
\newblock {\em Quantum Information \& Computation}, 14(5\&6):361--383, 2014.

\bibitem{Borodin77}
A.~Borodin.
\newblock On relating time and space to size and depth.
\newblock {\em SIAM Journal on Computing}, 6:733--744, 1977.

\bibitem{BorodinCP83}
A.~Borodin, S.~Cook, and N.~Pippenger.
\newblock Parallel computation for well-endowed rings and space-bounded
  probabilistic machines.
\newblock {\em Information and Control}, 58:113--136, 1983.

\bibitem{brandao2011quasipolynomial}
F.~Brand{\~a}o, M.~Christandl, and J.~Yard.
\newblock A quasipolynomial-time algorithm for the quantum separability
  problem.
\newblock In {\em Proceedings of the 43rd Annual ACM Symposium on Theory of
  Computing}, pages 343--352, 2011.

\bibitem{BravermanG14}
Mark Braverman and Ankit Garg.
\newblock Small value parallel repetition for general games.
\newblock In {\em Proceedings of the 47th Annual ACM on Symposium on Theory of
  Computing}, pages 335--340. ACM, 2015.

\bibitem{Bravyi11qsat}
S.~Bravyi.
\newblock Efficient algorithm for a quantum analogue of 2-{SAT}.
\newblock {\em Contemporary Mathematics}, 536:33--48, 2011.

\bibitem{BrunnerPAGMS08}
N.~Brunner, S.~Pironio, A.~Acin, N.~Gisin, A.~M\'ethot, and V.~Scarani.
\newblock Testing the dimension of {Hilbert} spaces.
\newblock {\em Physical Review Letters}, 100:210503, 2008.

\bibitem{Cabello01}
A.~Cabello.
\newblock Bell's theorem without inequalities and without probabilities for two
  observers.
\newblock {\em Physical Review Letters}, 86:1911--1914, 2001.

\bibitem{ChaillouxS13parallel}
A.~Chailloux and G.~Scarpa.
\newblock Parallel repetition of entangled games with exponential decay via the
  superposed information cost.
\newblock In {\em Automata, Languages, and Programming}, volume 2014 of {\em
  Lecture Notes in Computer Science}, pages 296--307. Springer, 2014.

\bibitem{chiesa2013improved}
A.~Chiesa and M.~Forbes.
\newblock Improved soundness for {QMA} with multiple provers.
\newblock {\em Chicago Journal of Theoretical Computer Science}, 2013:1--23,
  2013.

\bibitem{Clauser:69a}
J.~Clauser, M.~Horne, A.~Shimony, and R.~Holt.
\newblock Proposed experiment to test local hidden-variable theories.
\newblock {\em Physical Review Letters}, 23:880--884, 1969.

\bibitem{CleveEMM98}
R.~Cleve, A.~Ekert, C.~Macchiavello, and M.~Mosca.
\newblock Quantum algorithms revisited.
\newblock {\em Proceedings of the Royal Society}, A454:339--354, 1998.

\bibitem{CleveHTW04}
R.~Cleve, P.~H{\o}yer, B.~Toner, and J.~Watrous.
\newblock Consequences and limits of nonlocal strategies.
\newblock In {\em Proceedings of the 19th Conference on Computational
  Complexity}, pages 236--249, 2004.

\bibitem{CleveSUU08xor}
R.~Cleve, W.~Slofstra, F.~Unger, and S.~Upadhyay.
\newblock Perfect parallel repetition theorem for quantum {XOR} proof systems.
\newblock {\em Computational Complexity}, 17(2):282--299, 2008.

\bibitem{Connes76}
A.~Connes.
\newblock Classification of injective factors cases ${II}_1$, ${II}_\infty$,
  ${III}_\lambda$, $\lambda\neq 1$.
\newblock {\em Annals of Mathematics}, 104(1):73--115, 1976.

\bibitem{Cook71}
S.~Cook.
\newblock The complexity of theorem proving procedures.
\newblock In {\em Proceedings of the 3rd Annual ACM Symposium on Theory of
  Computing}, pages 151--158, 1971.

\bibitem{cooney2011rank}
T.~Cooney, M.~Junge, C.~Palazuelos, and D.~P{\'e}rez-Garc{\'\i}a.
\newblock Rank-one quantum games.
\newblock {\em Computational Complexity}, 24(1):133--196, 2011.

\bibitem{cubitt2014complexity}
T.~Cubitt and A.~Montanaro.
\newblock Complexity classification of local {H}amiltonian problems.
\newblock In {\em Proceedings of the 55th Annual IEEE Symposium on Foundations
  of Computer Science}, pages 120--129, 2014.

\bibitem{DamgaardL09}
I.~Damg{\aa}rd and C.~Lunemann.
\newblock Quantum-secure coin-flipping and applications.
\newblock In {\em Advances in Cryptology -- ASIACRYPT 2009}, volume 5912 of
  {\em Lecture Notes in Computer Science}, pages 52--69. Springer, 2009.

\bibitem{DinurS14}
I.~Dinur and D.~Steurer.
\newblock Analytical approach to parallel repetition.
\newblock In {\em Proceedings of the 46th Annual ACM Symposium on Theory of
  Computing}, pages 624--633, 2014.

\bibitem{DSV13full}
I.~Dinur, D.~Steurer, and T.~Vidick.
\newblock A parallel repetition theorem for entangled projection games.
\newblock In {\em Proceedings of the 29th Conference on Computational
  Complexity}, pages 197--208, 2014.

\bibitem{Edmonds65a}
J.~Edmonds.
\newblock Minimum partition of a matroid into independent subsets.
\newblock {\em Journal of Research of the National Bureau of Standards Section
  B: Mathematics and Mathematical Physics}, 69B(1--2):67--72, 1965.

\bibitem{Edmonds65}
J.~Edmonds.
\newblock Paths, trees, and flowers.
\newblock {\em Canadian Journal of Mathematics}, 17(3):449--467, 1965.

\bibitem{FeiGolLovSafSze96JACM}
U.~Feige, S.~Goldwasser, L.~Lov\'asz, S.~Safra, and M.~Szegedy.
\newblock Interactive proofs and the hardness of approximating cliques.
\newblock {\em Journal of the ACM}, 43(2):268--292, 1996.

\bibitem{Feige2000}
U.~Feige and J.~Kilian.
\newblock Two-prover protocols---low error at affordable rates.
\newblock {\em SIAM Journal on Computing}, 30(1):324--346, 2000.

\bibitem{FeiLov92STOC}
U.~Feige and L.~Lov\'{a}sz.
\newblock Two-prover one-round proof systems: Their power and their problems.
\newblock In {\em Proceedings of the 24th Annual ACM Symposium on Theory of
  Computing}, pages 733--744, 1992.

\bibitem{FitzsimonsV15}
J.~Fitzsimons and T.~Vidick.
\newblock A multiprover interactive proof system for the local {Hamiltonian}
  problem.
\newblock In {\em Proceedings of the 6th Conference on Innovations in
  Theoretical Computer Science}, pages 103--112, 2015.

\bibitem{Fortnow89}
L.~Fortnow.
\newblock The complexity of perfect zero-knowledge.
\newblock In S.~Micali, editor, {\em Randomness and Computation}, volume~5 of
  {\em Advances in Computing Research}, pages 327--343. Greenwich: JAI Press,
  1989.

\bibitem{fortnow1989complexity}
L.~Fortnow.
\newblock {\em Complexity-Theoretic Aspects of Interactive Proof Systems}.
\newblock PhD thesis, Massachusetts Institute of Technology, 1989.

\bibitem{FortnowR99}
L.~Fortnow and J.~Rogers.
\newblock Complexity limitations on quantum computation.
\newblock {\em Journal of Computer and System Sciences}, 59(2):240--252, 1999.

\bibitem{FortnowRS88}
L.~Fortnow, J.~Rompel, and M.~Sipser.
\newblock On the power of multi-prover interactive protocols.
\newblock In {\em Proceedings of the 3rd Annual Structure in Complexity Theory
  Conference}, pages 156--161, 1988.

\bibitem{Fritz12tsirelson}
T.~Fritz.
\newblock Tsirelson's problem and {K}irchberg's conjecture.
\newblock {\em Reviews in Mathematical Physics}, 24(05):1250012, 2012.

\bibitem{Gharibian10}
S.~Gharibian.
\newblock Strong {NP}-hardness of the quantum separability problem.
\newblock {\em Quantum Information \& Computation}, 10(3):343--360, 2010.

\bibitem{GharibianK12hierarchy}
S.~Gharibian and J.~Kempe.
\newblock Hardness of approximation for quantum problems.
\newblock In {\em Automata, Languages, and Programming}, volume 7391 of {\em
  Lecture Notes in Computer Science}, pages 387--398. Springer, 2012.

\bibitem{gharibian2014ground}
S.~Gharibian and J.~Sikora.
\newblock Ground state connectivity of local {Hamiltonians}.
\newblock Available as arXiv.org e-Print 1409.3182, 2014.

\bibitem{Goldreich02}
O.~Goldreich.
\newblock Zero-knowledge twenty years after its invention.
\newblock Electronic Colloquium on Computational Complexity Report 2002/186,
  2002.

\bibitem{GoldreichSV98}
O.~Goldreich, A.~Sahai, and S.~Vadhan.
\newblock Honest verifier statistical zero knowledge equals general statistical
  zero knowledge.
\newblock In {\em Proceedings of the 30th Annual ACM Symposium on Theory of
  Computing}, pages 23--26, 1998.

\bibitem{GoldwasserMR85}
S.~Goldwasser, S.~Micali, and C.~Rackoff.
\newblock The knowledge complexity of interactive proof systems.
\newblock In {\em Proceedings of the 17th Annual ACM Symposium on Theory of
  Computing}, pages 291--304, 1985.

\bibitem{GoldwasserMR89}
S.~Goldwasser, S.~Micali, and C.~Rackoff.
\newblock The knowledge complexity of interactive proof systems.
\newblock {\em SIAM Journal on Computing}, 18(1):186--208, 1989.

\bibitem{GoldwasserS89}
S.~Goldwasser and M.~Sipser.
\newblock Private coins versus public coins in interactive proof systems.
\newblock In S.~Micali, editor, {\em Randomness and Computation}, volume~5 of
  {\em Advances in Computing Research}, pages 73--90. Greenwich: JAI Press,
  1989.

\bibitem{gosset2013quantum}
D.~Gosset and D.~Nagaj.
\newblock Quantum {3-SAT} is {QMA1}-complete.
\newblock In {\em Proceedings of the 54th Annual IEEE Symposium on Foundations
  of Computer Science}, pages 756--765, 2013.

\bibitem{GottesmanI2009}
D.~Gottesman and S.~Irani.
\newblock The quantum and classical complexity of translationally invariant
  tiling and {Hamiltonian} problems.
\newblock In {\em Proceedings of the 50th Annual IEEE Symposium on Foundations
  of Computer Science}, pages 95--104, 2009.

\bibitem{aharonov2002quantum}
A.~Grilo, I.~Kerenidis, and J.~Sikora.
\newblock Quantum {NP} - a survey.
\newblock Available as arXiv.org e-Print quant-ph/0210077, 2002.

\bibitem{grilo2014qma}
A.~Grilo, I.~Kerenidis, and J.~Sikora.
\newblock {QMA} with subset state witnesses.
\newblock Available as arXiv.org e-Print 1410.2882, 2014.

\bibitem{Grover96}
L.~Grover.
\newblock A fast quantum mechanical algorithm for database search.
\newblock In {\em Proceedings of the 28th Annual ACM Symposium on Theory of
  Computing}, pages 212--219, 1996.

\bibitem{Gurvits03}
L.~Gurvits.
\newblock Classical deterministic complexity of {Edmonds'} problem and quantum
  entanglement.
\newblock In {\em Proceedings of the 35th Annual ACM Symposium on Theory of
  Computing}, pages 1--19, 2003.

\bibitem{Gutoski05}
G.~Gutoski.
\newblock Upper bounds for quantum interactive proofs with competing provers.
\newblock In {\em Proceedings of the 20th Conference on Computational
  Complexity}, pages 334--343, 2005.

\bibitem{Gutoski09}
G.~Gutoski.
\newblock {\em Quantum Strategies and Local Operations}.
\newblock PhD thesis, University of Waterloo, 2009.

\bibitem{GutoskiHMW15}
G.~Gutoski, P.~Hayden, K.~Milner, and M.~Wilde.
\newblock Quantum interactive proofs and the complexity of separability
  testing.
\newblock {\em Theory of Computing}, 11(3):59--103, 2015.

\bibitem{GutoskiW05}
G.~Gutoski and J.~Watrous.
\newblock Quantum interactive proofs with competing provers.
\newblock In {\em Proceedings of the 22nd Symposium on Theoretical Aspects of
  Computer Science}, volume 3404 of {\em Lecture Notes in Computer Science},
  pages 605--616. Springer, 2005.

\bibitem{GutoskiW07}
G.~Gutoski and J.~Watrous.
\newblock Toward a general theory of quantum games.
\newblock In {\em Proceedings of the 39th Annual ACM Symposium on Theory of
  Computing}, pages 565--574, 2007.

\bibitem{GutoskiW13}
G.~Gutoski and X.~Wu.
\newblock Parallel approximation of min-max problems.
\newblock {\em Computational Complexity}, 22(2):385--428, 2013.

\bibitem{HallgrenKSZ08}
S.~Hallgren, A.~Kolla, P.~Sen, and S.~Zhang.
\newblock Making classical honest verifier zero knowledge protocols secure
  against quantum attacks.
\newblock In {\em Proceedings of the 35th International Colloquium on Automata,
  Languages and Programming}, volume 5126 of {\em Lecture Notes in Computer
  Science}, pages 592--603. Springer, 2008.

\bibitem{HallgrenSS11}
S.~Hallgren, A.~Smith, and F.~Song.
\newblock Classical cryptographic protocols in a quantum world.
\newblock In {\em Advances in Cryptology -- CRYPTO 2011}, volume 6841 of {\em
  Lecture Notes in Computer Science}, pages 411--428. Springer, 2011.

\bibitem{harrow2013testing}
A.~Harrow and A.~Montanaro.
\newblock Testing product states, quantum {M}erlin--{A}rthur games and tensor
  optimization.
\newblock {\em Journal of the ACM}, 60(1):3, 2013.

\bibitem{Hastad01}
J.~H{\aa}stad.
\newblock Some optimal inapproximability results.
\newblock {\em Journal of the ACM}, 48:798--859, 2001.

\bibitem{HaydenMW14}
P.~Hayden, K.~Milner, and M.~Wilde.
\newblock Two-message quantum interactive proofs and the quantum separability
  problem.
\newblock {\em Quantum Information \& Computation}, 14(5\&6):384--416, 2014.

\bibitem{Holenstein09}
T.~Holenstein.
\newblock Parallel repetition: Simplifications and the no-signaling case.
\newblock {\em Theory of Computing}, 5:141--172, 2009.

\bibitem{Ito14parallelization}
T.~Ito.
\newblock Parallelization of entanglement-resistant multi-prover interactive
  proofs.
\newblock {\em Information Processing Letters}, 114(10):579--583, 2014.

\bibitem{ItoKM09}
T.~Ito, H.~Kobayashi, and K.~Matsumoto.
\newblock Oracularization and two-prover one-round interactive proofs against
  nonlocal strategies.
\newblock In {\em Proceedings of the 24th Conference on Computational
  Complexity}, pages 217--228, 2009.

\bibitem{ItoKW12}
T.~Ito, H.~Kobayashi, and J.~Watrous.
\newblock Quantum interactive proofs with weak error bounds.
\newblock In {\em Proceedings of the 3rd Conference on Innovations in
  Theoretical Computer Science}, pages 266--275, 2012.

\bibitem{IV12}
T.~Ito and T.~Vidick.
\newblock A multi-prover interactive proof for {NEXP} sound against entangled
  provers.
\newblock {\em Proceedings of the 53rd IEEE Symposium on Foundations of
  Computer Science}, pages 243--252, 2012.

\bibitem{JainJUW11}
R.~Jain, Z.~Ji, S.~Upadhyay, and J.~Watrous.
\newblock {QIP = PSPACE}.
\newblock {\em Journal of the ACM}, 58(6):30, 2011.

\bibitem{Jain12unique}
R.~Jain, I.~Kerenidis, G.~Kuperberg, M.~Santha, O.~Sattath, and S.~Zhang.
\newblock On the power of a unique quantum witness.
\newblock {\em Theory of Computing}, 8(17):375--400, 2012.

\bibitem{JainPY13parallel}
R.~Jain, A.~Pereszl{\'e}nyi, and P.~Yao.
\newblock A parallel repetition theorem for entangled two-player one-round
  games under product distributions.
\newblock In {\em Proceedings of the 29th Conference on Computational
  Complexity}, pages 209--216, 2014.

\bibitem{JainUW09}
R.~Jain, S.~Upadhyay, and J.~Watrous.
\newblock Two-message quantum interactive proofs are in {PSPACE}.
\newblock In {\em Proceedings of the 50th Annual IEEE Symposium on Foundations
  of Computer Science}, pages 534--543, 2009.

\bibitem{JainW09}
R.~Jain and J.~Watrous.
\newblock Parallel approximation of non-interactive zero-sum quantum games.
\newblock In {\em Proceedings of the 24th Conference on Computational
  Complexity}, pages 243--253, 2009.

\bibitem{janzing2005non}
D.~Janzing, P.~Wocjan, and T.~Beth.
\newblock ``{Non-identity-check}'' is {QMA}-complete.
\newblock {\em International Journal of Quantum Information}, 3(3):463--473,
  2005.

\bibitem{Ji2015}
Z.~Ji.
\newblock Classical verification of quantum proofs.
\newblock Available as arXiv.org e-Print 1505.07432, 2015.

\bibitem{JungeNPPSW11tsirelson}
M.~Junge, M.~Navascues, C.~Palazuelos, D.~Perez-Garcia, V.~Scholz, and
  R.~Werner.
\newblock Connes' embedding problem and {Tsirelson's} problem.
\newblock {\em Journal of Mathematical Physics}, 52(1):012102, 2011.

\bibitem{Kale07}
S.~Kale.
\newblock {\em Efficient Algorithms Using the Multiplicative Weights Update
  Method}.
\newblock PhD thesis, Princeton University, 2007.

\bibitem{Karp72}
R.~Karp.
\newblock Reducibility among combinatorial problems.
\newblock In R.~Miller and J.~Thatcher, editors, {\em Complexity of Computer
  Computations}, pages 85--103. Plenum Press, New York, 1972.

\bibitem{KayeLM07}
P.~Kaye, R.~Laflamme, and M.~Mosca.
\newblock {\em An Introduction to Quantum Computing}.
\newblock Oxford University Press, 2007.

\bibitem{KempeKR06lh}
J.~Kempe, A.~Kitaev, and O.~Regev.
\newblock The complexity of the local {H}amiltonian problem.
\newblock {\em SIAM Journal on Computing}, 35(5):1070--1097, 2006.

\bibitem{KKMTV11}
J.~Kempe, H.~Kobayashi, K.~Matsumoto, B.~Toner, and T.~Vidick.
\newblock Entangled games are hard to approximate.
\newblock {\em SIAM Journal on Computing}, 40(3):848--877, 2011.

\bibitem{KempeKMV09}
J.~Kempe, H.~Kobayashi, K.~Matsumoto, and T.~Vidick.
\newblock Using entanglement in quantum multi-prover interactive proofs.
\newblock {\em Computational Complexity}, 18:273--307, 2009.

\bibitem{kempe20033}
J.~Kempe and O.~Regev.
\newblock 3-local {H}amitonian is {QMA}-complete.
\newblock {\em Quantum Information \& Computation}, 3(3):258--264, 2003.

\bibitem{KV11parallel}
J.~Kempe and T.~Vidick.
\newblock Parallel repetition of entangled games.
\newblock In {\em Proceedings of the 43rd Annual ACM Symposium on Theory of
  Computing}, pages 353--362, 2011.

\bibitem{Kirchberg93qwep}
E.~Kirchberg.
\newblock On non-semisplit extensions, tensor products and exactness of group
  {$C^*$}-algebras.
\newblock {\em Inventiones Mathematicae}, 112(1):449--489, 1993.

\bibitem{KitaevSV02}
A.~Kitaev, A.~Shen, and M.~Vyalyi.
\newblock {\em Classical and Quantum Computation}, volume~47 of {\em Graduate
  Studies in Mathematics}.
\newblock American Mathematical Society, 2002.

\bibitem{KitaevW00}
A.~Kitaev and J.~Watrous.
\newblock Parallelization, amplification, and exponential time simulation of
  quantum interactive proof systems.
\newblock In {\em Proceedings of the 32nd Annual ACM Symposium on Theory of
  Computing}, pages 608--617, 2000.

\bibitem{Knill96}
E.~Knill.
\newblock Quantum randomness and nondeterminism.
\newblock Technical Report LAUR-96-2186, Los Alamos National Laboratory, 1996.
\newblock Available as arXiv.org e-Print quant-ph/9610012.

\bibitem{Kobayashi03}
H.~Kobayashi.
\newblock Non-interactive quantum perfect and statistical zero-knowledge.
\newblock In {\em Proceedings of the 14th International Symposium on Algorithms
  and Computation}, volume 2906 of {\em Lecture Notes in Computer Science},
  pages 178--188. Springer, 2003.

\bibitem{Kobayashi08}
H.~Kobayashi.
\newblock General properties of quantum zero-knowledge proofs.
\newblock In {\em Proceedings of the 5th IACR Theory of Cryptography
  Conference}, volume 4948 of {\em Lecture Notes in Computer Science}, pages
  107--124. Springer, 2008.

\bibitem{KobayashiLN13}
H.~Kobayashi, F.~Le Gall, and H.~Nishimura.
\newblock Stronger methods of making quantum interactive proofs perfectly
  complete.
\newblock In {\em Proceedings of the 4th Conference on Innovations in
  Theoretical Computer Science}, pages 329--352, 2013.

\bibitem{KobMat03JCSS}
H.~Kobayashi and K.~Matsumoto.
\newblock Quantum multi-prover interactive proof systems with limited prior
  entanglement.
\newblock {\em Journal of Computer and System Sciences}, 66(3):429--450, 2003.

\bibitem{KobayashiMY03multiplemerlins}
H.~Kobayashi, K.~Matsumoto, and T.~Yamakami.
\newblock Quantum {M}erlin--{A}rthur proof systems: Are multiple {M}erlins more
  helpful to {A}rthur?
\newblock In {\em Algorithms and Computation}, volume 2906 of {\em Lecture
  Notes in Computer Science}, pages 189--198. Springer, 2003.

\bibitem{LeungTW13}
D.~Leung, B.~Toner, and J.~Watrous.
\newblock Coherent state exchange in multi-prover quantum interactive proof
  systems.
\newblock {\em Chicago Journal of Theoretical Computer Science}, 2013:11, 2013.

\bibitem{Levin73}
L.~Levin.
\newblock Universal sequential search problems ({English} translation).
\newblock {\em Problems of Information Transmission}, 9(3):265--266, 1973.

\bibitem{Liu06consistency}
Y.-K. Liu.
\newblock Consistency of local density matrices is {QMA}-complete.
\newblock In {\em Approximation, Randomization, and Combinatorial Optimization.
  Algorithms and Techniques}, volume 4110 of {\em Lecture Notes in Computer
  Science}, pages 438--449. Springer, 2006.

\bibitem{LundFKN92}
C.~Lund, L.~Fortnow, H.~Karloff, and N.~Nisan.
\newblock Algebraic methods for interactive proof systems.
\newblock {\em Journal of the ACM}, 39(4):859--868, 1992.

\bibitem{MarriottW05}
C.~Marriott and J.~Watrous.
\newblock Quantum {Arthur--Merlin} games.
\newblock {\em Computational Complexity}, 14(2):122--152, 2005.

\bibitem{McKagueYS12rigidity}
M.~McKague, T.~Yang, and V.~Scarani.
\newblock Robust self-testing of the singlet.
\newblock {\em Journal of Physics A: Mathematical and Theoretical},
  45(45):455304, 2012.

\bibitem{Mermin90}
D.~Mermin.
\newblock Simple unified form for the major no-hidden-variables theorems.
\newblock {\em Physical Review Letters}, 65(27):3373--3376, 1990.

\bibitem{MittalS07}
R.~Mittal and M.~Szegedy.
\newblock Product rules in semidefinite programming.
\newblock In {\em Fundamentals in Computation Theory}, volume 4639 of {\em
  Lecture Notes in Computer Science}, pages 435--445. Springer, 2007.

\bibitem{NPA08NJP}
M.~Navascu{\'{e}}s, S.~Pironio, and A.~Ac\'in.
\newblock A convergent hierarchy of semidefinite programs characterizing the
  set of quantum correlations.
\newblock {\em New Journal of Physics}, 10(7):073013, 2008.

\bibitem{Neff94}
C.~Neff.
\newblock Specified precision polynomial root isolation is in {NC}.
\newblock {\em Journal of Computer and System Sciences}, 48(3):429--463, 1994.

\bibitem{NielsenC00}
M.~Nielsen and I.~Chuang.
\newblock {\em Quantum Computation and Quantum Information}.
\newblock Cambridge University Press, 2000.

\bibitem{nishimura2004polynomial}
H.~Nishimura and T.~Yamakami.
\newblock Polynomial time quantum computation with advice.
\newblock {\em Information Processing Letters}, 90(4):195--204, 2004.

\bibitem{donnellPCP}
R.~O'Donnell.
\newblock A history of the {PCP} theorem, 2005.
\newblock Available at
  \url{http://courses.cs.washington.edu/courses/cse533/05au/pcp-history.pdf}.

\bibitem{Okamoto00}
T.~Okamoto.
\newblock On relationships between statistical zero-knowledge proofs.
\newblock {\em Journal of Computer and System Sciences}, 60(1):47--108, 2000.

\bibitem{TerhalO082d}
R.~Oliveira and B.~Terhal.
\newblock The complexity of quantum spin systems on a two-dimensional square
  lattice.
\newblock {\em Quantum Information \& Computation}, 8(10):900--924, 2008.

\bibitem{Ozawa04qwep}
N.~Ozawa.
\newblock About the {QWEP} conjecture.
\newblock {\em International Journal of Mathematics}, 15(05):501--530, 2004.

\bibitem{Ozawa13connes}
N.~Ozawa.
\newblock About the {C}onnes embedding conjecture.
\newblock {\em Japanese Journal of Mathematics}, 8(1):147--183, 2013.

\bibitem{PapadimitriouY84facets}
C.~Papadimitriou and M.~Yannakakis.
\newblock The complexity of facets (and some facets of complexity).
\newblock {\em Journal of Computer and System Sciences}, 28(2):244--259, 1984.

\bibitem{Peres90}
A.~Peres.
\newblock Incompatible results of quantum measurements.
\newblock {\em Physics Letters A}, 151(3--4):107--108, 1990.

\bibitem{Raz98}
R.~Raz.
\newblock A parallel repetition theorem.
\newblock {\em SIAM Journal on Computing}, 27(3):763--803, 1998.

\bibitem{RegevV13}
O.~Regev and T.~Vidick.
\newblock Quantum {XOR} games.
\newblock In {\em Proceedings of the 28th Conference on Computational
  Complexity}, pages 144--155, 2013.

\bibitem{ReichardtUV12}
B.~Reichardt, F.~Unger, and U.~Vazirani.
\newblock A classical leash for a quantum system: Command of quantum systems
  via rigidity of {CHSH} games.
\newblock Available as arXiv.org e-Print 1209.0448, 2012.

\bibitem{ReichardtUV13}
B.~Reichardt, F.~Unger, and U.~Vazirani.
\newblock Classical command of quantum systemes.
\newblock {\em Nature}, 496(7446):456--460, 2013.

\bibitem{Rosgen08b}
B.~Rosgen.
\newblock Additivity and distinguishability of random unitary channels.
\newblock {\em Journal of Mathematical Physics}, 49(10):102107, 2008.

\bibitem{Rosgen08a}
B.~Rosgen.
\newblock Distinguishing short quantum computations.
\newblock In {\em Proceedings of the 25th International Symposium on
  Theoretical Aspects of Computer Science}, pages 597--608, 2008.

\bibitem{Rosgen09}
B.~Rosgen.
\newblock {\em Computational Distinguishability of Quantum Channels}.
\newblock PhD thesis, University of Waterloo, 2009.

\bibitem{RosgenW05}
B.~Rosgen and J.~Watrous.
\newblock On the hardness of distinguishing mixed-state quantum computations.
\newblock In {\em Proceedings of the 20th Conference on Computational
  Complexity}, pages 344--354, 2005.

\bibitem{SahaiV03}
A.~Sahai and S.~Vadhan.
\newblock A complete promise problem for statistical zero-knowledge.
\newblock {\em Journal of the ACM}, 50(2):196--249, 2003.

\bibitem{ScholzW08tsirelson}
V.~Scholz and R.~Werner.
\newblock Tsirelson's problem.
\newblock Available as arXiv.org e-Print 0812.4305.

\bibitem{Shamir92}
A.~Shamir.
\newblock {IP} $=$ {PSPACE}.
\newblock {\em Journal of the ACM}, 39(4):869--877, 1992.

\bibitem{Shor94}
P.~Shor.
\newblock Algorithms for quantum computation: discrete logarithms and
  factoring.
\newblock In {\em Proceedings of the 35th Annual IEEE Symposium on Foundations
  of Computer Science}, pages 124--134, 1994.

\bibitem{Shor97}
P.~Shor.
\newblock Polynomial-time algorithms for prime factorization and discrete
  logarithms on a quantum computer.
\newblock {\em SIAM Journal on Computing}, 26(5):1484--1509, 1997.

\bibitem{Slofstra11xor}
W.~Slofstra.
\newblock Lower bounds on the entanglement needed to play {XOR} non-local
  games.
\newblock {\em Journal of Mathematical Physics}, 52(10):102202, 2011.

\bibitem{SpekkensR01}
R.~Spekkens and T.~Rudolph.
\newblock Degrees of concealment and bindingness in quantum bit-commitment
  protocols.
\newblock {\em Physical Review A}, 65(1):123410, 2001.

\bibitem{Terhal04}
B.~Terhal.
\newblock Is entanglement monogamous?
\newblock {\em IBM Journal of Research and Development}, 48:71--78, 2004.

\bibitem{toner2009monogamy}
B.~Toner.
\newblock Monogamy of non-local quantum correlations.
\newblock In {\em Proceedings of the Royal Society of London A: Mathematical,
  Physical and Engineering Sciences}, volume 465, pages 59--69, 2009.

\bibitem{TsirelsonProblem}
B.~Tsirelson.
\newblock Bell inequalities and operator algebras.
\newblock Available at
  \url{http://www.tau.ac.il/~tsirel/download/bellopalg.pdf}.

\bibitem{Tsirelson87}
B.~{Tsirel'son}.
\newblock Quantum analogues of the \uppercase{B}ell inequalities:
  \uppercase{T}he case of two spatially separated domains.
\newblock {\em Journal of Soviet Mathematics}, 36:557--570, 1987.

\bibitem{Unruh12}
D.~Unruh.
\newblock Quantum proofs of knowledge.
\newblock In {\em Advances in Cryptology -- Eurocrypt 2012}, volume 7237 of
  {\em Lecture Notes in Computer Science}, pages 135--152. Springer, 2012.

\bibitem{ValiantV86unique}
L.~Valiant and V.~Vazirani.
\newblock {NP} is as easy as detecting unique solutions.
\newblock {\em Theoretical Computer Science}, 47:85--93, 1986.

\bibitem{vanDamH03}
W.~van Dam and P.~Hayden.
\newblock Universal entanglement transformations without communication.
\newblock {\em Physical Review A}, 67(6):060302, 2003.

\bibitem{Graaf97}
J.~van~de Graaf.
\newblock {\em Towards a Formal Definition of Security for Quantum Protocols}.
\newblock PhD thesis, Universit\'e de Montr\'eal, 1997.

\bibitem{VertesiP09}
T.~V\'ertesi and K.~P\'al.
\newblock Bounding the dimension of bipartite quantum systems.
\newblock {\em Physical Review A}, 79:042106, 2009.

\bibitem{WarmuthK06}
M.~Warmuth and D.~Kuzmin.
\newblock Online variance minimization.
\newblock In {\em Proceedings of the 19th Annual Conference on Learning
  Theory}, volume 4005 of {\em Lecture Notes in Computer Science}, pages
  514--528. Springer, 2006.

\bibitem{Watrous99}
J.~Watrous.
\newblock {PSPACE} has constant-round quantum interactive proof systems.
\newblock In {\em Proceedings of the 40th Annual IEEE Symposium on Foundations
  of Computer Science}, pages 112--119, 1999.

\bibitem{Watrous00group}
J.~Watrous.
\newblock Succinct quantum proofs for properties of finite groups.
\newblock In {\em Proceedings of the 41st Annual IEEE Symposium on Foundations
  of Computer Science}, pages 537--546, 2000.

\bibitem{Watrous02}
J.~Watrous.
\newblock Limits on the power of quantum statistical zero-knowledge.
\newblock In {\em Proceedings of the 43rd Annual IEEE Symposium on Foundations
  of Computer Science}, pages 459--468, 2002.

\bibitem{Watrous03}
J.~Watrous.
\newblock {PSPACE} has constant-round quantum interactive proof systems.
\newblock {\em Theoretical Computer Science}, 292(3):575--588, 2003.

\bibitem{Watrous09}
J.~Watrous.
\newblock Zero-knowledge against quantum attacks.
\newblock {\em SIAM Journal on Computing}, 39(1):25--58, 2009.

\bibitem{Wehner06}
S.~Wehner.
\newblock Entanglement in interactive proof systems with binary answers.
\newblock In {\em Proceedings of the 23rd Annual Symposium on Theoretical
  Aspects of Computer Science}, volume 3884 of {\em Lecture Notes in Computer
  Science}, pages 162--171, 2006.

\bibitem{Wu10}
X.~Wu.
\newblock Equilibrium value method for the proof of {QIP=PSPACE}.
\newblock Available as arXiv.org e-Print 1004.0264, 2010.

\bibitem{WuCY14parallel}
X.~Wu, K.-M. Chung, and H.~Yuen.
\newblock Parallel repetition for entangled $k$-player games via fast quantum
  search.
\newblock In {\em Proceedings of the 30th Conference on Computational
  Complexity}, volume~33 of {\em Leibniz International Proceedings in
  Informatics (LIPIcs)}, pages 512--536. Schloss Dagstuhl--Leibniz-Zentrum fuer
  Informatik, 2015.

\end{thebibliography}
\end{document}